\documentclass[12pt, notitlepage, letterpaper]{report}
\pdfoutput=1

\usepackage[utf8]{inputenc}
\usepackage{setspace}
\usepackage{pdfsync}

\def\showauthornotes{0}
\def\showkeys{0}
\def\showdraftbox{0}
\def\confversion{0}
\def\isdraft{0}

\usepackage{xspace,enumerate}
\usepackage{amsmath,amssymb}
\usepackage{amsthm}
\ifnum\confversion=0
	\usepackage[toc,page]{appendix}
\fi
\usepackage{thmtools}
\usepackage{thm-restate}
\usepackage{color,graphicx}
\usepackage{boxedminipage}
\usepackage{makecell}
\usepackage{tabularx}
\usepackage{relsize}

\ifnum\showkeys=1
\usepackage[color]{showkeys}
\fi

\definecolor{darkred}{rgb}{0.5,0,0}
\definecolor{darkgreen}{rgb}{0,0.35,0}
\definecolor{darkblue}{rgb}{0,0,0.55}

\usepackage[dvipsnames]{xcolor}

\ifnum\confversion=1
\usepackage[pdfstartview=FitH,pdfpagemode=UseNone,colorlinks=false,linkcolor=darkblue,filecolor=darkred,citecolor=darkgreen,urlcolor=darkred,pagebackref,bookmarks=false]{hyperref}
\else
\usepackage[pdfstartview=FitH,pdfpagemode=UseNone,colorlinks,linkcolor=NavyBlue,filecolor=blue,citecolor=OliveGreen,urlcolor=NavyBlue,pagebackref]{hyperref}
\fi

\usepackage[capitalise,nameinlink]{cleveref}
\usepackage[T1]{fontenc}

\usepackage{mathtools,dsfont}

\usepackage{mathpazo}
\usepackage{microtype}
\usepackage[top=1in, bottom=1in, left=1.25in, right=1.25in]{geometry}

\ifnum\showauthornotes=1
\newcommand{\Authornote}[3]{{\sf\small\color{#3}{[#1: #2]}}}
\newcommand{\Authorcomment}[2]{{\sf \small\color{gray}{[#1: #2]}}}
\newcommand{\Authorfnote}[2]{\footnote{\color{red}{#1: #2}}}
\else
\newcommand{\Authornote}[3]{}
\newcommand{\Authorcomment}[2]{}
\newcommand{\Authorfnote}[2]{}
\fi

\ifnum\showdraftbox=1

\else

\fi

\newtheorem{theorem}{Theorem}[section]

\newtheorem{observation}[theorem]{Observation}
\newtheorem{definition}[theorem]{Definition}
\newtheorem{notation}[theorem]{Notation}
\newtheorem{lemma}[theorem]{Lemma}
\newtheorem{remark}[theorem]{Remark}
\newtheorem{proposition}[theorem]{Proposition}
\newtheorem{corollary}[theorem]{Corollary}
\newtheorem{claim}[theorem]{Claim}
\newtheorem{fact}[theorem]{Fact}

\newtheorem{example}[theorem]{Example}

\newtheorem{algo}[theorem]{Algorithm}

\newenvironment{algorithm}[3]
        {\noindent\begin{boxedminipage}{\textwidth}\begin{algo}[#1]\ \par
        {\begin{tabular}{r l}
        \textbf{Input} & #2\\
        \textbf{Output} & #3
        \end{tabular}\par\enskip}}
        {\end{algo}\end{boxedminipage}}

\def\FullBox{\hbox{\vrule width 6pt height 6pt depth 0pt}}

\def\qed{\ifmmode\qquad\FullBox\else{\unskip\nobreak\hfil
\penalty50\hskip1em\null\nobreak\hfil\FullBox
\parfillskip=0pt\finalhyphendemerits=0\endgraf}\fi}

\def\qedsketch{\ifmmode\Box\else{\unskip\nobreak\hfil
\penalty50\hskip1em\null\nobreak\hfil$\Box$
\parfillskip=0pt\finalhyphendemerits=0\endgraf}\fi}

\newenvironment{proofsketch}{\begin{trivlist} \item {\bf
Proof Sketch:~~}}
  {\qedsketch\end{trivlist}}

\def\to{\rightarrow}
\def\eps{\varepsilon}
\def\epsilon{\varepsilon}

\def\eps{\epsilon}

\def\phi{\varphi}
\def\cal{\mathcal}

\def\implies{\Rightarrow}
\newcommand{\sub}{\ensuremath{\subseteq}}

\newcommand{\defeq}{:=}

\renewcommand{\bar}{\overline} 

\newcommand{\given}{\;\ifnum\currentgrouptype=16 \middle\fi \vert\;}

\newcommand{\ie}{i.e.,\xspace}
\newcommand{\eg}{e.g.,\xspace}
\newcommand{\etal}{et al.\xspace}
\newcommand{\cf}{{\it cf.,}}

\newcommand{\mper}{\,.}
\newcommand{\mcom}{\,,}

\newcommand{\R}{{\mathbb R}}
\newcommand{\E}{{\mathbb E}}

\newcommand{\N}{{\mathbb{N}}}

\newcommand{\F}{{\mathbb F}}

\newcommand{\pmone}{\{-1,1\}\xspace}

\newcommand{\indicator}[1]{\mathds{1}_{\{#1\}}}

\usepackage{nicefrac}

\let\nfrac=\nicefrac

\newcommand{\abs}[1]{\ensuremath{\left\lvert #1 \right\rvert}}

\newcommand{\norm}[1]{\ensuremath{\left\lVert #1 \right\rVert}}
\newcommand{\smallnorm}[1]{\ensuremath{\lVert #1 \rVert}}

\newcommand{\ip}[2] {\ensuremath{\left\langle #1 , #2 \right\rangle}}

\newcommand{\one}{{\mathbf{1}}}

\newcommand{\Esymb}{\mathbb{E}}
\newcommand{\Psymb}{\mathbb{P}}

\newcommand{\Varsymb}{\mathrm{Var}}
\DeclareMathOperator*{\ExpOp}{\Esymb}

\makeatletter
\def\Pr#1{%
    \ProbabilityRender{\Psymb}{#1}%
}

\def\Ex#1{%
    \ProbabilityRender{\Esymb}{#1}%
}

\def\tildeEx#1{%
    \ProbabilityRender{\widetilde{\Esymb}}{#1}%
}

\def\condPE#1#2{%
	\@ifnextchar\bgroup
	{\ConditionalProbabilityRender{\widetilde{\Esymb}}{#1}{#2}}
	{\ProbabilityRender{\widetilde{\Esymb}}{#1 \given #2}}
}

\def\tildeVar#1{
	\ProbabilityRender{\widetilde{\Varsymb}}{#1}
}

\def\tildeCov#1#2{
	\ProbabilityRender{\tildecov}{#1,#2}
}

\def\ConditionalProbabilityRender#1#2#3#4{
	\renderwithdist{#1}{#2}{#3 \given #4}	
}

\def\ProbabilityRender#1#2{
  \@ifnextchar\bgroup%
  {\renderwithdist{#1}{#2}}
   {\singlervrender{#1}{#2}}
}
\def\singlervrender#1#2{%
   \ensuremath{\mathchoice
       {{#1}\left[ #2 \right]}
       {{#1}[ #2 ]}
       {{#1}[ #2 ]}
       {{#1}[ #2 ]}
   }
}
\def\renderwithdist#1#2#3{%
   \@ifnextchar\bgroup
   {\superfancyrender{#1}{#2}{#3}}
   {\ensuremath{\mathchoice
      {\underset{#2}{#1}\left[ #3 \right]}
      {{#1}_{#2}[ #3 ]}
      {{#1}_{#2}[ #3 ]}
      {{#1}_{#2}[ #3 ]}
     }
   }
}
\def\superfancyrender#1#2#3#4#5{
   \ensuremath{\mathchoice
      {\underset{#1}{{#1}}\left#4 #3 \right#5}
      {{#1}_{#2}#4 #3 #5}
      {{#1}_{#2}#4 #3 #5}
      {{#1}_{#2}#4 #3 #5}
   }
}
\makeatother

\newcommand{\conv}[1]{\mathrm{conv}\inparen{#1}}

\newfont{\inhead}{eufm10 scaled\magstep1}
\newcommand{\deffont}{\sf}

\newcommand{\calB}{{\cal B}}
\newcommand{\calC}{{\cal C}}
\newcommand{\calD}{{\cal D}}
\newcommand{\calE}{{\cal E}}
\newcommand{\calF}{{\cal F}}

\newcommand{\calH}{{\cal H}}
\newcommand{\calJ}{{\cal J}}
\newcommand{\calL}{{\cal L}}

\newcommand{\calO}{{\cal O}}
\newcommand{\calR}{{\cal R}}
\newcommand{\calX}{{\cal X}}
\newcommand{\calT}{{\cal T}}
\newcommand{\calS}{{\cal S}}

\newcommand{\poly}{{\mathrm{poly}}}
\newcommand{\polylog}{{\mathrm{polylog}}}

\newcommand{\suchthat}{{\;\; : \;\;}}

\DeclareMathOperator\supp{Supp}

\DeclareMathOperator{\rank}{\operatorname {rk}}

\DeclareMathOperator{\sgn}{\operatorname{sgn}}

\DeclareMathOperator*{\argmin}{\arg\!\min}
\DeclareMathOperator*{\argmax}{\arg\!\max}

\renewcommand{\iff}{\ensuremath{\Leftrightarrow}}
\renewcommand{\bar}[1]{\ensuremath{\overline{#1}}}
\newcommand{\ceil}[1]{\ensuremath{\left\lceil #1 \right\rceil}}

\newcommand{\bigoh}{\operatorname{O}}

\newcommand{\bigomega}{\mathop{\Omega}}

\newcommand{\inparen}[1]{\left(#1\right)}             
\newcommand{\inbraces}[1]{\left\{#1\right\}}           
\newcommand{\insquare}[1]{\left[#1\right]}             

\newcommand{\Szemeredi}{Szemer\'edi\xspace}

\newcommand{\Caratheodory}{Carath\'eodory\xspace}
\newcommand{\Zemor}{Z\'emor\xspace}

\newcommand{\cC}{\calC}
\newcommand{\cD}{\calD}
\newcommand{\cE}{\calE}
\newcommand{\cF}{\calF}
\newcommand{\cB}{\calB}
\newcommand{\cR}{\calR}

\newcommand{\out}{\mathrm{out}}
\newcommand{\inn}{\mathrm{in}}
\newcommand{\dec}{\mathrm{dec}}
\newcommand{\cx}{\calC_X}
\newcommand{\cz}{\calC_Z}
\newcommand{\dx}{\calD_X}
\newcommand{\dz}{\calD_Z}
\newcommand{\ex}{\calE_X}
\newcommand{\ez}{\calE_Z}
\newcommand{\fx}{\calF_X}
\newcommand{\fz}{\calF_X}
\newcommand{\hx}{H_X}
\newcommand{\hz}{H_Z}
\newcommand{\subs}[2]{{#1}_{\scaleto{{#2}\mathstrut}{6.5pt}}}
\newcommand{\dxh}{\overline{h}}
\newcommand{\dxz}{\overline{z}}

\DeclarePairedDelimiter\angles{\langle}{\rangle}
\DeclarePairedDelimiter\braces{\lbrace}{\rbrace}
\DeclarePairedDelimiter\brackets{\lbrack}{\rbrack}

\newcommand{\distLperp}[1]{\Delta_{L, \cz^\perp}\parens*{#1\,,h}}
\newcommand{\distR}[1]{\Delta_{R}\parens*{#1\,,h}}

\newcommand{\dist}[1]{\Delta \parens*{ #1 }}
\newcommand{\partmin}{\psi}
\newcommand{\disR}[2]{\Delta_{R}\parens*{#1\,,#2}}

\newcommand{\phix}{\subs{\varphi}{X}}
\newcommand{\phiz}{\subs{\varphi}{Z}}
\newcommand{\wphix}{\widetilde{\phix}}
\newcommand{\wphiz}{\widetilde{\phiz}}
\newcommand{\unconx}{\subs{\varphi}{X}^{\raisebox{0.3ex}{$\scriptscriptstyle-1$}}}
\newcommand{\unconz}{\subs{\varphi}{Z}^{\raisebox{0.3ex}{$\scriptscriptstyle-1$}}}

\newcommand{\chiTan}{\chi}
\DeclareMathOperator{\fold}{\operatorname{Fold}}

 \newcommand\SetSymbol[1][]{%
     \nonscript\:#1\vert
     \allowbreak
     \nonscript\:
     \mathopen{}}
  \DeclarePairedDelimiterX\Set[1]\{\}{%
     \renewcommand\given{\SetSymbol[\delimsize]}
     #1
}

\declaretheorem[sibling=theorem]{summary}

\DeclareSymbolFont{extraup}{U}{zavm}{m}{n}
\DeclareMathSymbol{\varheart}{\mathalpha}{extraup}{86}
\DeclareMathSymbol{\vardiamond}{\mathalpha}{extraup}{87}

\def\matr#1{\mathsf{#1}}

\def\one{\mathbf 1}

\def\Cc{\mathcal C}

\def\OPT{\mathsf{OPT}}

\def\SDP{\mathsf{SDP}}

\def\Aye{\matr A}

\def\Ess{\matr S}

\def\Jay{\matr J}

\def\ess{\mathfrak s}
\def\tee{\mathfrak t}

\DeclareMathOperator{\Tr}{Tr}

\DeclareMathOperator{\PExp}{\widetilde \E}
\DeclarePairedDelimiter\set{\lbrace}{\rbrace}
\DeclarePairedDelimiter\parens{\lparen}{\rparen}

\DeclareMathOperator{\bias}{bias}

\DeclareMathOperator{\dsum}{dsum}
\DeclareMathOperator{\dprod}{dprod}

\newcommand{\swap}[2]{\Ess_{#1,#2}}

\newcommand{\CUT}{\textup{CUT}}

\newcommand{\SCUT}{\textup{CUT}_{\pm}}

\newcommand{\error}{\delta}

\newcommand{\sset}{\chi}

\newcommand{\AN}{\alpha_{\textup{\textbf{AN}}}}

\newcommand{\tuples}[1]{W(#1)}

\usepackage{tikz}
\newcommand*\circled[1]{\tikz[baseline=(char.base)]{
            \node[shape=circle,draw,inner sep=1pt] (char) {#1};}}

\DeclareMathOperator{\zigzag}{\circled{{\rm z}}}

\usepackage{tikz}
\usetikzlibrary{shapes.symbols}

\usepackage{scalerel}

\usepackage{csquotes}
\usepackage{tcolorbox}

\newcommand{\snote}[1]{\Authornote{Shashank}{#1}{blue}}
\newcommand{\mnote}[1]{\Authornote{Madhur}{#1}{magenta}}

\DeclareMathOperator{\tow}{Tow}
\newcommand{\trans}{\mathsf{T}}
\newcommand{\fbar}{\overline{f}}

\def\pspacing{-10pt}

\def\li{{ \ell }}
\def\ri{{ r }}

\def\TanC{{ \calC^{Tan} }}
\def\AELC{{ \calC^{AEL} }}

\def\indi#1{{\one \{ #1 \} }}
\def\dupPE#1{
	\ProbabilityRender{\widetilde{\Esymb}^*}{#1}
}

\DeclareMathOperator{\agr}{\operatorname {agr}}
\DeclareMathOperator{\tildecov}{\widetilde{\operatorname {Cov}}}
\DeclareMathOperator{\dupCov}{\widetilde{\operatorname {Cov}}_{\widetilde{\Esymb}^*}}
\def\zee{{ \mathbf{Z} }}
\def\dis{{ \Delta }}
\def\tee{{ \theta }}
\def\chiTan{{ \chi}}
\def\chiAEL{{ \overline{\chi} }}

\title{Continuous Optimization for Decoding Errors}
\author{Shashank Srivastava}
\makeatletter\let\Title\@title\makeatother

\begin{document}
\pagenumbering{roman}
\begin{titlepage}
	\ifnum\isdraft=1
		\begin{tcolorbox}
			\begin{center}
        		\Huge
				DRAFT
			\end{center}
		\end{tcolorbox}
	\fi
    \begin{center}
		\ifnum\isdraft=0
			\vspace*{0.8cm}
		\fi
            
        \Huge
        \textbf{\Title}
            
        \vspace{0.8cm}
        \LARGE \textbf{Shashank Srivastava}
        \vspace{0.8cm}
        
        \large

        A thesis submitted\\
        in partial fulfillment of the requirements for\\
        the degree of\\
        \vspace{0.8cm}
        Doctor of Philosophy in Computer Science\\
        \vspace{0.8cm}
        at the\\

        \vspace{0.8cm}
        TOYOTA TECHNOLOGICAL INSTITUTE AT CHICAGO\\
        Chicago, Illinois\\

        \vspace{0.8cm}
        August 2024\\
            
        \vspace{0.8cm}
        \vspace{0.8cm}

        Thesis Committee:\\
        \vspace{0.4cm}
        Pravesh K. Kothari\\
        \vspace{0.4cm}
        Yury Makarychev\\
        \vspace{0.4cm}
        Madhur Tulsiani (Thesis Advisor)\\

    \end{center}
\end{titlepage}

\thispagestyle{empty}
\vspace*{1cm}
\begin{abstract}
Error-correcting codes are one of the most fundamental objects in pseudorandomness, with applications in communication, complexity theory, and beyond. Codes are useful because of their ability to support decoding, which is the task of recovering a codeword from its noisy copy. List decoding is a relaxation where the decoder is allowed to output a list of codewords, and has seen tremendous progress over the last 25 years. In this thesis, we prove new algorithmic and combinatorial results about list decoding. 

We describe a list decoding framework that reduces the task of efficient decoding to proving distance in certain restricted proof systems. We then instantiate this framework for Tanner codes of Sipser and Spielman [IEEE Trans. Inf. Theory 1996] and Alon-Edmonds-Luby (AEL) distance amplification [FOCS 1995] of unique decodable base codes to get the first polynomial time list decoding algorithms for these codes up to their respective Johnson bounds. We also discuss extensions to the quantum version of AEL distance amplification, yielding polynomial-time list decodable quantum LDPC codes.

We next give an alternate viewpoint of the list decoding framework based on abstract regularity lemmas instead of convex hierarchies. We show how to efficiently implement the regularity lemma for the case of Ta-Shma’s explicit codes near the Gilbert-Varshamov bound [STOC 2017]. This leads to a near-linear time algorithm for unique decoding of Ta-Shma’s codes.

We also give new combinatorial results that improve known list sizes beyond the Johnson bound. Firstly, we adapt the AEL amplification to construct a new family of explicit codes that can be combinatorially list decoded to the optimal error correction radius. This is the first example of such a code that does not use significant algebraic ingredients. Secondly, we present list size improvements for Folded Reed-Solomon codes, improving the state of the art list size among explicit list decoding capacity achieving codes.
\end{abstract}

\newpage
\chapter*{Acknowledgements}
I must start by thanking my advisor Madhur Tulsiani for his constant support and encouragement over the years. From teaching me the basics of research, to being an oracle for my many technical queries, Madhur has shown incredible patience. He has been generous with his time, pointed with his advice, and flexible with his expectations. His influence on me stretches well past academics, and like everything he chooses, he excels at being a mentor. I would like to thank other members of my thesis committee, Pravesh Kothari and Yury Makarychev, who have been supportive throughout the process with their feedback and encouragement.

I am particularly thankful to my co-author Fernando Granha Jeronimo, who always had valuable advice on how to address daunting challenges during my PhD, both technical and non-technical. Fernando has always taken out time for me from his own busy schedule, and I have learned a lot from him. I am also thankful to my other co-authors Vedat Levi Alev, Dylan Quintana and Tushant Mittal, who have all been fun people to do research with and learn things from!

While these individuals have all left a positive impact on me, it was amplified manifold because I met them in the amazing atmosphere that TTIC provides. From the very helpful admin staff, to the many workshops, and an overall atmosphere that encourages collaboration and socializing, TTIC gets a lot right when it comes to being a graduate school. Shout-out to Mary who single-handedly brings so much warmth to the institute. Thank you Adam, Alicia, Amy, Brandie, Chrissy, Erica, Jessica and everyone else who keeps things running smoothly! 

Talking about institutional support, I took some amazing and useful theory courses at TTIC offered by Julia Chuzhoy, Yury Makarychev, Nathan Srebro and Madhur Tulsiani, as well as many great courses at UChicago CS and Math. I should also thank Matthew Turk and Avrim Blum, under whose leadership TTIC continues to provide a nurturing environment for students like me.

The students at TTIC also formed the backbone of my social life in Chicago. Other than people already mentioned above, I should thank Akilesh, Amin, Ankita, Anmol, Goutham, Han, Kavya, Kshitij, Max, Naren, Nirmit, Omar, Pushkar, Shubham and Sudarshan. Special thanks to Akash, Mrinal and Rachit who welcomed me to the TTIC theory group - our trips to Devon are one of my favorite memories from Chicago. Rachit also showed me how to enjoy Chicago, and we spent countless hours walking or in the CTA together. 

Beyond TTIC, I received a lot of support from my friends Sahila and Shreyasi, for which I am grateful to both of them. Several friends from my four years at IIT Kharagpur deserve a mention as well, and they are Abhinav, Abhishek, Arafat, Astha, Pranjal, Rajat and Rakshit. Mentorship by KGP professors Rogers Mathew and Sudebkumar Pal sent me along the direction of theoretical CS, and I thank them both.

Finally, I thank my family for their unconditional love and support, and putting up with my annoying habits. Unbeknownst to them, their upbringing must have set me off on an academic journey some time in my childhood, and this thesis is the culmination of that journey. It is therefore fitting that I dedicate this thesis to them.
\tableofcontents

\listoffigures
\listoftables

\newpage
\pagenumbering{arabic}
\setcounter{page}{1}
\chapter{Introduction}\label{chap:intro}
The world today communicates at a scale unimaginable a few decades ago, and it is arguably this communication that has driven much of our progress over the last half a century. This communication takes place over extremely large networks, consisting of varying electronic architecture, geography and adversaries. This makes the communication inherently susceptible to noise, and error correction is the task of removing this noise.

Error-correcting codes, or just \emph{codes}, are objects designed to withstand noise, which are placed in a software layer over the hardware to achieve error correction. It has been known since the work of Shannon \cite{S48} that randomly chosen codes have the combinatorial structure to support optimal noise tolerance, but not necessarily to support efficient algorithms. Designing explicit codes and their associated encoding/decoding algorithms therefore pose interesting challenges in pseudorandomness and algorithm design respectively.


In this thesis, we present new algorithmic and combinatorial results for several explicit code families. Before discussing these results in detail, let us discuss introduce some basics of coding theory.

\section{Basics of Coding Theory}

\begin{definition}[Code, rate and distance]
	A code $\calC$ of blocklength $n$ and  alphabet size $q$ is a subset of $[q]^n$, and each element of this subset is called a codeword. Two important parameters of a code are its rate and distance:
\begin{itemize}
\item Rate $R = \frac{\log_q |\calC|}{n}$.
\item Distance $\delta = \min_{\substack{f_1, f_2 \in \calC \\ f_1\neq f_2}} \Delta(f_1,f_2)$.
\end{itemize}
\end{definition}
In the definition above, $\Delta(f_1,f_2)$ is the normalized Hamming distance between $f_1$ and $f_2$ viewed as strings over $[q]$. 

\begin{definition}[Linear and LDPC codes]
The code $\calC$ is linear if $[q]$ can be identified with the finite field $\F_q$ (that is, $q$ is a prime power), and the subset $\calC$ is a linear subspace of $\F_q^n$. A parity check matrix $H$ of a code $\calC$ is a matrix such that $\ker H = \calC$, and $\calC$ is a low-density parity check (LDPC) code if it has a parity check matrix that has only constantly many non-zero entries in each row.
\end{definition}

While designing codes, it is desirable to maximize both rate and distance, but it can be shown a high rate precludes high distance and a high distance precludes high rate. The possible rate vs distance tradeoff is of central importance in coding theory, and many questions related to it remain unresolved, especially in the case of binary codes with $q=2$. An infinite family of codes with growing blocklength is called \emph{good} if both rate and distance are bounded below by an absolute constant.

For communication, $q^{R\cdot n}$ messages are placed into a bijection with the code. Then, if Alice wants to send Bob a message, she sends the corresponding codeword. Because the codewords are all separated from each other, even a somewhat corrupted codeword may be identified with the original codeword, and therefore the intended message. More formally, if the fraction of positions in $[n]$ where corruptions occur is at most $\frac{\delta}{2}$, then by the triangle inequality, a corrupted codeword may be uniquely mapped back to the uncorrupted codeword, and this is called unique decoding.

One limitation of unique decoding is that the fraction of errors corrected can never exceed $1/2$. List decoding is a relaxation of unique decoding where a corrupted codeword may be mapped to a small (say, polynomial in $n$) list of codewords it could have come from. Unlike unique decoding, list decoding allows correcting a fraction of errors arbitrarily close to 1, and this feature has led to applications in complexity theory and other areas of pseudorandomness. 

Both unique decoding and list decoding have natural algorithmic questions associated to them: does there exist an efficient algorithm that takes as input a corrupted codeword (with a promise on the amount of corruption), and outputs the list of codewords close to it? In the case of unique decoding, when the promised number of errors is at most half the distance, this list must be of size at most 1.
%

If $g \in [q]^n$ is the corrupted codeword received by Bob, we denote by $\calL(g,\eta)$ the list of codewords that are at a distance of at most $\eta$ from $g$. Rephrasing the above, the combinatorial bounds ask for an upper bound on the size of $\calL(g,\eta)$ and the algorithmic challenge is to output the list $\calL(g,\eta)$.
\vspace{\pspacing}
\section{Background and Motivation}
%

There has been tremendous progress in the area of list decoding over the last 25 years, starting from the works of Guruswami and Sudan \cite{S97, GS99} who showed that widely used Reed-Solomon codes can be efficiently list decoded upto the Johnson bound. 

\begin{theorem}[Johnson Bound \cite{Joh62, G01}]\label{thm:johnson}
	Let $\calC$ be a code with distance $\delta$ and alphabet size $q$. Then there is a threshold $\calJ_q(\delta) \defeq (1-1/q) \cdot \inparen{ 1- \sqrt{1- \frac{q}{q-1}\delta}} \in \inparen{\frac{\delta}{2},\delta}$ such that the code is combinatorially list decodable upto radius $\calJ_q(\delta)$. More precisely, for any $g\in [q]^n$ and $\eps > 0$,
	\begin{enumerate}[(i)]
		\item $|\calL(g, \calJ_q(\delta))| \leq (q-1)\cdot n$.
		\item $|\calL(g, \calJ_q(\delta) - \eps)| \leq \calO_{\eps}(1)$.
	\end{enumerate}
\end{theorem}

Reed-Solomon (RS) codes are based on evaluations of bounded degree polynomials over large finite fields. An RS code of rate $R$ has distance $1-R$, and this is optimal by the Singleton bound. The alphabet size of RS codes however is at least $n$, and in particular grows with the blocklength. The Johnson bound corresponding to this distance and large alphabet size is $1-\sqrt{R}$, and so RS codes allow correcting $1-\sqrt{R}$ using rate $R$ codes.

The technique of interpolation-based decoding pioneered by Guruswami and Sudan has since been used for list decoding variety of algebraic codes \cite{GS99, GI03, GR08, GW11, GX13, Kop15, BHKS23}, yielding both combinatorial and algorithmic results. This includes construction of explicit codes of rate $R$ and list decoding radius arbitrarily close to $1-R$, achieving the so-called list decoding capacity over large alphabets. They were also combined with combinatorial operations on codes to get list decodable codes with more desirable properties such as smaller alphabet size.

While algebraic codes such as Reed-Solomon and Reed-Muller continue to be the most well-studied family of codes, constructions based on expander graphs can enjoy some features missing in algebraic constructions. For example, some expander-based codes are LDPC and can be unique decoded in truly linear time \cite{SS96, GI05}. 
A number of recent breakthrough code constructions have also been based on expanders or high-dimensional expanders. This includes Ta-Shma's construction of explicit binary codes with near-optimal rate-distance tradeoff \cite{TS17}, and Dinur et al's contruction of locally testable codes with constant distance, constant rate and constant locality \cite{DELLM22}.

However, when it comes to list decoding algorithms (or even combinatorial list decodability in interesting parameter regimes), we know very few techniques in the absence of algebraic structure in a code. In view of the interesting code families based purely on combinatorial/spectral properties, it is desirable to reproduce the success of algebraic interpolation based list decoding for a broader class of codes.
%
%

%


\section{Overview of Results}

In this section, we give a brief overview of the results appearing in this thesis. We focus on the broader context and techniques used, and leave the technical details to later chapters.
\subsection{A Generic Framework for List Decoding}\label{sec:framework}
As mentioned previously, the rate vs distance tradeoff for binary codes remains poorly understood. The Plotkin bound says that any binary code with distance $1/2$ or more can only have polynomially many codewords, and therefore must have a vanishing rate. For distance $1/2-\eps$, it is known that the best rate achievable is $\calO(\eps^2\log(1/\eps))$ \cite{MRRW77, Alon09}, while the Gilbert-Varshamov bound says that random (linear) codes have a rate of $\Omega(\eps^2)$. In 2017, Ta-Shma \cite{TS17} made a breakthrough by constructing \emph{explicit} codes close to the GV bound with a rate of $\eps^{2+o(1)}$. This code is based on a direct sum operation based on a pseudorandom hypergraph constructed via modifications to expander walks.

While trying to design algorithms for decoding direct sum codes motivated by \cite{TS17}, Alev et al \cite{AJQST20} introduced the technique of entropy maximization for list decoding codes. The idea is to find a \emph{single} object in a one-shot optimization, from which the list of codewords may be extracted. Entropy maximization for this object ensures that it \emph{covers} every codeword in the list. Similar ideas have also been used in the setting of list decodable learning \cite{KarmalkarKK19, RY20}. This part of the algorithm was later used in essentially the same form by Jeronimo et al \cite{JQST20} and Richelson and Roy \cite{RR23}, who gave the first algorithms for unique decoding and list decoding up to Johnson bound respectively for Ta-Shma's code. 

However, this entropy proxy was defined in terms of a convex relaxation for the Ta-Shma code that crucially relied upon the direct sum structure. In \cref{chap:framework}, we show that this idea can be adapted for any code, and in fact, it can be used to give an alternate proof of the combinatorial Johnson bound (\cref{thm:johnson}).

\begin{definition}
Let $\chi: [q] \rightarrow \R^{q-1}$ be a map such that
\[
	\ip{\chi(x)}{\chi(y)} = \begin{cases}
		\frac{-1}{q-1} \qquad \qquad &x\neq y, \\
		1 & x = y
	\end{cases}
\]
Such a map exists because the Gram matrix is positive semi-definite and of rank $q-1$.
\end{definition}

In the binary case, the above definition just maps $\{0,1\}$ to $\{-1,1\}$, and for large alphabets it maps them to the indicator vectors. The next lemma shows that maximizing a simple entropy proxy for a distribution conditioned on the distribution being in a Hamming ball ensures the distribution is supported on every codeword in the said ball.
\begin{lemma}[Covering Lemma]\label{lem:covering_intro}
	Let $\calC$ be a code with alphabet $q$ and distance $\delta$. Let $g \in \F_q^n$ be a corrupted codeword. Let $\calD$ be a distribution over $\calC$ such that
	\begin{enumerate}[(i)]
		\item $\Ex{h \sim \calD}{\Delta(g,h)} < \calJ_q(\delta)$, and
		\item the entropy functional $\Psi(\calD) = - \Ex{i\in [n]}{\norm{\Ex{h\sim \calD}{\chi(h_i)}}^2}$ is maximized among all distributions satisfying (i).
	\end{enumerate}
	Then, $\calL(g,\calJ_q(\delta)) \sub \supp(\calD)$.
\end{lemma}

\begin{corollary}[Johnson bound]
	$|\calL(g,\calJ_q(\delta))| \leq (q-1)n+1$.
\end{corollary}

\begin{proof}
	The functional $\Psi(\calD)$ only depends on the functions $\norm{\Ex{h\sim \calD}{\chi(h_i)}}^2$ for $i\in [n]$, which in turn depend on a total of $(q-1)n$ real-valued functionals on $\calD$. Then, by Tchakaloff's theorem, which is a simple extension of \Caratheodory's theorem \cite{COA20}, we can find a distribution $\calD'$ with the $\Psi(\calD') = \Psi(\calD)$ but with $|\supp(\calD')| \leq (q-1)n+1$. Since it is also true that $\calL(g,\calJ_q(\delta)) \sub \supp(\calD')$, it follows that 
	\[
		|\calL(g,\calJ_q(\delta)| \leq |\supp(\calD')| \leq (q-1)n+1.
	\]
\end{proof}
Further, we show that the techniques used by \cite{AJQST20, JQST20, RR23} for list decoding Ta-Shma codes can be extended to a wide class of codes whose distance proof is based on \emph{local} properties combined with spectral expansion for the local-to-global jump. The idea is to algorithmize the above proof of Johnson bound by replacing the set of distributions of codewords by a relaxation that is easier to optimize over via convex optimization techniques. Next, we argue that this relaxation is still tight enough to prove distances similar to true codewords. Combining the two parts yields a list decoding algorithm.

This leads to the first polynomial time list decoding algorithm for Tanner codes of Sipser and Spielman \cite{SS96}, which are an important class of LDPC codes. It also leads to an algorithm that reduces list decoding a distance amplification scheme of Alon, Edmonds and Luby \cite{AEL95} to unique decoding of the base code. This distance amplification scheme has found numerous applications \cite{GI05, KMRZS17, GKORZS18, BGG22}, but all list decoding algorithms for it relied upon efficient list recovery of the base code. List recovery is a generalization of list decoding, and therefore our reduction to just unique decoding of the base code significantly lowers the requirements posed. All of these algorithms are based on the above mentioned proof of the Johnson bound, and so work upto the corresponding Johnson bounds for these codes.

The convex relaxations used above are based on the Sum-of-Squares (SoS) hierarchy, which is a series of increasingly tight relaxations parameterized by a degree parameter $t$. Recall that a true distribution can be represented in terms of its (exponentially many) marginal distributions. A degree-$t$ SoS relaxation for these distributions maintains marginals for only sets of size at most $t$ that are consistent with each other. Another property of true distributions is that its moment matrix is positive semidefinite, and the SoS relaxation imposes this constraint on the $n^{\calO(t)}$ marginals as well. Another viewpoint of the SoS hierarchy is that it can be used to efficiently discover certificates for statements that can be proved using non-negativity of degree-$t$ polynomials. The main advantage of using SoS is that it supports optimization in time $n^{\calO(t)}$ using semidefinite programming.

The algorithms for both Tanner codes and AEL amplification use an SoS relaxation for a large but constant $t$, combined with a random conditioning based rounding from \cite{BRS11, AJT19}.

\subsubsection*{Comparison to Existing Techniques}
Most existing list decoders in literature are based on the framework of Guruswami and Sudan \cite{S97, GS99, PV05} for algebraic codes and its extensions. These interpolation-based decoders proceed by learning a structured object (bivariate or multivariate polynomials of low-degree) using a completely unstructured received word $g$. Then, the algorithm forgets about $g$, and the entire list is extracted out of the learned structured object.

We believe that the list decoding machinery developed in \cite{AJQST20, JQST20, RR23} and this thesis can be seen as a broad framework that implements the above scheme for general codes, regardless of any algebraic structure in the code. The structured object is now a distribution over codewords (or a regularity lemma decomposition, as we will see in \cref{sec:regularity}), which contains the entire list.

Therefore, this framework gives a natural starting point for designing list decoding algorithms for any code. Our results can be seen as a successful implementation of this scheme for multiple classes of codes that use spectral expanders, but the scheme itself is quite general.

\vspace{\pspacing}
\subsubsection*{Other Applications of the Framework}

We briefly mention some extensions of the entropy maximization machinery for decoding.

\begin{enumerate}
\item By suitably modifying the entropy function, all these list decoding algorithms can be adapted for list-recovery as well as weighted list-recovery, which are generalizations of list decoding useful when performing further operations such as concatenation on a code. It can also handle the presence of both errors and erasures while decoding.
\item If the code obtained by AEL were to be finally concatenated again to obtain a binary code instead of a large alphabet code, then a simple modification in the entropy function to write it in terms of the final binary code allows us to decode up to the Johnson radius corresponding to the product distance. This is different from using the list-recovery properties of the outer code for soft decoding, which incurs some loss \cite{GuruswamiS02}.
\item This machinery can also be adapted to list decode the recent construction of locally testable codes with constant rate, constant distance and constant locality by \cite{DELLM22}. That code is similar to the Tanner construction, and a straightforward adaptation of our Tanner code decoder extends to those $c^3$-LTCs.
\item While we focus on the list decoding results for these codes, the notion of a distance proof for pseudocodewords can also be used for unique-decoding up to half the distance. In fact, one may skip the entropy maximization step in this case, and simply ask for the closest pseudocodeword to a received word. This may be useful in boosting unique-decoding from a very small fraction of errors, to half the distance.
\end{enumerate}
\vspace{\pspacing}

\subsection{List Decodable Quantum LDPC Codes}\label{sec:quantum}
We saw how SoS based algorithms can successfully list decode an important class of LDPC codes constructed using bipartite spectral expanders upto the Johnson bound, namely the Tanner codes of Sipser and Spielman \cite{SS96}. In fact, these graph based constructions were the first explicit construction of a classical LDPC code with constant distance and constant rate, more than 30 years after Gallager proved their existence \cite{Gallager62}. LDPC codes are known for their fast, simple and parallelizable unique decoders, and have therefore been widely used in practice.

Quantum codes are a generalization of classical linear codes set to play an important role in the realization of Quantum Computing. We will focus on the subclass of Calderbank-Shor-Steane (CSS) codes which are closest to classical linear codes. Due to the features above, as well as other complications that arise in quantum computing architectures, quantum codes that are LDPC (QLDPC) are particularly important for applications. In fact, it is desirable that not only is every parity check sparse, but that it is also \emph{spatially} not spread out.

\begin{definition}[CSS Codes]
A CSS Code over $\F_2$ can be described as a pair of linear codes $\calC = (\calC_X,\calC_Z)$ with a shared blocklength $n$ and the property that $\calC_X^{\perp} \sub \calC_Z$ (which implies $\calC_Z^{\perp} \sub \calC_X$). The distance of $\calC$ is $\min( d_X, d_Z)$ where 
\begin{align*}
	d_X = \min_{f \in \calC_X \backslash \calC_Z^{\perp}} \frac{|f|}{n} \qquad \text{and} \qquad d_Z = \min_{f \in \calC_Z \backslash \calC_X^{\perp}} \frac{|f|}{n}
\end{align*}

and the rate is $r_X+r_Z - 1$ where
\begin{align*}
	r_X = \frac{\log |\calC_X|}{n} \qquad \text{and} \qquad 	r_Z = \frac{\log |\calC_Z|}{n}
\end{align*}

Further, $\calC$ is LDPC if both $\calC_X$ and $\calC_Z$ have parity check matrices with sparse rows.
\end{definition}

One of the central open problems in quantum error correction was to design QLDPC codes with constant distance and constant rate, until it was recently solved by Panteleev and Kalachev \cite{PK22} using Left-Right Square Cayley complexes which can be seen as higher dimensional analogs of the Ramanujan Cayley graphs. In fact, unlike in the classical case, we didn't even know of the existence of such quantum codes since we do not know of any randomized constructions either. 
Building on the work of Panteleev and Kalachev \cite{PK22}, Leverrier and \Zemor \cite{LZ22} constructed related quantum LDPC codes and paired them with an iterative decoder that can decode errors upto a constant fraction of the distance in linear time and is also parallelizable.

The AEL distance amplification was adapted to the quantum setting by Bergamaschi, Golowich and Gunn \cite{BGG22}. The quantum Singleton bound says that the distance of a rate $R$ quantum code is at most $\frac{1-R}{2}$, and just like in the classical case, this bound is achieved by quantum Reed-Solomon codes (which are not LDPC though) over large alphabet.
\cite{BGG22} applied the quantum AEL amplification to the good QLDPC codes of \cite{PK22, LZ22} to obtain QLDPC codes approaching the quantum Singleton bound.

These codes are shown to be unique decodable in \cite{BGG22} up to radius $\frac{1-R}{4}$, and one would like our techniques from \cref{sec:framework} to be able to list decode beyond this unique decoding radius. However, a straightforward adaptation of our covering lemma and distance proof based techniques quickly breaks down. This is because $\calC_Z^{\perp} \sub \calC_X$ necessarily contains sparse vectors, and notion of distance for CSS codes is only defined up to shifts within cosets of $\calC_Z^{\perp}$. 

With the exponential size of $\calC_Z^{\perp}$, this algebraic modification becomes challenging to implement in our largely analytic framework. The problem of distinguishing $\calC_Z^{\perp}$ from $\calC_X \backslash \calC_Z^{\perp}$ for certain quantum LDPC codes has even been the source of strong lower bounds for the Sum-of-Squares hierarchy \cite{HL22}.
%

WLOG, say we wish to prove a lower bound on $d_X$, that is, if $x,y\in \calC_X$ are such that $y \not\in x+\calC_Z^{\perp}$, then $d_X = \Omega(1)$. A typical proof starts with $x$ and $y$, and considers $x'$ that is defined to be the element of $x+\calC_Z^{\perp}$ closest to $y$. Obviously, $\Delta(x,y) \geq \Delta(x',y)$. The optimality of $x'$ with respect to distance from $y$ gives it some additional structure, and with this additional structure assumed, one proves a classical-like distance $\Delta(x',y)=0$ or $\Delta(x',y)=\Omega(1)$. Coming to algorithms, this last classical-like distance is often captured by the SoS machinery from previous section, but the trivial statement $\Delta(x,y) \geq \Delta(x',y)$ need not be a low-degree proof!

Nevertheless, we show in \cite{MST23} that for the case of quantum AEL amplification, one can find an $x''$ for which $\Delta(x,y)\geq \Delta(x'',y)$ and $\Delta(x'',y) = \Omega(1)$ are both low-degree proofs. However, this $x''$ is chosen explicitly instead of as an optimal point for some optimization problem, and while it lowers the distance to $y$, it need not be the closest to $y$ among $x+\calC_Z^{\perp}$. Therefore, we choose to call it a partial minimizer. Our algorithm only uses a function $\Gamma(x'')$ in the algorithm which is designed so that $\Gamma(x) = \Gamma(x'')$, and in particular, the algorithm not need to compute $x''$ corresponding to some pseudocodeword $x$. This is important since there are many different $x''$ for the same $x$ based on different values of $y$.

Moreover, we use the fact that the space $\calC_Z^{\perp}$ can be decomposed into a space generated by certain new AEL generators and another space generated by base code generators. 
Since these AEL generators are highly structured and explicit, they can be argued about with Sum-of-Squares of large enough degree without needing the algebraic structure. For the base code generators, we show that a unique decoder already encapsulates the ability to deal with the low-weight elements of $\calC_Z^{\perp}$.

\subsection{Near-linear time decoding via Regularity}\label{sec:regularity}
All the algorithms discussed so far have been based on the Sum-of-Squares hierarchy combined with random conditioning based rounding, leading to poor runtimes (even though polynomial-time) with large exponents of $n$. In this section, we show that ideas inspired from the covering lemma approach have the potential to lead to much faster algorithms. In particular, we will present a near-linear time algorithm from \cite{JST21} for Ta-Shma's codes that can unique decode up to half the distance, as well as list decode up to a radius that approaches $1/2$ as distance approaches $1/2$. We will call this \emph{gentle} list decoding to differentiate it from list decoding along the Johnson bound. 

Our starting point is a combinatorial regularity lemma that proves a weak form of Johnson bound, in the spirit of weak regularity lemma of Frieze and Kannan \cite{FriezeK99}.

\begin{lemma}[Weak Regularity Lemma]\label{lem:regularity}
	Let $g\in \R^n$ be a unit norm vector, and $\calF \sub \R^n$ be a family of unit norm vectors to be thought of as \emph{distinguishers}. 
	
	Then, for any $\eta > 0$, there exists an $h = \sum_{i=1}^k c_i f_i$, with $f_i \in \calF$, $c_i \in \R$, and $k\leq 1/\eta^2$, such that $g$ and $h$ are $\eta$-indistinguishable by $\calF$. That is, for any $f\in \calF$,
	\[
		\ip{g-h}{f} = \ip{g-\sum_{i=1}^k c_i f_i}{f} \leq \eta
	\]
	Moreover, we may ensure that $c_i >0$ and $\sum_{i} c_i \leq 1/\eta$.
\end{lemma}

The lemma above works with any inner product, and the norms used are the ones induced by the corresponding inner product. This lemma appeared first in the work of Frieze and Kannan \cite{FriezeK99} on approximation schemes for dense Max-Cut, where $g$ is the adjacency matrix of the input graph and the family $\calF$ are rank-1 $\pm 1$ matrices. It has since been generalized significantly \cite{LS07, Gowers10}, and new proofs have been found based on running gradient descent for minimizing the convex function $\Gamma(h) = \sup_{f\in \calF} \ip{g-h}{f}$. The optimum of this function is at most 0 since $h=g$ is feasible, and gradient descent will lead us to a point $h$ with $\Gamma(h) \leq \eta$ in $1/\eta^2$ steps. The gradient at every step is some $f\in \calF$, so that the update looks like $h \leftarrow h+ \eta \cdot f$, which gives the desired structure for $h$.

The regularity lemma is also useful for illustrating the structure vs pseudorandomness paradigm. We are decomposing any $g = h + (g-h)$, where $h$ is to be thought of as the structured part, and $g-h$ is pseudorandom with respect to $\calF$. Moreover, the bigger the family $\calF$, the stronger the pseudorandomness property, but weaker the structure of $h$. We will exploit this handle on structure vs pseudorandomness to design our algorithms. But first, let us prove a slightly weak form of Johnson bound, which is particularly useful in the large distance regime. For simplicity, we focus on the binary case, since our main application is to decoding of Ta-Shma's codes.

\begin{corollary}[Weak Johnson bound]\label{cor:weak_johnson}
	Let $\calC$ be an binary code with distance greater than $\frac{1-\eps}{2}$. Then for any received word $g$, $\abs{\calL(g, \frac{1-2\sqrt{\eps}}{2})} \leq \frac{1}{\eps}$.
\end{corollary}

\begin{proof}
	Embed each $\F_2^n$ into $\R^n$ by mapping $0$ to $1$ and $1$ to $-1$, as used before in the proof of Johnson bound. The family $\calF$ is taken to be all codewords of $\calC$, and it holds from the distance of $\calC$ that that for any $f, f'\in \calF$, $\ip{f}{f'} \leq \eps$.
	
	Applying the regularity \cref{lem:regularity} with $\eta$, we obtain an $h = \sum_{i=1}^k c_i f_i$, where $f_i$ are codewords, such that for all $f\in \calF$,
	\[
		\ip{g-h}{f} \leq \eta \implies \ip{h}{f} \geq \ip{g}{f} - \eta
	\]
	Let $f^* \in \calL(g,\frac{1-2\sqrt{\eps}}{2})$, so that $\ip{g}{f^*} > 2\sqrt{\eps}$. Choosing $\eta = \sqrt{\eps}$, this means
	\begin{align*}
		\ip{h}{f^*} &\geq \ip{g}{f^*} - \eta > 2\sqrt{\eps} - \sqrt{\eps} = \sqrt{\eps} \\
		\ip{\sum_i c_i f_i}{f^*} &> \sqrt{\eps} \\
		\ip{\frac{\sum_i c_i f_i}{\sum_i c_i}}{f^*} &> \frac{1}{\sum_i c_i} \cdot \sqrt{\eps} ~\geq \eta \cdot \sqrt{\eps} = \eps
	\end{align*}
	Using the distance of $\calC$, this implies that $f^*$ must be one of $f_i$ that appear in $h = \sum_i c_i f_i$. Since there are only $1/\eta^2 = 1/\eps$ many of them, this gives an upper bound on the list size.
\end{proof}

\subsubsection*{Efficient Regularity Lemma}
To use the regularity lemma algorithmically, one needs to implement the gradient oracle, which can be hard in general. However, as suggested before, the regularity lemma can also be used by enlarging the set of distinguishers to be some superset of codewords. In the extreme case, when the set of distinguishers is every binary string, implementing the regularity lemma algorithmically becomes easy, but the structure of $h$ is not useful enough.

Therefore, we wish to find a relaxation of the set of codewords (say, pseudocodewords) which supports efficient optimization of linear forms, but this relaxation still needs to be tight enough to satisfy the distance property. Looking back at the proof of \cref{cor:weak_johnson}, regardless of whether the set of pseudocodewords have distance, one may always conclude $\ip{f^*}{f_i} > \eps$ for some $i \in [k]$. This means that we only want the set of pseudocodewords to support an efficient decoder that works upto the distance $\frac{1-\eps}{2}$. When the pseudocodeword $f_i$ is a true codeword, then this is trivial as the only codeword in a ball around $f_i$ is $f_i$ itself.

Roughly speaking, true codewords in Ta-Shma's code correspond to $k^{th}$ tensor power of base codewords, restricted to a pseudorandom $k$-regular (ordered) hypergraph instead of all $k$-tuples. For the case of Ta-Shma's code, we show that the set of rank-1 tensor products of $t$ many $\pm 1$ strings 
satisfies both of the conditions above simultaneously. The strings used in this tensor product no longer need to be confined to the base code. 

This apriori requires solving $k$-wise cut-norm problem, which is hard in general. However, we show that for expander walks, a combination of 2-wise cut norm \cite{AN04} and an iterative dense model theorem \cite{ReingoldTTV08} allows us to solve the linear form optimization problem to implement the gradient oracle. Moreover, the SDP-based algorithm of \cite{AN04} can be made to work in near-linear time based on the techniques of \cite{LeeP20}. The decoder for these rank-1 tensor based pseudocodewords is then simply to decode in the base code from each string that is a factor in the tensor product.

The algorithm crucially uses the pseudorandom hypergraph to be able to use 2-wise cut-norm iteratively. In the process, it incurs an error of $\approx \eps^{1/k}$ compared to the error of $\calO(\sqrt{\eps})$ for the non-algorithmic regularity lemma where $\eps$ is the bias of the code. This stops us from getting an algorithm that decodes close to the Johnson bound. However, our ideas are sufficient to get a unique decoder and a gentle list decoder for Ta-Shma's code. It is still open to decode Ta-Shma's code to the Johnson bound in near-linear time, and one potential avenue would be to implement this regularity scheme with error $\approx \sqrt{\eps}$.
\begin{theorem}[\cite{JST21}]
	Suppose $\calC$ is a binary $\eps$-balanced code with rate $\eps^{2+o(1)}$ obtained via the direct-sum based distance amplification scheme of Ta-Shma applied to a base code $\calC_0$, where $\calC_0$ can be unique-decoded in time $\calT(n)$. Then there is an algorithm that runs in time $\calO_{\eps}(n \cdot \polylog(n))\cdot \calT(n)$ and unique-decodes $\calC$ upto half its distance.
\end{theorem}


%
\subsection{Decodability beyond Johnson Bound}\label{sec:beyond_johnson}
One major drawback of the SoS-based list decoding framework is that it is limited to list decoding upto the Johnson bound for a code. This is because it can be seen as an efficient implementation of the Covering Lemma based proof of the combinatorial Johnson bound. Since the Johnson bound is known to be tight in general, we do not expect improvements to the Covering Lemma in general.

However, to harness the full potential of list decodability, and in particular construct codes that achieve the list decoding capacity, it becomes important to list decode beyond the Johnson bound of best possible distance. 
%
%
We know that for any $R\in (0,1)$ and $\eps>0$, there exist codes (in particular, random codes) of rate $R$, distance $1-R$ that can be list decoded upto radius $1-R-\eps$, even though the corresponding Johnson bound is only $1-\sqrt{R}$. An explicit family of such "list decoding capacity achieving" codes were constructed by \cite{GR08}, called Folded Reed-Solomon Codes. The proof of their list decodability used algebraic arguments relying on the algebra inherent in the construction of the code. This has also been the case for later constructions that achieve list decoding capacity \cite{GW11, Kop15, KMRZS17, GX22, KRSW23}. 

We saw AEL amplified codes provide an alternative to RS codes by approaching the same optimal rate-distance tradeoff. A natural question is whether we can modify the near-Singleton bound codes constructed via \cite{AEL95} in a way that retains the optimal rate-distance tradeoff but improves the list decoding radius beyond the Johnson bound. The AEL amplification can be seen as a local-to-global phenomenon for rate-distance tradeoff, since the AEL code inherits its rate-distance tradeoff from a constant sized code placed on vertex neighborhoods of the expander graph. This local to global transfer only happens when supported by a high-rate constant distance code however.

Indeed, when the base code is replaced by a stronger high-rate list recoverable code such as the one from \cite{GR08}, then one can even do a local to global transfer of list decodability properties. Can we transfer the list decodability properties even if the base code continues be an arbitrary high-rate code with constant distance? This is important for applications such as constructing LDPC codes achieving list decoding capacity.

Recently, in an ongoing work, we show that AEL amplification can indeed be used for constructing capacity-achieving codes, based only on expanders and any base code that has high-rate with constant distance.
As expected for the local-to-global phenomenon, the amplification now uses, in addition to a spectral expander, a constant sized inner code that is chosen to be a capacity achieving code with constant list sizes. Examples of such codes include random linear codes \cite{AGL24}, Reed-Solomon codes with random evaluation points \cite{ST20, BGM23, GZ23, AGL24}, folded Reed-Solomon codes \cite{KRSW23, Tamo24}.

This yields explicit codes achieving list decoding capacity with constant alphabet size and constant list size, as well as the first constant alphabet explicit code family near the generalized Singleton bound $\frac{2}{3}\cdot (1-R)$ for lists of size 2. Moreover, by starting with an LDPC and linear-time unique decodable code, our capacity achieving codes can also be made LDPC and linear-time unique decodable, thereby resolving an open problem of \cite{MosheiffRRSW19}.

\begin{theorem}[Ongoing work]
	For every $\eps >0$ and $R\in (0,1)$, there is an explicit family of codes constructed using AEL amplification with rate $R$ and distance at least $1-R-\eps$ which is combinatorially list decodable up to $\frac{2}{3}\cdot (1-R- \eps)$ with list size $2$. The alphabet size of the code is $2^{\poly(1/\eps)}$.
\end{theorem}

\begin{theorem}[Ongoing work]\label{thm:ael_capacity}
	For every $\eps >0$ and $R\in (0,1)$, there is an explicit family of codes constructed using AEL amplification with rate $R$ and distance at least $1-R-\eps$ which is combinatorially list decodable up to $1-R-\eps$ with list size and alphabet size dependent only on $\eps$ and independent of the blocklength.
\end{theorem}

The list size as well as alphabet size for the codes in \cref{thm:ael_capacity}, though constant, grow as a tower of height $\poly(1/\eps)$, where $\eps$ is the gap to capacity. Moreover, these results are just combinatorial in nature, and leave open the interesting task of efficiently decoding these codes up to the capacity. 

These codes are based on a new analysis of AEL amplification that enables a local-to-global transfer of combinatorial list decodability from the inner code to $\AELC$. The ideas are based on using erasures to guarantee slightly non-trivial list decodability for interleaved codes. Such a result was already known when interleaving Reed-Solomon codes due to the multivariate interpolation technique of \cite{CS03, PV05} and other works thereafter, but our argument works for interleaving of general codes. Somewhat surprisingly, an argument analogous to the Schwartz-Zippel lemma can be made to work for interleaved codes, even when the codes being interleaved do not have any algebraic structure. We believe that these ideas shed new light on interpolation-based decoding, and they may find further applications.
\subsection{List Size Improvement for Folded Reed-Solomon Codes}\label{sec:folded_rs}
While the codes from previous section are interesting because of new techniques, we have known list decoding capacity achieving codes since the work of Guruswami and Rudra \cite{GR08}. These codes are based on \emph{folding} the usual Reed-Solomon code. Folding is a simple operation where for each string, $m$ different alphabet symbols in $[q]$ are treated single symbol of $[q^m]$. Therefore, folding transforms a string in $[q]^n$ to $[q^m]^{n/m}$. 

It is not difficult to see that folding preserves the rate, and the distance cannot decrease. If we fold a code on the Singleton bound such as a Reed-Solomon code, then the distance must also be preserved. The main advantage of folding Reed-Solomon (RS) codes is that the list decoding radius improves to beyond Johnson bound, and is at least $1-R-\eps$ if $m = \frac{1}{\eps^2}$. This was shown by Guruswami and Rudra \cite{GR08}, building on the work on Parvaresh and Vardy \cite{PV05}, and a simpler proof and algorithm was found in \cite{Gur11} using observations from \cite{Vadhan12}.

However, it is important that this folding for RS codes is done along a specific algebraic structure. Let $\F_q$ be a field, and $\gamma \in \F_q^*$ be a primitive element of $\F_q$, so that every non-zero element of $\F_q$ can be written as $\gamma^i$ for some integer $i\geq 0$. The (full-length) Reed-Solomon code of blocklength $n$ and rate $R$ has messages corresponding to polynomials of degree $<Rn$, and the codewords are given by the map
\[
	f(X) \rightarrow \insquare{ f(1), f(\gamma), \cdots ,f(\gamma^{n-1}) } \in \F_q^n
\]
where $n=q-1$.
The folded Reed-Solomon code has the same set of messages, but the encoding changes to
\[
	f(X) \rightarrow \insquare{ \begin{pmatrix}
	f(1) \\
	f(\gamma)\\
	\vdots \\
	f(\gamma^{m-1})
	\end{pmatrix}, 
	\begin{pmatrix}
	f(\gamma^m) \\
	f(\gamma^{m+1})\\
	\vdots \\
	f(\gamma^{2m-1})
	\end{pmatrix},
	\cdots ,\begin{pmatrix}
	f(\gamma^{n-m}) \\
	f(\gamma^{n-m+1})\\
	\vdots \\
	f(\gamma^{n-1})
	\end{pmatrix} } \in (\F_q^m)^{n/m}
\]

The main result from \cite{Gur11} says that if $m=1/\eps^2$, then for any $g\in (\F_q^m)^{n/m}$, the list $\calL(g,1-R-\eps)$ is contained in an affine subspace of dimension at most $\calO(1/\eps)$. This immediately gives an upper bound of $n^{\calO(1/\eps)}$ for the list size, proving that folded RS codes combinatorially achieve list decoding capacity. \cite{Gur11} also showed that a basis for the affine subspace can be found in $\calO(n^2)$ time, and so there also exists an efficient algorithm for this list decoding.

Note that the bounds on both list size and alphabet size are large polynomials in $n$. \cite{GR08} showed in their original paper on capacity achieving codes that the AEL amplification can also be used for alphabet size reduction to a constant independent of $n$ (but dependent on $\eps$). For list size improvement, \cite{Gur11} isolated a pseudorandom object called \emph{subspace evasive sets} such that no affine subspace of small dimension can intersect with a subspace evasive set in more than $\calO_{\eps}(1)$ points. Thus, if the message polynomials for folded RS codes were chosen from such a set, the lists would be of size at most $\calO_{\eps}(1)$. Moreover, \cite{Gur11} showed the existence of such subspace evasive sets with large enough size that the loss in rate due to pre-encoding is negligible.

Explicit subspace evasive sets were then constructed by Dvir and Lovett \cite{DL12}, giving codes decodable up to $1-R-\eps$ with list size $(1/\eps)^{\calO(1/\eps)}$. There have also been attempts to use algebraic-geometric (AG) codes \cite{Gur09, GX12, GX22}, as well as to use variants of subspace evasive-ness \cite{GX13, GK16, GRZ21}, to reduce the alphabet size, list size and/or decoding time.

\subsubsection*{List size of Folded RS codes}
Somewhat surprisingly, it was shown in \cite{KRSW23} that folded RS codes themselves, without any modification, have much smaller list sizes than previously thought. They proved an upper bound of $(1/\eps)^{\calO(1/\eps)}$ using a general theorem on the intersection of Hamming balls and affine subspaces, matching the list size obtained by \cite{DL12} using subspace evasive sets. Their proof was somewhat simplified by \cite{Tamo24}, and was based on earlier ideas on subspace designs from \cite{GK16}.

For decoding upto $1-R-\eps$, it is known that the list size must be at least $1/\eps$. It is known that random puncturings of RS and folded RS codes achieve capacity with this optimal list size. However, for explicit codes, there is a gap between best possible list size and best known list size. In \cref{chap:list_size}, we give an elementary proof of the $(1/\eps)^{1/\eps}$ bound based on intersections between Hamming balls and affine subspaces. In fact, we prove a stronger result by showing that when decoding $m$-folded RS codes up to $\approx \frac{s}{s+1}\inparen{1-R}$, where $m$ is sufficiently large compared to $s$ and $s\geq 1$, the list size is bounded by $s\cdot (s+1)^{s-2}$. 

This implies the results of \cite{KRSW23, Tamo24}, but also works for fixed small values of $s=2,3,\cdots$. The case $s=1$ is just unique decoding. We note that the case of $s=2$ and the corresponding list size of 2 was also shown by \cite{Tamo24}, but his method did not generalize to $s>3$. 

We use many of the same techniques as earlier works, but structure our proof in a bottom-up manner instead of a top-down manner. We start by showing a simple combinatorial argument that shows that an affine subspace of dimension 1, or a line, can intersect a Hamming ball of radius $\frac{s}{s+1}\Delta$ in at most $s$ points. This relies on the simple observation that given a line, $[n]$ can be divided into two sets $S$ and $\overline{S}$ such that all the points on the line agree on $\overline{S}$, and any two points on the line differ \emph{everywhere} on $S$. Thus, the restriction to $S$ can be seen as a distance 1 code, and moreover $|S| \geq \Delta n$.

Therefore, the agreement sets between codewords and the received word must be disjoint over $S$, and if there were $s+1$ codewords in the list, one of these agreement sets must have size at most $\frac{|S|}{s+1}$ (when restricted to $S$). This codeword and the received word differ in at least $\frac{s}{s+1}\cdot |S| \geq \frac{s}{s+1}\Delta n$ positions, contradicting its membership in the list.

We then use an induction on the dimension of the affine subspace, and the capacity result is obtained for $s\approx 1/\eps$ and dimension $\approx 1/\eps$.

\subsubsection*{Using the Folded Structure}
Until now, our results are based on the above general argument applied to the RS code underlying the folded RS code. Next, we use the folded structure to gain improvements in the list size. Roughly speaking, our induction step above corresponds to counting how many agreement sets a fixed coordinate can belong to. For the 1-dimensional case, the disjointness of agreement sets is just the statement that each coordinate can belong to at most 1 agreement set. For higher dimensions, fixing a coordinate to be an agreement fixes its value, and therefore reduces the dimension of the affine subspace by 1. We then use the inductive hypothesis to obtain a weak version of the disjointness, and a double counting argument similar to 1-dimensional case finishes the proof.

For a folded RS code, fixing a particular coordinate to be an agreement actually gives us multiple equations, and can reduce the dimension by much more than 1. In an ideas case, all of these equations will be linearly independent, and we can fix the entire codeword after fixing a single coordinate. This would again look like the 1-dimensional case above. Unfortunately, such a strong guarantee need not hold. However, \cite{GK16} proved a global upper bound on the sum of rank deficit we see in each coordinate. This is based on the notion of a folded Wronskian determinant criterion for linear independence, and this part of the proof was also used by \cite{KRSW23, Tamo24}. However, with our sharper bottom-up structure of the proof, we are able to improve the list size to $(s-1)^2+1$. 

This gives a list size bound of 2 for decoding up to $\frac{2}{3}(1-R)$, and in the regime of list decoding capacity with $s\approx 1/\eps$, the list size is bounded by $\calO(1/\eps^2)$.
\vspace{\pspacing}
\section{Organization}

We start in \cref{chap:prelims} with a description of common code families and coding theoretic operations we will be encountering throughout this thesis. In addition, we include preliminaries on the Sum-of-Squares hierarchy of convex relaxations, which is one of the main technical ingredients used in the algorithmic results later.

The results in \cref{chap:framework} that describe a general list decoding framework along with concrete instantiations for Tanner codes and AEL amplification are based on joint work with Jeronimo and Tulsiani \cite{JST23}. 

In \cref{chap:quantum}, we introduce quantum CSS codes and describe the challenges the framework from \cref{chap:framework} faces in the quantum setting. The construction of list decodable QLDPC codes in this chapter is based on joint work with Mittal and Tulsiani \cite{MST23}.

In \cref{chap:regularity}, we describe a near-linear time algorithm for unique decoding Ta-Shma's codes based on efficient regularity lemmas for sparse but pseudorandom hypergraphs, such as the one used in Ta-Shma's code construction. This is based on joint work with Jeronimo and Tulsiani \cite{JST21}.

In \cref{chap:capacity}, we show our construction of new codes that have bounded list size upto optimal error radius. That is, these codes achieve list decoding capacity, and can be constructed by applying a suitable modification of the AEL distance amplification procedure to an arbitrary base code. Unfortunately, we do not know of an efficient algorithm to decode up to this radius however. This is based on unpublished results from ongoing work.

Finally, in \cref{chap:list_size}, we present list size improvements for Folded Reed-Solomon codes. First, we show how elementary combinatorial arguments can give a simplified proof of the results in \cite{KRSW23, Tamo24}, and then a more involved argument using folded Wronskian determinants from \cite{GK16} can give a significantly improved bound on the list size. This is also based on unpublished results from ongoing work.

In \cref{appendix:ta-shma}, we include the details of the code construction by Ta-Shma \cite{TS17}, which is mostly needed in \cref{chap:regularity}.

\chapter{Preliminaries}\label{chap:prelims}
%
%
For a bipartite graph $G=(L,R,E)$ where $L$ is the set of left vertices, and $R$ is the set of right
vertices, we index the left set by $\li$ and the right set by $\ri$. For a vertex $\li \in L$, 
we denote the set of edges incident to it by $N_L(\li)$ (left neighborhood), and the set of edges
incident to  $\ri\in R$ is denoted by $N_R(\ri)$(right neighborhood). 
We use $\li \sim \ri$ to denote that the vertex $\li \in L$ is adjacent to the vertex $\ri \in R$, that is, $(\li,\ri)\in E$. 

Fix an arbitrary ordering of the edges. Then there are bijections between the sets $E$, $L \times
[d]$, and $R \times [d]$, 
given by taking $(\li,i)$ to be the $i^{th}$ edge incident on $\li$, and similarly for $R \times [d]$.
%
Henceforth, we will implicitly assume such an ordering of the edges is fixed, and use the resulting bijections.
%


\begin{definition}
	Let $[q]$ be a finite alphabet and let $f,g\in [q]^n$. Then the (fractional) distance
        between $f,g$ is defined as \[ \dis(f,g) = \Ex{i\in [n]}{ \indi{f_i \neq g_i}} \mper \]
\end{definition}

\begin{definition}[Code, distance and rate]
	A code $\calC$ of block length $n$, distance $\delta$ and rate $\rho$ over the alphabet size $q$ is a set $\calC \subseteq [q]^n$ with the following properties
	\begin{enumerate}[(i)]
		\item $\rho = \frac{\log_{q} |\calC|}{n}$
		\item $\delta = \min_{\substack{h_1,h_2\in \calC \\ h_1\neq h_2}} \dis(h_1,h_2)$
	\end{enumerate}
	Such codes are succinctly represented as $[n,\delta,\rho]_q$.
	We say $\calC$ is a linear code if $[q]$ can be identified with a finite field and $\calC$ is a linear subspace of $[q]^n$.
\end{definition}

We generally think of codes as an infinite family such that $n\rightarrow \infty$. A code family, or just code, is \emph{good} if both $\rho$ and $\delta$ are bounded below by a constant independent of $n$ as $n\rightarrow \infty$.
\begin{definition}[List of codewords]
	For any $g\in [q]^n$, the list of codewords in $\calC$ that are at a distance less than
        $\eta$ from $g$ is denoted by $\calL(g,\eta)$. That is $\calL(g,\eta) = \inbraces{h\in \calC \suchthat \dis(h,g) < \eta }$.
\end{definition}
\begin{definition}[Unique and List Decoding]
	A list decoding algorithm for a code $\calC$ of distance $\delta$ takes as input $g\in [q]^n$ and a decoding radius parameter $\eta$, and outputs $\calL(g,\eta)$. When $\eta<\delta/2$, then this list is of size at most 1, and the algorithm is called unique decoding algorithm.
\end{definition}

Efficient list decoding algorithms that work up to radius $\eta$ can only work if the list $\calL(g,\eta)$ is of bounded size. We will call a code $\calC$ \emph{combinatorially list decodable} up to $\eta$ if $\calL(g,\eta)$ is of size at most $\poly(n)$.

The rate of the code corresponds to how densely packed it is in the ambient space $\F_q^n$, and the distance is a measure of how separated the codewords are. There is a tension between these quantities therefore, which is captured in the following simple relationship called the Singleton bound.

\begin{theorem}[Singleton bound]
	Any code with rate $\rho$ and distance $\delta$ must have $\rho+\delta \leq 1$.
\end{theorem}

As we will see soon, Reed-Solomon codes achieve the Singleton bound, and are therefore called Maximum Distance Separable (MDS) codes. In particular, one may obtain codes with distance $1-\eps$ with a constant rate $\eps$. 

However, Reed-Solomon are defined over very large (in fact, growing with $n$) alphabets. For small and fixed alphabet like $q=2$, the Plotkin bound says that distances beyond the threshold $1-\frac{1}{q}$ are out of reach for constant rate codes.
\begin{theorem}[Plotkin bound]
	A code with alphabet size $q$ and $\delta\geq 1-\frac{1}{q}$ must have rate $\rho \rightarrow 0$ as $n\rightarrow \infty$.
\end{theorem}
In fact, the following bound is not difficult to show, and we include a proof for completeness.
\begin{theorem}\label{thm:strong_plotkin}
	A code with alphabet size $q$ and $\delta > 1-\frac{1}{q^2}$ can have only $q$ codewords. Therefore, the rate $\rho$ is at most $\frac{1}{n}$.
\end{theorem}

\begin{proof}
Suppose the code is of size $M$, and we arrange all codewords into a matrix of size $M\times n$, so that each row corresponds to a codeword, and each column corresponds to a coordinate $i\in [n]$. We use a double counting argument to upper and lower bound the total number differences between all $M^2$ pairs of codewords. For pairs where both codewords are equal, the number of differences is 0, and it is $>(1-\frac{1}{q^2})n$ for any pair where the two codewords differ.

To upper bound this quantity, let $m_{ij}$ denote the count of symbol $j \in [q]$ in column $i\in [n]$. Then the coordinate $i$ will not contribute differences to $\sum_j m_{ij}^2$ pairs out of $M^2$, and the remaining $M^2-\sum_j m_{ij}^2$ pair of codewords must have a difference in coordinate $i$.
\begin{align*}
	M(M-1)\inparen{1-\frac{1}{q^2}}n &< \sum_{i\in[n]} \insquare{M^2 - \sum_{j\in [q]} m_{ij}^2} \\
	&\leq \sum_{i\in[n]} \insquare{M^2 - \frac{1}{q} \inparen{\sum_{j\in [q]} m_{ij}}^2} \\
	&= \sum_{i\in[n]} \insquare{M^2 - \frac{1}{q} M^2} \\
	&=nM^2\inparen{1-\frac{1}{q}}
\end{align*}
That is,
\[
	\inparen{1-\frac{1}{M}} \cdot \inparen{1+\frac{1}{q}} < 1 \implies M < q+1.
\]
\end{proof}

\section{Reed-Solomon Codes and their Variants}
The most well-studied family of codes is the Reed-Solomon code, which is based on evaluations of low-degree polynomials over finite fields. We refer the reader to \cite{GRS23} for a thorough introduction to well-studied code families. We collect some basic facts that we will need throughout this thesis.

\begin{definition}[Reed-Solomon (RS) codes]
	Given a finite field $\F_q$ and $n$ distinct points $\alpha_1,\cdots ,\alpha_n \in \F_q$, the Reed-Solomon code of rate $R$ is defined to be
	\[
		\calC^{RS} = \inbraces{ \inparen{f(\alpha_1), f(\alpha_2), \cdots , f(\alpha_n) } \suchthat f(X)\text{ is a polynomial in }\F_q[X] \text{ of degree }< R\cdot n }
	\]
	Because no two distinct degree $<Rn$ polynomials can agree on more than $Rn$ evaluation points, the distance of the code is at least $1-R$, which is optimal by the Singleton bound. The alphabet of $\calC^{RS}$ is $\F_q$, and the blocklength is $n$.
\end{definition}

Note that the field $\F_q$ needs to have at least $n$ distinct points for this construction to work, and therefore the alphabet size for RS codes grows with $n$.

Efficient decoders for unique decoding and list decoding RS codes up to Johnson bound are well known \cite{GS99, GRS23}. The corresponding decoding radii are $\frac{1-R}{2}$ and $1-\sqrt{R}$ respectively. It was shown recently that an RS code with randomly chosen evaluation points can be combinatorially list decoded up to $1-R$ with optimal list size, achieving the generalized Singleton bound \cite{BGM23, GZ23, AGL24}. However, no explicit constructions or efficient algorithms are known.

After the work of Guruswami and Sudan \cite{GS99} gave an algorithm to decode up to $1-\sqrt{R}$, there were attempts to improve this decoding radius all the way to $1-R$, which would be the best possible. While this has not been achieved for Reed-Solomon codes themselves, we now know of codes closely related to RS codes that achieve this optimal error correction radius. Such codes were discovered by Parvaresh and Vardy \cite{PV05}, Guruswami and Rudra \cite{GR08}, Guruswami and Wang \cite{GW11}, Kopparty \cite{Kop15}, Bhandari et al \cite{BHKS23}, etc. These codes also come with efficient algorithms to perform this decoding, including near-linear time decoders \cite{GHKS24}.

One of these codes that achieves the list decoding capacity is the folded Reed-Solomon code. Recall that in the usual code, the codeword associated to a degree $<Rn$ polynomial is the evaluation of the polynomial on $n$ distinct points. In $m$-folded RS codes, $m$ different evaluations are clubbed into a single alphabet symbol, so that the alphabet of the code becomes $\F_q^m$. The blocklength also changes to $n/m$, but the rate is preserved, and the distance cannot decrease.

Moreover, this folding is done using a primitive element $\gamma$ of $\F_q$, so that the symbols $f(\alpha), f(\gamma \alpha), \cdots ,f(\gamma^{m-1}\alpha)$ are folded into a single bigger symbol. This is important for the algebraic proof of list decodability beyond Johnson bound. The encoding map now looks like
\[
	f(X) \rightarrow \insquare{ \begin{pmatrix}
	f(1) \\
	f(\gamma)\\
	\vdots \\
	f(\gamma^{m-1})
	\end{pmatrix}, 
	\begin{pmatrix}
	f(\gamma^m) \\
	f(\gamma^{m+1})\\
	\vdots \\
	f(\gamma^{2m-1})
	\end{pmatrix},
	\cdots ,\begin{pmatrix}
	f(\gamma^{n-m}) \\
	f(\gamma^{n-m+1})\\
	\vdots \\
	f(\gamma^{n-1})
	\end{pmatrix} } \in (\F_q^m)^{n/m}
\]
where $f(X) \in \F_q[X]^{<Rn}$.
The following key claim about list decodability of folded RS codes was proven in \cite{Gur11}
\begin{theorem}\label{thm:lin_alg_rs}
	Let $\calC^{FRS}$ be the $m$-folded Reed-Solomon code defined based on polynomials in $\F_q[X]^{<Rn}$, so that the the field size $q>n$, blocklength is $N=n/m$, rate is $R$ and distance is $1-R$.
	
	For any integer $s$, where $1\leq s \leq m$, and for any $g\in (\F_q^m)^N$, the list $\calL(g,\frac{s}{s+1}\inparen{1-\frac{m}{m-s+1}R})$ is contained in an $(s-1)$-dimensional affine subspace of $\F_q[X]^{<Rn}$.
	
	Moreover, there is an algorithm that given $g$, can find a basis for this affine subspace in time $\calO((n\log q)^2)$.
\end{theorem}
The codes achieving list decoding capacity are obtained by setting $s=1/\eps$ and $m=1/\eps^2$.
\section{Interleaving of codes}
We saw how folding as a coding theoretic operation can enhance the list decoding properties of a code at the cost of increasing the alphabet size. Another similar operation is interleaving of two codes $\calC_1$ and $\calC_2$ of the same blocklength (say $n$), denoted $\calC_1 \odot \calC_2$. We first define interleaving of two strings $x,y \in [q]^n$ to be a third string $z\in ([q]^2)^n$, defined as $z_i = (x_i,y_i)$. That is, the interleaved string simply writes down on coordinate $i$ both the symbols that appear on coordinate $i$ in $x$ and $y$. We use the notation $z = x \odot y$.

To extend this definition to codes, for every pair of codewords $(f_1,f_2) \in \calC_1 \times \calC_2$ (Cartesian product of sets), we include the string $f_1\odot f_2$ in $\calC_1 \odot \calC_2$. The interleaving operation may be extended to $k>2$ strings in a natural fashion by including the symbols on each of the $k$, and accordingly increasing the alphabet size of the new string to be $q^k$. Similar to above, we may also extend the definition to interleaving of $k$ codes by considering all $k$-tuples of codewords in $\calC_1 \times \cdots \times \calC_k$. Formally,
\[
	\calC_1 \odot \calC_2 = \inbraces{ \inparen{ \begin{matrix} f_1(1) \\	f_2(1) \end{matrix}, \begin{matrix} f_1(2) \\	f_2(2) \end{matrix}, \cdots , \begin{matrix} f_1(n) \\	f_2(n) \end{matrix} } \suchthat f_1 \in \calC_1, f_2 \in \calC_2}
\]
where $f_1(i)$ denotes the $i^{th}$ coordinate of $f_1$, and likewise for $f_2$.

The code $\calC \odot \calC \odot \cdots \odot \calC$, where the $k$ copies of the same code are interleaved is denoted in shorthand as $\calC^{\odot k}$.

The number of codewords in $\calC^{\odot k}$ is $|\calC|^k$, and the alphabet size also changes from $q$ in $\calC$ to $q^k$ in $\calC^{\odot k}$, so that the rate is unchanged. It is not difficult to show that the distance is also unchanged when interleaving. We will see in \cref{chap:capacity} that the list decoding radius shifts, if allowing for lists of size $|\calC|^{k-1}$ out of $|\calC|^k$. This slightly non-trivial claim is the workhorse of our capacity achieving codes in \cref{chap:capacity}.

Such a statement was previous known for certain algebraic codes, most famously when the codes being interleaved are RS codes, based on interpolation techniques.
\section{Expander graphs}
Expander graphs are a key ingredient in many code constructions due to their pseudorandom properties.

\begin{definition}[Non-bipartite expander]
	Let $G=(V,E)$ be a $d$-regular graph, and let $A_G$ be the normalized adjacency matrix of $G$ so that its top eigenvalue is 1. The graph $G$ is called an $(n,d,\lambda)$-expander if $A_G$ has all the other eigenvalues bounded in absolute value by $\lambda$. When $n$ and $d$ are clear from context, we call such graphs simply $\lambda$-expanders.
\end{definition}

Infinite families of $(n,d,\lambda)$-expanders, with growing $n$ as $d$ and $\lambda$ are constant, can be derived based on Ramanujan graphs of \cite{LPS88} as long as $\lambda \geq \frac{2\sqrt{d-1}}{d}$. The main property of an expander graph that we will use is that they satisfy the following expander mixing lemma:

\begin{lemma}[Expander Mixing Lemma (EML)]
Let $G=(V,E)$ be a $\lambda$-expander. Let $f,g : V\rightarrow \R$ be two functions on $V$. The expander mixing lemma says that the for an expander graph, the average of $f(u)g(v)$ is same whether $u$ and $v$ are chosen based on edges of an expander or $u$ and $v$ are chosen from all pairs in $V\times V$.
	\[
		\abs{\Ex{(u,v) \in E}{f(u)\cdot g(v)} - \Ex{u\in V}{f(u)}\cdot \Ex{v\in V}{g(v)}} \leq \lambda \cdot \Ex{u\in V}{f(u)^2}^{1/2} \cdot \Ex{v\in V}{g(v)^2}^{1/2}
	\]
\end{lemma}

This reduction in number of choices from $n^2$ to $\frac{nd}{2}$ upto $\lambda$ error is what allows expanders to be useful in many derandomization applications.

We will be also be heavily using bipartite expanders, which can be obtained as 2-covers of usual expanders, and also satisfy a corresponding expander mixing lemma.

\begin{definition}[Bipartite expander]
	Let $G=(L,R,E)$ be a $d$-regular bipartite graph, and let $A_G$ be the $L \times R$ normalized biadjacency matrix of $G$ so that its top singular value is 1. The graph $G$ is called an $(n,d,\lambda)$-bipartite-expander if $A_G$ has all the other singular values bounded by $\lambda$. Again, when $n$ and $d$ are clear from context, we will call $G$ simply a $\lambda$-bipartite-expander. Moreover, when it is clear from context that we are dealing with bipartite graphs, we will simply call $G$ a $\lambda$-expander.
\end{definition}

\begin{lemma}[Expander Mixing Lemma (EML) on bipartite graphs]\label{lem:eml_bipartite}
Let $G=(L,R,E)$ be a $\lambda$-bipartite expander. Let $f : L\rightarrow \R$ and $g:R\rightarrow \R$ be two functions on $L$ and $R$ respectively. The EML for bipartite graphs says that
	\[
		\abs{\Ex{(\li,\ri) \in E}{f(\li)\cdot g(\ri)} - \Ex{\li\in L}{f(\li)}\cdot \Ex{\ri\in R}{g(\ri)}} \leq \lambda \cdot \Ex{\li \in L}{f(\li)^2}^{1/2} \cdot \Ex{\ri\in R}{g(\ri)^2}^{1/2}
	\]
\end{lemma}

\section{Graph based codes}
\paragraph{Expander codes.}
We recap the construction from~\cite{Zemor01}, building on ideas from \cite{Tanner81, SS96},
which constructs an infinite family of good codes starting from any good (inner) linear code over a
small fixed block length of rate larger than $1/2$. The code $\calC_0$ is also referred to as the
base code for $\calC^{Tan}$.
\begin{definition}
	Given an inner linear code $\calC_0$ on alphabet $[q]$ and block length $d$, and a $d$-regular graph $G(V,E)$, we define the Tanner code $\TanC$ as
	\[
		\TanC = \{ h: E\rightarrow [q] \suchthat h|_{N(v)} \in \calC_0, \forall v \in V\}
	\]
	where $h|_S$ for $S\subseteq E$ denotes the restriction of $h$ to the set of coordinates indexed by $S$.
	In this work, we will only use Tanner codes defined on bipartite graphs. 
\end{definition}
By using an infinite family of graphs with constant degree $d$, we get an infinite family of codes based on an inner code of block length $d$.

\begin{theorem}[\cite{SS96}, \cite{Zemor01}]
	Let the distance and rate of the inner code $\calC_0$ be $\delta_0$ and $\rho_0$ respectively, and those of associated $\TanC$ be $\delta$ and $\rho$. If $G$ is an $(n,d,\lambda)$-expander, then $\delta \geq \delta_0\cdot (\delta_0-\lambda)$ and $\rho\geq 2\rho_0-1$.
\end{theorem}
The codes above were shown to be linear time decodable up to their unique decoding radius (for the designed distance) in \cite{Zemor01} and \cite{SkaRoth03}.
\paragraph{Alon-Edmonds-Luby distance amplification}\label{sec:AEL_prelims}

The following distance amplification scheme was introduced in \cite{ABNNR92}, \cite{AEL95} and used by \cite{GI05} to design linear time unique decodable near-MDS codes.

\begin{definition}[Concatenated codes]
	Fix an $(n,d,\lambda)$-expander $G(L,R,E)$. Let $\calC_{1}$ be an $[n,\delta_{1},r_{1}]_{q_1}$ code and let $\calC_{0}$ be a $[d,\delta_{0},r_{0}]_{q_0}$ code with $q_1 = |\calC_0|$. 
	
	We define the concatenation of $f\in [q_1]^L$ with $\calC_{0}: [q_1]\rightarrow [q_0]^d$ as $f^*: E \rightarrow [q_0]$, defined as
	\[
		f^{*}_{\calC_{0}}(e) ~=~ \calC_{0}(f(\li))(j)
	\]
	where $\li$ is the left endpoint of edge $e$ and $e$ is the $j^{th}$ edge incident on $\li$.
	The concatenated code $\calC^{*}_{\calC_{0}}(\calC_1)$ of block length $n\cdot d$ and alphabet $[q_0]$ is defined to be,
	\[
		\calC^{*}_{\calC_{0}}(\calC_{1}) ~=~ \{ f^{*}_{\calC_{0}} \suchthat f \in \calC_{1}\}
	\]
	When clear from context, we will omit $\calC_{0},\calC_1$ in the above notation to call the concatenated code $\calC^*$.
\end{definition}

\begin{claim}
	$\dis(f^{*}_{\calC_{0}},g^{*}_{\calC_{0}}) ~\geq~ \delta_{0}\cdot \dis(f,g)$, which also
        implies $\Delta(\calC^{*}_{\calC_{0}}(\calC_{1})) \geq \delta_{0} \cdot \delta_{1}$.
\end{claim}


\begin{definition}[AEL Codes]
	Fix an $(n,d,\lambda)$-expander $G(L,R,E)$. Let $\calC_{1}$ be an $[n,\delta_{1},r_{1}]_{q_1}$ code and let $\calC_{0}$ be a $[d,\delta_{0},r_{0}]_{q_0}$ code with $q_1 = |\calC_0|$. 
We define the AEL-encoding $f^{AEL}_{\calC_{0}}: R \rightarrow [q_0]^d$ as
	\[
		f^{AEL}_{\calC_{0}} (\ri) ~=~ \left( f^{*}_{\calC_{0}}(e_1),f^{*}_{\calC_{0}}(e_2),\cdots,f^{*}_{\calC_{0}}(e_d) \right)
	\]
	where $e_1,e_2,\cdots,e_d$ are the $d$ edges incident on $\ri$.
	The AEL code $\calC^{AEL}_{\calC_{0}}(\calC_{1}) \subseteq [q_0^d]^{n}$ is defined as 
	\[
		\calC^{AEL}_{\calC_{0}}(\calC_{1}) ~=~ \{ f^{AEL}_{\calC_{0}} \suchthat f \in \calC_{1}\}
	\]
	When clear from context, we will omit $\calC_{0},\calC_1$ in the above notation to call the AEL code $\AELC$.
\end{definition}
Alon, Edmonds and Luby, proved the following result, which shows that the construction can be used
to amplify the distance to $\delta_0$, by choosing $\lambda$ sufficiently small.
\begin{theorem}[\cite{AEL95}]\label{thm:ael_distance}
	$\dis(f^{AEL}_{\calC_{0}},g^{AEL}_{\calC_{0}}) ~\geq~ \delta_{0}-
        \frac{\lambda}{\dis(f,g)}$, which also implies $\Delta(\calC^{AEL}_{\calC_{0}}(\calC_{1}))
        ~\geq~ \delta_{0} - \frac{\lambda}{\delta_{1}}$.
\end{theorem}



A codeword $f$ of $\AELC$ technically belongs to the space $[q_0^d]^R$. However, we will often choose to identify codewords of $\AELC$ as belonging to $[q_0]^E$. It is clear that the two spaces are in bijection with each other and in fact, one can just \emph{fold} or \emph{unfold} the symbols to move from one space to the other. Choosing $f$ to be in $[q_0]^E$ allows us to talk about $f$ viewed from left vertex set $L$ or right vertex set $R$ simply by changing the distance functions. Let $f,g \in \AELC$, then we define the following three distance functions:
\begin{align*}
	\dis^L(f,g) &\defeq \Ex{\li\in L}{ \indi{ f_{N_L(\li)} \neq g_{N_L(\li)}}} \\
	\dis(f,g) &\defeq \Ex{e\in E}{\indi{ f_e \neq g_e}} = \Ex{\li \in L}{\dis(f_{N_L(\li)},g_{N_L(\li)})} = \Ex{\ri \in R}{\dis(f_{N_R(\ri)},g_{N_R(\ri)})}\\
	\dis^R(f,g) &\defeq \Ex{\ri\in R}{ \indi{ f_{N_R(\ri)} \neq g_{N_R(\ri)}}} 
\end{align*}

With this notation, \cref{thm:ael_distance} can be stated in a simpler form.
\begin{theorem}[Restatement of \cref{thm:ael_distance}]
	$\dis^R(f,g) \geq \delta_0 - \frac{\lambda}{\dis^L(f,g)}$.
\end{theorem}
\section{Sum-of-Squares hierarchy}
The sum-of-squares hierarchy of semidefinite programs (SDPs) provides a family of increasingly
powerful convex relaxations for several optimization problems. 
Each ``level" $t$ of the hierarchy is parameterized by a set of constraints corresponding to
polynomials of degree at most $t$ in the optimization variables. While the relaxations in the
hierarchy can be viewed as  semidefinite programs of size $n^{O(t)}$ \cite{BS14, Laurent09}, 
it is often convenient to view the solution as a linear operator, called the ``pseudoexpectation" operator.
%

%
%
\vspace{-5 pt}
\paragraph{Pseudoexpectations}
Let $t$ be an positive even integer and fix an alphabet $[q]$. An SoS solution of degree $t$, or a pseudoexpectation of SoS-degree $t$, over the variables $\zee = \{Z_{i,j}\}_{i\in[m],j\in[q]}$ is represented by a linear operator $ \tildeEx{\cdot}: \R[\zee]^{\leq t} \rightarrow \R$ such that:
\vspace{-5 pt}
\begin{enumerate}[(i)]
    \item $\tildeEx{1} = 1$.
    \item $\tildeEx{p^2} \geq 0$ if $p$ is a polynomial in $\zee = \{Z_{i,j}\}_{i\in [m],j\in [q]}$ of degree $\leq t/2$.
\end{enumerate}
\vspace{-5 pt}
 Note that linearity implies $\tildeEx{p_1} + \tildeEx{p_2} = \tildeEx{p_1+p_2}$ and $\tildeEx{c\cdot
  p_1} = c \cdot \tildeEx{p_1}$ for $c\in \R$, for $p_1, p_2 \in \R[\zee]^{\leq t}$.
This also allows for a succinct representation of $\tildeEx{\cdot}$ using any basis for $\R[\zee]^{\leq t}$.

The set of all pseudoexpectations should be seen as a relaxation for the set of all possible
(distributions over) assignments to $m$ variables in alphabet $[q]$.
Indeed, any assignment $f: [m] \rightarrow [q]$, can be seen as a pseudoexpectation 
which assigns the value $1$ to a
monomial consistent with $f$ and $0$ otherwise. This can be extended via linearity to all
polynomials, and then by convexity of the constraints to all distributions over assignments.
However, the reverse is not true when $t < m$, and there can be degree-$t$ pseudoexpectations which do
not correspond to any genuine distribution.

It is possible to optimize over the set of degree-$t$ pseudoexpectations in time $m^{O(t)}$ via SDPs
(under certain conditions on the bit-complexity of solutions~\cite{OD16, RW17:sos}).
We next define what it means for pseudoexpectations to satisfy some problem-specific constraints.

\begin{definition}[Constrained Pseudoexpectations]\label{def:constraints_on_sos}
Let $\calS = \inbraces{f_1 = 0, \ldots, f_m = 0, g_1 \geq 0, \ldots, g_r \geq 0}$ be a system of
polynomial constraints, with each polynomial in $\calS$ of degree at most $t$. We say $\tildeEx{\cdot}$ is a pseudoexpectation operator respecting $\calS$, if in addition to the above conditions, it also satisfies
	\begin{enumerate}
	\item $\tildeEx{p \cdot f_i} = 0$,  $\forall i \in [m]$ and $\forall p$ such that $\deg(p \cdot f_i) \leq t$.
	\item $\tildeEx{p^2 \cdot \prod_{i \in S} g_i} \geq 0$, $\forall S \subseteq [r]$ and $\forall p$ such that $\deg(p^2 \cdot \prod_{i \in S} g_i) \leq t$.
	\end{enumerate}
\end{definition}
\paragraph{Local constraints and local functions.}
Any constraint that involves at most $k$ variables from $\zee$, with $k\leq t$, can be written as a degree-$k$ polynomial, and such constraints may be enforced into the SoS solution.
%
%
In particular, we will always consider the following canonical constraints on the variables $\zee$.
\ifnum\confversion=1
\begin{align*}
&Z_{i,j}^2 = Z_{i,j},\ \forall i\in[m],j\in[q] \\
\text{and} \quad &\sum_j Z_{i,j} = 1,\ \forall i\in[m] \mper
\end{align*}
\else
\[
Z_{i,j}^2 = Z_{i,j},\ \forall i\in[m],j\in[q] 
\quad \text{and} \quad 
\sum_j Z_{i,j} = 1,\ \forall i\in[m] \mper
\]
\fi
%
%
%
We will also consider additional constraints and corresponding polynomials, defined by ``local" functions. For any $f\in [q]^m$ and $M\sub [m]$, we use $f_M$ to denote the restriction $f|_M$, and $f_i$ to denote $f_{\{i\}}$ for convenience.
\begin{definition}[$k$-local function]
	A function $\mu: [q]^m \rightarrow \R$ is called $k$-local if there is a set $M\subseteq [m]$ of size $k$ such that $\mu(f)$ only depends on $\inbraces{f(i)}_{i\in M}$, or equivalently, $\mu(f)$ only depends on $f|_M$.
	
	If $\mu$ is $k$-local, we abuse notation to also use $\mu: [q]^M \rightarrow \R$ with $\mu(\alpha) = \mu(f)$ for any $f$ such that $f|_M=\alpha$. It will be clear from the input to the function $\mu$ whether we are using $\mu$ as a function on $[q]^m$ or $[q]^M$.
\end{definition}

Let $\mu:[q]^m\rightarrow \R$ be a $k$-local function that depends on coordinates $M\subseteq [m]$ with $|M|=k$. Then $\mu$ can be written as a degree-$k$ polynomial $p_{\mu}$ in $\zee$:
\[
	p_{\mu}(\zee) = \sum_{\alpha \in [q]^M} \inparen{\mu(\alpha) \cdot\prod_{i\in M} Z_{i,\alpha_i}}
\]


With some abuse of notation, we let $\mu(\zee)$ denote $p_{\mu}(\zee)$. We will use such $k$-local
functions inside $\tildeEx{\cdot}$ freely without worrying about their polynomial
representation. For example, $\tildeEx{ \indi{\zee_{i} \neq j}}$ denotes $\tildeEx{ 1- Z_{i,j}}$. 
%
The notion of $k$-local functions can also be extended from real valued functions to vector valued functions in a straightforward way.

\begin{definition}[vector-valued local functions]
A function $\mu: [q]^m \rightarrow \R^N$ is $k$-local if the $N$ real valued functions corresponding to the $N$ coordinates are all $k$-local. Note that these different coordinate functions may depend on different sets of variables, as long as the number is at most $k$ for each of the functions.
\end{definition}
\vspace{-5 pt}
\paragraph{Local distribution view of SoS}

It will be convenient to use a shorthand for the function $\indi{\zee_S = \alpha}$, and we will use $\zee_{S,\alpha}$. Likewise, we use $\zee_{i,j}$ as a shorthand for the function $\indi{\zee_i = j}$. That is, henceforth,
\ifnum\confversion=1
\begin{align*}
	&\tildeEx{\zee_{S,\alpha}} = \tildeEx{\indi{\zee_S = \alpha}} = \tildeEx{ \prod_{s\in S}
                                    Z_{s,\alpha_s}}  \\
	\text{and } \quad & \tildeEx{\zee_{i,j}} = \tildeEx{\indi{\zee_i = j}} = \tildeEx{ Z_{i,j}}.
\end{align*}
\else
\begin{align*}
	\tildeEx{\zee_{S,\alpha}} ~=~ \tildeEx{\indi{\zee_S = \alpha}} = \tildeEx{ \prod_{s\in S}
                                    Z_{s,\alpha_s}} 
\qquad \text{and} \qquad
	\tildeEx{\zee_{i,j}} ~=~ \tildeEx{\indi{\zee_i = j}} = \tildeEx{ Z_{i,j}}
\end{align*}
\fi


Note that for any $S \subseteq [m]$ with $\abs{S} = k \leq t/2$,
\ifnum\confversion=1
\begin{gather*}
	\sum_{ \alpha \in [q]^{k}} \tildeEx{\zee_{S,\alpha}} = 
 \tildeEx{ \prod_{s\in S} \inparen{ \sum_{j\in [q]} Z_{s,j}} } = 1
\\
	\tildeEx{\zee_{S,\alpha}} = \tildeEx{ \prod_{s\in S} Z_{s,\alpha_s}} = \tildeEx{ \prod_{s\in
            S} Z^2_{s,\alpha_s}} \geq 0 \mper
\end{gather*}
\else
\[
	\sum_{ \alpha \in [q]^{k}} \tildeEx{\zee_{S,\alpha}} = 
 \tildeEx{ \prod_{s\in S} \inparen{ \sum_{j\in [q]} Z_{s,j}} } = 1
\qquad \text{and} \qquad
	\tildeEx{\zee_{S,\alpha}} = \tildeEx{ \prod_{s\in S} Z_{s,\alpha_s}} = \tildeEx{ \prod_{s\in
            S} Z^2_{s,\alpha_s}} \geq 0 \mper
\]
\fi
Thus, the values $\inbraces{\tildeEx{\zee_{S, \alpha}}}_{\alpha\in [q]^S}$ define a distribution
over $[q]^k$, referred to as the local  distribution for $\zee_S$.

%
%
Let $\mu: [q]^m \rightarrow\R$ be a $k$-local function for $k\leq t/2$, depending on $M \subseteq
[m]$. Then, $\tildeEx{\mu(\zee)}$ can be seen as the expected value of the function $\mu$ under the
local distribution for $M$, since
\ifnum\confversion=1
\begin{align*}
	\tildeEx{\mu(\zee)} 
~=~& \tildeEx{\sum_{\alpha\in [q]^M} \inparen{\mu(\alpha) \cdot\prod_{i\in M} Z_{i,\alpha_i}}}\\
~=~& \sum_{\alpha\in [q]^M} \mu(\alpha) \cdot \tildeEx{\prod_{i\in M} Z_{i,\alpha_i}}\\
~=~& \sum_{\alpha\in [q]^M} \mu(\alpha) \cdot \tildeEx{\zee_{M,\alpha}} \mper
\end{align*}
\else
\begin{align*}
	\tildeEx{\mu(\zee)} 
~=~ \tildeEx{\sum_{\alpha\in [q]^M} \inparen{\mu(\alpha) \cdot\prod_{i\in M} Z_{i,\alpha_i}}}
~=~ \sum_{\alpha\in [q]^M} \mu(\alpha) \cdot \tildeEx{\prod_{i\in M} Z_{i,\alpha_i}}
~=~ \sum_{\alpha\in [q]^M} \mu(\alpha) \cdot \tildeEx{\zee_{M,\alpha}} \mper
\end{align*}
\fi

\begin{claim}
	Let $\tildeEx{\cdot}$ be a degree-$t$ pseudoexpectation. For $k \leq t/2$, let $\mu_1,\mu_2$
        be two $k$-local functions on $[q]^m$, depending on the same set of coordinates $M$, and
        $\mu_1(\alpha) \leq \mu_2(\alpha) ~~\forall \alpha \in [q]^M$. Then $\tildeEx{\mu_1(\zee)} \leq \tildeEx{\mu_2(\zee)}$.
%
\end{claim}

\begin{proof}
Let $\calD_M$ be the local distribution induced by $\tildeEx{\cdot}$ for $\zee_M$. Then
$\tildeEx{\mu_1(\zee)} = \Ex{\alpha \sim \calD_M}{\mu_1(\alpha)}$, and $\tildeEx{\mu_2(\zee)} =
\Ex{\alpha\sim \calD_M}{\mu_2(\alpha)}$, which implies $\tildeEx{\mu_1(\zee)} \leq \tildeEx{\mu_2(\zee)}$.
%
\end{proof}
The previous claim allows us to replace any local function inside $\tildeEx{\cdot}$ by another local function that dominates it. We will make extensive use of this fact.
\vspace{-5 pt}
\paragraph{Covariance for SoS solutions}
Given two sets $S,T \sub [m]$ with $|S|,|T|\leq k/4$, we can define the covariance between indicator random variables of $\zee_S$ and $\zee_T$ taking values $\alpha$ and $\beta$ respectively, according to the local distribution over $S \cup T$. This is formalized in the next definition.
\begin{definition}
Let $\tildeEx{\cdot}$ be a pseudodistribution operator of SoS-degree-$t$, and $S,T$ are two sets of
size at most $t/4$, and $\alpha\in [q]^S$, $\beta\in [q]^T$, we define the pseudo-covariance and
pseudo-variance,
\ifnum\confversion=1
\small
\begin{gather*}
\tildecov(\zee_{S,\alpha},\zee_{T,\beta}) 
= \tildeEx{ \zee_{S,\alpha} \cdot \zee_{T,\beta} } - \tildeEx{\zee_{S,\alpha}} \tildeEx{\zee_{T,\beta}} \\	
\tildeVar{\zee_{S,\alpha}} ~=~ \tildecov(\zee_{S,\alpha},\zee_{S,\alpha})
\end{gather*}
\normalsize
\else
\begin{gather*}
\tildecov(\zee_{S,\alpha},\zee_{T,\beta}) 
~=~ \tildeEx{ \zee_{S,\alpha} \cdot \zee_{T,\beta} } - \tildeEx{\zee_{S,\alpha}} \tildeEx{\zee_{T,\beta}} \\
\tildeVar{\zee_{S,\alpha}} ~=~ \tildecov(\zee_{S,\alpha},\zee_{S,\alpha})
\end{gather*}
\fi
The above definition is extended to pseudo-covariance and pseudo-variance for pairs of sets $S,T$, 
as the sum of absolute value of pseudo-covariance for all pairs $\alpha,\beta$ :
\ifnum\confversion=1
\begin{gather*}
\tildecov(\zee_S,\zee_T) 
~=~ \sum_{\alpha\in [q]^S \atop \beta\in [q]^T} \abs{ \tildecov(\zee_{S,\alpha},\zee_{T,\beta}) } \\
\tildeVar{\zee_S} ~=~ \sum_{\alpha\in [q]^S} \abs{ \tildeVar{\zee_{S,\alpha} } }
\end{gather*}
\else
\[
\tildecov(\zee_S,\zee_T) 
~=~ \sum_{\substack{\alpha\in [q]^S \\ \beta\in [q]^T}} \abs{ \tildecov(\zee_{S,\alpha},\zee_{T,\beta}) }
\qquad \text{and} \qquad		
\tildeVar{\zee_S} ~=~ \sum_{\alpha\in [q]^S} \abs{ \tildeVar{\zee_{S,\alpha} } }
\]
\fi
\end{definition}

%
We will need the fact that $\tildeVar{\zee_S}$ is bounded above by 1, since,
\begin{align*}
\tildeVar{\zee_S} 
&~=~ \sum_{\alpha} \abs{\tildeVar{\zee_{S,\alpha}}} \\
&~=~ \sum_{\alpha}\inparen{
          \tildeEx{\zee_{S,\alpha}^2} - \tildeEx{\zee_{S,\alpha}}^2} \\
&~\leq~ \sum_{\alpha} \tildeEx{\zee_{S,\alpha}^2} \\
&~=~ \sum_{\alpha} \tildeEx{\zee_{S,\alpha}} 
~=~ 1
\end{align*}

%
\vspace{-5 pt}
\paragraph{Conditioning SoS solutions.}
We will also make use of conditioned pseudoexpectation operators, which may be defined in a way
similar to usual conditioning for true expectation operators, as long as the event we condition on
is local. 
The conditioned SoS solution is of a smaller degree, but continues to respect the constraints that original solution respects.

\begin{definition} Let $F \subseteq [q]^m$ be subset (to be thought of as an event) such that $\one_F:[q]^m \rightarrow \{0,1\}$ is a $k$-local function. Then for every $t>2k$, we can condition a pseudoexpectation operator of SoS-degree $t$ on $F$ to obtain a new conditioned pseudoexpectation operator $\condPE{\cdot}{E}$ of SoS-degree $t-2k$, as long as $\tildeEx{\one^2_F(\zee)}>0$. The conditioned SoS solution is given by
\[
	\condPE{ p(\zee)}{F(\zee) } \defeq \frac{\tildeEx{p(\zee) \cdot \one^2_{F}(\zee)}}{\tildeEx{\one^2_{F}(\zee)}}
\]
where $p$ is any polynomial of degree at most $t-2k$.
\end{definition}

We can also define pseudocovariances and pseudo-variances for the conditioned SoS solutions.
\begin{definition}
	Let $F\sub [q]^m$ be an event such that $\one_F$ is $k$-local, and let $\tildeEx{\cdot}$ be a pseudoexpectation operator of degree $t$, with $t>2k$. Let $S,T$ be two sets of size at most $\frac{t-2k}{2}$ each. Then the pseudocovariance between $\zee_{S,\alpha}$ and $\zee_{T,\beta}$ for the solution conditioned on event $F$ is defined as,
\ifnum\confversion=1
\begin{multline*}
\tildecov(\zee_{S,\alpha},\zee_{T,\beta} \vert F) = \\
\tildeEx{ \zee_{S,\alpha} \zee_{T,\beta} \vert F} - \tildeEx{\zee_{S,\alpha} \vert F} \tildeEx{\zee_{T,\beta} \vert F}
\end{multline*}
\else
\begin{align*}
\tildecov(\zee_{S,\alpha},\zee_{T,\beta} \vert F) 
~=~ \tildeEx{ \zee_{S,\alpha} \zee_{T,\beta} \vert F} - \tildeEx{\zee_{S,\alpha} \vert F} \tildeEx{\zee_{T,\beta} \vert F}
\end{align*}
\fi
\end{definition}

We also define the pseudocovariance between $\zee_{S,\alpha}$ and $\zee_{T,\beta}$ after
conditioning on a random assignment for some $\zee_V$ with $V\sub [m]$. 
Note that here the random assignment for $\zee_V$ is chosen according to the local distribution for
the set $V$.

\begin{definition}[Pseudocovariance for conditioned pseudoexpectation operators]
\ifnum\confversion=1
\begin{multline*}
\tildecov(\zee_{S,\alpha},\zee_{T,\beta} \vert \zee_V) 
= \\ \sum_{\gamma \in [q]^V}   \tildecov(\zee_{S,\alpha},\zee_{T,\beta} \vert \zee_V = \gamma) \cdot \tildeEx{\zee_{V,\gamma}}
	\end{multline*}
\else
\begin{align*}
\tildecov(\zee_{S,\alpha},\zee_{T,\beta} \vert \zee_V) 
~=~ \sum_{\gamma \in [q]^V}   \tildecov(\zee_{S,\alpha},\zee_{T,\beta} \vert \zee_V = \gamma) \cdot \tildeEx{\zee_{V,\gamma}}
\end{align*}
\fi
\end{definition}

And we likewise define $\tildeVar{\zee_{S,\alpha} \vert \zee_V}$, $\tildecov(\zee_S, \zee_T \vert \zee_V)$ and $\tildeVar{\zee_S \vert \zee_V}$.

\section{SoS relaxations for codes}\label{sec:sos_and_codes}
For both Tanner codes and AEL codes, we will identify $[m]$ with $E$ so that the SoS solutions will be relaxations to the assignments to edges of a bipartite $(n,d,\lambda)$-expander.

\paragraph{Pseudocodewords for Tanner Codes}

Let $G(L,R,E)$ be the bipartite $(n,d,\lambda)$-expander on which the Tanner code is defined, and let $\calC_0\subseteq [q]^d$ be the inner code. The SoS variables will be $\zee = \{Z_{e,j}\}_{e\in E,j\in [q]}$. 

\begin{definition}[Tanner Pseudocodewords]
For $t\geq 2d$, we define a degree-$t$ Tanner pseuocodeword to be a degree-$t$ pseudoexpectation operator $\tildeEx{\cdot}$ on $\zee$ respecting the following constraints:
\ifnum\confversion=1
\begin{align*}
&\forall \li \in L, \quad \zee_{N_L(\li)} \in \calC_0 \\
\text{and } \quad &\forall \ri \in R,\quad \zee_{N_R(\ri)} \in\calC_0
\end{align*}
\else
\begin{align*}
\forall \li \in L, \quad \zee_{N_L(\li)} \in \calC_0 
\qquad \text{and} \qquad          
\forall \ri \in R,\quad \zee_{N_R(\ri)} \in\calC_0
\end{align*}
\fi
Again, since these constraints are $d$-local, it is sufficient to enforce that certain degree-$d$ polynomials are zero (respected by pseudoexpectation) to enforce these constraints. 
In particular, each parity check in the parity check matrix of $\calC_0$ will correspond to a monomial of size at most $d$ that we can enforce to be equal to 1.
\end{definition}

We can also define a generalization of distance between two codewords to include pseudocodewords.
\begin{definition}[Distance from a pseudocodeword]
The distance between a pseudocodeword $\tildeEx{\cdot}$ of SoS-degree $t\geq 2d$ and a codeword $h$ of $\TanC$ is defined as
\ifnum\confversion=1
\begin{align*}
\dis(\tildeEx{\cdot},h) 
~\defeq~& \Ex{e\in E}{\tildeEx{\indi{\zee_e \neq h_e}}} \\
~=~& \Ex{\li\in L}{\tildeEx{\dis(\zee_{N_L(\li)},  h_{N_L(\li)})}}
\end{align*}
\else
\[
\dis(\tildeEx{\cdot},h) 
~\defeq~ \Ex{e\in E}{\tildeEx{\indi{\zee_e \neq h_e}}} 
~=~ \Ex{\li\in L}{\tildeEx{\dis(\zee_{N_L(\li)},  h_{N_L(\li)})}}
\]
\fi
\end{definition}

\paragraph{Pseudocodewords for AEL}

Let $G(L,R,E)$ be the bipartite $(n,d,\lambda)$-expander on which the AEL code is defined, and let $\calC_0\subseteq [q_0]^d$ be the inner code. The SoS variables will be $\zee = \{Z_{e,j}\}_{e\in E,j\in [q_0]}$. 

\begin{definition}[AEL Pseudocodewords]
	For $t\geq 2d$, we define a degree-$t$ AEL pseuocodeword to be a degree-$t$ pseudoexpectation operator $\tildeEx{\cdot}$ on $\zee$ respecting the following constraints:
	\begin{align*}
		\forall \li \in L,\quad \zee_{N_L(\li)} \in \calC_0
	\end{align*}
\end{definition}

Next we define the distances between a pseudocodeword and a codeword of $\AELC$. 
\begin{definition}[Distance from a pseudocodeword]
	The left, middle and right distances between a pseudocodeword $\tildeEx{\cdot}$ of SoS-degree $t\geq 2d$ and a codeword $h \in \AELC$ are defined as
	\begin{align*}
		\dis^L(\tildeEx{\cdot},h) &\defeq \Ex{\li}{\tildeEx{\indi{\zee_{N_L(\li)} \neq h_{N_L(\li)}}}} \\
	\dis(\tildeEx{\cdot},h) &\defeq \Ex{e}{\tildeEx{\indi{\zee_e \neq h_e}}} \\
		\dis^R(\tildeEx{\cdot},h) &\defeq \Ex{\ri\in R}{\tildeEx{\indi{\zee_{N_R(\ri)}\neq  h_{N_R(\ri)}}}}
	\end{align*}
\end{definition}
%
%


\chapter{A General Framework for List Decoding}\label{chap:framework}
\label{sec:intro}
%
%
%
%
%

Expander graphs have been a powerful tool for the construction of codes with several interesting
properties, and a variety of applications.
A (very) small subsample of the list of applications already includes the seminal
constructions of expander codes~\cite{SS96, Zemor01}, widely used distance amplification
constructions~\cite{ABNNR92, AEL95}, as well as recent breakthrough constructions of
$\epsilon$-balanced codes~\cite{TS17}, locally testable codes~\cite{DELLM22}, and quantum LDPC
codes~\cite{PK22, LZ22}. 
A detailed account of the rich interactions between coding theory and expander graphs, and
pseudorandom objects in general, can be found in several excellent surveys and
textbooks on these areas~\cite{GSurvey04, Vadhan12, HLW06, GRS23}. 

The combinatorial and spectral structure of codes based on expander graphs often leads to very
efficient algorithms for unique decoding. However, obtaining list decoding for constructions based on
expanders often requires incorporating additional algebraic structure in the construction, to take
advantage of the well-established machinery for list decoding using polynomials~\cite{Guruswami:survey}.
While there are certainly important counterexamples to the above statement, such as the
expander-based codes of Guruswami and Indyk~\cite{GI03} and Ta-Shma~\cite{TS17} which
allow for list decoding, we know of few \emph{general techniques to exploit expansion for
  list decoding.}
In this work, we consider the question of finding techniques for list decoding from errors,
which can work in settings where no algebraic structure may be available, such as the decoding of LDPC
codes constructed from expander graphs.

Building on the significant body of work for LP decoding of expander codes~\cite{Feldman03, FWK05}, 
we consider the question of decoding as an optimization problem, which can be approached
via convex relaxations. 
We show that stronger relaxations obtained via the Sum-of-Squares (SoS) hierarchy of semidefinite
programs, can in fact be used to obtain list decoding algorithms for several code constructions
based on expanders. 
These hierarchies can be viewed as proof systems~\cite{FKP19}, with relaxations at a level $t$ of
the hierarchy corresponding to proofs which can be carried out by reasoning about sum-of-squares of
polynomials of degree at most $t$ in the optimization variables. The proof system corresponding to a
small number of levels of the SoS hierarchy turns out to be powerful enough
to capture the distance proofs for several expander-based codes, when the proofs rely on spectral
properties of expander graphs.
Combined with generic ``covering lemmas'' which ensure that the solutions to these relaxations
do not completely ignore any codeword in the list, these can be used to design
list decoding algorithms for several families of codes based on expanders, up to the Johnson bound
where the list size is known to be bounded.

\section{Our Results}
\vspace{-5 pt}
\subsection{Tanner Codes}
Low-Density Parity Check (LDPC) codes were introduced by a foundational work of
Gallager~\cite{Gallager62} and graph-based constructions were obtained by Tanner~\cite{Tanner81}.
Sipser and Spielman~\cite{SS96} gave the first constructions of Tanner codes with distance bounds
based on the expansion of the graph, which also admitted a linear time (unique) decoding algorithm. 
An elegant construction based on bipartite spectral expanders,  with particularly simple
(linear-time) unique-decoding algorithms, was given by \Zemor~\cite{Zemor01}.
Variants of these constructions have led to applications~\cite{RU:book} and have also been used as
building blocks in the recent constructions of locally testable codes by Dinur \etal~\cite{DELLM22}
and quantum LDPC codes by Panteleev and Kalachev~\cite{PK22} (see also \cite{LZ22}).

There exist highly efficient algorithms for the unique-decoding of these codes from both probabilistic
and adversarial errors, based on combinatorial arguments, linear programming
relaxations~\cite{Feldman03, ADS12, FWK05} and message passing
algorithms~\cite{Guruswami:MP-survey, RU:book}. 
In the setting of \emph{erasures} where the location of the corruptions in the transmitted codeword is known,
recent work has also led to linear-time list decoding algorithms~\cite{RZWZ21, HW15}, which also
work for the more general task of list recovery in the large alphabet (high-rate) case~\cite{HW15}.
However, to the best of our knowledge, no list decoding algorithms are known in the more challenging
(and common) setting of \emph{errors} when the location of the corruptions are unknown, even though
random ensembles of LDPC codes are even known to combinatorially achieve list-decoding
capacity~\cite{MosheiffRRSW19}, and thus have bounded list sizes up to optimal error radii.

We show that relaxations obtained via the SoS hierarchy can be used list-decode \Zemor's
construction of Tanner codes~\cite{Zemor01}, up to the Johnson bound (which is an error-radius where
list sizes are always known to be bounded). 
Our proof technique can also be extended to work for other constructions of Tanner codes where the
proof for the distance of the code is based on spectral arguments, but is easiest to illustrate in
the context of \Zemor's construction. 
We briefly recall the construction before describing our result.

Given a bipartite $d$-regular graph $G=(L,R,E)$ with $\abs{L}=\abs{R}=n$, the Tanner code
is of blocklength $m = \abs{E} = nd$. 
Given an alphabet size $q$, the code consists of all edge-labelings $f \in [q]^m$, such that the
labels in the neighborhood of every vertex
\footnote{One can also consider variants where the base code $\calC_0$ is different for different
  vertices, but this does not make a difference for our purposes.}, 
belong to a ``base code'' $\calC_0 \subseteq [q]^d$.
When the base code has (fractional) distance $\delta_0$ and $G$ has (normalized) second singular
value at most $\lambda$ for the biadjacency matrix, the distance of the Tanner code is known to be at least $\delta = \delta_0 \cdot
(\delta_0 - \lambda)$. 
The Johnson bound for distance $\delta$ and alphabet size $q$ is defined as
$\calJ_q(\delta) \defeq (1 - \nfrac{1}{q}) \cdot \inparen{ 1 - \inparen{1 - \nfrac{q \cdot
      \delta}{(q-1)}}^{1/2}}$, and is always greater than the unique decoding radius $\delta/2$.
\ifnum\confversion=1
We prove the following (see full version for a formal statement).
\else
We prove the following.
\fi

\begin{theorem}\ifnum\confversion=0[Informal version of \cref{thm:tanner-decoding}]\fi
\ Given a Tanner code $\calC$ as above and $\eps > 0$, there is a deterministic algorithm based on
$q^{O(d)}/\eps^4$ levels of the SoS hierarchy, which given an arbitrary $g \in [q]^m$, runs in time
$n^{q^{O(d)}/\eps^4}$, and recovers the list of codewords within distance
$\calJ_q(\delta) - \eps$ ~of $g$. 
\end{theorem}
Note that one can think of $q, d$ in \Zemor's construction and $\eps$ above as constants (in fact
$d$ is required to be constant for LDPC codes), in which case the above running time is polynomial in
$n$. 
Of course, these running times are no match for the linear-time unique decoding, and erasure
list-decoding, algorithms available for these codes. 
However, we view the techniques used in the proof of the above algorithm as a first step towards
identifying the right structures, and designing truly efficient algorithms, to take advantage of
expansion for list-decoding of LDPC codes (as has proved to be the case for several SoS-based algorithms in the past). 

Our techniques also extend to yield a similar statement for the recent construction of locally
testable codes by Dinur \etal~\cite{DELLM22}, which are Tanner codes on a different structure
called a ``square Cayley complex''.
These are constructed using a group $H$ and two generator sets $A,B \subseteq H$, with each
generator set individually defining an expanding Cayley graph on $H$ (with second eigenvalue bounded
by $\lambda$).  The sizes of the generator
sets (equal to the graph degree) are taken as constant, say $\abs{A} = \abs{B} =  d$. The
construction relies on base codes $\calC_A, \calC_B \subseteq [q]^d$, with distances, say $\delta_A$
and $\delta_B$, and is known to have distance at least $\delta = \delta_A \cdot \delta_B \cdot
\inparen{\max\{\delta_A, \delta_B\} - \lambda}$.
\begin{theorem}\ifnum\confversion=0[Informal version of \cref{thm:square-decoding}]\fi
\ Given a code  $\calC^{SCC}$ with block length $m$ and alphabet $[q]$, 
supported on a square Cayley complex as described above, and $\eps > 0$, there is a deterministic
algorithm based on $q^{O(d^2)}/\eps^{O(1)}$ levels of the SoS hierarchy, which given an arbitrary $g \in [q]^m$, 
recovers the list of codewords within distance $\calJ_q(\delta) - \eps$ ~of $g$. 
\end{theorem}
\vspace{-5 pt}
\subsection{Distance Amplified Codes}
The proofs for the distance of the above Tanner codes, are also very similar to the ones used for
analyzing the distance amplification procedure of Alon, Edmonds, and Luby~\cite{AEL95} (AEL), based on
expander graphs.
While there are several variants of this construction discussed in the literature, we 
will discuss a version of the AEL construction~\cite{Kopparty} which is particularly close to the
Tanner code construction of \Zemor.
Given a $d$-regular bipartite graph $G=(L,R,E)$ with second singular value $\lambda$, an ``outer" code $\calC_1 \subseteq [q_1]^n$ with distance $\delta_1$, and an ``inner" code $\calC_0 \subseteq [q]^d$ with $\abs{\calC_0} = q_1$, the AEL procedure constructs a new code $\calC^{AEL} \subseteq [q^d]^n$ with distance at least $\delta_0 - \frac{\lambda}{\delta_1}$. 
Thus, it yields constructions with arbitrarily large block lengths that inherit the parameters of the small inner code. 

The AEL procedure has been used as an important ingredient for obtaining optimal rate-distance
tradeoffs in several constructions, such as the capacity-achieving list decodable codes by Guruswami
and Rudra~\cite{GR08}. The amplification is achieved via a simple redistribution of symbols using the expander, and the construction also preserves several interesting local properties of the outer code, such as the property of being LDPC, or locally testable, or locally correctable~\cite{KMRZS17, GKORZS18}. 
We refer the reader to the discussion in \cite{KMRZS17} for an excellent account of the applications and properties of the AEL construction.

The AEL procedure has been used to construct several list decodable codes, including some of the
results cited above, and a quantum analogue of the construction was also used recently by Bergamaschi \etal~\cite{BGG22} to obtain quantum codes meeting the Singleton bound (via quantum list decoding). We will see more about the quantum extension of AEL in the next chapter.

However, for the resulting code $\calC^{AEL}$ to be list-decodable, one often needs to assume stronger properties such as list-recovery for the outer code $\calC_1$. 
Since these stronger properties may not always be available (for example, when one wants to preserve some local properties for $C_1$ like being LDPC), we again consider the question of finding techniques which can allow for list-decoding $\calC^{AEL}$ for expanding graphs $G$, without relying on additional structure from $\calC_1$.

We show that relaxations based on the SoS hierarchy, can be used to list decode the distance-amplified code $\calC^{AEL}$, even when the outer code is only assumed to be \emph{unique decodable}. In particular, we prove the following result\ifnum\confversion=1 \ (a formal statement appears in full version)\fi:
\begin{theorem}\ifnum\confversion=0[Informal version of \cref{thm:list_decoding_ael}]\fi
\ Let $\calC^{AEL}$ be a distance-amplified code as above, with the outer code $\calC_1$ taken to be unique decodable from radius $\delta_{dec}$ in  time $T(n)$, and let $\delta = \delta_0 - \frac{\lambda}{\delta_{dec}}$. Then, for every $\eps > 0$,
there is a deterministic algorithm based on $q^{O(d)}/\eps^4$ levels of the SoS hierarchy, which
given an arbitrary $g \in [q^d]^n$, runs in time
$n^{q^{O(d)}/\eps^4} + O(T(n))$, and recovers the list of codewords within distance
$\calJ_{q^d}(\delta) - \eps$ ~of $g$.
\end{theorem}
\vspace{-5 pt}
We note that the decoding radius for the above algorithm is $\calJ\inparen{\delta_0 - \frac{\lambda}{\delta_{dec}}}$ instead of $\calJ\inparen{\delta_0 - \frac{\lambda}{\delta_{1}}}$, which would be the Johnson bound for the true distance of the code $\calC^{AEL}$. 
However, in applications of AEL, one often chooses parameters so that the distance of code is about
$\delta_0$, and the effect of the second term is minimized by choosing a small $\lambda$. 
When $\calC_1$ is known to unique-decodable up to a smaller radius $\delta_{dec}$, one can still obtain list-decodable codes up to (nearly) the Johnson bound by choosing $G$ to be a sufficiently good expander (with small $\lambda$).

\vspace{-5 pt}
\subsection{Techniques}
As mentioned earlier, our techniques are based on using the Sum-of-Squares hierarchy of convex
relaxations for an optimization problem related to the decoding problem. 
In unique-decoding algorithms based on the LP relaxations, the optimization objective is to find the
closest codeword to a given received word, and the correctness of the decoding procedure often
relies on the LP being integral for an appropriate range of parameters.
In contrast, algorithms for list-decoding actually need to ensure that the solution to the convex
relaxation has sufficient information about all codewords in the list, and so it is important that
the solution is \emph{not} integral but rather a ``maximally-convex'' combination, covering all of
the list elements. 
This can be ensured by statements which we call ``covering lemmas'', which are discussed in more detail in \ifnum\confversion=1 \cref{sec:overview}. \else \cref{sec:covering}. \fi
The proofs for the covering lemmas are based on the techniques from~\cite{AJQST20}, where these were
used for the list-decoding of direct-sum codes.

A second key component of our proof, which makes the SoS hierarchy particularly appealing to work
with, is for the relaxation to be able to capture global properties of the code, such as the
distance. 
While local properties of the code, such as the structures of the base/inner codes are enforced
through explicit constraints included in the relaxation, the global property of distance is a
nontrivial consequence of these constraints.
However, the \emph{proofs} of these distance properties are spectral in nature, for the codes we consider
here, which makes them discoverable by the SoS hierarchy. This key idea is due to the work of Richelson and Roy \cite{RR23}.

For example, the proofs rely on statements such as the expander mixing lemma, which can viewed as a
consequence of statements like $\ip{f}{A f} \leq \lambda \cdot \norm{f}^2$ when $\ip{f}{1} = 0$, and
$A$ is the (normalized) adjacency matrix of a graph with second eigenvalue at most $\lambda$. 
Taking $\Pi$ to be the projector to the space orthogonal to the all-ones vector, we can re-write the
above inequality as $\ip{f}{(\lambda \cdot \Pi - \Pi A \Pi)f} \geq 0$. 
However, note that the matrix $\lambda \cdot \Pi - \Pi A \Pi$ is actually a positive semidefinite
matrix, which means that the expression $\ip{f}{(\lambda \cdot \Pi - \Pi A \Pi)f}$ is a
\emph{sum-of-squares} of linear forms in the entries of $f$. 
The SoS hierarchy can be viewed as a proof system, where a solution to the level-$2t$ relaxation
can be seen as satisfying all inequalities which can be derived using sum-of-squares of polynomials
of degree at most $t$. 
We can show that this means that the solutions (after some modification) satisfy some codeword-like
properties, using which it is possible to appeal to a \emph{unique} decoding algorithm to recover
one element from the list from one such ``codeword-like'' SoS solution.
\ifnum\confversion=1
See the full version for a formal statement about these codeword-like ``distance certificates'' for SoS solutions.
\else
These codeword-like ``distance certificates'' for SoS solutions are developed in \cref{sec:distance}.
\fi

Broadly speaking, our techniques can be seen as part of the ``Proofs to Algorithms'' paradigm based
on the Sum-of-Squares method~\cite{FKP19}. 
Our covering lemmas for SoS relaxations yield  a generic framework for converting SoS proofs of
distance for \emph{any code}, to list decoding algorithms which work up to the Johnson bound. 
This framework can also capture the results by Richelson and Roy~\cite{RR23} for list
decoding Ta-Shma's codes up to Johnson bound. However, we omit this proof since Ta-Shma's construction is quite technical 
and the ideas are more easily illustrated for Tanner and AEL codes. In \cref{chap:regularity}, we will see a different way of decoding Ta-Shma's codes that suffices for unique decoding and runs in near-linear time.

As mentioned in \cref{chap:intro}, list decoding algorithms often need to rely on algebraic structure, and are thus particularly
well suited to work with large alphabets (fields). One then obtains algorithms for small-alphabet
codes via techniques such as concatenation and list recovery.
On the other hand, the techniques based on convex relaxations discussed above seem to work well
directly over small alphabets.
\vspace{\pspacing}
\subsection{Related work}
In terms of techniques, the works most directly related to ours are those using similar SoS
relaxations for list-decoding of Ta-Shma's codes~\cite{AJQST20, JQST20, RR23}. 
In particular, the proofs of the covering lemmas follow the approach of Alev \etal~\cite{AJQST20},
and idea of viewing the proof of distance as implementable in the SoS hierarchy was also used by
 Richelson and Roy~\cite{RR23}. 
A precursor to much of this research on list-decoding, 
is the result of Dinur \etal~\cite{DHKLNTS19}, which suggested
the approach of using semidefinite programming and expansion for list decoding of codes obtained via
an earlier distance amplification procedure of Alon \etal~\cite{ABNNR92}, which can be seen as a
special case of the AEL distance amplification.

Another important work, related to the use of SoS hierarchy for decoding LDPC codes, is the
\emph{lower bound} of Ghazi and Lee~\cite{GL18} for using the SoS hierarchy to decode random LDPC
codes. 
However, the lower bound shows that the relaxation for finding the optimal (closest) codeword may
have value much better than the true optimal codeword, when the decoding radius is larger than that
of LP decoding, thus showing that the SoS relaxations may not be integral. 
On the other hand, the relaxations we use do not optimize for the closest codeword, but rather go
through covering lemmas.
A recent work of Chen \etal~\cite{CCLO22} also shows significantly improved distance bounds, and
improved unique-decoding bounds for the expander codes of Sipser and Spielman~\cite{SS96}. 
Our results do not apply for the codes considered in their work in a black box fashion, since the analysis is based on
lossless \emph{vertex expansion}, for which we do not always know of spectral certificates.

Our work can also be seen as obtaining ``sparse'' analogues of the results of Gopalan, Guruswami, and
Raghavendra~\cite{GGR09} for the list-decoding of tensor and interleaved codes, which can be viewed
as replacing the bipartite expanders in \Zemor and AEL constructions respectively, by a complete bipartite graph.
Indeed, viewing AEL as a sparse analog of interleaving will play a key role in \cref{chap:capacity}.
%
%
%
\section{A Technical Overview}
\label{sec:overview}

Our proof can be viewed as an algorithmic implementation of the proof of Johnson bound, and the
proofs of distance, for the relevant codes.
We start with an overview of the proof for Tanner codes. The argument is very similar for all codes
considered here, substituting an appropriate proof of distance in each case. 
\vspace{-5 pt}
\paragraph{Johnson bound via covering lemmas.}
We first prove the Johnson bound via a statement we will call a ``covering lemma'', which can then
be generalized to work with convex relaxations.
Given an $[m, \delta, \rho]_q$ code $\calC$ and a received word $g \in [q]^m$, our goal is to
bound the list of radius $\calL(g,\calJ_q(\delta)) = \inbraces{h\in \calC \suchthat \dis(h,g) <
  \calJ_q(\delta) }$, where $\calJ_q(\delta)$ denotes the Johnson bound. 
It will be convenient to work with $\beta = 1 -  \frac{q \cdot \delta}{q-1}$, and the Johnson bound can be
defined using the equation
\[
\inparen{1 - \frac{q \cdot \calJ_q(\delta)}{q-1}}^2 ~=~ \inparen{1 - \frac{q \cdot \delta}{q-1}} ~=~ \beta \mper
\]
Also, we can map elements of $[q]$ to corners $u_1, \ldots, u_q$ of the simplex in $\R^{q-1}$, which are unit
vectors satisfying $\ip{u_i}{u_j} = \nfrac{-1}{(q-1)}$ for $i \neq j$. Applying
this map, say $\chi$, pointwise to $g, h \in [q]^m$, gives
\ifnum\confversion=1
\begin{align*}
\ip{\chi(g)}{\chi(h)} ~=~& \Ex{i \in [m]}{\ip{\chi\inparen{g(i)}}{\chi\inparen{h(i)}}} \\
 ~=~& 1 - \Delta(g,h) - \frac{\Delta(g,h)}{q-1} \\
 ~=~& 1 - \frac{q \cdot \Delta(g,h)}{q-1} \mper
\end{align*}
\else
\[
\ip{\chi(g)}{\chi(h)} ~=~ \Ex{i \in [m]}{\ip{\chi\inparen{g(i)}}{\chi\inparen{h(i)}}} ~=~ 1 - \Delta(g,h) - \frac{\Delta(g,h)}{q-1} ~=~ 1 - \frac{q \cdot
  \Delta(g,h)}{q-1} \mper
\]
\fi
Thus, given $g$ and $\beta$ as above, we can write the list as 
$\calL = \inbraces{h\in \calC \suchthat \ip{\chi(g)}{\chi(h)} > \sqrt{\beta}}$. 
The covering lemma \ifnum\confversion=0 in \cref{sec:covering} \fi shows that given a set $\calF$ of unit vectors 
in an inner-product space, and a unit vector $\tilde{g}$ satisfying $\ip{\tilde{g}}{f} > \sqrt{\beta}$ for
all $f \in \calF$, there exists $g_0$ in the \emph{convex hull} $\conv{\calF}$ satisfying
$\ip{g_0}{f} > \beta$ for all $f \in \calF$. \ifnum\confversion=1 (See the full version for a proof based on minimizing $\ell_2$ norm of $g_0$ subject to $g_0 \in \conv{\calF}$.)\fi

Instantiating this with $\tilde{g} = \chi(g)$ and $\calF = \inbraces{\chi(h) \suchthat h \in
  \calL}$ gives $g_0 \in \conv{\calF}$ satisfying $\ip{g_0}{\chi(h)} > \beta$ for all $h
\in \calL$. 
Additionally, $g_0$ can be chosen to be supported on at most $m \cdot (q-1) +1$ elements of $\calF$ via
\Caratheodory theorem.
Since $\ip{\chi(h_1)}{\chi(h_2)} \leq \beta$ for any $h_1 \neq h_2$ in $\calC$, any $h \in \calF
\setminus \supp(g_0)$ will satisfy $\ip{g_0}{\chi(h)} \leq \beta$, which is a contradiction. 
Thus, we must have $\calL \subseteq \supp(g_0)$ implying $\abs{\calL} \leq (q-1) \cdot m + 1$.
While this is a weaker bound on the list size (but can be independent of $m$ via approximate \Caratheodory
theorems), each step of the above proof can be extended to work well with convex relaxations.
\vspace{-5 pt}
\paragraph{Algorithmic covering lemmas via SoS.}
Note that in the above proof, it suffices to have $g_0$ in the convex hull of \emph{all} codewords
\ie $g_0 \in \conv{\chi(h) \suchthat h \in \calC}$, instead of only the codewords in $\calL$.
Since the convex hull of all the codewords is still a difficult set to optimize over, we instead
work with degree-$t$ \emph{pseudocodewords} defined as solutions to an SoS relaxation of degree-$t$,
respecting certain local constraints corresponding to the code (see \cref{chap:prelims}).

Moreover, the covering lemma used above can be proved by finding the $g_0 \in \conv{\calF}$, which
minimizes $\norm{g_0}$ while satisfying $\ip{g_0}{g} > \sqrt{\beta}$.
Analogously, we consider solutions to the SoS relaxation, given as pseudocodewords $\tildeEx{\cdot}$
satisfying $\ip{\tildeEx{\chi(\zee)}}{\chi(g)} > \sqrt{\beta}$ and minimizing
$\norm{\tildeEx{\chi(\zee)}}$. Note that here $\chi(\zee)$ is a (vector-valued) 1-local function, and
thus the vector $\tildeEx{\chi(\zee)}$ is well defined.
A similar relaxation was also used for list-decoding of direct-sum codes in \cite{AJQST20}.
\vspace{-5 pt}
\paragraph{Choosing from the ``support'' via conditioning.}
The next step of the proof of Johnson bound can be seen as saying that for all $h_1$ such that
$\chi(h_1) \in \supp(g_0)$ and $h_2 \in \calL$, we must have $\ip{\chi(h_1)}{\chi(h_2)} = 1$ (when
they are equal) or $\ip{\chi(h_1)}{\chi(h_2)} \leq \beta$ (when they are distinct codewords).

We will do this in two steps. The first is to develop a good proxy for ``$\chi(h_1) \in \supp(g_0)$''
since we are working with the pseudocodeword $\tildeEx{\cdot}$, which is not a
convex combination of codewords.
Instead, it suffices to look at pseudocodewords which have small
pseudo-covariance across a typical pair of left-right vertices in the bipartite graph defining the Tanner code \ie
\[
\Ex{\substack{\ell \in L \\ r \in R}}{\tildecov(\zee_{N_L(\ell)},\zee_{N_R(r)})} ~\leq~ \eta \mper
\]
Note that when $\tildeEx{\cdot}$ corresponds to an actual codeword (or any
integral solution) the covariance will be 0. Thus, the above is a weakening of the notion of
``vertex of the convex hull''.  We call such a solution, an \emph{$\eta$-good} pseudocodeword.
A similar definition was also used by Richelson and Roy~\cite{RR23} in their list-decoding algorithm
for Ta-Shma's codes~\cite{TS17}, and also in earlier works~\cite{AJQST20, JQST20}.

An argument of Barak, Raghavendra, and Steurer~\cite{BRS11} shows that conditioning the starting SoS
solution on the values few randomly chosen variables, leads to an $\eta$-good solution. 
In fact, one either has small covariance in the sense above, or conditioning on the variables in
$N_R(r)$ for a random $r \in R$ reduces the average variance (seen from the left) by
$\Omega_{q,d}(\eta^2)$. Since the (pseudo-)variance is non-negative, this process terminates,
yielding an $\eta$-good solution.
The argument can also be made deterministic by enumerating over all (constant-sized) subsets to
condition on.

\vspace{-5 pt}
\paragraph{Proof of distance for good pseudocodewords.}
Given an $\eta$-good pseudocodeword $\tildeEx{\cdot}$ as defined above, we can now prove that for
any $h \in \calL$, we must have 
\ifnum\confversion=1
\begin{align*}
\ip{\tildeEx{\chi(\zee)}}{\chi(h)} &~\geq~ 1 - O(\eta) \\
\text{or} \qquad
\ip{\tildeEx{\chi(\zee)}}{\chi(h)} &~\leq~ \beta + O(\eta) \mper
\end{align*}
\else
\[
\ip{\tildeEx{\chi(\zee)}}{\chi(h)} ~\geq~ 1 - O(\eta) \qquad \text{or} \qquad
\ip{\tildeEx{\chi(\zee)}}{\chi(h)} ~\leq~ \beta + O(\eta) \mper
\]
\fi
%
%
To argue the above dichotomy, we will switch to the distances instead of inner products, as the
proof closely follows the distance proof of Tanner codes. Let $\inbraces{X_{\ell}}_{\ell \in L}$ and
$\inbraces{Y_r}_{r \in R}$ be ensembles of $d$-local functions defined as 
$X_{\ell}(\zee) = \one\inbraces{\zee_{N_L(\ell)} \neq h_{N_{L}(\ell)}}$ 
and 
$Y_{r}(\zee) = \one\inbraces{\zee_{N_R(r)} \neq h_{N_{R}(r)}}$. Using the fact that the distance of
the base code $\calC_0$ is at least $\delta_0$, and the fact that the local constraints for $\calC_0$
are respected by the pseudocodewords, one can show that 
\ifnum\confversion=1
\begin{align*}
\Delta(\tildeEx{\cdot}, h) 
~\defeq~& \Ex{e \in E}{\tildeEx{\one\inbraces{\zee_e \neq h_e}}} \\
~\geq~& \delta_0 \cdot \Ex{\ell \in L}{\tildeEx{X_{\ell}(\zee)}} \mper
\end{align*}
\else
\[
\Delta(\tildeEx{\cdot}, h) 
~\defeq~ \Ex{e \in E}{\tildeEx{\one\inbraces{\zee_e \neq h_e}}} 
~\geq~ \delta_0 \cdot \Ex{\ell \in L}{\tildeEx{X_{\ell}(\zee)}} \mper
\]
\fi
Taking the geometric mean with a similar inequality for $\inbraces{Y_r}_{r \in R}$ gives
\[
\Delta(\tildeEx{\cdot}, h) 
~\geq~ 
\delta_0 \cdot \inparen{\Ex{ \substack{\ell \in L \\ r \in R}}{\tildeEx{X_{\ell}(\zee)} \cdot \tildeEx{Y_{r}(\zee)}}}^{1/2}
\mper
\]
On the other hand, a variant of the expander mixing lemma for pseudoexpectations (see \ifnum\confversion=1 full version\else \cref{sec:distance}\fi), together with the $\eta$-good property and the simple observation that \[
\indi{\zee_e \neq h_e} \leq X_{\li}(\zee) \cdot Y_{\ri}(\zee), \] gives
\ifnum\confversion=1
\small
\begin{align*}
&\ \Delta(\tildeEx{\cdot}, h) \\
&\leq \Ex{\ell \sim r}{\tildeEx{X_{\ell}(\zee) \cdot Y_r(\zee)}}  \\
&\leq \Ex{\ell \in L \atop r \in R}{\tildeEx{X_{\ell}(\zee) \cdot Y_r(\zee)}} + \\
 &\hspace*{2cm} \lambda \cdot \inparen{\Ex{\ell \in L \atop r \in R}{\tildeEx{X_{\ell}^2(\zee)} \cdot \tildeEx{Y_{r}^2(\zee)}}}^{1/2} \\
&\leq \Ex{\ell \in L \atop r \in R}{\tildeEx{X_{\ell}(\zee)} \cdot \tildeEx{Y_r(\zee)}} + \\
& \hspace*{2cm} \eta + \lambda \cdot \inparen{\Ex{\ell \in L \atop r \in R}{\tildeEx{X_{\ell}(\zee)} \cdot
       \tildeEx{Y_{r}(\zee)}}}^{1/2} \mper
\end{align*}
\normalsize
\else
\begin{align*}
\Delta(\tildeEx{\cdot}, h) 
&~\leq~ \Ex{\ell \sim r}{\tildeEx{X_{\ell}(\zee) \cdot Y_r(\zee)}}  \\
&~\leq~ 
 \Ex{\substack{\ell \in L \\ r \in R}}{\tildeEx{X_{\ell}(\zee) \cdot Y_r(\zee)}} 
~+~ 
\lambda \cdot \inparen{\Ex{\substack{\ell \in L \\ r \in R}}{\tildeEx{X_{\ell}^2(\zee)} \cdot
                                                                         \tildeEx{Y_{r}^2(\zee)}}}^{1/2}
  \\
&~\leq~
 \Ex{\substack{\ell \in L \\ r \in R}}{\tildeEx{X_{\ell}(\zee)} \cdot \tildeEx{Y_r(\zee)}} ~+~ \eta   
~+~ 
\lambda \cdot \inparen{\Ex{ \substack{\ell \in L \\ r \in R}}{\tildeEx{X_{\ell}(\zee)} \cdot
       \tildeEx{Y_{r}(\zee)}}}^{1/2} \mper
\end{align*}
\fi
Note that the last line also used $\tildeEx{X_{\ell}^2(\zee)} = \tildeEx{X_{\ell}(\zee)}$ (and
similarly for $Y_r$).
Comparing the two bounds and solving the resulting quadratic inequality in $\tau =
(\Ex{\ell,r}{\tildeEx{X_{\ell}(\zee)} \cdot \tildeEx{Y_r(\zee)}})^{1/2}$ yields the dichotomy.
\vspace{-5 pt}
\paragraph{Completing the argument.}
Starting the algorithmic covering lemma with parameter $\beta + \eps$ ensures that with positive
(constant) probability, for any fixed $h \in \calL$, the conditioning yields an $\eta$-good
pseudocodeword satisfying $\ip{\tildeEx{\chi(\zee)}}{\chi(h)} \geq \beta + \nfrac{\eps}{2}$. 
By the above dichotomy, we must be in the first case (for sufficiently small $\eta$). 
%
%
A simple averaging argument then shows that considering $h' \in [q]^m$ with $h'(e) = \argmax_{j \in
  [q]}\inbraces{\tildeEx{\indi{\zee_{e}=j}}}$ gives $\Delta(h',h) = O(\eta)$.
Given such an $h'$, the codeword $h$ can be recovered by unique decoding.
%

The above argument can also be made to work for other codes, where the distance proof is a spectral argument
based on expanders, such as the codes obtained via the distance amplification of Alon, Edmonds and
Luby~\cite{AEL95}. 
One can simply use pseudocodewords for the appropriate code, and substitute the corresponding distance
proof in the argument above.


%
\section{Covering Lemma and Johnson Bounds}
\label{sec:covering}
In this section, we will introduce our abstract covering lemma, and then use it to give algorithm-friendly proofs of (known) Johnson bounds \cite{G01} by showing that there is a distribution over codewords that covers the list. Next, we will show that if we are willing to work with distributions over degree-$t$ pseudocodewords (or just pseudocodewords, due to convexity), then such a pseudocodeword may be found in time $n^{\calO(t)}$.

\subsection{Covering Lemma}

\begin{lemma}[Covering Lemma]\label{lem:covering}
	Let $\calH$ be an inner product space, and let $\calF$ be a family of unit vectors in $\calH$. Suppose there exists a $g \in\calH$ of unit norm such that for any $f\in \calF$, $\ip{g}{f} > \eps$. Then, there exists a $g_0 \in \conv{\calF}$ such that for any $f\in\calF$, $\ip{g_0}{f} > \eps^2$.
\end{lemma}

\begin{proof}
Consider the set $T = \inbraces{v\in \conv{\calF} \suchthat \ip{v}{g} >\eps}$, which is non-empty because any $f\in \calF$ also belongs to $T$. Let $g_0 = \argmin_{v\in T} \norm{v}^2$. We will show that $g_0$ must have the property that $\ip{g_0}{f}>\eps^2$ for any $f\in\calF$. Note that $\norm{g_0}\cdot \norm{g}\geq \ip{g_0}{g}>\eps$ means $\norm{g_0}>\eps$.

Suppose not, then there exists an $f\in\calF$ such that $\ip{g_0}{f} \leq \eps^2$. For an $\tee \in[0,1]$ to be chosen later, consider $g_1 = \tee g_0+(1-\tee)f$, which is also in $\conv{\calF}$, and $\ip{g_1}{g} = \tee \cdot \ip{g_0}{g}+(1-\tee)\cdot \ip{f}{g} >\eps$.
\begin{align*}
	\norm{g_1}^2 &~=~ \ip{\tee g_0+(1-\tee)f}{\tee g_0+(1-\tee)f} \\
	&~=~ \tee^2 \cdot \norm{g_0}^2+2\tee(1-\tee) \cdot \ip{g_0}{f}+(1-\tee)^2 \cdot \norm{f}^2\\
	&~\leq~ \tee^2 \cdot \norm{g_0}^2+2\tee(1-\tee)\eps^2+(1-\tee)^2\\
	&~=~\tee^2\cdot (\norm{g_0}^2-\eps^2) +\eps^2 \cdot (\tee+(1-\tee))^2 + (1-\tee)^2 (1-\eps^2)\\
	&~=~\tee^2\cdot (\norm{g_0}^2-\eps^2) +\eps^2 + (1-\tee)^2 \cdot (1-\eps^2)\\
	&~=~\tee^2 \cdot (1-\eps^2 + \norm{g_0}^2-\eps^2) - 2\tee \cdot (1-\eps^2) +1
\end{align*}
At $\tee=1$, this expression equals $\norm{g_0}^2$. The minimum of this quadratic function of $\tee$
is achieved at $\tee= \frac{2(1-\eps^2)}{2(1-\eps^2+\norm{g_0}^2-\eps^2)}$, which is $<1$ as
$\norm{g_0} >\eps$. Thus, we can reduce $\norm{g_1}^2$ further to strictly less than $\norm{g_0}^2$
by choosing $\tee$ to the above value, which contradicts the optimality of $g_0$.
\end{proof}
\subsection{Johnson bounds}

We will prove the standard $q$-ary Johnson bound and the version from \cite{G01} which receives weight $W_{i,j}$ for each $j\in [q]$ for every coordinate $i\in [n]$. First we define some functions to embed the received word $g$ (or received weights) in $\R^{(q-1)n}$, which will be the inner product space where we apply the covering lemma.

\begin{definition}
Fix $q \in \N$. We denote by $\chi_q: [q]\rightarrow \R^{q-1}$ any function that satisfies
\begin{align*}
	\ip{\chi_q(j_1)}{\chi_q(j_2)} = \begin{cases}
		1 &j_1=j_2 \\
		\frac{-1}{q-1} &j_1\neq j_2
	\end{cases}
\end{align*} 
For $f\in [q]^n$, we use $\chi_q(f)$ to denote the vector in $\R^{(q-1)n}$ obtained by applying $\chi_q$ on $f$ coordinate-wise. Note that \[\norm{\chi_q(f)}^2 ~=~ \Ex{i\in[n]}{\norm{\chi_q(f_i)}^2} = 1\]\\
When clear from context, we may omit $q$ to write $\chi(f)$.
\end{definition}

Such a $\chi_q$ exists because the corresponding $q\times q$ Gram matrix is positive semidefinite and of rank $q-1$.

\begin{observation}
For $f_1,f_2 \in [q]^n$, if $\dis(f_1,f_2) = (1-\frac{1}{q})(1-\beta)$, then
\[
	\ip{\chi(f)}{\chi(g)} = \beta
\]
\end{observation}

\begin{theorem}[$q$-ary Johnson bound]\label{thm:johnson_bound}
Let $\calC \subseteq [q]^n$  be a code of distance at least $\dis(\calC) = (1-\frac{1}{q})(1-\beta)$. For any $g\in [q]^n$, 
\[
	\abs{\calL \inparen{ g,\inparen{1-\frac{1}{q}} \cdot (1-\sqrt{\beta}) } } ~\leq~ (q-1)\cdot n+1
\]
\end{theorem}

\begin{proof}

Let $\calL = \calL \inparen{ g,\inparen{1-\frac{1}{q}}\cdot (1-\sqrt{\beta}) } = \inbraces{h_1,h_2,\cdots,h_m}$. Then for any $h_i\in \calL$,
\[
	\ip{\chi(g)}{\chi(h_i)} > \sqrt{\beta}
\]
By \cref{lem:covering}, there exists a $g_0 \in \conv{\inbraces{\chi(h_1),\chi(h_2),\cdots,\chi(h_m)}}$ such that $\ip{g_0}{\chi(h)}>\beta$ for every $h\in \calL$. Since $g_0 \in \R^{(q-1)n}$, by the \Caratheodory theorem, we may assume that $g_0$ can be written as a convex combination of at most $(q-1)n+1$ elements of $\inbraces{\chi(h_1),\chi(h_2),\cdots,\chi(h_m)}$. Let this subset of $\inbraces{\chi(h_1),\chi(h_2),\cdots,\chi(h_m)}$ be $\supp(g_0)$. 

For any $h\in \calL$,
\begin{align*}
	\ip{g_0}{\chi(h)} > \beta &~\implies~
	\exists \chi(h_0)\in \supp(g_0) \text{ such that } \ip{ \chi(h_0) }{\chi(h)} > \beta \\
	&~\implies~ \exists \chi(h_0)\in \supp(g_0) \text{ such that } \ip{\chi(h_0)}{\chi(h)}=1 \\
	&~\implies~ \exists \chi(h_0)\in \supp(g_0) \text{ such that } h_0=h \\
	&~\implies~ \chi(h) \in \supp(g_0)
\end{align*}
Here we used the fact that for any $h_i,h_j\in \calL$, due to distance of the code,
\[
	\ip{\chi(h_i)}{\chi(h_j)} \leq \beta \quad \text{ or }\quad \ip{\chi(h_i)}{\chi(h_j)} =1
\]
Finally, we have concluded that $\inbraces{\chi(h_1),\chi(h_2),\cdots,\chi(h_m)} \sub \supp(g_0)$ and using $|\supp(g_0)| \leq (q-1)\cdot n+1$ gives $m=|\calL| \leq (q-1)\cdot n +1$.
\end{proof}

Note that in the above proof, $g_0\in \conv{\inbraces{\chi(h_1),\chi(h_2),\cdots,\chi(h_m)}} \sub \conv{\chi(\calC)}$. Indeed, even if we minimize the norm as in the proof of \cref{lem:covering} over $\conv{\chi(\calC)}$, we will still get the covering property. This will be useful when trying to find the cover via an efficient algorithm, where we do not know the list apriori.

We now prove the more general weighted version of Johnson bound, which captures list recovery as a special case.

\begin{theorem}[Weighted Johnson bound \cite{G01}]\label{thm:weighted_johnson_bound}
	Let distance of code be at least $(1-\frac{1}{q})(1-\beta)$. Given weights $\{w_{i,j}\}_{i\in [n],j\in [q]}$, let $W_i = \sum_j w_{i,j}$ and $W_i^{(2)} = \sum_j w_{i,j}^2$. The number of codewords $h$ that satisfy
	\[
		\Ex{i}{\frac{w_{i,h_i}}{W_i}} > \frac{1}{q} + \sqrt{\inparen{1-\frac{1}{q}} \cdot
                  \Ex{i}{\frac{W_i^{(2)}}{W_i^2} -\frac{1}{q}} \cdot \beta }
	\]
	is at most $n(q-1)+1$.
\end{theorem}
\begin{proof}
Embed the received weights in $i^{th}$ position as $\chi(\{w_{i,j}\}_{j\in [q]}) = \sum_j \frac{w_{i,j} \chi_q(j)}{W_i}$, and append all these $n$ vectors, each of dimension $q-1$, and normalize to a unit vector to form the final embedded vector $u$. This normalizing factor is $ \inparen{\ExpOp_{i}{\norm{\frac{\sum_j w_{ij} \chi_q(j)}{W_i} }^2}}^{1/2} $, which we can simplify as
\begin{align*}
	\Ex{i}{\norm{\frac{\sum_j w_{ij} \chi_q(j)}{W_i} }^2} &~=~ \Ex{i}{\frac{1}{W_i^2} \left( \sum_j w_{ij}^2 -\frac{1}{q-1} \sum_{j_1\neq j_2} w_{ij_1}w_{ij_2}\right) }	\\
	&~=~ \Ex{i}{\frac{1}{W_i^2} \left( \frac{q}{q-1} \sum_j w_{ij}^2 -\frac{1}{q-1} (\sum_{j} w_{ij})^2 \right) }\\
	&~=~ \frac{q}{q-1} \cdot \Ex{i}{ \frac{1}{W_i^2} \left( W_i^{(2)} -\frac{1}{q}W_i^2 \right) }\\
	&~=~ \frac{q}{q-1} \cdot \Ex{i}{ \frac{W_i^{(2)}}{W_i^2} -\frac{1}{q} }
\end{align*}
We then have
\begin{align*}
	 \inparen{\frac{q}{q-1} \cdot \Ex{i}{ \frac{W_i^{(2)}}{W_i^2} -\frac{1}{q} }}^{1/2} \cdot  ~\ip{u}{\chi_q(h)} &~=~ \Ex{i}{\frac{1}{W_i} \left( \sum_{j} w_{i,j} \ip{\chi_q(j)}{\chi_q(h_i)}\right) } \\
	&~=~ \Ex{i}{\frac{1}{W_i} \left(  w_{i,h_i} - \frac{\sum_{j\neq f_i} w_{i,j}}{q-1} \right)} \\
	&~=~ \Ex{i}{ \frac{q}{q-1} \frac{w_{i,h_i}}{W_i} - \frac{1}{q-1} } 
	~=~ \frac{q}{q-1} \cdot \Ex{i}{ \frac{w_{i,h_i}}{W_i} - \frac{1}{q} } 
\end{align*}

For a codeword $h$ satisfying $\Ex{i}{\frac{w_{i,h_i}}{W_i}} > \frac{1}{q} +
\sqrt{(1-\frac{1}{q}) \cdot \Ex{i}{\frac{W_i^{(2)}}{W_i^2} -\frac{1}{q}} \cdot \beta }$, we
have $\ip{u}{\chi(h)} >\sqrt{\beta}$
and the rest of the proof follows as in \cref{thm:johnson_bound}, by using the fact that any two codewords $h_1,h_2$ satisfy $\ip{\chi(h_1)}{\chi(h_2)} \leq \beta$.
\end{proof}

\subsection{Algorithmic covering lemma}

As mentioned earlier, the above lemma guarantees the existence of a distribution over the list that agrees simultaneously with the entire list, but finding such a distribution may be hard without knowing the list already! The other option is to look for a distribution over all codewords, but such a polytope will have exponential number of vertices and it's not clear how to optimize over it efficiently.

We instead relax to allow degree-$t$ pseudodistributions, represented as degree-$t$ pseudoexpectation operators, which can be optimized over in time $n^{O(t)}$.

We recall that $\chiTan: [q] \rightarrow \R^{q-1}$ is a function that is extended pointwise to $\chiTan:[q]^E \rightarrow \inparen{\R^{q-1}}^{|E|}$. Note that $\chiTan$ is a 1-local vector-valued function on $[q]^E$, and so we can use $\tildeEx{\chiTan(\zee)} \in \R^E$ which will satisfy
\[
	\tildeEx{\chiTan(\zee)}(e) = \tildeEx{\chi_q(\zee_e)}
\]
%

\subsubsection{Tanner code}

Let $q \geq 2$ be an integer. Fix $G(L,R,E)$ to be an $(n,d,\lambda)$-expander, and $\calC_0 \subseteq [q]^d$ to be an inner code. Let $\TanC \subseteq [q]^E$ be the Tanner code determined by $G$ and $\calC_0$. 

\begin{lemma}\label{lem:abstract_algorithmic_covering}
	Fix $\gamma>0$. Let $u\in \R^{(q-1)|E|}$ with $\norm{u}_2 =1$. For any $t\geq d$, there exists a pseudocodeword $\tildeEx{\cdot}$ of SoS-degree $t$ such that for any $h \in \TanC$ that satisfies $\ip{u}{\chiTan(h)}>\gamma$, 
	\[
		\ip{\tildeEx{\chiTan(\zee)}}{\chiTan(h)} = \Ex{e\in E}{\ip{\tildeEx{\chi_q(\zee_e)}}{\chi_q(h(e))}} > \gamma^2
	\]
	Moreover, this pseudocodeword $\tildeEx{\cdot}$ can be found in time $n^{\bigoh(t)}$.
\end{lemma}
\begin{proof}
	Define the convex function $\Psi(\tildeEx{\cdot}) = \ip{\tildeEx{\chiTan(\zee)}}{\tildeEx{\chiTan(\zee)}}$ and solve the convex program in \cref{tab:Tanner_SoS}.
	
\begin{table}[h]
\hrule
\vline
\begin{minipage}[t]{0.99\linewidth}
\vspace{-5 pt}
{\small
\begin{align*}
    &\mbox{minimize}\quad ~~ \Psi\left(\tildeEx{\cdot}\right)
    \\
&\mbox{subject to}\quad \quad ~\\
    &\qquad \ip{\tildeEx{\chiTan(\zee)}}{u} ~>~ \gamma\label{cons:agreement-ld}    \\
&\qquad \tildeEx{\cdot} \text{is a pseudocodeword of SoS-degree } t
\end{align*}}
\vspace{-14 pt}
\end{minipage}
\hfill\vline
\hrule
\caption{Finding cover for the list $\calL$.}
\label{tab:Tanner_SoS}
\end{table}
	
	Let $\tildeEx{\cdot}$ be the pseudocodeword obtained, with $\Psi^*=\Psi(\tildeEx{\cdot})$. We will use the optimality of $\tildeEx{\cdot}$ to argue that for any $h\in \calL$,
	\[
		\ip{\tildeEx{\chiTan(\zee)}}{\chi(h)} > \gamma^2
	\]
	
	Suppose not. Then there exists $h\in \TanC$ such that $\ip{\tildeEx{\chiTan(\zee)}}{\chiTan(h)} \leq \gamma$. Then, for some $\tee \in[0,1]$ to be chosen later, consider 
	\[
		\PExp_{\tee}[\cdot] = (1-\tee)\tildeEx{\cdot} + \tee\PExp^{(h)}[\cdot]
	\]
	Here, we think of $h$ as an SoS-degree-$t$ pseudocodeword, so that by convexity, $\PExp_{\tee}[\cdot]$ is also an SoS-degree-$t$ pseudocodeword. We will show towards a contradiction that $\Psi(\PExp_{\tee}[\cdot]) < \Psi^*$.
	\begin{align*}
		\Psi(\PExp_{\tee}[\cdot])  &= \ip{\PExp_{\tee}[\chiTan(\zee)]}{\PExp_{\tee}[\chiTan(\zee)]} \\
		&= \ip{(1-\tee)\cdot \tildeEx{\chiTan(\zee)} + \tee\cdot \PExp^{(h)}[\chiTan(\zee)]}{(1-\tee)\cdot \tildeEx{\chiTan(\zee)}+\tee\cdot \PExp^{(h)}[\chiTan(\zee)]} \\
		&= \ip{(1-\tee) \cdot \tildeEx{\chiTan(\zee)} + \tee\cdot \chiTan(h)}{(1-\tee)\cdot \tildeEx{\chiTan(\zee)}+\tee \cdot \chiTan(h)} \\
		&= (1-\tee)^2 \cdot \ip{\tildeEx{\chiTan(\zee)}}{\tildeEx{\chiTan(\zee)}} +2\cdot \tee(1-\tee)\cdot \ip{\tildeEx{\chiTan(\zee)}}{\chiTan(h)} + \tee^2\cdot \ip{\chiTan(h)}{\chiTan(h)}\\
		&\leq (1-\tee)^2\cdot \Psi^* +2\cdot \tee(1-\tee)\cdot \gamma +\tee^2
	\end{align*}
	
	Optimizing over $\tee$, we choose $\tee^* = \frac{\Psi^* - \gamma}{\Psi^*-2\gamma+1}$. As long as $\Psi^* > \gamma$, we get optimal $\tee^* >0$, which implies $\Psi(\PExp_{\tee^*}[\cdot]) < \Psi\inparen{\PExp_{\tee}[\cdot]{\Big \vert}_{\tee=0}} = \Psi^*$, which would be a contradiction.
\end{proof}

\begin{lemma}\label{lem:cover_for_list_tanner}
	Fix $\eps>0$. Let the distance of $\TanC$ be $\delta$, and let $g$ be a received word.
	For any $t\geq d$, there exists a degree-$t$ pseudocodeword $\tildeEx{\cdot}$ such that for any $h\in \TanC$ such that $\dis(g,h)<\calJ(\delta)-\eps$,
	\[
		\dis(\tildeEx{\cdot},h) < \delta-2\eps \cdot \sqrt{1-\frac{q}{q-1} \cdot\delta}
	\]
	Moreover, this pseudocodeword $\tildeEx{\cdot}$ can be found in time $n^{\bigoh(t)}$.
\end{lemma}

\begin{proof}
	We use $g$ to construct a vector $u$, which we can use to find the required pseudocodeword via \cref{lem:abstract_algorithmic_covering}.
	
	Let $\delta = \inparen{1-\frac{1}{q}}(1-\beta)$ so that $\calJ(\delta) = \inparen{1-\frac{1}{q}}(1-\sqrt{\beta})$. Let $u = \chi(g)$. 
	
	For any $h\in \calL(g,\calJ(\delta)-\eps)$,
	\[
		\dis(g,h) < \calJ(\delta)-\eps = \inparen{1-\frac{1}{q}}(1-\sqrt{\beta}) -\eps \implies  \ip{\chi(g)}{\chi(h)} > \sqrt{\beta} + \frac{q}{q-1}\cdot \eps
	\]
	 Therefore, using \cref{lem:abstract_algorithmic_covering} and vector $u$, we can find a degree-$t$ pseudocodeword $\tildeEx{\cdot}$ such that for any $h\in \calL(g,\calJ(\delta)-\eps)$,
	 \[
	 	\ip{\tildeEx{\chiTan(\zee)}}{\chiTan(h)} > \inparen{\sqrt{\beta}+\frac{q}{q-1}\cdot \eps}^2 > \beta + 2\sqrt{\beta} \cdot \frac{q}{q-1} \cdot \eps
	 \]
	 Writing again in terms of distances and using $\beta = 1-\frac{q}{q-1}\delta$, this means,
	 \[
	 	\dis(\tildeEx{\cdot},h) < \inparen{1-\frac{1}{q}}(1-\beta -2\sqrt{\beta}\cdot
                \frac{q}{q-1}\cdot \eps) = \delta -2\sqrt{1-\frac{q}{q-1}\delta}\cdot \eps \mper
                \qquad\qquad \qquad \qedhere
	 \]
\end{proof}
	Next we show that \cref{lem:abstract_algorithmic_covering} can also be used to efficiently find a cover for the list when dealing with list recovery, by using the given weights to construct a modified vector $u$ (which will be the same vector embedding that was used in proof of \cref{thm:weighted_johnson_bound}).
\begin{lemma}\label{lem:cover_list_recovery_tanner}
	Fix $\eps>0$. Let the distance of $\TanC$ be $\delta$, and let the given weights be $\{w_{e,j}\}_{e\in E,j\in [q]}$. Assume that the weights are normalized so that $\sum_j w_{e,j}=1$ and denote  $W_e^{(2)} = \sum_j w_{e,j}^2$. \\
	For any $t\geq d$, there exists a pseudocodeword $\tildeEx{\cdot}$ of SoS-degree $t$ such that for any $h\in \TanC$ that satisfies
	\[
		\Ex{e}{w_{e,h(e)}} > \frac{1}{q} + \sqrt{\inparen{1-\frac{1}{q}  -\delta} \cdot \inparen{\Ex{e}{W_e^{(2)} -\frac{1}{q}}} } + \eps
	\]
	also satisfies $\dis(\tildeEx{\cdot},h) < \delta - \bigomega_{q,\delta,W}(\eps)$.
	
	Moreover, this pseudocodeword $\tildeEx{\cdot}$ can be found in time $n^{\bigoh(t)}$.
\end{lemma}

\begin{proof}
Again, let $\delta = \inparen{1-\frac{1}{q}}(1-\beta)$. We use the same vector $u$ to embed the given weights, with the property that for any $h$ that satisfies 
	\begin{align*}
		\Ex{e}{w_{e,h(e)}} &> \frac{1}{q} + \sqrt{\inparen{1-\frac{1}{q}  -\delta} \cdot \inparen{\Ex{e}{W_e^{(2)} -\frac{1}{q}}} } + \eps \\
		&= \frac{1}{q} + \sqrt{\inparen{1-\frac{1}{q}} \cdot \inparen{\Ex{e}{W_e^{(2)} -\frac{1}{q}}} \beta} + \eps
	\end{align*}
	now satisfies \[ \ip{u}{\chi(h)}> \sqrt{\beta} + \sqrt{\frac{q}{q-1}} \frac{ \eps }{ \sqrt{\Ex{e}{W_e^{(2)} -\frac{1}{q}}}} \]
Next, we appeal to \cref{lem:abstract_algorithmic_covering} with $u$ to efficiently obtain a degree-$t$ pseudocodeword $\tildeEx{\cdot}$ such that
	\[
		\ip{\tildeEx{\chiTan(\zee)}}{\chiTan(h)} > \beta + 2 \sqrt{\beta} \sqrt{\frac{q}{q-1}} \frac{ \eps }{ \sqrt{\Ex{e}{W_e^{(2)} -\frac{1}{q}}}}
	\]
	In terms of distance, this means,
\begin{align*}
\dis(\tildeEx{\cdot},h) &~<~ \inparen{1-\frac{1}{q}} \inparen{1-\beta -2\sqrt{\beta}
                          \sqrt{\frac{q}{q-1}} \frac{\eps }{ \sqrt{\Ex{e}{W_e^{(2)} -\frac{1}{q}}}} } \\
		&~=~ \delta -2 \eps \cdot \sqrt{1-\frac{q}{q-1}\delta}
           \cdot\sqrt{\frac{1-\frac{1}{q}}{\Ex{e}{W_e^{(2)} -\frac{1}{q}}}} \mper
                \qquad\qquad \qquad \qquad  \qedhere
	\end{align*} 
\end{proof}

\subsubsection{AEL Code}

Fix $q\in \N$. Fix $G=(L,R,E)$ to be an $(n,d,\lambda)$-expander, $\calC_0 \subseteq [q]^d$ to be an inner code, and $\calC_1 \subseteq [|\calC_0|]^n$ to be an outer code. Let $\AELC \subseteq [q^d]^R$ be the Tanner code determined by $G$, $\calC_0$ and $\calC_1$. 

\begin{definition}
	Let $\chiAEL: [q]^E \rightarrow \inparen{\R^{q^d-1}}^R$ be defined as
	\[
		\inparen{\chiAEL(f)}(\ri) = \chi_{q^d}(f_{N_R(\ri)})
	\]
	so that $\chiAEL$ is a $d$-local vector-valued function on $[q]^E$.
\end{definition}

\begin{observation}
If $f_1,f_2\in [q]^E$ with $\dis_{AEL}(f_1,f_2) = (1-\frac{1}{q^d})(1-\beta)$, then $\ip{\chiAEL( f_1)}{\chiAEL(f_2)} = \beta$.
\end{observation}

\begin{lemma}\label{lem:abstract_algorithmic_covering_ael}
	Fix $\gamma>0$. Let $u\in (\R^{q^d-1})^R$ with $\norm{u}=1$. For any $t\geq d$, there exists a degree-$t$ pseudocodeword $\tildeEx{\cdot}$ such that for any $h\in \AELC$ that satisfies $\ip{u}{\chiAEL(h)} > \gamma$, 
	\[
		\ip{\tildeEx{\chiAEL(\zee)}}{\chiAEL(h)} = \Ex{\ri \in R}{\ip{\tildeEx{\chi_{q^d}(\zee_{N_R(\ri)})}}{\chi_{q^d}(h(\ri))}} > \gamma^2
	\]
	Moreover, this pseudocodeword $\tildeEx{\cdot}$ can be found in time $n^{\bigoh(t)}$.
\end{lemma}

\begin{proof}
The proof is very similar to the proof of \cref{lem:abstract_algorithmic_covering}, except we replace embedding function $\chiTan$ by $\chiAEL$.

Define the quantity $\Psi(\tildeEx{\cdot}) = \ip{\tildeEx{\chiAEL(\zee)}}{\tildeEx{\chiAEL(\zee)}}$ and solve the following convex program.
	
\begin{table}[h]
\hrule
\vline
\begin{minipage}[t]{0.99\linewidth}
\vspace{-5 pt}
{\small
\begin{align*}
    &\mbox{minimize}\quad ~~ \Psi\left(\tildeEx{\cdot}\right)
    \\
&\mbox{subject to}\quad \quad ~\\
    &\qquad \ip{\tildeEx{\chiAEL(\zee)}}{\chiAEL(g)} >  \gamma\label{cons:agreement-ld}
    \\
&\qquad \tildeEx{\cdot} \text{is a pseudocodeword of SoS-degree } t
\end{align*}}
\vspace{-14 pt}
\end{minipage}
\hfill\vline
\hrule
\caption{Finding cover for the list $\calL$.}
\end{table}

Let $\tildeEx{\cdot}$ be the pseudocodeword obtained, with $\Psi^*=\Psi(\tildeEx{\cdot})$. Suppose it does not have the covering property. Then there exists $h\in \calL$ such that $\ip{\tildeEx{\chiAEL(\zee)}}{\chiAEL(h)} \leq \gamma$. Then, for some $\tee \in[0,1]$ to be chosen later, consider 
	\[
		\PExp_{\tee}[\cdot] = (1-\tee)\tildeEx{\cdot} + \tee\PExp^{(h)}[\cdot]
	\]
We again have,
\begin{align*}
		\Psi(\PExp_{\tee}[\cdot])  &= \ip{\PExp_{\tee}[\chiAEL(\zee)]}{\PExp_{\tee}[\chiAEL(\zee)]} \\
		&= \ip{(1-\tee) \tildeEx{\chiAEL(\zee)} + \tee\PExp^{(h)}[\chiAEL(\zee)]}{(1-\tee)\tildeEx{\chiAEL(\zee)}+\tee \PExp^{(h)}[\chiAEL(\zee)]} \\
		&= \ip{(1-\tee) \tildeEx{\chiAEL(\zee)} + x\chiAEL(h)}{(1-\tee)\tildeEx{\chiAEL(\zee)}+\tee\chiAEL(h)} \\
		&= (1-\tee)^2 \ip{\tildeEx{\chiAEL(\zee)}}{\tildeEx{\chiAEL(\zee)}} +2\tee(1-\tee)\ip{\tildeEx{\chiAEL(\zee)}}{\chiAEL(h)} + \tee^2\ip{\chiAEL(h)}{\chiAEL(h)}\\
		&\leq (1-\tee)^2 \Psi^* +2\tee(1-\tee)\gamma +\tee^2
	\end{align*}
	
	Optimizing over $\tee$, we choose $\tee^* = \frac{\Psi^* - \gamma}{\Psi^*-2\gamma+1}$. We will obtain a contradiction since $\Psi^* > \gamma^2$.
\end{proof}

\begin{lemma}\label{lem:cover_for_list_ael}
Fix $\eps>0$. Let the distance of $\AELC$ be $\delta$, and let $g$ be a received word. For any $t\geq d$, there exists a degree-$t$ pseudocodeword $\tildeEx{\cdot}$ such that for any $h \in \AELC$ such that $\dis(g,h) < \calJ(\delta) - \eps$,
\[
	\Ex{\ri}{\indi{ \tildeEx{\zee_{N_R(\ri)} \neq h_{\ri}}}} < \delta - 2 \eps \cdot \sqrt{1-\frac{q^d}{q^d-1} \cdot \delta}
\]
\end{lemma}

\begin{proof}
Same proof as the proof of \cref{lem:cover_for_list_tanner}, with the alphabet changed. The received word $g$ can be used to construct a unit vector $u = \chiAEL(g)$, which is then used via \cref{lem:abstract_algorithmic_covering_ael} to find a pseudocodeword with the required covering property.
\end{proof}

\begin{lemma}\label{lem:cover_for_list_recovery_ael}
	Fix $\eps>0$. Let the distance of $\AELC$ be $\delta$, and let the given weights be $\{w_{\ri,j}\}_{\ri\in R,j\in [q^d]}$. Assume that the weights are normalized so that $\sum_j w_{\ri,j}=1$ and denote  $W_{\ri}^{(2)} = \sum_j w_{{\ri},j}^2$. \\
	For any $t\geq d$, there exists a degree-$t$ pseudocodeword $\tildeEx{\cdot}$ such that for any $h\in \AELC$ that satisfies
	\[
		\Ex{\ri}{w_{\ri,h(\ri)}} > \frac{1}{q^d} + \sqrt{\inparen{1-\frac{1}{q^d}  -\delta} \cdot \inparen{\Ex{\ri}{W_{\ri}^{(2)} -\frac{1}{q^d}}} } + \eps
	\]
	also satisfies $\dis(\tildeEx{\cdot},h) < \delta - \bigomega_{q,d,\delta,W}(\eps)$.
	
	Moreover, this pseudocodeword $\tildeEx{\cdot}$ can be found in time $n^{\bigoh(t)}$.
\end{lemma}

\begin{proof}
	Same proof as the proof of \cref{lem:cover_list_recovery_tanner}.
\end{proof}


%
%
%
%
%
%

%
%
%
%
\section{Sum-of-Squares Proofs of Distance}
\label{sec:distance}
We will be proving that pseudocodewords satisfying certain $\eta$-good property defined below have the same distance properties as true codewords, up to $\eta$ error.
\begin{definition}\label{def:eta_good}
A pseudocodeword of SoS-degree at least $2d$ is $\eta$-good if
\begin{align*}
    \Ex{\li,\ri}{\tildecov[\zee_{N_L(\li)},\zee_{N_R(\ri)}]} \leq \eta
\end{align*}
\end{definition}

\begin{observation}
    A true codeword is a 0-good pseudocodeword.
\end{observation}

The $\eta$-good property is useful to change the pseudoexpectation of a product of functions into the product of pseudoexpectations of those functions. We establish a formal claim about this in the next lemma. The terms $\norm{X_{\li}}_{\infty}$ and $\norm{Y_{\ri}}_{\infty}$ should be just seen as normalizing the scale, and indeed we will only use functions that are bounded in infinity norm by 1.

\begin{lemma}
	Let $\{X_{\li}\}_{\li\in L}$ and $\{Y_{\ri}\}_{\ri\in R}$ be two collections of $d$-local functions on $[q]^E$ such that for every $\li\in L$, $X_{\li}(f)$ only depends on $f|_{N_L(\li)}$ and for every $\ri\in R$, $Y_{\ri}(f)$ only depends on $f|_{N_R(\ri)}$. Then, for an $\eta$-good pseudocodeword $\tildeEx{\cdot}$,
	\[
		\Ex{\li,\ri}{\tildeEx{X_{\li}(\zee)Y_{\ri}(\zee)}} ~\leq~ \Ex{\li,\ri}{\tildeEx{X_{\li}(\zee)}\tildeEx{Y_{\ri}(\zee)}} + \eta \left( \max_{\li} \norm{X_{\li}}_{\infty} \right) \left( \max_{\ri} \norm{Y_{\ri}}_{\infty}\right)
	\]
\end{lemma}

\begin{proof}
	For any $\li$ and $\ri$,
{\small
	\begin{align*}
		&\tildeEx{X_{\li}(\zee)Y_{\ri}(\zee)} - \tildeEx{X_{\li}(\zee)}\tildeEx{Y_{\ri}(\zee)} \\
		&~=~ \sum_{\substack{\alpha \in [q]^{N_L(\li)} \\ \beta\in[q]^{N_R(\ri)}}}
                     \tildeEx{X_{\li}(\alpha) \prod_{s\in N_L(\li)} \zee_{s,\alpha_s} \cdot Y_{\ri}(\beta) \prod_{t\in N_R(\ri)} \zee_{t,\beta_t}} - 
		 \sum_{\substack{\alpha \in [q]^{N_L(\li)} \\ \beta\in[q]^{N_R(\ri)}}} \cdot \tildeEx{X_{\li}(\alpha) \prod_{s\in N_L(\li)} \zee_{s,\alpha_s}} \tildeEx{Y_{\ri}(\beta) \prod_{t\in N_R(\ri)} \zee_{t,\beta_t}} \\
		&~=~ \sum_{\alpha,\beta} X_{\li}(\alpha) Y_{\ri}(\beta) \cdot \inparen{\tildeEx{
           \prod_{s\in N_L(\li)} \zee_{s,\alpha_s} \cdot \prod_{t\in N_R(\ri)} \zee_{t,\beta_t}} -
           \tildeEx{\prod_{s\in N_L(\li)} \zee_{s,\alpha_s}} \cdot \tildeEx{\prod_{t\in N_R(\ri)} \zee_{t,\beta_t}}} \\
		&~\leq~  \norm{X_{\li}}_{\infty} \norm{Y_{\ri}}_{\infty} \cdot \mathlarger{\sum}_{\alpha,\beta}
           \abs{ \tildeEx{ \prod_{s\in N_L(\li)} \zee_{s,\alpha_s} \prod_{t\in N_R(\ri)}
           \zee_{t,\beta_t}} - \tildeEx{\prod_{s\in N_L(\li)} \zee_{s,\alpha_s}} \cdot \tildeEx{\prod_{t\in N_R(\ri)} \zee_{t,\beta_t}}} \\
		&~=~  \norm{X_{\li}}_{\infty} \norm{Y_{\ri}}_{\infty} \cdot \tildecov[\zee_{N_L(\li)},\zee_{N_R(\ri)}]
	\end{align*}		
}	
	Averaging over $\li$ and $\ri$, we get
\begin{align*}
\Ex{\li,\ri}{\tildeEx{X_{\li}(\zee)Y_{\ri}(\zee)} - \tildeEx{X_{\li}(\zee)}\tildeEx{Y_{\ri}(\zee)}} 
	&~\leq~  \Ex{\li,\ri}{\norm{X_{\li}}_{\infty} \norm{Y_{\ri}}_{\infty} \cdot \tildecov[\zee_{N_L(\li)},\zee_{N_R(\ri)}]} \\
	&~\leq~ \max_{\li}\norm{X_{\li}}_{\infty} \cdot \max_{\ri} \norm{Y_{\ri}}_{\infty} \cdot \Ex{\li,\ri}{\tildecov[\zee_{N_L(\li)},\zee_{N_R(\ri)}]} \\
	&~\leq~ \eta \cdot \max_{\li}\norm{X_{\li}}_{\infty} \cdot \max_{\ri}
   \norm{Y_{\ri}}_{\infty}
\end{align*}
\end{proof}
The proofs of distance for both Tanner and AEL codes go via Expander Mixing Lemma (EML), and so we establish the analog of EML for pseudocodewords. Morally speaking, EML allows us to change measure from the edges of an expander to the complete (bipartite) graph for product functions.  First we prove a version of EML for vector valued functions, and then show that since pseudoexpectation operators can be written in terms of certain underlying vectors, they also satisfy a version of EML. Note that this step does not require any $\eta$-good property.
	
	\begin{lemma}[EML for vector-valued functions]
	\label{high_dimensional_eml}
		Let $\{v_{\li}\}_{\li \in L}$ and $\{u_{\ri}\}_{\ri \in R}$ be a collection of vectors in $\R^N$. Then,
		\[
			\abs{\Ex{\li \sim \ri}{\ip{v_{\li}}{u_{\ri}}} - \Ex{\li,\ri}{\ip{v_{\li}}{u_{\ri}}}} \leq \lambda \sqrt{\Ex{\li}{\norm{v_{\li}}^2}} \sqrt{\Ex{\ri}{\norm{u_{\ri}}^2}}
		\]
	\end{lemma}
	
	\begin{proof}
		Usual EML applied coordinate-wise.
	\end{proof}
	
\begin{lemma}[EML for pseudoexpectations]\label{lem:eml_for_pexp}
	Let $\{X_{\li}\}_{\li\in L}$ and $\{Y_{\ri}\}_{\ri\in R}$ be two collections of $d$-local functions on $[q]^E$ such that for every $\li\in L$, $X_{\li}(f)$ only depends on $f|_{N_L(\li)}$ and for every $\ri\in R$, $Y_{\ri}(f)$ only depends on $f|_{N_R(\ri)}$.
Then for a $\lambda$-spectral expander, we have
	\[
		\abs{\Ex{\li\sim \ri}{\tildeEx{X_{\li}(\zee) Y_{\ri}(\zee) }} -
                  \Ex{\li,\ri}{\tildeEx{X_{\li}(\zee) Y_{\ri}(\zee)}}} \leq \lambda
                \sqrt{\Ex{\li}{\tildeEx{X_{\li}(\zee)^2}}} \sqrt{\Ex{\ri}
                  {\tildeEx{Y_{\ri}(\zee)^2}}} \mper
	\]
\end{lemma}
\begin{proof}	
Consider the $2n\times 2n$ matrix $M$, with
\[
		M_{ij} = \begin{cases}
			\tildeEx{X_iX_j}, & 1\leq i\leq n, 1\leq j\leq n \\
			\tildeEx{X_iY_{j-n}}, & 1\leq i\leq n, n+1\leq j\leq 2n \\
			\tildeEx{Y_{i-n}X_j}, &n+1\leq i\leq 2n, 1\leq j\leq n\\
			\tildeEx{Y_{i-n}Y_{j-n}} &n+1\leq i\leq 2n,n+1\leq j\leq 2n
		\end{cases}
\]
For any vector $v = (x_1,x_2,\cdots,x_n,y_1,y_2,\cdots,y_n)$, we show that $v^TMv\geq 0$, so that $M$ is PSD.
	\begin{align*}
		v^TMv &~=~ \Ex{i,j}{M_{ij} x_ix_j + M_{i,j+n}x_iy_j + M_{i+n,j} y_ix_j + M_{i+n,j+n} y_iy_j}\\
		&~=~ \Ex{i,j}{\tildeEx{X_iX_j}x_ix_j+\tildeEx{X_iY_j}x_iy_j+\tildeEx{Y_iX_j}y_ix_j+\tildeEx{Y_iY_j}y_iy_j} \\
		&~=~ \Ex{i,j}{\tildeEx{x_i x_j X_iX_j}+\tildeEx{x_iy_jX_iY_j}+\tildeEx{y_ix_jY_iX_j}+\tildeEx{y_iy_jY_iY_j}} \\
		&~=~ \Ex{i,j}{\tildeEx{(x_i X_i + y_i Y_i)(x_jX_j+y_jY_j)}} \\
		&~=~ \tildeEx{\Ex{i,j}{(x_i X_i + y_i Y_i)(x_jX_j+y_jY_j)}} \\
		&~=~ \tildeEx{\Ex{i}{(x_i X_i + y_i Y_i)}^2}  ~\geq~ 0\\
	\end{align*}
Therefore there exist vectors $\{v_{\li}\}_{\li \in L}$ and $\{u_{\ri}\}_{\ri\in R}$ such that 			
\[		
\tildeEx{X_{\li} Y_{\ri}} = \ip{v_{\li}}{u_{\ri}}, \quad
\tildeEx{X_{\li}^2} = \ip{v_{\li}}{v_{\li}}, \quad \text{and} \quad
\tildeEx{Y_{\ri}^2} = \ip{u_{\ri}}{u_{\ri}}
\]
Applying \cref{high_dimensional_eml} to the collection of vectors obtained above, we immediately obtain,
\[
		\abs{\Ex{\li\sim \ri}{\tildeEx{X_{\li}Y_{\ri}}} - \Ex{\li,\ri}{\tildeEx{X_{\li}
                      Y_{\ri}}}} ~\leq~ \lambda \sqrt{\Ex{\li}{\tildeEx{X_{\li}^2}}}
                \sqrt{\Ex{\ri}{\tildeEx{Y_{\ri}^2}}} \qquad \qedhere
\]
\end{proof}

\subsection{Tanner Code}
Suppose we are working with a Tanner code $\TanC$ with inner code $\calC_0$ of distance $\delta_0$, so that the distance of $\TanC$ is at least $\delta_0(\delta_0-\lambda)$. We show that $\eta$-good pseudocodewords satisfy a similar distance property, up to error $\eta$.
\begin{lemma}[Distance of Tanner code]
	\label{lem:distance_from_codeword}
    The distance between an $\eta$-good pseudocodeword $\tildeEx{\cdot}$ and a true codeword $h$ is at least $\delta_0(\delta_0-\lambda) -2\eta \frac{\delta_0}{\delta_0-\lambda}$, or at most $\frac{4\eta^2}{(\delta_0-\lambda)^2} +\frac{\eta (\delta_0+\lambda)}{\delta_0-\lambda}$.
In particular, if $\lambda \leq \delta_0/3$ and $\eta\leq \delta_0^2/9$, then $\dis(h,\tildeEx{\cdot}) \leq 3\eta$ or $\dis(h,\tildeEx{\cdot}) \geq \delta_0(\delta_0-\lambda) - 3\eta$.
\end{lemma}
\begin{proof}
Let $X_{\li}(\zee) \defeq \indi{\zee_{N_L(\li)}\neq h_{N_L(\li)}}$ and $Y_{\ri}(\zee) \defeq
\indi{\zee_{N_R(\ri)}\neq h_{N_R(\ri)}}$, and let $\tau$ denote the quantity
$\sqrt{\Ex{\li}{\tildeEx{X_{\li}(\zee)}}} \cdot \sqrt{\Ex{\ri}{\tildeEx{Y_{\ri}(\zee)}}}$. Then, we have
\begin{align*}
    		\dis\inparen{h,\tildeEx{\cdot}} &~=~ \Ex{e}{\tildeEx{\indi{\zee_e\neq h_e}}} \\
    		&~\leq~ \Ex{\li \sim \ri}{\tildeEx{X_{\li}(\zee) \cdot Y_{\ri}(\zee) }} \\
    		&~\leq~ \Ex{\li , \ri}{\tildeEx{ X_{\li}(\zee) \cdot Y_{\ri}(\zee)}} + \lambda \cdot \sqrt{\Ex{\li}{\tildeEx{X_{\li}(\zee)^2}}} \cdot \sqrt{\Ex{\ri}{\tildeEx{Y_{\ri}(\zee)^2}}}\\
    		&~=~ \Ex{\li , \ri}{\tildeEx{X_{\li}(\zee) \cdot Y_{\ri}(\zee)}} + \lambda \cdot \sqrt{\Ex{\li}{\tildeEx{X_{\li}(\zee)}}} \cdot \sqrt{\Ex{\ri}{\tildeEx{Y_{\ri}(\zee)}}}\\
    		&~=~ \Ex{\li , \ri}{\tildeEx{X_{\li}(\zee) \cdot Y_{\ri}(\zee) }} + \lambda \cdot \tau\\
    		&~=~ \Ex{\li , \ri}{\tildeEx{X_{\li}(\zee)} \cdot \tildeEx{Y_{\ri}(\zee) }} + \eta
           \cdot 1 \cdot 1 +\lambda \cdot \tau \\ &~=~ \tau^2 +\lambda \cdot \tau +\eta
\end{align*}
On the other hand,
\begin{align*}
   		\dis\inparen{h,\tildeEx{\cdot}} &= \Ex{e}{\tildeEx{\indi{\zee_e\neq h_e}}} \\
   		&~=~ \Ex{\li}{\Ex{e\in N_L(\li)}{\tildeEx{\indi{\zee_e\neq h_e}}}} \\
   		&~=~ \Ex{\li}{\tildeEx{\dis\inparen{\zee_{N_L(\li)}, h_{N_L(\li)}}}} \\
   		&~\geq~ \Ex{\li}{\tildeEx{0\cdot \indi{\zee_{N_L(\li)} = h_{N_L(\li)}} + \delta_0 \cdot \indi{\zee_{N_L(\li)} \neq h_{N_L(\li)}}}} \\
   		&~=~ \Ex{\li}{\tildeEx{ \delta_0 \cdot \indi{\zee_{N_L(\li)} \neq h_{N_L(\li)}}}} \\
   		&~=~ \Ex{\li}{\tildeEx{ \delta_0 \cdot X_{\li}(\zee)}} \\
   		&~=~ \delta_0 \cdot  \Ex{\li}{\tildeEx{ X_{\li}(\zee)}} \\
\end{align*}
Likewise, $\dis\inparen{h,\tildeEx{\cdot}} \geq \delta_0 \cdot \Ex{\ri}{\tildeEx{Y_{\ri}(\zee)}}$, and so,
\[
   		\dis\inparen{h,\tildeEx{\cdot}} ~\geq~ \delta_0 \cdot
                \sqrt{\Ex{\li}{\tildeEx{X_{\li}(\zee)}}} \cdot
                \sqrt{\Ex{\ri}{\tildeEx{Y_{\ri}(\zee)}}} ~=~ \delta_0 \cdot \tau
\]
Comparing, we get, $\tau^2 +\lambda \cdot \tau +\eta \geq \delta_0 \cdot \tau$, which means,
\[
\tau ~\geq~ \frac{(\delta_0-\lambda) + \sqrt{(\delta_0-\lambda)^2-4\eta}}{2}
\qquad \text{or} \qquad
\tau ~\leq~ \frac{(\delta_0-\lambda) - \sqrt{(\delta_0-\lambda)^2-4\eta}}{2}
\]
In the first case, we have
\begin{align*}
\dis\inparen{h,\tildeEx{\cdot}} ~\geq~ \delta_0 \cdot \tau 
    		&~\geq~ \delta_0 \cdot \frac{(\delta_0-\lambda) + \sqrt{(\delta_0-\lambda)^2-4\eta}}{2}\\
    		&= \frac{\delta_0(\delta_0-\lambda)}{2} \inparen{ 1+ \sqrt{1-\frac{4\eta}{(\delta_0 - \lambda)^2}} }\\
    		&\geq \frac{\delta_0(\delta_0-\lambda)}{2} \inparen{ 1+ 1-\frac{4\eta}{(\delta_0 - \lambda)^2} }\\
    		&= \delta_0(\delta_0-\lambda) - 2\eta \cdot \frac{\delta_0}{\delta_0-\lambda} \mper
    	\end{align*}
Also, in the second case, we have
\begin{align*}
    		\tau ~\leq~ \frac{(\delta_0-\lambda) - \sqrt{(\delta_0-\lambda)^2-4\eta}}{2} 
    		&~=~ \frac{\delta_0-\lambda}{2} \left( 1- \sqrt{1-\frac{4\eta}{(\delta_0-\lambda)^2}} \right)\\
    		&~\leq~ \frac{\delta_0-\lambda}{2} \left( 1- 1+\frac{4\eta}{(\delta_0-\lambda)^2} \right)\\
    		&~=~ \frac{2\eta}{\delta_0-\lambda} \mcom\\
    	\end{align*}
which gives
\[
\dis(g,\tildeEx{\cdot}) ~\leq~ \tau^2 +\lambda \tau +\eta 
~\leq~ \frac{4\eta^2}{(\delta_0-\lambda)^2} +\frac{2\eta \lambda}{\delta_0-\lambda}+\eta 
~=~ \frac{4\eta^2}{(\delta_0-\lambda)^2} +\frac{\eta (\delta_0+\lambda)}{\delta_0-\lambda}
\qedhere \] 
\end{proof}

\subsection{AEL Code}
Let $\calC_1$ be an outer code on an $(n,d,\lambda)$-expander graph $G(L,R,E)$ and let $\calC_0$ be the inner code. Let $\AELC$ be the code obtained by redistributing symbols along the edges of $G$ and then collecting them on vertices of $R$, as explained in \cref{sec:AEL_prelims}.

Let $\delta_0$ be the distance of $\calC_0$, so that (designed) distance of $\AELC$ is $\delta = \delta_0 - \frac{\lambda}{\delta_1}$. Let $h \in [q_0]^E$ be a codeword in $\AELC$. We show that an $\eta$-good pseudocodeword that has some left-distance from $h$ has a much larger right-distance from $h$.

\begin{lemma}[Distance of AEL Code]\label{lem:AEL_amplification}
	For an $\eta$-good pseudocodeword $\tildeEx{\cdot}$ and a codeword $h \in \AELC$, 
	\[
		\dis^R(\tildeEx{\cdot},h) \geq \delta_0 - \frac{\lambda+\eta}{\dis^L(\tildeEx{\cdot},h)}
	\]
\end{lemma}

\begin{proof}
	We establish upper and lower bounds on $\dis(\tildeEx{\cdot},h)$.
	\begin{align*}
		\dis(\tildeEx{\cdot},h) &= \Ex{e}{\tildeEx{\indi{\zee_e\neq h_e}}} \\
		&~\geq~ \Ex{\li\in L}{\tildeEx{0\cdot \indi{\zee_{N_L(\li)} = h_{N_L(\li)}} + \delta_0 \cdot \indi{\zee_{N_L(\li)}\neq h_{N_L(\li)}}}} \\
		&~=~ \delta_0 \Ex{\li\in L}{\tildeEx{\indi{\zee_{N_L(\li)}\neq h_{N_L(\li)}}}} \\
		&~=~ \delta_0 \dis^L(\tildeEx{\cdot},h)
	\end{align*}
For the upper bound, we again rely on the expander mixing lemma:
\begin{align*}
		\dis(\tildeEx{\cdot},h) &~=~ \Ex{e}{\tildeEx{\indi{\zee_e\neq h_e}}} \\
		&~\leq~ \Ex{\li\sim \ri}{\tildeEx{\indi{\zee_{N_L(\li)}\neq h_{N_L(\li)}} \indi{\zee_{N_R(\ri)}\neq h_{N_R(\ri)}}}} \\
		&~\leq~ \Ex{\li ,\ri}{\tildeEx{\indi{\zee_{N_L(\li)}\neq h_{N_L(\li)}} \indi{\zee_{N_R(\ri)}\neq h_{N_R(\ri)}}}} +\lambda\\
		&~\leq~ \Ex{\li ,\ri}{\tildeEx{\indi{\zee_{N_L(\li)}\neq h_{N_L(\li)}}} \tildeEx{\indi{\zee_{N_R(\ri)}\neq h_{N_R(\ri)}}}} +\lambda +\eta\\
		&~=~ \Ex{\li}{\tildeEx{\indi{\zee_{N_L(\li)}\neq h_{N_L(\li)}}}} \Ex{\ri}{\tildeEx{\indi{\zee_{N_R(\ri)}\neq h_{N_R(\ri)}}}} +\lambda +\eta\\
		&~=~ \dis^L(\tildeEx{\cdot},h) \cdot \dis^R(\tildeEx{\cdot},h) +\lambda+\eta
	\end{align*}
Dividing the two bounds by $\dis^L(\tildeEx{\cdot},h)$ and rearranging, we finally get,
\[
\dis^R(\tildeEx{\cdot},h) ~\geq~ \delta_0 - \frac{\lambda+\eta}{\dis^L(\tildeEx{\cdot},h)} \mper
\qquad \qquad \qquad \qquad \qedhere
\]
\end{proof}


%
\section{Correlation Reduction via Conditioning}
\label{sec:conditioning}
We will use the following claim from \cite{BRS11} (see Lemma 5.2 there) that says that if $\zee_S$ and $\zee_T$ have a large covariance, then conditioning on $\zee_T$ reduces the variance of $\zee_S$ significantly.
\begin{lemma}
	Let $\tildeEx{\cdot}$ be a pseudoexpectation operator of SoS-degree $t$ with associated pseudocovariance and pseudovariance operators. Assume $S,T$ are sets such that $|S|+|T|\leq t/2$, then,
	\[
		\tildeVar{\zee_S | \zee_T} \leq \tildeVar{\zee_S} - \frac{1}{q^{|T|}} \sum_{\alpha \in [q]^S, \beta\in [q]^T} \frac{(\tildeCov{\zee_{S,\alpha}}{\zee_{T,\beta}})^2}{\tildeVar{\zee_{T,\beta}}}
	\]
\end{lemma}
In particular, observe that pseudovariances are non-increasing under conditioning.
The next lemma shows that if the average covariance across all pairs $(\li,\ri)$ is $\eta$, then conditioning on a random vertex in $R$ will reduce the average variance in $L$ in expectation by $\Omega(\eta^2)$. Then, \cref{lem:low_covariance_solution} will use that this cannot happen more than $\calO(1/\eta^2)$ times, and then we must end up with a conditioned pseudoexpectation operator which has low average covariance, that is, it is $\eta$-good.
\begin{lemma}\label{lem:conditioning_reduces_variance}
    Let $\eta < \Ex{\li,\ri}{\tildeCov{\zee_{N_L(\li)}}{\zee_{N_R(\ri)}}} $. Then,
    \[
        \Ex{\ri \in R}{\Ex{\li}{ \tildeVar{\zee_{N_L(\li)} \vert \zee_{N_R(\ri)}}}} < \Ex{\li}{ \tildeVar{\zee_{N_L(\li)}}} - \frac{1}{q^{2d}}{\eta^2}
    \]
\end{lemma}
\begin{proof}

    \begin{align*}
        \Ex{\ri \in R}{\Ex{\li}{\tildeVar{\zee_{N_L(\li)}|\zee_{N_R(\ri)}}}} &= \Ex{\li,\ri}{\tildeVar{\zee_{N_L(\li)}|\zee_{N_R(\ri)}}}\\
        &\leq \Ex{\li,\ri}{\tildeVar{\zee_{N_L(\li)}} - \frac{1}{q^d} \sum_{\alpha, \beta} \frac{(\tildeCov{\zee_{N_L(\li),\alpha}}{\zee_{N_R(\ri),\beta}})^2}{\tildeVar{\zee_{N_R(\ri),\beta}}}} \\
        &\leq \Ex{\li,\ri}{\tildeVar{\zee_{N_L(\li)}} - \frac{1}{q^d} \sum_{\alpha, \beta} \left( \tildeCov{\zee_{N_L(\li),\alpha}}{\zee_{N_R(\ri),\beta}} \right)^2} \\
        &\leq \Ex{\li,\ri}{\tildeVar{\zee_{N_L(\li)}} - \frac{1}{q^{3d}} \left( \sum_{\alpha, \beta} \abs{ \tildeCov{\zee_{N_L(\li),\alpha}}{\zee_{N_R(\ri),\beta}}}\right)^2} \\
        &= \Ex{\li}{\tildeVar{\zee_{N_L(\li)}}} - \frac{1}{q^{3d}} \Ex{\li,\ri}{\left( \tildeCov{\zee_{N_L(\li)}}{\zee_{N_R(\ri)}} \right)^2} \\
        &\leq \Ex{\li}{\tildeVar{\zee_{N_L(\li)}}} - \frac{1}{q^{3d}} \left( \Ex{\li,\ri}{ \tildeCov{\zee_{N_L(\li)}}{\zee_{N_R(\ri)}} } \right)^2 
\end{align*}
\end{proof}
\begin{lemma}\label{lem:low_covariance_solution}
	Let $\eta>0$ be arbitrarily small. Given any SoS solution $\tildeEx{\cdot}$ of degree $\geq 2d\left(\frac{q^{3d}}{\eta^2}+1\right)$, there exists a number $k^* \leq q^{3d}/\eta^2 $ such that 
	\[
		\Ex{v_1,v_2,\cdots,v_{k^*}}{\Ex{\li,\ri}{\tildecov[\zee_{N_L(\li)},\zee_{N_R(\ri)} \vert \zee_{N_R(v_1)},\zee_{N_R(v_2)},\cdots,\zee_{N_R(v_{k^*})}]}} \leq \eta
	\]
\end{lemma}
\begin{proof}
Consider $\Phi_k(\tildeEx{\cdot}) \defeq \Ex{v_1,v_2,\cdots,v_k}{\Ex{\li}{\tildeVar{\zee_{N_L(\li)}|\zee_{N_R(v_1)},\zee_{N_R(v_2)},\cdots,\zee_{N_R(v_k)}}}}$. We know that 
\[
1~\geq~ \Phi_0 ~\geq~ \Phi_1~\geq~ \cdots ~\geq~ \Phi_{q^{3d}/\eta^2} ~\geq~ 0
\]
so there exists a $k^* \leq q^{3d}/\eta^2$ such that $\Phi_{k^*}(\tildeEx{\cdot}) - \Phi_{k^*+1}(\tildeEx{\cdot}) \leq \frac{1}{q^{3d}/\eta^2} = \eta^2/q^{3d}$.
By contrapositive of \cref{lem:conditioning_reduces_variance}, this means that 
	\[
		\Ex{v_1,v_2,\cdots,v_{k^*}}{\Ex{\li,\ri}{\tildecov[\zee_{N_L(\li)},\zee_{N_R(\ri)} \vert \zee_{N_R(v_1)},\zee_{N_R(v_2)},\cdots,\zee_{N_R(v_{k^*})}]}} \leq \eta
\qedhere	
	\]
\end{proof}


%
\section{List Decoding up to Johnson Bound}
\label{sec:decoding}
In this section, we combine different pieces of the proof to give list decoding algorithms up to the Johnson bound. Note that in both Tanner and AEL cases, we reduce to unique decoding of either the same code or the base code, and this unique decoding needs to be done from pseudocodewords instead of codewords. We will handle this slight strengthening of unique decoding via randomized rounding in \cref{sec:decoding_from_fractional}.
\subsection{Tanner code}\label{sec:list_decoding_tanner}
Let $\TanC$ be a Tanner code on an $(n,d,\lambda)$-expander graph $G(L,R,E)$, with $\calC_0$ as the inner code. Let $\delta_0$ be the distance of $\calC_0$, so that (designed) distance of $\TanC$ is $\delta = \delta_0(\delta_0-\lambda)$. Assume $\lambda \leq \delta_0/3$. Given $g\in [q]^E$, we wish to recover the list $\calL(g,\calJ(\delta)-\eps)$. As $\eps \rightarrow 0$, the decoding radius gets arbitrarily close to the Johnson bound.

\begin{theorem}[List decoding Tanner codes]\label{thm:tanner-decoding}
There is a deterministic algorithm based on $\calO_{q,d}(1/\eps^4)$ levels of the SoS-hierarchy that given $g$ runs in time $n^{O_{q,d}(1/\eps^4)}$ time and computes the list $\calL(g,\calJ(\delta)-\eps))$.
\end{theorem}
\begin{proof}
We apply the algorithmic covering \cref{lem:cover_for_list_tanner} to obtain a pseudocodeword $\tildeEx{\cdot}$ of SoS-degree $t \geq 2d\left( \frac{q^{3d}}{\eta^2} + 2\right)$ such that for any $h \in \calL(g,\calJ(\delta)-\eps)$, we know that $\dis(\tildeEx{\cdot},h) \leq \delta - \eps_2$ for $\eps_2 = 2\eps \cdot \sqrt{1-\frac{q}{q-1} \delta} \geq \Omega(\eps)$. We will choose $\eta$ later, and note that the choice of $\eta$ does not change $\eps_2$.
Henceforth, we fix an $h\in \calL((g,\calJ(\delta) - \eps)$, so that $\dis(\tildeEx{\cdot},h) \leq \delta -\eps_2$. Our goal is to recover $h$.

Pick a random $k \in \{1,\cdots,\ceil{\frac{q^{3d}}{\eta^2}}\}$. From \cref{lem:low_covariance_solution}, we know that with probability at least $\frac{\eta^2}{q^{3d}}$, 
\begin{align}\label{eqn:random_conditioning}
\Ex{v_1,v_2,\cdots ,v_k}{\Ex{\li,\ri}{\tildecov[\zee_{N_L(\li)},\zee_{N_R(\ri)} |
  \zee_{N_R(v_1)},\zee_{N_R(v_2)},\cdots ,\zee_{N_R(v_k)}]}} \leq \eta \mper
\end{align}
	We assume that we found a $k$ such that \cref{eqn:random_conditioning} holds. Let $V$ be the (random) set of $k$ vertices we condition on, that is, $V = \{v_1,v_2,\cdots,v_k\}$, and let $N_R(V) = \cup_{v\in V} N_R(v)$. Then,
	\[
		\Ex{V\sub R, |V|=k}{\Ex{\li,\ri}{\tildecov[\zee_{N_L(\li)},\zee_{N_R(\ri)} | \zee_{N_R(V)}]}} \leq \eta
	\]
	More explicitly, conditioning on $N_R(V)$ involves sampling an assignment for $N_R(V)$ according to the local distribution of $\zee_{N_R(V)}$. Let this random assignment be $\beta$, and we get
	\[
		\Ex{\substack{V \sub R,|V|=k \\ \beta \sim \zee_{N_R(V)}}}{\Ex{\li,\ri}{\tildecov[\zee_{N_L(\li)},\zee_{N_R(\ri)} | \zee_{N_R(V)}=\beta]}} \leq \eta
	\]
By Markov's inequality,
	\[
	\Pr{\substack{V \sub R,|V|=k \\ \beta \sim \zee_{N_R(V)}}}{\Ex{\li,\ri}{\tildecov[\zee_{N_L(\li)},\zee_{N_R(\ri)} | \zee_{N_R(V)}=\beta]}> \frac{\eps_2}{15}} \leq \frac{15\eta}{\eps_2}
	\]
By choosing $\eta=\eps_2^2/60\delta$, we get
\begin{equation}\label{eqn:low_covariance}
	\Pr{\substack{V \sub R,|V|=k \\ \beta \sim \zee_{N_R(V)}}}{\Ex{\li,\ri}{\tildecov[\zee_{N_L(\li)},\zee_{N_R(\ri)} | \zee_{N_R(V)}=\beta]} \leq \frac{\eps_2}{15}} \geq 1- \eps_2/4\delta
\end{equation}
For some $V,\beta$ such that \cref{eqn:low_covariance} holds, we will be using $\tildeEx{~ \cdot~ |
  \zee_{N_R(V)}=\beta}$ as an $\eta$-good pseudocodeword. Using $\eps_2 < \delta$, note that 
\begin{align*}
		\frac{\eps_2}{15} < \frac{\delta}{15} = \frac{\delta_0\cdot (\delta_0-\lambda)}{15}< \frac{\delta_0^2}{9}
\end{align*}
and $\lambda\leq \delta_0/3$ so that the conditions of \cref{lem:distance_from_codeword} are satisfied. 
	 
	 For this $\eta$-good pseudocodeword, we need to argue that it is still close to $h$ that we are trying to find. This is easy to ensure in expectation, and we again appeal to Markov's inequality to say that it also holds with significant probability, up to some loss in distance. By the law of total expectation, for any $V \sub R$,
\[
		\Ex{\beta \sim \zee_{N_R(V)}}{\dis\inparen{\tildeEx{~\cdot~ | \zee_{N_R(V)}=\beta},h}} = \dis(\tildeEx{\cdot},h) \leq \delta-\eps_2
\]
Averaging over all $V\sub R$ of size $k$,
\begin{equation}
		\Ex{\substack{V \sub R,|V|=k \\ \beta \sim \zee_{N_R(V)}}}{\dis\inparen{\tildeEx{\ \cdot\ | \zee_{N_R(V)}=\beta},h}} \leq \delta-\eps_2
\end{equation}
Again, we claim via Markov's inequality that a significant fraction of conditionings must end up being not too far from $f$.
\begin{equation}\label{eqn:good_agreement}
		\Pr{\substack{V \sub R,|V|=k \\ \beta \sim \zee_{N_R(V)}}}{\dis\inparen{ \tildeEx{\ \cdot\ | \zee_{N_R(V)}=\beta},h} \leq \delta-\frac{\eps_2}{2}} \geq \frac{\eps_2/2}{\delta-\eps_2+\frac{\eps_2}{2}} \geq \frac{\eps_2}{2\delta}
\end{equation}
Henceforth, we fix a conditioning $(V,\beta)$ with $\beta\in [q]^{N_R(V)}$ such that events in both \cref{eqn:low_covariance} and \cref{eqn:good_agreement} happen. Note that by a union bound, a random $(V,\beta)$ has this property with probability at least $\eps_2/4\delta$. Fix such a conditioning, and let the conditioned pseudoexpectation be \[ \dupPE{\cdot} = \tildeEx{~\cdot~ | \zee_{N_R(V) = \beta}}\] and the corresponding covariance operator be $\dupCov[\cdot]$ Note that the degree of $\dupPE{\cdot}$ is $t-2d\cdot k \geq 4d$. From definition, we know that
\begin{align}
\dis(\dupPE{\cdot},h) \leq \delta-\frac{\eps_2}{2} \label{eqn:good_agreement_new_operator}\\
\Ex{\li,\ri}{\dupCov[ \zee_{N_L(\li)}, \zee_{N_R(\ri)}]} \leq \frac{\eps_2}{15}\label{eqn:low_covariance_new_operator}
\end{align}
From \cref{eqn:low_covariance_new_operator} and \cref{lem:distance_from_codeword}, we know that $\dis(\dupPE{\cdot},h) \leq \eps_2/5$ or $\dis(\dupPE{\cdot},h) \geq \delta - \eps_2/5$. The latter is impossible because of \cref{eqn:good_agreement_new_operator}, and so we must have $\dis(\dupPE{\cdot},h) \leq \eps_2/5 < \delta/5$.
	
Finally, we use \cref{lem:unique_decoding} to recover $h$ using $\dupPE{\cdot}$ with probability at least $1/5$.
The algorithm succeeds if 
\begin{enumerate}
\item a $k$ is picked so that \cref{eqn:random_conditioning} holds, 
\item $(V,\beta)$ is picked so that events in both \cref{eqn:low_covariance} and \cref{eqn:good_agreement} happen,
\item the call to \cref{lem:unique_decoding} succeeds.
\end{enumerate}
The success probability is therefore at least
\[
\Pr{\text{success}} ~\geq~ \frac{\eta^2}{q^{2d}} \cdot \frac{\eps_2}{4\delta} \cdot \frac{1}{5} ~\geq~
\Omega \left( \frac{\eps_2^5}{\delta^3\cdot q^{2d}} \right) ~=~ \Omega_{q,d}(\eps^5) \mper
\]
We have shown that for any $h\in \calL((g,\calJ(\delta)-\eps)$, the algorithm above outputs $h$ with probability at least $\Omega_{q,d}(\eps^5)$. Note that this implicitly proves an upper bound on the list of $\calO_{q,d}(1/\eps^5)$.
	
Therefore, the random choices that the algorithm makes lead it to different elements of the list. We next argue that we can derandomize all random choices in the algorithm, so that all elements of the list can be found with a deterministic algorithm.
	\begin{enumerate}
		\item For the random choice of $k$, we can try out all possible $q^{3d}/\eta^2$ values for $k$. 
		\item For random $(V,\beta)$, we can again try out all possible values, which are at most $n^k\cdot 2^k \leq n^{\calO_{q,d}(1/\eps^4)}$ in number.
		\item \cref{lem:unique_decoding} can be derandomized using a standard threshold rounding argument, as argued in \cref{lem:derandomized_decoding_from_distributions}.
	\end{enumerate}
Thus, the final algorithm starts with an empty list and goes over all the deterministic choices above. Every $h\in \calL(g,\calJ(\delta)-\eps)$ will be discovered in at least one of these deterministic steps, and we can efficiently check whether $\Delta(g,h) <\calJ(\delta)-\eps$. If yes, $h$ is added to the output list.
\end{proof}

\subsection{AEL Code}\label{sec:list_decoding_ael}
Let $\AELC$ be an AEL code determined by an $(n,d,\lambda)$-expander graph $G(L,R,E)$, an inner code $\calC_0$ of distance $\delta_0$, rate $r_0$, alphabet size $q_0$, and an outer code $\calC_1$ of distance $\delta_1$, rate $r_1$ and alphabet size $q_1 = |\calC_0|$. The code $\AELC$ is of alphabet size $q_0^d$, rate $r_0r_1$ and (designed) distance $\delta_0-\frac{\lambda}{\delta_1}$. 
\begin{theorem}[List Decoding AEL codes]\label{thm:list_decoding_ael}
	Suppose the code $\calC_1$ can be efficiently unique-decoded from radius $\delta_{dec}$. Assume $\lambda \leq \kappa \cdot \delta_{dec} \leq \kappa \cdot \delta_1$, so that the distance of $\AELC$ is at least $\delta_0-\kappa$. Then for any $\eps > 0$, the code $\AELC$ can be list decoded from a radius of $\calJ(\delta_0-\kappa)-\eps$ by using $\calO_{q,d,\delta_{dec}}(1/\eps^4)$ levels of SoS-hierarchy, in time $n^{\calO_{q,d,\delta_{dec}}(1/\eps^4)}$.
\end{theorem}
\begin{proof}
Let $g\in [q_0^d]^R$ be a received word. Recall that the distance of an AEL codeword $h$ with $g$ is given by $\dis^R(g,h) = \Ex{\ri\in R}{\indi{g(\ri)\neq h_{N_R(\ri)}}}$.
	
We again start by applying the algorithmic covering \cref{lem:cover_for_list_ael} to get a pseudocodeword $\tildeEx{\cdot}$ of SoS-degree $t\geq 2d(\frac{q^{3d}}{\eta^2}+2)$ such that for every $h \in \calL(g,\calJ(\delta_0-\kappa)-\eps)$,
\[
\Ex{\ri}{\tildeEx{\indi{\zee_{N_R(\ri)} \neq h_{N_R(\ri)}}}} = \dis^R(\tildeEx{\cdot},h) \leq (\delta_0-\kappa)-\eps_2
\]
Here $\eta>0$ is a small constant to be chosen later, and $\eps_2 = 2\eps \sqrt{1-\frac{q^d}{q^d-1}(\delta_0-\kappa)} \geq \Omega(\eps)$.
	Henceforth, we fix an $h\in \calL\left((g,\calJ(\delta_0-\kappa) - \eps \right)$, so that $\dis^R(\tildeEx{\cdot},h) \leq \delta_0-\kappa -\eps_2$. Our goal is to recover $h$.
	
	Pick a random $k \in \{1,\cdots,\ceil{\frac{q^{3d}}{\eta^2}}\}$. From \cref{lem:low_covariance_solution}, we know that with probability at least $\frac{\eta^2}{q^{3d}}$,
	\[
		\Ex{v_1,v_2,\cdots ,v_k}{\Ex{\li,\ri}{\tildecov[\zee_{N_L(\li)},\zee_{N_R(\ri)} | \zee_{N_R(v_1)},\zee_{N_R(v_2)},\cdots ,\zee_{N_R(v_k)}]}} \leq \eta
	\]
We assume that we found a $k$ such that \cref{eqn:random_conditioning} holds. Let $V$ be the (random) set of $k$ vertices we condition on, that is, $V = \{v_1,v_2,\cdots,v_k\}$, and let $N_R(V) = \cup_{v\in V} N_R(v)$. Then,
\[
\Ex{V\sub R, |V|=k}{\Ex{\li,\ri}{\tildecov[\zee_{N_L(\li)},\zee_{N_R(\ri)} | \zee_{N_R(V)}]}} \leq \eta
\]
More explicitly, conditioning on $N_R(V)$ involves sampling an assignment for $N_R(V)$ according to the local distribution of $\zee_{N_R(V)}$. Let this random assignment be $\beta$, and we get
\[
\Ex{\substack{V \sub R,|V|=k \\ \beta \sim
    \zee_{N_R(V)}}}{\Ex{\li,\ri}{\tildecov[\zee_{N_L(\li)},\zee_{N_R(\ri)} | \zee_{N_R(V)}=\beta]}}
\leq \eta \mper
\]
By Markov's inequality, and by choosing $\eta = \frac{\eps_2^2\delta_{dec}}{16(\delta_0-\kappa)}$,
	\begin{gather}
	\Pr{\substack{V \sub R,|V|=k \\ \beta \sim \zee_{N_R(V)}}}{\Ex{\li,\ri}{\tildecov[\zee_{N_L(\li)},\zee_{N_R(\ri)} | \zee_{N_R(V)}=\beta]}> \frac{\delta_{dec} \cdot \eps_2}{4}} \leq \frac{4\eta}{\delta_{dec}\cdot\eps_2} \leq \frac{\eps_2}{4(\delta_0-\kappa)} \\
	\Pr{\substack{V \sub R,|V|=k \\ \beta \sim \zee_{N_R(V)}}}{\Ex{\li,\ri}{\tildecov[\zee_{N_L(\li)},\zee_{N_R(\ri)} | \zee_{N_R(V)}=\beta]} \leq \frac{\delta_{dec} \cdot \eps_2}{4}} \geq 1-\frac{\eps_2}{4(\delta_0-\kappa)}\label{eqn:eta_good_ael}
	\end{gather}
	As in the Tanner case, we next claim that the distance is preserved with significant probability when conditioning randomly. Let $h \in \calL(g,\calJ(\delta_0-\kappa)-\eps)$ so that $\dis^R(\tildeEx{\cdot},h) < \delta_0-\kappa -\eps_2$. Using a similar argument as in the Tanner case,
	\begin{equation}
		\Ex{(V,\beta)}{\dis^R\left( \tildeEx{\ \cdot\ | \zee_{N_R(V)}=\beta},h\right)} \leq (\delta_0-\kappa)-\eps_2
	\end{equation}
which allows us to claim via Markov's inequality that
\begin{align}\label{eqn:good_agreement_ael}
		\Pr{(V,\beta)}{\dis^R \left( \tildeEx{\ \cdot\ | \zee_{N_R(V)}=\beta},h \right) \leq (\delta_0-\kappa)-\frac{\eps_2}{2}} ~\geq~ \frac{\eps_2/2}{(\delta_0-\kappa)-\eps_2+\eps_2/2} 
		~\geq~ \frac{\eps_2}{2(\delta_0-\kappa)}
	\end{align}
Again, let $(V,\beta)$ be a conditioning such that events in both \cref{eqn:eta_good_ael} and \cref{eqn:good_agreement_ael} hold (which happens with probability at least $\frac{\eps_2}{4(\delta_0-\kappa)}$). Let $\dupPE{\cdot} = \tildeEx{\cdot | \zee_{N_R(V)} =\beta}$ be an $\eta$-good pseudocodeword that satisfies \cref{eqn:good_agreement_ael}.
This means
\[
\delta_0 - \frac{\lambda+\eta}{\dis^L(\dupPE{\cdot},h)} ~\leq~ \dis^R(\dupPE{\cdot},h) ~\leq~
(\delta_0-\kappa)-\eps_2/2 
\]
Rearranging, we get that
\[
\kappa+\eps_2/2 ~\leq~ \frac{\lambda+\eta}{\dis^L(\dupPE{\cdot},h)}
~\leq~ \frac{\lambda+\delta_{dec}\cdot\eps_2/4}{\dis^L(\dupPE{\cdot},h)}\\
~\leq~ \frac{\kappa \cdot \delta_{dec}+\delta_{dec}\cdot\eps_2/4}{\dis^L(\dupPE{\cdot},h)} \mcom
\]
which gives the required bound on $\dis^L(\dupPE{\cdot},h)$ as
\[
\dis^L(\dupPE{\cdot},h) ~\leq~ \delta_{dec} - \frac{\delta_{dec} \eps_2/4}{\kappa+\eps_2/2} 
		~\leq~ \delta_{dec} - \frac{\delta_{dec} \eps_2}{4\delta_0}
\]
Finally, we use \cref{lem:unique_decoding_ael} to find $h$ using $\dupPE{\cdot}$ with probability at least $\eps_2/4\delta_0$.
The final success probability is at least
\[
\frac{\eta^2}{q^{3d}} \cdot \frac{\eps_2}{4(\delta_0-\kappa)} \cdot \frac{\eps_2}{4\delta_0} ~\geq~
\Omega\inparen{\frac{\eps_2^6\delta_{dec}^2}{q^{3d}\delta_0^4}} ~\geq~
\Omega_{q,d,\delta_{dec}}(\eps^6) \mper
	\]
	Just as in the case of Tanner codes, this algorithm can be derandomized by trying out all possible random choices made by the algorithm, to give a deterministic algorithm that recovers the list.
\end{proof}
Note that while the theorem above deals with list decoding, it can be easily adapted for list recovery by replacing the use of \cref{lem:cover_for_list_ael} by \cref{lem:cover_for_list_recovery_ael}.
Next, we use the AEL amplification scheme to construct near-MDS codes list decodable up to the Johnson bound.

\begin{theorem}\label{thm:near_mds_main}
	For any $\nfrac{1}{2} >\eps_1, \eps_2>0$, there is an infinite family of codes $\calC$ of blocklength $n$ with the following properties:
	\begin{enumerate}[(i)]
		\item The rate of the code is $\rho$ and distance is at least $1-\rho-\eps_1$.
		\item The code is over an alphabet of size $2^{\calO(\eps_1^{-6}\log(1/\eps_1))}$.
		\item The code can be list decoded from radius $\calJ(1-\rho - \eps_1)-\eps_2$ in time $n^{\calO_{\eps_1}(1/\eps_2^4)}$.
	\end{enumerate}
\end{theorem}

\begin{proof}
We sketch how to instantiate \cref{thm:list_decoding_ael} to obtain such codes.

Suppose we are working with $(n,d,\lambda)$-expander.

Choose the inner code $\calC_0$ to be a Reed Solomon code of rate $\rho_0$, distance $1-\rho_0$ and alphabet size $q_0=d$ (or any MDS code). Choose the outer code $\calC_1$ over alphabet of size $|\calC_0| = d^{\rho_0\cdot d}$ to have rate $1-\eps_1$ that can be unique decoded from radius $\delta_{dec}=\Omega(\eps_1^2)$, such as the one constructed in \cite{GI05}.

Let $\lambda = \kappa\delta_{dec}$ with $\kappa =\eps_1$, so that $\lambda \leq \Theta(\eps_1^3)$ and $d=\Theta(1/\eps_1^6)$.

The rate of the final AEL code is $\rho \defeq (1-\eps_1) \rho_0$, and the distance is at least
\begin{align*}
	(1-\rho_0) - \frac{\lambda}{\delta_1} &~\geq~ 1-\frac{\rho}{1-\eps_1} - \frac{\lambda}{\delta_{dec}} \\
	&~\geq~1-\rho -2\eps_1 \rho -\eps_1 \\
	&~\geq~1-\rho - 3\eps_1
\end{align*}

The alphabet size is $q_0^d = d^d  = 2^{ \calO\inparen{\eps_1^{-6}\log(1/\eps_1)}}$.

For list decodability, we use \cref{thm:list_decoding_ael} to claim that the above code can be list decoded from $\calJ(1-\rho-3\eps_1)-\eps_2$ in time $n^{\calO_{q_0,d,\delta_{dec}}(1/\eps_2^4)} = n^{\calO_{\eps_1}(1/\eps_2^4)}$.

Replace $\eps_1$ by $\eps_1/3$ to get the final result.
\end{proof}

Note that we can also deal with $\calC_1$ that can be decoded from smaller radius like $\calO(\eps_1^3)$, by suitably adjusting $\lambda$ and paying the cost in alphabet size. Above, we have not chosen parameters optimally to keep the exposition simple. The alphabet size in the code constructed in \cite{GI05} was smaller than what we ask for here, but alphabet size can always be increased while preserving rate, distance and unique decoding radius by folding multiple symbols together into one. This looks like multiple symbols of the outer code being assigned to the same left node in the AEL construction.

\subsection{Decoding from fractional vectors}\label{sec:decoding_from_fractional}
This section has auxiliary claims needed to finish the list decoding algorithm proof. In the Tanner code case, we reduce to unique decoding of the same code from an arbitrarily small radius. In AEL, we reduce to unique decoding of the base (outer) code. In both these cases, we have some $\eps$ slack, that is sufficient for randomized rounding to produce a (corrupted) word within the unique decoding radius.
\begin{lemma}\label{lem:decoding_from_distributions}
	Let $\calC$ be an $[n,\delta,\rho]_q$ code, which is unique decodable from distance $\delta_{dec}\leq \delta/2$ in time $\calT(n)$. Given a collection of distributions $\calD = \inbraces{\calD_i}_{i\in [n]}$, each of them supported on $[q]$, there is a unique codeword $h$ that satisfies
	\[
		\Ex{i}{\Ex{j\sim \calD_i}{\indi{h_i \neq j}}} \le \delta_{dec}-\eps
	\]
	This codeword $h$ can be found in time $\calO(qn)+\calT(n)$ with probability at least $\eps/\delta_{dec}$.
\end{lemma}
\begin{proof}
First, we show uniqueness of $h$. Let $h, h'$ be two codewords in $\calC$ such that 
\[		
\Ex{i}{\Ex{j\sim \calD_i}{\indi{h_i \neq j}}} ~\le~ \delta_{dec}-\eps 
\qquad \text{and} \qquad 
\Ex{i}{\Ex{j\sim \calD_i}{\indi{h'_i \neq j}}} ~\le~ \delta_{dec}-\eps
\]
For any $j\in [q]$, we have $\indi{h_i\neq h'_i} \leq \indi{h_i\neq j}+\indi{h'_i\neq j}$. Therefore,
\begin{align*}
\dis(h,h') ~=~ \Ex{i}{\indi{h_i\neq h'_i}} &~\leq~ \Ex{i}{\Ex{j\sim \calD_i}{\indi{h_i\neq j}+\indi{h'_i\neq j}}} \\
		&~=~ \Ex{i}{\Ex{j\sim \calD_i}{\indi{h_i\neq j}}}+\Ex{i}{\Ex{j\in \calD_i}{\indi{h'_i\neq j}}} 
		~\leq~ 2\delta_{dec}-2\eps < \delta \mper
\end{align*}
By distance property of the code, this means that $\dis(h,h') = 0$, or that $h=h'$.

The algorithm to find $h$ is to independently sample from every distribution to get a random $g \in [q]^n$, and then issue a unique decoding call from $g$. We show that with significant probability, $h$ lies in the $\delta_{dec}$ radius ball around $g$, which will show that algorithm succeeds with that probability.
\begin{align*}
\Ex{g}{\dis(g,h)} ~=~ \Ex{g}{\Ex{i}{\indi{g_i\neq h_i}}} 
		~=~ \Ex{i}{\Ex{g}{\indi{g_i\neq h_i}}} 
		~=~ \Ex{i}{\Ex{g_i \sim \calD_i}{\indi{g_i\neq h_i}}} 
		~\leq~ \delta_{dec}-\eps \mper
	\end{align*}
Thus, by Markov's inequality, we have, $\Pr{g}{\dis(g,h) \leq \delta_{dec}} \geq
\frac{\eps}{\delta_{dec}}$, which proves the claim.
\end{proof}
\begin{remark}
		The success probability in the \cref{lem:decoding_from_distributions} can be amplified by repeated sampling. Moreover, the fact that we reduce to the unique decoding algorithm of $\calC$ is not important, as it is also possible to use a list decoding algorithm for $\calC$ from distance $\delta_{dec}$ to find $h$, as long as the sampled $g$ satisfies $\delta(g,h) < \delta_{dec}$. In that case, we output a random element of list obtained, which incurs an additional loss of $1/L$ factor, where $L$ is the list size guaranteed by list decoding algorithm for $\calC$ up to $\delta_{dec}$.
	\end{remark}
	
	\begin{remark}
		As shown in \cref{lem:derandomized_decoding_from_distributions}, this argument can be derandomized using threshold rounding. The use of \cref{lem:decoding_from_distributions} in the next two lemmas can therefore be replaced by \cref{lem:derandomized_decoding_from_distributions}.
	\end{remark}
	
	Next, we use pseudocodewords that lie in the unique decoding ball to construct the collection of distributions needed by \cref{lem:decoding_from_distributions}, for the Tanner and AEL cases.
	\begin{lemma}[Unique decoding from Tanner pseudocodewords]\label{lem:unique_decoding}
		Let $\TanC$ be a code as in \cref{sec:list_decoding_tanner}, with distance $\delta = \delta_0(\delta_0-\lambda)$ and $\lambda\leq \delta_0/3$, and in particular it can be unique decoded from radius $\delta/4$ in time $\calO(|E|)$. Given a Tanner pseudocodeword $\tildeEx{\cdot}$ such that $\dis(\tildeEx{\cdot},h) <\delta/5$, we can find $h$ in time $\calO(|E|)$ with probability at least $1/5$.
	\end{lemma}
	\begin{proof}
		The pseudocodeword gives a collection of distributions $\calD = \inbraces{\calD_e}_{e\in E}$, each distribution over $[q]$. The $e^{th}$ distribution $\calD_e$ gives a weight of $\tildeEx{\indi{\zee_e=j}}$ to the value $j\in [q]$.
		
		The distance $\dis(\tildeEx{\cdot},h)$ translates to
		\begin{align*}
			\dis(\tildeEx{\cdot},h) &= \Ex{e}{\tildeEx{\indi{\zee_e \neq h_e}}} \\
			&= \Ex{e}{\tildeEx{ \sum_{j\in [q]}  \indi{\zee_e =j }\indi{h_e \neq j}}} \\
			&=\Ex{e}{ \sum_{j\in [q]} \indi{h_e \neq j} \tildeEx{  \indi{\zee_e =j }}} \\
			&=\Ex{e}{ \Ex{j\sim \calD_e}{ \indi{h_e \neq j}}} \\
		\end{align*}
		We can therefore use this collection of distributions to obtain a codeword $h\in \TanC$ via \cref{lem:decoding_from_distributions} with $\delta_{dec} = \delta/4$ and $\eps = \delta/20$ such that $\dis(\tildeEx{\cdot},h) <\delta/5$.
	\end{proof}

	Next, we use similar ideas to round and decode from AEL pseudocodewords. We borrow the terminology used for AEL codes from \cref{sec:list_decoding_ael}.
	\begin{lemma}[Unique decoding from AEL pseudocodewords]\label{lem:unique_decoding_ael}
		Let $\AELC$ be a code as in \cref{sec:list_decoding_ael}, with $\lambda\leq \kappa\cdot \delta_{dec}$  and distance at least $\delta_0 - \kappa$. Assume that the outer code $\calC_1$ can be unique decoded from radius $\delta_{dec}$ in time $\calT(n)$. Given an AEL pseudocodeword $\tildeEx{\cdot}$ and $h \in \AELC$ such that $\dis^L(\tildeEx{\cdot},h) \leq \delta_{dec}-\eps$, we can find $h$ in time $\calO(n) + \calT(n)$ with probability at least $\eps/\delta_{dec}$.
	\end{lemma}

	\begin{proof}
		First, we use the given pseudocodewords to build a collection of distributions $\calD = \inbraces{\calD_{\li}}_{\li \in L}$, each distribution over $[q_1]$. Recall that $\calC_0$ can be seen as a map from $[q_1]$ to $[q_0]^d$. The ${\li}^{th}$ distribution gives a weight of $\tildeEx{\indi{\zee_{N_L(\li)} = \calC_0(\alpha)}}$ to $\alpha \in [q_1]$. 
		
		Let $\overline{h}$ be the codeword in $\calC_1$ corresponding to the codeword $h$. That is, $\overline{h}$ is such that $\calC_0(\overline{h}_{\li}) = h_{N_L(\li)}$. With the collection of distributions $\calD$ defined, we relate $\dis^L(\tildeEx{\cdot},h)$ to agreement of $\overline{h}$ with $\calD$.
		\begin{align*}
			\dis^L(\tildeEx{\cdot},h) &= \Ex{\li}{\tildeEx{\indi{\zee_{N_L(\li)} \neq h_{N_L(\li)}}}} \\
			&= \Ex{\li}{\tildeEx{ \sum_{\alpha\in [q_1]} \indi{\zee_{N_L(\li)} = \calC_0(\alpha)} \indi{h_{N_L(\li)} \neq \calC(\alpha)} }} \\
			&=\Ex{\li}{\sum_{\alpha\in [q_1]} \indi{h_{N_L(\li)} \neq \calC_0(\alpha)} \tildeEx{ \indi{\zee_{N_L(\li)} = \calC_0(\alpha)}  }} \\
			&=\Ex{\li}{ \Ex{\alpha\sim \calD_{\li}}{ \indi{h_{N_L(\li)} \neq \calC_0(\alpha)}}} \\
			&=\Ex{\li}{ \Ex{\alpha\sim \calD_{\li}}{ \indi{\overline{h}_{\li} \neq \alpha}}}
		\end{align*}
		
		We call \cref{lem:decoding_from_distributions} for the code $\calC_1$ with the collection of distributions $\calD$ and find $\overline{h}$, and therefore $h$, with probability at least $\eps/\delta_{dec}$.
	\end{proof}
	
	To end this section, we note that the rounding from fractional vectors above can be derandomized through a standard method known as threshold rounding. Similar ideas are used to derandomize the classical Generalized Minimum Distance decoding for concatenated codes.
	
	\begin{lemma}\label{lem:derandomized_decoding_from_distributions}
	Let $\calC$ be an $[n,\delta,\rho]_q$ code, which is unique decodable from distance $\delta_{dec}\leq \delta/2$ in time $\calT(n)$. Given a collection of distributions $\calD = \inbraces{\calD_i}_{i\in [n]}$, each of them supported on $[q]$ described as a collection $q$ weights that sum to 1, there is a unique codeword $h$ that satisfies
	\[
		\Ex{i}{\Ex{j\sim \calD_i}{\indi{h_i \neq j}}} \le \delta_{dec}
	\]
	This codeword $h$ can be found in time $\calO(qn)\cdot\calT(n)$ with a deterministic algorithm.
	
	\begin{proof}
		The uniqueness of $h$ is as in proof of \cref{lem:decoding_from_distributions}.
		
		Let the weight on $j \in [q]$ according to $\calD_i$ be $w_{ij}$, so that $\sum_{j\in [q]} w_{ij} = 1$. We replace the randomized rounding of \cref{lem:decoding_from_distributions} by the following process:
		\begin{enumerate}[(i)]
			\item Define the cumulative sums $c_{ij} = \sum_{j' \leq j}{ w_{ij'}}$, so that $c_{iq} = 1$. Define $c_{i0} = 0$.
			\item For each $i\in [n]$, embed $\calD_i$ into the interval $[0,1]$ as $q+1$ points $(c_{i0} = 0, c_{i1}, c_{i2},\cdots , c_{iq} = 1)$.
			\item Choose $\theta \in [0,1]$ uniformly at random.
			\item Build $h'\in [q]^n$ coordinate-wise as follows. For $i\in [n]$, if 
			\[	
				\theta\in \left[c_{i(j-1)}, c_{ij}\right),
			\]
			then $h'_i = j$. This ensures that $\Pr{\theta}{h'_i = j} = w_{ij}$.
		\end{enumerate}
		We show that $h'$ has the same distance from $h$ in expectation as the collection of distributions $\inbraces{\calD_i}_{i\in [n]}$.
		\begin{align*}
			\Ex{\theta \in [0,1]}{\Delta(h',h)} &= \Ex{\theta \in [0,1]}{\Ex{i}{\indi{h'_i \neq h_i}}} \\
			&= \Ex{\theta \in [0,1]}{\Ex{i}{ \indi{\theta \not\in [c_{i(h_i-1)}, c_{ih_i}) }}} \\
			&= \Ex{i}{ \Ex{\theta \in [0,1]}{\indi{\theta \not\in [c_{i(h_i-1)}, c_{ih_i}) }}} \\
			&= \Ex{i}{1 - w_{ih_i}} \\
			&= \Ex{i}{\Ex{j\sim \calD_i}{\indi{h_i \neq j}}} \leq \delta_{dec}
		\end{align*}
		Therefore, rounding according to a random threshold $\theta$ produced an $h'$ that is at most $\delta_{dec}$ distance away from $h$. As the final step to derandomization, note that two thresholds $\theta_1,\theta_2$ produce the exact same $h'$ if there is no point from step (ii) above embedded between $\theta_1$ and $\theta_2$. Total number of points embedded is at most $q\cdot n$, and so it suffices to try at most $\calO(q\cdot n)$ many thresholds to produce all the different $h'$ possible - one of which must be $\delta_{dec}$-close to $h$.
	\end{proof}
\end{lemma}

%
\section{List Decoding Concatenated Codes}
\label{sec:concat}

In this section, we will adapt the techniques developed earlier to decode concatenated codes. Concatenation is a useful operation to obtain codes with smaller alphabet size starting from a code over large alphabet. Previous works on list decoding of concatenated codes \cite{GS00, GR08, GuruswamiS02} seem to all rely on list recovery of outer code, with intricate weights to be passed along with inner codewords. We will only use list decodability of outer code.

First, we show that the covering lemma based argument can be used to decode the concatenated code up to the Johnson radius corresponding to product of decoding radius of outer code and the distance of inner code. That is, if the distance of inner code is $\delta_0$, distance of outer code is $\delta_1$, and the decoding radius of outer code is $\delta_{dec}$, we can decode the concatenated code up to radius $\calJ(\delta_{dec} \cdot \delta_0)$. 

Moreover, since concatenated codes do not involve expansion for their distance proof, we do not need to deal with SoS-based pseudocodewords or any low-covariance conditions. In fact, our pseudocodewords will just be local distributions over the inner code for each coordinate of the outer code. Since this set of pseudocodewords can be described as the feasible set corresponding to linear constraints over $\calO(n)$ variables, we can minimize the appropriate norm (which is a convex function) via Ellipsoid method to get a pseudocodeword with the covering property in time $n^{\calO(1)}$.

Note that this is weaker than decoding up to $\calJ(\delta_1\cdot\delta_0)$, which is the Johnson radius corresponding to the true distance of the concatenated code. In \cref{sec::concat_outer_ael}, we will see that by using the decoder of outer code in a white box way, we can get to the Johnson bound $\calJ(\delta_1\cdot\delta_0)$ when the outer code supports list decoding through our Covering Lemma based machinery, like in the case of near-MDS codes of \cref{thm:near_mds_main}.

\subsection{List decoding arbitrary concatenated codes}

Let the $[n,\delta_1,\rho_1]_{q_1}$ outer code be $\calC_1$  and the $[d,\delta_0,\rho_0]_{q_0}$ inner code be $\calC_0$, with $q_1 = q_0^{\rho_0 \cdot d}$. Let the concatenated code be $\calC^*$ with distance at least $\delta = \delta_0\cdot \delta_1$. A codeword $h^*=\calC_0(h)$ of the concatenated code can be seen as a tuple $h^* = (h^*_1,h^*_2,\cdots,h^*_n)$ where each $h^*_i\in \calC_0$, or $h^*_i = \calC_0(h_i)$ for some $h_i\in [q_1]$. Note that not all tuples of this form are codewords, and $(\calC_0(f_1),\calC_0(f_2),\cdots,\calC_0(f_n)) \in \calC^*$ iff $(f_1,f_2,\cdots,f_n)\in \calC_1$. 
	
\begin{definition}
	A pseudocodeword of the concatenated code is a psuedoexpectation operator $\tildeEx{\cdot}$ of degree $d$ over the variables $\zee = \inbraces{Z_{i,j,k}}_{i\in [n],j\in [d],k\in[q_0]}$ that respects the following $d$-local constraints:
	\begin{enumerate}[(i)]
		\item $Z_{i,j,k}^2 ~=~ Z_{i,j,k}$
		\item For every $i\in[n],j\in [d]$, $\sum_{k\in[q_0]} Z_{i,j,k} =1$
		\item	 $\forall i\in [n], \quad (\zee_{i,1},\zee_{i,2},\cdots,\zee_{i,d}) \in \calC_0$.
	\end{enumerate}
\end{definition}

We no longer enforce the constraint for non-negativity of squares of polynomials, and in fact, our pseudoexpectation operators are just a collection of $n$ distributions $\inbraces{\calD_i}_{i\in [n]}$ over $[q_1]$. The weight assigned to $f\in [q_1]$ in distribution $\calD_i$ is $\tildeEx{\indi{\zee_i = \calC_0(f)}}$.

Following is the natural generalization of distances to pseudocodewords.

\begin{definition}[Distance from a pseudocodeword]
The distance of a pseudocodeword $\tildeEx{\cdot}$ and a codeword $h^*$ is defined as
	\[
		\dis(\tildeEx{\cdot},h^*) = \Ex{i\in [n]}{\tildeEx{\dis(\zee_i,h^*_i)}}
	\]
\end{definition}

\begin{lemma}\label{lem:decode_from_pseudo}
	Assume that $\calC_1$ can be list-decoded from radius $\delta_{dec}$ in time $\calT(n)$ with list size $L$.
	
	For any $\eps>0$, there is a deterministic algorithm that given a pseudocodeword $\tildeEx{\cdot}$ outputs the list 
	\[
		\calL(\tildeEx{\cdot},\delta_0\cdot \delta_{dec}) := \inbraces{h^* \in \calC^* \suchthat \dis(h^*,\tildeEx{\cdot}) < \delta_0 \cdot \delta_{dec}}
	\]
	and runs in time $2^{\calO(d)} n^{\calO(1)}+\calO(2^d \cdot n)\cdot \calT(n)$.
\end{lemma}

\begin{proof}
	We use $\tildeEx{\cdot}$ to get local distributions $\calD_i$ over $[q_1]$ for each coordinate $i$. 
	\begin{align*}
		\dis(\tildeEx{\cdot},h^*) = \Ex{i\in [n]}{\tildeEx{\dis(\zee_i,h^*_i)}} &= \Ex{i\in [n]}{\Ex{f\sim \calD_i}{\dis \inparen{ \calC_0(f),h^*_i }}}\\
		&\geq \Ex{i\in [n]}{\Ex{f\sim \calD_i}{\indi{f\neq h_i}\dis \inparen{ \calC_0(f),h^*_i } }}\\
		&\geq \Ex{i\in [n]}{\Ex{f\sim \calD_i}{\indi{f\neq h_i} \delta_0}}\\
		&\geq \Ex{i\in [n]}{\Ex{f\sim \calD_i}{\indi{f\neq h_i}\dis \inparen{ \calC_0(f),h^*_i } }}\\
		&= \delta_0 \cdot \Ex{i\in [n]}{\Ex{f\sim \calD_i}{\indi{f\neq h_i}}}
	\end{align*}
	Any codeword $h^* =\calC_0(h) \in \calL$ must have the property that $\dis(\tildeEx{\cdot},h^*) <\delta_{dec} \cdot \delta_0$, so that
	\[
		\Ex{i}{\Ex{f\sim \calD_i}{\indi{f\neq h_i}}} < \delta_{dec}.
	\]
	Finally, we use \cref{lem:derandomized_decoding_from_distributions} to decode from the above agreement. The only modification is that since we might potentially deal with the \emph{list} decoding algorithm of outer code $\calC_1$, we take a union of all the lists generated by the different calls corresponding to different thresholds, and then prune it finally. For any $h^* = \calC_0(h)\in \calL$, there is some threshold for which the $[q_1]^n$ string generated in \cref{lem:derandomized_decoding_from_distributions} will be at distance $<\delta_{dec}$ from $h$. Therefore, $h$ will be contained in at least one of the lists discovered by the algorithm.
\end{proof}

\begin{lemma}\label{lem:concat_covering}
	For any $\eps>0$, there is an algorithm that given $g\in [q_0]^{nd}$ and $\delta>0$, finds a pseudocodeword $\tildeEx{\cdot}$ such that $\calL (g,\calJ(\delta)-\eps) \subseteq \calL(\tildeEx{\cdot},\delta-\eps_2)$ in time $2^{\calO(d)}\cdot n^{\calO(1)}$ for some $\eps_2>0$.
\end{lemma}

\begin{proof}[Proof Sketch]
	This is again an algorithmic implementation of the covering lemma, with the distributions over codewords to be relaxed to distributions over pseudocodewords as defined above. Minimizing the appropriate norm while optimizing over the convex set of pseudocodewords gives us this covering property. We omit the details since the argument is very similar to \cref{lem:cover_for_list_tanner}.
	
	Since the above optimization problem minimizes a convex quadratic function with linear constraints on at most $2^d\cdot n$ variables, we can get a running time of $2^{\calO(d)}\cdot n^{\calO(1)}$.
\end{proof}

\begin{theorem}
	Let $\calC_0$ be a binary inner code of blocklength $d$, distance $\delta_0$ and rate $\rho_0$. Also let $\calC_{1}$ be an outer code of blocklength $n$, distance $\delta_1$ and rate $\rho_{1}$ on an alphabet of size $|\calC_0|=2^{\rho_0d}$. Assume that $\calC_1$ can be list-decoded from radius $\delta_{dec}$ in time $\calT(n)$ with list size $L$.
	
	Then the code $\calC^*$ obtained by concatenating $\calC_1$ with $\calC_0$ can be list decoded up to a radius of $\calJ(\delta_{dec}\delta_0)$ in time $2^{\calO(d)}n^{\calO(1)} + \calO(2^d \cdot n) \cdot \calT(n)$.
\end{theorem}

\begin{proof}
	Given $g$ such that we wish to find $\calL = \calL(g,\calJ(\delta_{dec}\cdot \delta_0) -\eps)$, we first use \cref{lem:concat_covering} to find a pseudocodeword $\tildeEx{\cdot}$ that covers the list, and then use it via \cref{lem:decode_from_pseudo} to find $\calL$.
\end{proof}

\subsection{List Decoding with outer AEL codes}\label{sec::concat_outer_ael}

In this section, we show that when the near-MDS codes obtained using AEL distance amplification are used for concatenation with smaller alphabet codes, we can list decode the smaller alphabet code up to its Johnson bound. This will not be done by a black-box call to the list decoding/list recovery algorithm of the outer code, but we will crucially use the SoS-based list decoding strategy for the outer code.

To the best of our knowledge, list-decoding to the Johnson bound of the concatenated code has not been achieved for the Reed-Solomon outer code. This shows that while our near-MDS codes via AEL construction match Reed-Solomon codes in terms of list decoding radius, they have some extra desirable features. Of course, the runtime of our algorithms is quite poor compared to the near-linear time Reed-Solomon decoders.

The pseudocodewords for the concatenated code will be concatenations of the pseudocodewords of outer AEL code. Since the covering lemma works irrespective of the code we are working with, we can get a cover for the list of (final) codewords by efficiently optimizing over the pseudocodewords of outer code. Then a simple argument shows that the distance property for outer pseudocodewords translates to distance property for concatenated pseudocodewords.

Let us recall some notation for AEL Codes. Let $\AELC$ be an AEL code determined by an $(n,d,\lambda)$-expander graph $G=(L,R,E)$, an inner code $\calC_0$ of distance $\delta_0$, rate $r_0$, alphabet size $q_0$, and an outer code $\calC_1$ of distance $\delta_1$, rate $r_1$ and alphabet size $q_1 = |\calC_0|$. The code $\AELC$ is of alphabet size $q_0^d$, rate $r_0r_1$ and (designed) distance $\delta_0-\frac{\lambda}{\delta_1}$.

Suppose we concatenate $\AELC$ with an $[d_2,\delta_2,r_2]_{q_2}$ code $\calC_2$ such that $q_0^d = q_2^{r_2\cdot d_2}$, to obtain a new $\left[ nd_2, \inparen{\delta_0-\frac{\lambda}{\delta_1}} \delta_2, r_0r_1r_2\right]_{q_2}$ code $\calC_3$. Note that for the code $\calC_3$, distance is defined for $f_1,f_2\in [q_2]^{nd_2}$ as 
\begin{equation}\label{eqn:dist_for_concat}
	\dis^*(f_1,f_2) = \Ex{\ri\in R,j\in [d_2]}{\indi{f_1(\ri,j) \neq f_2(\ri,j)}}
\end{equation}
which can also be extended to distances between a pseudocodeword $\tildeEx{\cdot}$ and a codeword $h^*\in \calC_3$ as
\[
	\dis^*(\tildeEx{\cdot}, h^*) = \Ex{\ri\in R,j\in [d_2]}{\tildeEx{\indi{\calC_2(\zee_{N_R(\ri)})(j) \neq h^*(r,j) }}}
\]
If the codeword $h^*\in \calC_3$ is obtained by concatenating $h\in \AELC$ with $\calC_2$, we denote $h^* = \calC_2(h)$, and the above distance expression is the same as
\[
	\dis^*(\tildeEx{\cdot}, \calC_2(h)) = \Ex{\ri\in R,j\in [d_2]}{\tildeEx{\indi{\calC_2(\zee_{N_R(\ri)})(j) \neq \calC_2(h_{N_R(\ri)}) (j) }}}
\]
\begin{theorem}
	Let the code $\calC_1$ be decodable from radius $\delta_{dec}$, and $\lambda\leq \kappa \cdot \delta_{dec}$, so that the distance of code $\calC_3$ is at least $\delta_2(\delta_0-\kappa)$. For every $\eps > 0$, the code $\calC_3$ can be list decoded up to the radius $\calJ_{q_2} \inparen{ \inparen{\delta_0-\kappa} \delta_2 } - \eps$ in time $n^{\calO_{q,d,\delta_{dec}}(1/\eps^4)}$.
\end{theorem}

\begin{proof}[Proof Sketch]
We consider the same pseudocodewords as we did for AEL codes, and recall that we proved the following distance property for $\eta$-good pseudocodewords in \cref{lem:AEL_amplification}. For any $h\in \AELC$,

\[
	\Ex{\ri}{\tildeEx{\indi{\zee_{N_R(\ri)} \neq h_{N_R(\ri)}}}} \geq \delta_0 - \frac{\lambda+\eta}{\Ex{\li}{\tildeEx{\indi{\zee_{N_L(\li)} \neq h_{N_L(\li)}}}}}
\]
We extend the above distance property to the distance according to final $\calC_3$,
\begin{align*}
	\Ex{\ri\in R,j\in [d_2]}{\tildeEx{\indi{\calC_2(\zee_{N_R(\ri)})(j) \neq \calC_2(h_{N_R(\ri)})(j)}}}
&= \Ex{\ri}{\tildeEx{\dis\inparen{\calC_2(\zee_{N_R(\ri)}) , \calC_2(h_{N_R(\ri)})}}} \\
	&\geq \delta_2 \cdot \Ex{\ri}{\tildeEx{\indi{\zee_{N_R(\ri)} \neq h_{N_R(\ri)}}}} \\
	&\geq \delta_2 \inparen{ \delta_0 - \frac{\lambda+\eta}{\Ex{\li}{\tildeEx{\indi{\zee_{N_L(\li)} \neq h_{N_L(\li)}}}}} }
\end{align*}

That is,
\begin{equation}\label{eqn:concat_ael_distance}
	\dis^*\inparen{\tildeEx{\cdot}, \calC_2(h)} \geq \delta_2 \inparen{ \delta_0 - \frac{\lambda+\eta}{ \dis^L \inparen{\tildeEx{\cdot},h} }}
\end{equation}
The key distance property established, the rest of the argument is same as while decoding AEL codes. Given $g$ to be decoded, let's call the list of codewords at distance less than $\calJ(\delta_2(\delta_0-\kappa)) - \eps$ as $\calL$. We find a pseudocodeword $\tildeEx{\cdot}$ that is close to all the codewords in $\calL$. For any $h^*\in \calL$ such that $h^* = \calC_2(h)$,
\[
	\dis^*(\tildeEx{\cdot}, h^*) = \dis^*(\tildeEx{\cdot}, \calC_2(h)) < \delta_2(\delta_0-\kappa)-\eps_2
\]
where  $\eps_2 = 2\sqrt{1-\frac{q_2}{q_2-1} \delta_2(\delta_0-\kappa)}\cdot \eps = \Theta(\eps)$.

The covering lemma used here minimizes $\ell_2$-norm of the embedding corresponding to the concatenated code $\calC_3$ with the alphabet size $q_2$, while the relaxation for pseudocodeword was defined for $\AELC$ which does not depend on $\calC_2$ at all. However, this embedding according to $\calC_3$ is a linear function of $\tildeEx{\cdot}$, therefore it is still a convex function that is minimized.

By conditioning $\tildeEx{\cdot}$, we obtain another pseudocodeword $\dupPE{\cdot}$ which is $\eta$-good and retains its closeness to $\calC_2(h)$. By choosing $\lambda$ and $\eta$ small enough, \cref{eqn:concat_ael_distance} allows us to conclude that 
\[ 
	\dis^L \inparen{\tildeEx{\cdot},h} = \Ex{\li}{\tildeEx{\indi{\zee_{N_L(\li)} \neq h_{N_L(\li)}}}} < \delta_{dec}-\eps_3
\] 
for some $\eps_3>0$, which is sufficient to find $h$, and therefore $\calC_2(h)$, via \cref{lem:unique_decoding_ael}.

As before, this algorithm can be derandomized by going over all random choices which must discover the entire list $\calL(g,\calJ_{q_2}((\delta_0-\kappa)\delta_2)-\eps)$.
\end{proof}

Note that this argument does not place any restriction on $\delta_0$, and in particular, if we choose $\kappa$ to be much smaller than $\delta_0$, we can decode arbitrarily close to the Johnson radius corresponding to the product bound $\calJ(\delta_0 \cdot\delta_2)$, even for small values of $\delta_0$ like $1/3$. In contrast, the existing list decoding algorithms for concatenated codes via list recovery of outer Reed-Solomon code \cite{GuruswamiS02} only approach this Johnson bound when the outer distance is very close to 1.

In fact, we can also decode up to the Johnson bound of product of distances for any outer code that supports list decoding up to Johnson bound via our Covering Lemma/SoS-based machinery. This also includes Tanner code in particular.
%
\section{List Decoding Codes on Square Cayley Complex}
\label{sec:square_complex}

First, let us set some notation for the quadripartite left-right square Cayley complex on the group $G$ of size $n$ with generator sets $A,B$ of size $d$ each. Let $\calS$ denote the set of squares. The functions $X,Y,U,V$ are 4 bijections from $G\times A \times B$ to $\calS$, with the following property:
\[
	X(g,a,b) = Y(ga,a^{-1},b) = U(bga,a^{-1},b) = V(bg,a,b^{-1})
\]
The set $X(g,\cdot,b)\subseteq \calS$ is defined as $\{X(g,a,b) \suchthat a\in A\}$, and likewise for $X(g,\cdot,\cdot)$, $X(g,a,\cdot)$ and corresponding $Y,U,V$ sets. The sets $X(g,\cdot,\cdot),Y(g,\cdot,\cdot),U(g,\cdot,\cdot)$ and $V(g,\cdot,\cdot)$ should be seen as the analogs of sets $N_L(\li)$ and $N_R(\ri)$ from the Tanner codes.

The code on this square Cayley complex is then defined as 
\begin{align*}
	\calC^{SCC} = \{h \in [q]^{\calS} \suchthat \forall g\in G,\ & h|_{X(g,\cdot,\cdot)} \in C_A\otimes C_B, \\ 
	&h|_{Y(g,\cdot,\cdot)} \in C_A\otimes C_B, \\
	&h|_{U(g,\cdot,\cdot)} \in C_A\otimes C_B, \\
	&h|_{V(g,\cdot,\cdot)} \in C_A\otimes C_B\}
\end{align*}

where $C_A$ and $C_B$ are inner codes of blocklength $d$ each, and $\calC_A\otimes \calC_B$ is their tensor code.

For the code defined by left-right Cayley complex with inner codes $C_A, C_B$ with parameters $(d,\delta_A,r_A)$ and $(d,\delta_B,r_B)$ respectively, the distance of $\calC^{SCC}$ is lower bounded \cite{DELLM22} by \[ \delta = \delta_A \delta_B(\max(\delta_A,\delta_B) - \lambda)\]Note that the distance of tensor code $\calC_A\otimes \calC_B$ is at least $\delta_A \cdot \delta_B$. 

\begin{theorem}\label{thm:square-decoding}
For every $\eps>0$, there is an algorithm based on $\calO_{q,d}(1/\eps^4)$ levels of the SoS-hierarchy that runs in time $n^{\calO_{q,d}(1/\eps^4)}$ and can list decode $\calC^{SCC}$ up to $\calJ_q(\delta)- \eps$.
\end{theorem}

\begin{proof}[Proof Sketch]
	We outline the proof by once again focusing on a proof of distance for the appropriate notion of $\eta$-good pseudocodewords, and combining this proof with the covering lemma can be done as in the case of Tanner codes. Assume WLOG that $\delta_A\geq \delta_B$.
	
	Consider the SoS relaxation where variables correspond to squares: $\zee =\inbraces{Z_{s,j}}_{s\in \calS,j\in [q]}$. The SoS-degree of this relaxation is at least $d^2$ so that we can enforce all the inner code constraints by making the SoS relaxation respect such constraints explicitly.
	
	Note that for a fixed $b$, the set of squares $\{X(g,a,b) \suchthat g\in G, a\in A\}$ can be seen as the edges of an expander code, with the inner code as $C_A$. This motivates the following definition.
	
	A pseudocodeword is called $\eta$-good if for all $b \in B$,
	\begin{align}\label{eqn:eta_good_ltc}
		\Ex{g_1,g_2}{\tildeCov{\zee_{X(g_1,\cdot,b)}}{\zee_{Y(g_2,\cdot,b)}}} \leq \eta
	\end{align}
	
	As we proved before, this implies that for any $h\in \calC^{SCC}$, 
	\[
		\tau_b = \sqrt{\Ex{g}{\tildeEx{\indi{\zee_{X(g,\cdot,b)}\neq h_{X(g,\cdot,b)}}}}} \sqrt{\Ex{g}{\tildeEx{\indi{\zee_{Y(g,\cdot,b)}\neq h_{Y(g,\cdot,b)}}}}} 
	\] 
	satisfies
	\[
		\tau_b^2 -(\delta_A-\lambda)\tau_b+\eta \geq 0
	\]
	
	Let's prove upper and lower bounds on $\Ex{s}{\indi{Z_s\neq h_s}}$ in terms of $\tau_b$.
	
	\begin{align*}
		\Ex{s}{\indi{\zee_s \neq h_s}} &~=~ \Ex{g,a,b}{\tildeEx{\indi{\zee_{X(g,a,b)}\neq h_{X(g,a,b)}}}}  \\
		&~=~ \Ex{g}{\tildeEx{\dis(\zee_{X(g,\cdot,\cdot)}, h_{X(g,\cdot,\cdot)}}} \\
		&~\geq~ \Ex{g}{\tildeEx{0\cdot \indi{\zee_{X(g,\cdot,\cdot)} = h_{X(g,\cdot,\cdot)}}} + \delta_A \cdot \delta_B \cdot \indi{\zee_{X(g,\cdot,\cdot)} \neq h_{X(g,\cdot,\cdot)}}} \\
		&~=~ \delta_A\delta_B \Ex{g}{\tildeEx{\indi{\zee_{X(g,\cdot,\cdot )}\neq h_{X(g,\cdot,\cdot)}}}} \\
		&~\geq~ \delta_A\delta_B \Ex{g}{\tildeEx{\indi{\zee_{X(g,\cdot,b )}\neq h_{X(g,\cdot,b)}}}} 
	\end{align*}
	where the last inequality is true for any $b\in B$.
	
	Likewise, $\Ex{s}{\indi{\zee_s\neq h_s}} \geq \delta_A\delta_B \Ex{g}{\tildeEx{\indi{\zee_{Y(g,\cdot,b )}\neq h_{Y(g,\cdot,b)}}}}$. Therefore,
	
	\begin{align*}
		\Ex{s}{\indi{Z_s\neq h_s}} &\geq \delta_A\delta_B \sqrt{\Ex{g}{\tildeEx{\indi{\zee_{X(g,\cdot,b )}\neq h_{X(g,\cdot,b)}}}}} \sqrt{\Ex{g}{\tildeEx{\indi{\zee_{Y(g,\cdot,b )}\neq h_{Y(g,\cdot,b)}}}}} \\
		&= \delta_A\delta_B \tau_b
	\end{align*}
	
	For the upper bound,
	
	\begin{align*}
		\Ex{s}{\indi{Z_s \neq h_s}} &= \Ex{g,a,b}{\tildeEx{\indi{\zee_{X(g,a,b)}\neq h_{X(g,a,b)}}}} \\
		&= \Ex{b}{\Ex{g,a}{\tildeEx{\indi{\zee_{X(g,a,b)}\neq h_{X(g,a,b)}}}}} \\
		&\leq \Ex{b}{\tau_b^2 + \lambda \tau_b + \eta}
	\end{align*}
	
	Importantly, the lower bound works for any $b$, and so if any $\tau_b$ is large, we can conclude good distance. Otherwise, all $\tau_b$ are small, and then the upper bound works well.
	
	Suppose there is some $b$ for which $\tau_b \geq (\delta_A-\lambda) - \frac{2\eta}{\delta_A-\lambda}$. Then, the distance is at least $\delta_A\delta_B ( (\delta_A-\lambda) - \frac{2\eta}{\delta_A-\lambda})$.
	
	If not, then all $\tau_b$ are at most $\frac{2\eta}{\delta_A-\lambda}$. Then the distance,
	\begin{align*}
		\dis(\tildeEx{\cdot}, h) &\leq \Ex{b}{\tau_b^2 + \lambda\tau_b + \eta} \\
		&\leq \Ex{b}{\frac{4\eta^2}{(\delta_A-\lambda)^2} + \frac{\eta(\delta_A+\lambda)}{\delta_A-\lambda}} \\
		&= \frac{4\eta^2}{(\delta_A-\lambda)^2} + \frac{\eta(\delta_A+\lambda)}{\delta_A-\lambda}
	\end{align*}
	
	Putting everything together, if $\lambda\leq \delta_A/3$ and $\eta\leq \delta_A^2/9$, we get that 
	\[ \dis(\tildeEx{\cdot}, h) \geq \delta_A\delta_B(\delta_A-\lambda) - 3 \delta_B \eta \geq \delta -3\eta\]
	 or \[ \dis(\tildeEx{\cdot)}, h) \leq 3\eta\]
	
	Rest of the decoding argument is as before via covering lemma. We only need to show that the $\eta$-good property needed in \cref{eqn:eta_good_ltc} can be obtained by conditioning as before. In case we end up very close to a codeword, we can still unique decode via known unique decoding algorithms \cite{DELLM22} and \cref{lem:decoding_from_distributions}.
	
	The condition in \cref{eqn:eta_good_ltc} is that $d$ different average correlations are small. Recall that we ensured in the Tanner code case that covariance is small with probability $1-\gamma$ where $\gamma$ can be made arbitrarily small with SoS-degree. By a union bound, we can ensure that the covariance is small across all the $d$ bipartite graphs with probability at least $1-d\cdot \gamma$. Thus, by paying an additional factor of $d$ in the SoS-degree, we can ensure low covariance across all $b \in B$ as needed. 
	
	As before, this argument can be derandomized as well.
\end{proof}


\chapter{List Decodable Quantum LDPC Codes} \label{chap:quantum}
%
The area of quantum error correction has had tremendous progress in recent years, particularly in the construction of quantum low-density parity-check (QLDPC) codes. These are codes where membership can be tested via checks acting only on a small number of qubits, which is often an important property for physical implementations. 
Starting with the classic toric code by Kitaev~\cite{Kitaev03}, a sequence of works~\cite{TZ14, EKZ20, KT21, HHOD21, PK22, LZ22} has led to the construction of \emph{asymptotically good} QLDPC codes, with constant rate and constant relative distance. 

Specifically, the above constructions yield a special form of quantum code known as
Calderbank--Shor--Steane (CSS) code, which can be specified by a pair of subspaces $\cx,\cz
\subseteq \F_q^n$ satisfying $\cz^{\perp} \subseteq \cx$ (which also implies $\cx^{\perp} \subseteq
\cz$). In this work, we will always take quantum codes to mean quantum CSS codes, and refer to them
as QLDPC codes when $\cx^{\perp}$ and $\cz^{\perp}$ have generating sets consisting of sparse vectors. The code $\cC =
(\cx, \cz)$ is said to have blocklength $n$ and alphabet size $q$, with the relative distance $\delta$ and the rate $\rho$ defined as,
\begin{align*}
& \delta  ~=~ \frac{1}{n} \cdot \min\inbraces{\abs{c} ~\Big\vert~ c \in (\cx \setminus \cz^{\perp}) \cup (\cz \setminus \cx^{\perp})} 
\\  \text{and} \qquad &\rho ~=~ \frac{1}{n} \cdot \inparen{\dim (\cx) - \dim(\cz^\perp)} \mper
\end{align*}
The codes are said to be \textit{good codes} if both $\delta$ and $\rho$ are constants independent of $n$.
 
The first construction of good QLDPC codes was obtained by the recent breakthrough construction of
Panteleev and Kalachev~\cite{PK22}, which was also closely related to the independent construction
of good locally testable codes by Dinur, Evra, Livne, Lubotzky, and Mozes~\cite{DELLM22}. Subsequent variants of their
constructions by Leverrier and \Zemor~\cite{LZ22, LZ23:decodable} and by Dinur, Hsieh, Lin, and Vidick~\cite{DHLV23} have also obtained simpler descriptions and stronger properties, including \emph{algorithmic} guarantees such as the existence of linear-time (and even parallelizable) unique decoding algorithms, for decoding from errors up to a constant fraction of the blocklength. 
\paragraph{Improved error correction via list decoding.} 
Given the existence of codes with constant relative distance and associated unique decoding algorithms, it is natural to consider what is the \emph{maximum} fraction of errors that can be efficiently corrected, and whether even stronger forms of decoding are possible for such codes. 
%
%
As we saw before, for classical codes, it is possible to correct significantly more than $\delta/2$ fraction
of errors by relaxing the decoding task to \emph{list decoding}~\cite{G01}, where the goal is not to output a single codeword, but possibly a (small) list of codewords within a given error radius $\tau$ of the received word.  
In addition to allowing for recovery from $\tau > \delta/2$ fraction of errors, this is also helpful for several applications in complexity theory~\cite{Trevisan05} where tolerating a larger error radius is significantly more important than recovering a unique codeword.

In the quantum case, no-cloning restrictions forbid duplication of quantum states and require the task of list decoding to be defined more carefully. The appropriate analog in the quantum case is actually the \emph{classical} task of recovering a list of error patterns with at most $\tau$-fraction of errors that can lead to a given received word (or rather, a given syndrome, which corresponds to the output of the parity checks).
The question of constructing list decodable quantum codes (but not necessarily QLDPC codes) was considered by Leung and Smith~\cite{LS08}, and more recently, quantum codes admitting efficient list decoding algorithms were also constructed in the work of Bergamaschi, Golowich and Gunn~\cite{BGG22}.

Another form of decoding for quantum codes, which can go beyond the unique decoding radius $\delta/2$, is where one is still trying to recover a single quantum state, but is allowed to make an exponentially small (in the block-length $n$ of the code) error in the output quantum state. This notion of approximate quantum error-correcting codes (AQECCs) was previously considered by Cr\'{e}pau, Gottesman, and Smith~\cite{CGS05}. 
The work of Bergamaschi, Golowich, and Gunn was also motivated by the construction of AQECCs, and in fact, shows that existing constructions of purity testing codes~\cite{BCGST02} and robust secret sharing schemes~\cite{CDDFS15} can be combined with list decodable quantum codes, to obtain AQECCs (with the same decoding radius).
Thus, constructions of quantum codes with efficient list decoding can be used for both the above forms of error correction, beyond the unique-decoding radius.

%

%
\paragraph{Trading LDPC structure for list decodability.}
The constructions of list decodable codes in~\cite{BGG22} are based on applying and analyzing a
quantum analog of the distance amplification and alphabet reduction procedure of Alon, Edmonds, and
Luby~\cite{AEL95}. 
The quantum analog of the Alon--Edmonds--Luby (AEL) procedure takes a balanced bipartite graph ---
$G=(L,R,E)$ with $\abs{L}=\abs{R}=n$ and degree $d$ --- and two quantum CSS codes: an ``inner code''
$\cC$ in $\F_q^d$ and an ``outer code'' $\cD$ in $\F_{q^k}^n$ (for $k < d$). 
It combines these via concatenation and folding, to construct a new quantum code $\cF$
with improved distance properties.

The construction in \cite{BGG22} applies the AEL procedure using (CSS codes obtained from) folded
Reed--Solomon codes~\cite{GR08} as outer codes, which are known to have optimal list decoding properties as a function of the rate. 
However, since folded Reed--Solomon codes are not LDPC, the list decodable quantum codes and AQECCs in~\cite{BGG22} are not QLDPC codes.
They also consider a variant of their construction, applying the AEL procedure with outer codes
obtained from recent constructions of good (unique-decodable) QLDPC codes. However, the resulting codes can only be list decoded with access to a classical side channel, where one can transfer classical bits without errors. 
Thus, the construction can yield either (capacity-achieving) list decodable quantum codes or QLDPC codes (with optimal unique decoding properties), but not both at the same time. 

While the LDPC property is perhaps more significant for quantum codes, similar bottlenecks in the construction of LDPC codes with good list decoding properties, also arise in classical coding theory. 
List decoding guarantees are often obtained from algebraic constructions over large (polynomial in blocklength) alphabets, which do not necessarily allow for LDPC structure. We saw how to overcome these challenges in \cref{chap:framework}, and in this chapter, we extend those methods to the quantum setting.
\section{Our Results.} 
We construct QLDPC codes that are list decodable up to the Johnson bound, by giving a different analysis of the AEL distance amplification procedure.
In fact, we show that the AEL procedure yields a general method for obtaining list decodable codes, using only \emph{unique-decodability} of the starting code. 
Instantiating this with recent constructions of unique-decodable QLDPC codes, leads to new constructions of list decodable QLDPC codes.

Let $\delta_{\inn}$ and $\delta_{\out}$ denote respectively the distance for the inner and outer
codes used in the AEL procedure.
The analysis in~\cite{BGG22} proved (following the classical analysis) that when $G$
is a $\lambda$-expander, $\cF$ has distance $\delta ~\geq~ \delta_{\inn} - O\inparen{\frac{\lambda}{\delta_{\out}}}$.
However, their proof of list decodability for $\cF$ relied on the list-recoverability of the outer code. 

We show that when the outer code is unique-decodable up to distance $\delta_{\dec}$, the code $\cF$ can be list decoded up to distance $\mathcal{J}\inparen{\delta_{\inn} - \frac{\lambda}{\delta_{\dec}}}$, where $\mathcal{J}(\delta) = 1 - \sqrt{1-\delta}$ denotes the (alphabet-free) Johnson bound.
\begin{theorem}[Informal version of \cref{thm:list_decoding_quantum}]
Let $\cF = (\fx, \fz)$ be the CSS code obtained by applying quantum AEL amplification to the outer CSS code $\cD = (\dx,\dz)$ and inner CSS code $\cC = (\cx,\cz)$ using an $(n,d,\lambda)$--expander graph. 
Let $q$ and $\delta_{\inn}$ be respectively the alphabet size and distance for $\cC$ and let $\cD$ be efficiently unique-decodable from error radius $\delta_{\dec}$. Then, for any $\eps > 0$, $\cF$ can be decoded from error radius $\tau = \mathcal{J}\inparen{\delta_{\inn} - \frac{\lambda}{\delta_{\dec}}} - \eps$, in time $n^{q^{O(d)}/\eps^4}$.
\end{theorem}
For applications, one chooses $\lambda$ to be sufficiently small (and $q,d$ to be sufficiently large constants), so that the distance is $\delta \approx \delta_{\inn}$, and the list decoding radius is $\tau \approx \mathcal{J}(\delta_{\inn}) \approx \mathcal{J}(\delta) > \delta/2$.
Thus, the above codes can be list decoded up to the Johnson bound, which is strictly greater than the unique decoding radius.
When the (constant-sized) inner code $\cx,\cz \in \F_q^d$ is chosen to be an optimal CSS code matching the Singleton bound, we get $\delta \approx \delta_{\inn} \approx 1/2$ and $\tau \approx \mathcal{J}(\delta_{\inn}) \approx 0.293 > 0.25$.
It is not known how to construct QLDPC codes, or indeed even classical LDPC codes, which can be efficiently list decoded beyond the Johnson bound. 
\paragraph{List decoding from distance proofs.}
Our results are based on the framework developed in \cref{chap:framework} for list decoding of (classical) codes up to
the Johnson bound, using \enquote{covering lemmas} and \enquote{proofs of distance} that can be captured by convex relaxations in the SoS hierarchy. 
The covering lemma ensures that the solution to the SoS convex relaxation captures sufficient
information about each element of the list. Given the covering, the list elements are then isolated
using the fact that different codewords are sufficiently far from each other.
The SoS framework requires expressing the proof of distance for the code being decoded, in terms of inequalities obtained via non-negativity of sum-of-squares of low-degree polynomials (in formal variables corresponding to codeword symbols). 

For the case of codes obtained via the AEL distance amplification procedure, the framework in \cref{chap:framework} gives a reduction --- using convex relaxations in the SoS hierarchy --- from the task of list decoding the resulting codes to that of unique-decoding the outer code.
Our result can be viewed as a quantum analog of this reduction.

\subsection{Challenges in the Quantum setting} 
Extending the above framework and reduction for quantum codes faces some important bottlenecks. There are two key aspects to this generalization: one is that the notion of distance itself is different for quantum codes, and the second is that we need a proof that is expressible as sum-of-squares of low-degree polynomials.
%

For a classical code $C \subseteq \F_q^n$, the distance between $u,v \in C$ is defined as the Hamming distance $\Delta(u,v)$ between these vectors, which is easily expressible as a low-degree polynomial in (a real-valued embedding of) the coordinates of $u$ and $v$. 
For a CSS quantum code $\cF = (\fx,\fz)$ the distance between (say) $u,v \in \fx$ is defined as the distance between \emph{cosets} of $\fz^{\perp}$ containing $u$ and $v$ i.e. $\min_{w \in {\cF^\perp_Z}} \Delta(u, v-w)$. It is this minimization over the $w \in \fz^{\perp}$ that is difficult to capture in terms of low-degree polynomials over real variables (which are the objects appearing in solutions to SoS relaxations).

At this point, the reader may wonder how distances for QLDPC codes are proved in the first place. The proof of~\cite{BGG22} uses the relation between the unique decoding radii and avoids having to deal with $\fz^{\perp}$ explicitly. While it is an easy argument that the unique-decoding radius is half of the distance, this fact is not obviously captured as a statement in terms of low-degree polynomials.

When trying to prove distance between $u$ and $v$ in $\fx$, another technique is to replace $u$ by $u'$, where $u'$ is the closest element of $u+\fz^{\perp}$ to $v$. Clearly $\Delta(u,v)\geq \Delta(u',v)$. The optimality with respect to closeness to $v$ comes with additional structural properties for $u'$, and these structural properties are then used to prove a classical-like distance between $u'$ and $v$. For example, the distance between $u'$ and $v$ may be captured as a quadratic inequality $\Delta(u'v) \cdot (\Delta(u',v)-\delta)\geq 0$.

This technique was used in the proof of \cite{LZ22}. Interestingly, we expect that the proof of the quadratic inequality above is again a low-degree SoS proof since it uses spectral expansion, but the trivial statement $\Delta(u,v)\geq \Delta(u',v)$ may not have a low-degree SoS proof!

We note that minimization over cosets such as $\fz^{\perp}$ is of course implicit or explicit in
distance proofs of most QLDPC codes, including~\cite{BGG22}. The only difference in our case is
that we need to explicitly understand the part that can be captured by SoS relaxations.
To this end, we give a slightly more explicit linear algebraic description of the quantum codes
constructed by the AEL procedure in \cref{sec:ael}. We then use this description to give a new proof 
of distance amplification for quantum AEL.

%
%
The more explicit description of the concatenation and folding operations involved in AEL allows us to show that the space $\fz^{\perp}$ for such codes can be decomposed as $\fz^{\perp} = \cF_{\out} + \cF_{\inn}$, where $\cF_{\out}$ is a linear image of the space $\dz^{\perp}$ obtained from the outer code, and $\cF_{\inn} = \F_q^n \otimes \cz^{\perp}$ is obtained from $\cz^{\perp} \subseteq \F_q^d$ corresponding to the inner code.

We carry out the minimization over $\cF_{\inn}$ using an intermediate object we call a ``partial
minimizer''. Suppose we wish to prove distance between $u$ and $v$. As explained above, one approach would be to replace $u$ by $u'=\argmin_{u'\in u+\cF_{\inn}} \Delta(u',v)$, but in this case $\Delta(u,v)\geq \Delta(u',v)$ may not be a low-degree SoS proof. Partial minimizer can be seen as another $u''\in u+\cF_{\inn}$ such that $\Delta(u,v)\geq \Delta(u'',v)$ is a low-degree SoS statement (and in particular, is true), and that $u''$ has sufficient structure to still carry out a low-degree SoS distance proof, as one could have with $u'$.

When decoding, we use an SoS version of this partial minimizer to argue distance for pseudocodewords. Note that this partial minimizer is an object only needed for analysis, and the algorithm need not compute it explicitly. We then ``round'' the SoS solution to an element of $\F_{q^k}^n$, which is the ambient space for the outer code $\dx$, and the minimization over (the linear image of) $\dz^{\perp}$ is carried out implicitly by the (efficient) unique-decoder for the outer code.

\paragraph{Related work.} As mentioned earlier, our work is closely related to that of Bergamaschi, Golowich, and Gunn~\cite{BGG22} which considered AEL amplification in the context of quantum codes. Their results on distance amplification of QLDPC codes already yield codes with distance $\delta \approx 1/2$ and thus with unique-decoding radius $\delta/2 \approx 1/4$.
Our work is motivated by obtaining list decodability for such codes, beyond the unique-decodability threshold. 
The quantum version of AEL has also since been used in \cite{WLH23, GG23}.

Our techniques are based on the framework from \cref{chap:framework}, which uses SoS in the context of list decoding classical codes obtained via the AEL distance amplification procedure, up to the Johnson bound. Similar techniques, which can be seen as an SoS implementation of a distance proof, were also used by Richelson and Roy~\cite{RR23} for decoding the $\epsilon$-balanced code construction of Ta-Shma~\cite{TS17}. Both of these can be seen as instances of the more general ``proofs to algorithms'' paradigm used in the application of the sum-of-squares method to several statistical and combinatorial problems~\cite{FKP19}.
The use of spectral algorithms and semidefinite programming for list decoding was also used in the earlier works of Guruswami and Indyk~\cite{GI03} and Dinur \etal~\cite{DHKLNTS19}.

Another important work in understanding the limitations of SoS for reasoning about quantum codes, is the lower bound construction by Hopkins and Lin~\cite{HL22}. 
While their work shows that SoS relaxations cannot distinguish between $\fx$ and $\fz^{\perp}$ for some quantum codes constructed using expanders, our work can be viewed as proving that this is indeed possible with access to a decoder for the outer code, when $\cF = (\fx,\fz)$ is obtained via the AEL construction.

\section{Preliminaries}
\label{sec:prelims}

Throughout the paper, we will work with $\F_q$-vector spaces with the standard basis (unless specified otherwise). The spaces will be equipped with the Hamming metric and the canonical bilinear form $\angles{u,v} = \sum_i u_iv_i$, with respect to this basis. We define, $V^\perp = \Set{w\mid \ip{v}{w} = 0,\, \forall\, v \in V}$. We will also work with $\R$-vector spaces equipped with the expectation inner product and norm.

\subsection{Quantum CSS codes and list decoding}
%
%
Calderbank--Shor, and Steane (independently) showed that a pair of subspaces over $\F_q$ define a
quantum code if they satisfy a certain orthogonality condition. This construction is known as the CSS construction, and CSS codes form a subclass of \textit{stabilizer codes}.    
 



\begin{definition}[CSS Codes] Let $\cx,\cz \subseteq \F_q^n$ be $\F_q$-linear subspaces such that $\cz^\perp \subseteq \cx$. Then, $\cC = (\cx, \cz)$ defines a $[[n,k,\delta n]]_{q}$ quantum code, where 
\[ 
k =  \dim (\cx) - \dim(\cz^\perp), \;\; \delta = \frac{1}{n} \min \braces[\big]{ \abs{v}\;\big\vert\; v \in (\cx\setminus \cz^\perp) \cup (\cz\setminus \cx^\perp)}.
\]
 The CSS code is \textit{low-density parity check} (LDPC) if there exist (row and column)-sparse parity check matrices, $\hx, \hz$, over $\F_q$ such that $\cx = \ker \hx$ and $\cz = \ker \hz$.
\end{definition}

 

\paragraph{Vector space CSS codes}
The works~\cite{BGG22, GG23} show that the CSS construction can be generalized to the setup where the $\F_q$-linear subspaces have coefficients in a vector space $\F_q^b$.  This generalization enables the extension of the classical coding theoretic operation of folding to the quantum setup. We will work with this general definition throughout the paper. 

\paragraph{Folding} \textit{Folding} is an element-wise syntactic operation in which vectors in $\F_q^{bn}$ are viewed as vectors in $\parens*{\F_q^b}^n$. Formally, let $V \subseteq \F_q^{bn}$ be a vector space over $\F_q$. For a vector $v \in V$, denote \[\fold(v) = (v^{(1)}, \cdots, v^{(n)}) \in \parens*{\F_q^b}^n \text{ where } v^{(i)} = (v_{(i-1)b+1},\cdots, v_{ib})\in \F_q^b.\] 


\begin{definition}[Vector Space CSS code] Let $d$ be a positive integer and let $\cC = (\cx,\cz)$ be a $[[nb,k,\delta nb]]_q$ CSS code. Then, $\fold(\cC)$ defines a $[[n,k/b,\delta' n]]_{q,b}$ vector space CSS code wherein, 
 \[ 
\delta' = \frac{1}{n} \min \braces[\big]{ \abs*{\fold(v)}\;\big\lvert\; v \in (\cx\setminus \cz^\perp) \cup (\cz\setminus \cx^\perp)}.
\]
The weight, $\abs*{c}$, now is the Hamming weight over the alphabet $q^b$, \ie $|\fold(v)| = \abs{\braces{i \mid v^{(i)} \neq 0}}$. 	
\end{definition}

We will use the notation $[[n,k/b,\delta' n]]_{q,b}$ to specify that the code consists of $\F_q$-linear spaces that have been folded in blocks of size $b$. Unless specified otherwise, the folding will be done sequentially to the coordinates according to the fixed basis. {We will use the shortened notation $[[n,k]]$ when we do not require to address the distance of the code.  We will drop the notation $\fold(\cC)$ when we explicitly mention that a code is a vector space CSS code.}

\begin{remark} 	
 The \enquote{folded dimension} changes to $k/b$ to make it consistent with the classical notion. However,  both dimension and blocklength change by the same factor $b$, so that the rate is unchanged. Moreover, the dimensions as $\F_q$-subspaces remain unchanged after folding. 
\end{remark}

\paragraph{List decoding} We now formalize the notion of list decoding for quantum CSS codes and the folded codes. This is inspired by the classical definition but there are a couple of crucial changes in the quantum setup, (i) the input is no longer a corrupted codeword but a \textit{syndrome}, and (ii) the output is required to be a list \textit{containing} the list of possible errors. The first constraint is necessitated by the no-cloning theorem whereas the second relaxation is required as pruning the list can be hard, unlike in the classical case.

\begin{definition}[List of codewords]
Let $\cC = (\cx,\cz)$ be a vector space CSS code over $\parens*{\F_q^d}^n$. Let $B(g,\tau)$ denote a ball of fractional radius $\tau$, in the Hamming metric over $q^d$, around a vector $g$. For any pair of vectors $\subs{g}{X}, \subs{g}{Z} \in \parens*{\F_q^d}^n$, we define the following lists as lists of cosets of codewords, 
\begin{align*}
\calL_X (\subs{g}{X}, \tau) ~&=~ \braces[\Big]{\subs{h}{X} + \cz^\perp\, \mid \,\subs{h}{X} \in \cx,\,\; B(\subs{g}{X},\tau) \cap \subs{h}{X} + \cz^\perp \neq \emptyset},\\
\calL_Z (\subs{g}{Z}, \tau) ~&=~ \braces[\Big]{\subs{h}{Z} + \cx^\perp\, \mid \,\subs{h}{Z} \in \cx,\,\; B(\subs{g}{Z},\tau) \cap \subs{h}{Z} + \cx^\perp \neq \emptyset},\\
\calL (\subs{g}{X}, \subs{g}{Z},\tau) ~&=~ \calL_X (\subs{g}{X},\tau)\, \times\, \calL_Z (\subs{g}{Z},\tau).
\end{align*}
\end{definition}


We also define $\calL_e (\subs{g}{X}, \subs{g}{Z},\tau)$ as a list of cosets of errors, which is just the codeword list $\calL(\subs{g}{X}, \subs{g}{Z},\tau)$ shifted by $\subs{g}{X}$ and $\subs{g}{Z}$ respectively. When decoding from syndromes, all information about the original codeword is lost since the syndrome only depends on the error pattern. Therefore, one can only hope to output a list of errors rather than a list of codewords. In other words, since syndrome is invariant to translation by codewords, so should the output of a list decoding algorithm, and $\calL_e (\subs{g}{X}, \subs{g}{Z},\tau)$ is the translation-invariant version of $\calL (\subs{g}{X}, \subs{g}{Z},\tau)$.
\begin{align*}
\calL_e (\subs{g}{X}, \subs{g}{Z},\tau) ~&=~ \braces[\Big]{\parens[\big]{\subs{h}{X} -\subs{g}{X}  + \cz^\perp, \subs{h}{Z} - \subs{g}{Z}+ \cx^\perp}\, \big\vert \, \parens[\big]{\subs{h}{X} + \cz^\perp, \subs{h}{Z} +\cx^\perp} \in \calL (\subs{g}{X}, \subs{g}{Z},\tau)}\\[3pt]
~&=~ \parens*{\calL_X (\subs{g}{X},\tau) -\subs{g}{X}} \times \parens*{\calL_Z (\subs{g}{Z},\tau)-\subs{g}{Z}}\\[4pt]
~&=~\calL_e (\subs{g}{X}', \subs{g}{Z}',\tau) \qquad \text{for any}\qquad  \subs{g}{X}- \subs{g}{X}' \in \cx, \, \subs{g}{Z}- \subs{g}{Z}' \in \cz.
\end{align*} 
The last equality follows as the codes $\cx, \cz$ are linear and the Hamming metric is translation-invariant.


\begin{definition}[List Decodable vector space CSS codes, also in~\cite{BGG22}]
A quantum CSS code $\cC$ is $(\tau,L)$-\textit{list decodable} if for every $(\subs{g}{X}, \subs{g}{Z})$, the list size is bounded, \ie $\abs{\calL (\subs{g}{X}, \subs{g}{Z},\tau)} \leq L$.

 Fix a pair of parity check matrices $(\hx,\hz)$. We say that a code is \textit{efficiently list decodable} upto fractional radius $\tau$ if given $(\hx\subs{g}{X}, \hz\subs{g}{Z})$ such that $\subs{g}{X},\subs{g}{Z} \in B(0,\tau)$, there exists a $\poly(n)$-time algorithm that outputs a list that contains $\calL_e(\subs{g}{X}, \subs{g}{Z},\tau)$.  
\end{definition}

 \begin{observation} Assume that for any $(\subs{g}{X}, \subs{g}{Z})$, one can output the lists $\calL_X (\subs{g}{X}, \tau)$ and $\calL_Z (\subs{g}{Z}, \tau)$ in $\poly(n)$-time.  Then, the quantum CSS code is efficiently list decodable. 
 \end{observation}
 \begin{proof}
Given, $(\hx\subs{g}{X}, \hz\subs{g}{Z})$ one use Gaussian elimination to compute $(\subs{g}{X}', \subs{g}{Z}')$ such that $\subs{g}{X} - \subs{g}{X}' \in \cx$ and $\subs{g}{Z} - \subs{g}{Z}' \in \cz$. By the assumption we can compute $\calL_X (\subs{g}{X}',\tau)$ and $\calL_Z (\subs{g}{Z}',\tau)$, and therefore output, 
$\calL_e (\subs{g}{X}, \subs{g}{Z},\tau) = \parens[Bigg]{\calL_X (\subs{g}{X}',\tau) -\subs{g}{X}'} \times \parens[Bigg]{\calL_Z (\subs{g}{Z}',\tau)-\subs{g}{Z}'}$.
 \end{proof} 
 
 In summary, we can reduce the task of list decoding the classical codes $\fold(\ex)$ and $\fold(\ez)$ but \textit{upto cosets}. Moreover, our decoder will be symmetric and henceforth we will focus on the task of $X$-decoding.

 \begin{summary}[List decoding vector space codes]\label{sum:list_dec_vector_space}
 Let $\fold(\cE)$ be a vector space CSS code over $\parens*{\F_q^d}^n$ with given parity-check matrices $(\hx,\hz)$. The task of efficient list decoding $\fold(\cE)$ upto radius $\tau$ reduces to the following two tasks,
 \begin{itemize}
 	\item \textup{\textbf{X-decoding}}: Given as input $\subs{g}{X} \in \parens*{\F_q^d}^n$, output a list, $\calL_X'$ of cosets of codewords such that, 
 	\[ \calL_X (\subs{g}{X}, \tau)~\subseteq~ \calL_X' ~\subseteq~ \fold(\ex)/\fold(\ez^\perp).\]   
 	\item \textup{\textbf{Z-decoding}}: Given as input $\subs{g}{Z} \in \parens*{\F_q^d}^n$, output a list, $\calL_Z'$ of cosets of codewords such that,
 	\[ \calL_Z (\subs{g}{Z}, \tau)~\subseteq~ \calL_Z' ~\subseteq~ \fold(\ez)/\fold(\ex^\perp).\] 
 \end{itemize}
 	
 \end{summary}

\subsection{Duality preserving maps} 

To generalize the notion of concatenation to quantum CSS codes, we will need the notion of duality-preserving maps. These are needed to properly define the concatenated code such that the orthogonality constraint, $\cz^\perp\subseteq \cx$, is satisfied, thereby defining a quantum CSS code. This has been used in earlier works, for example, see~\cite{Ham08}. 

{A bilinear map $\ip{\cdot}{\cdot} :V \times W \to \F_q $ over $\F_q$-vector spaces $V,W$, is \textit{non-degenerate} if for any non-zero $x \in V$, the map $\ip{x}{\cdot} : W\to \F_q$ is not identically zero. Similarly,  for a non-zero $y \in W$, the map $\ip{\cdot}{y} : V\to \F_q$ is not the zero map.
}
\begin{definition}[Dual Systems and Basis]\label{def:dual}
	Let $V,W$ be $\F_q$-vector spaces of equal dimension. Let $\ip{\cdot}{\cdot}$ be a non-degenerate bilinear map $V \times W \to \F_q$. Then, $(V,W,\ip{\cdot}{\cdot})$ defines a \textit{dual system}. A basis $\{v_1,\cdots, v_b\}$ of $V$, and $\{w_1,\cdots, w_b\}$ of $W$ is said to be \textit{dual} if $\ip{v_i}{w_j} = 
	\delta_{ij}$ for all $i,j \in [b]$.\end{definition}

{The spaces $V,W$ are called \textit{dual spaces} as the bilinear map gives an isomorphism $W \to V^*$ defined as $w \mapsto \ip{w}{\cdot}$, which also proves the existence of such basis. }This map is injective as the bilinear map is non-degenerate. Now, one can use the canonical dual basis of $V^*$. Using a dual basis, one can construct \textit{duality-preserving maps} which are what we will need. 
%
\begin{definition}[Duality preserving map]\label{def:dual_maps}
	Let $(V_1,V_2, \ip{\cdot}{\cdot}_V)$ and $(W_1,W_2, \ip{\cdot}{\cdot}_W)$ be two dual systems. A pair of linear maps $(\phi_1,\phi_2)$ where $\phi_i: V_i \to W_i $ are \textit{duality preserving} if,
	\begin{equation}
   \ip{u}{v}_V = \ip{\phi_1(u)}{\phi_2(v)}_{W}  \; \; \forall\, u,v \in V.
   \end{equation}
\end{definition}

\begin{claim}\label{claim:dual_map}
	Let $(V_1, V_2)$ and be $(W_1, W_2)$ dual systems along with their pairs of dual bases. Then, $\phi_i: V_i \mapsto W_i$ that acts as identity with respect to these dual bases is a duality-preserving map.  
\end{claim}
\begin{proof}
Since the condition is bilinear, it suffices to prove it for any pair of basis vectors of $V_1, V_2$. 
\begin{equation*}
   \ip{u_i}{v_j}_V = \delta_{ij} =    \ip{x_i}{y_j}_V = \ip{\phi_1(u_i)}{\phi_2(y_j)}_{W}.\qedhere
   \end{equation*} 
\end{proof}


{ We now define two different dual systems we work with. Both use the canonical bilinear form over $\F_q^n$, albeit for different $n$. The first is the space $(\F_q^b,\F_q^b)$ which forms a dual system with the elementary basis as a dual basis. The second system will be subspaces of a CSS code. We now prove that this second system forms a dual system.
}


   \begin{lemma}\label{lem:compl}Let $\cC = (\cx, \cz)$ be a CSS code over $\F_q^d$ of dimension $k$, and let $\cx = \subs{W}{X} \oplus \cz^{\perp}$ and $\cz = \subs{W}{Z}  \oplus \cx^{\perp}$ respectively. Then, the canonical bilinear form over $\F_q$ is non-degenerate over $\subs{W}{X} \,\times \,\subs{W}{Z}$. 
 Therefore, there exists a pair of $\F_q$--linear isomorphisms $\phix: \F_{q}^{b} \to W_X$, $\phiz: \F_{q}^{b} \to W_Z$ such that, 
\[
\ip{x}{y}_{\F_q^b} = \ip{\phix(x)}{\phiz(y)}_{\F_q^d}  \; \; \forall\, x,y \in \F_{q^b} \mper
\]
\end{lemma}

\begin{proof} The space $(\F_q^k,\F_q^k)$ with the canonical bilinear form is a dual system with the elementary basis as a dual basis.  Therefore, if $(\subs{W}{X}, \subs{W}{Z} )$ equipped with the canonical bilinear form is a dual system, then~\cref{claim:dual_map} yields a duality-preserving isomorphism as needed. To show this, we only need to prove the non-degeneracy of the canonical form, and we will do so for one component as the argument is symmetric. 

Let $v \in \subs{W}{Z}$ be such that $\ip{u}{v} = 0$ for all $u \in \subs{W}{X} $. Since, $v \in  \subs{W}{Z} \subseteq \cz$, we have $\ip{w}{v} = 0$ for all $w \in \cz^\perp$. But, $\cx = \subs{W}{X} \oplus \cz^\perp$ and therefore, $v \in \cx^\perp$. Since, $\subs{W}{Z} \cap \cx^\perp = \{0\}$, $v$ must be $0$. 
\end{proof}

\paragraph{Changing base field} We now see that one can view a CSS code over $\F_{q^b}$ as a vector space CSS code over $\F_q^b$. To do this we first 
equip $\F_{q^b}$ with the \textit{trace form}.

  \begin{definition}[Trace Map]
	Let $\F_{q^b}$ be a degree $k$-extension of $\F_q$. The \textit{trace map}\footnote{We will drop the subscript as we will not work with multiple extensions.} is defined as 
	\[\mathsf{Tr}_{\,\F_{q^b}/\F_q} :\F_{q^b} \to \F_q, \;\;  x \mapsto x + x^q + x^{q^2} + \cdots + x^{q^b}.\]
{	The trace map  is $\F_q$-linear as for any $a \in \F_q$, $a^q = a$. }
\end{definition} 

It is a well-known fact that the trace map defines a non-degenerate $\F_q$-bilinear map over $\F_{q^b} \times \F_{q^b}$ defined as $\ip{x}{y}_\mathsf{Tr} = \mathsf{Tr} (xy)$. Therefore, $(\F_{q^b} ,\F_{q^b}, \ip{\cdot}{\cdot}_\mathsf{Tr} )$ forms a dual system, and so from~\cref{claim:dual_map}, there exists a pair of maps $\phix, \phi_z: \F_{q^b} \to \F_q^b$, that is duality preserving. We denote by $(\wphix, \wphiz)$, the map obtained on $\F_{q^b}^n$ obtained by applying $(\phix, \phiz)$ on each coordinate.


\begin{lemma}\label{lem:field_downgrade}
	Let $\cC$ be a $[[n,k,\delta n]]_{q^b}$ CSS code over $\F_{q^b}$, and let $(\phix,\phiz)$ be the duality preserving map defined above. Then $\cC' = \parens{\wphix(\cx), \wphiz(\cz)}$ is a $[[n,k, \delta n]]_{q, b}$ vector space CSS code, where  
\end{lemma}
\begin{proof}
To see that it is a CSS code, let $u  = \wphix(u') \in \wphiz(\cz)^\perp$. The vector $u$ can be written this way as $\wphix$ is an isomorphism. Then, for any $v \in \cz$,
\begin{align*}
 \ip{u'}{v} = \ip{\wphiz(u')}{\wphiz(v)} = 0. 	
\end{align*}
Thus, $u' \in \cz^\perp$ and therefore, $\wphiz(\cz)^\perp \subseteq \wphix(\cz^\perp) \subseteq \wphix(\cx)$ which is the CSS condition.

	Since $\cx$ is a $\F_{q^b}$- subspace, $\dim(\phix(\cx))$ over $\F_q$ is $b\cdot \dim(\cx)$. Thus, the dimension of the new (unfolded code) is $\dim(\phix(\cx)) - \dim(\phix(\cz)^\perp) = b\cdot k$ which becomes $k$ after folding. 
	The distance is unchanged as for $v = (v_1,\cdots, v_n) \in \F_{q^b}^n$ the inital weight is, 
	\[\abs{v} = \abs{\{ i \mid v_i\neq 0 \}}= = \abs{\{ i \mid \phix(v_i) \neq 0 \}} = \abs{\phix(v)}.\]
\end{proof}

\section{Concatenated Codes and AEL Amplification}
\label{sec:ael}

In this section, we formalize the quantum CSS generalizations of two operations on classical codes -- \textit{code concatenation} and AEL \textit{distance amplification} (\cite{AEL95}). We define it purely in linear algebraic terms but the recent work~\cite{BGG22} also formalizes this using the stabilizer code framework.

\subsection{Concatenation of CSS Codes}

To define a concatenation of CSS codes, one needs a pair of outer code and inner code that are compatible with respect to some parameters. Additionally in the quantum case, we need a pair of duality-preserving maps as defined in~\cref{lem:compl}.


\begin{definition}[Concatenated CSS Codes]
Given the following objects,	
\begin{itemize}
    \item \textbf{Outer Code} --- Let $\cD = (\dx, \dz)$ be a $[[n, k_{\out}, \delta_{\out}]]_{q,b_\out}$ vector space CSS code.
    \item \textbf{Inner Code} --- Let $\cC = (\cx, \cz)$ be a $[[d, k_\inn, \delta_{\inn}]]_{q,b_\inn}$ vector space CSS code that $b_\inn k_\inn = b_\out$. Let $\cx = W_X \oplus \cz^{\perp}$ and $\cz = W_Z \oplus \cx^{\perp}$ respectively (as $\F_q^{db_\inn}$-subspaces). 
    \item \textbf{Duality preserving maps} --- Let $(\phix,\phiz)$ be duality preserving maps as in~\cref{lem:compl} from $\F_{q}^{k}$ to $\subs{W}{X}$ and $\subs{W}{Z}$ respectively. Extend these maps to $(\F_{q}^k)^n$ by applying it to each coordinate {of the folded code}, and call the extended map $(\wphix, \wphiz)$.
\end{itemize}
\vspace{-5pt}One defines the concatenated CSS Code, $\cE = \cD \circ_\varphi \cC$ as $(\ex, \ez)$ with
    \begin{align*}
        \ex &=\wphix( \dx) + \mathbb{F}_{q}^n \otimes \cz^{\perp} \subseteq \F_q^{ndb_\inn}\\
        \ez &= \wphiz(\dz) + \mathbb{F}_{q}^n \otimes \cx^{\perp}\subseteq \F_q^{ndb_\inn}
    \end{align*}
\end{definition}

We now give an explicit description of the dual spaces that will prove that the concatenation operation defines a CSS code and also be useful in proving the distance of the final code. 

%
%

\begin{proposition}\label{prop:concat_dual}
    For the above definition of $(\ex, \ez)$, the dual spaces can be computed as follows,
    \begin{align*}
        \ex^{\perp} &= \wphiz(\dx^{\perp}) + \mathbb{F}_{q}^n \otimes \cx^{\perp} \\
        \ez^{\perp} &= \wphix(\dz^{\perp}) + \mathbb{F}_{q}^n \otimes \cz^{\perp}
    \end{align*}
    Therefore, $\cE$ is an $[[nd, {k_{\out}\cdot k_\inn}]]_{q, b_\inn}$ vector space CSS code.
    {Moreover, if $\cD$ is an LDPC code, and $d b_\inn$ is a constant, then $\cE$ is an LDPC code.}
\end{proposition}

\begin{proof} We prove the first equation as the proofs are symmetric.  Define, \[\mathcal{U}_X := \wphiz(\dx^{\perp}) + \mathbb{F}_{q}^n \otimes \cx^{\perp}.\] 

Since $\wphix$ is an $\F_q$-linear isomorphism $\dim(\wphix(U)) = \dim(U)$. Moreover, $\mathrm{im}(\wphiz)$ and $\cx^\perp$ are disjoint and thus their dimension add up.    Thus,
 \begin{align*}
 	\dim(\mathcal{U}_X) ~&=~ \dim(\dx^\perp) + n\, \dim(\cx^\perp) = nb_\out-\dim(\dx) + n (db_\inn-\dim(\cx)), \\
 	\dim(\ex) ~&=~ \dim(\dx) + n(\dim(\cz^\perp)) = \dim(\dx) + n(\dim(\cx) -k_\inn ) .
 \end{align*}
 
 Therefore, $\dim(\mathcal{U}_X) = ndb_\inn - \dim(\ex) =  \dim(\ex^\perp)$ and it suffices to show that $\,\mathcal{U}_X\subseteq \ex^\perp$.
 
Let $\alpha \in \mathcal{U}_X$ and $\beta \in \ex$. Using the definition of the spaces, we express them as,
\begin{align*}
 \alpha ~&=~ (\phiz(u^{(1)}) + x_1,\,\cdots\,, \phiz(u^{(n)}) + x_n ) \in \mathcal{U}_X \; \text{ where } x_i \in \cx^\perp,\, u =(u^{(1)},\cdots, u^{(n)}) \in\dx^\perp,\\   
 \beta ~&=~ (\phix(v^{(n)}) + z_1,\,\cdots\,, \phix(v^{(n)}) + z_n ) \in \ex  \;\;\, \text{ where } z_i \in \cz^\perp,\,  v = (v^{(1)},\cdots, v^{(n)}) \in\dx.
\end{align*}

Computing the inner product we get four kinds of terms,
\[
\ip{\alpha}{\beta} = \sum_{i=1}^n \parens[\Big]{\ip{\phiz\parens[\big]{u^{(i)}}}{\phix\parens[\big]{v^{(i)}}} + \ip{\phiz\parens[\big]{u^{(i)}}}{z_i} + \ip{x_i}{\phix\parens[\big]{v^{(i)}}} + \ip{x_i}{z_i} }  \]

Each of the last three terms is zero as, by definition, they belong to orthogonal spaces.
We are then left with the first term which can be calculated using the duality-preserving property,
\vspace{-0.5em}
\begin{align*}
   \ip{\alpha}{\beta} ~&=~ \sum_{i=1}^n \angles[\Big]{\phiz\parens[\big]{u^{(i)}}\;, \;\phix\parens[\big]{v^{(i)}}}_{\F_q^{db_\inn}} \\
   ~&=~ \sum_{i=1}^n  \ip{u^{(i)}}{v^{(i)}}_{\F_q^{b_\out}} &&\text{(Duality Preserving)}\\
   ~&=~  \ip{u}{v}_{\F_q^{nb_\out}} = 0 && (u \in \dx^\perp, v \in \dx). 
\end{align*}

This proves that $\mathcal{U}_X = \ex^\perp$. Moreover, if $\dx = V_X \oplus \dz^\perp$, the proof implies that $\ex = \phi_X(V_X) \oplus \ez^\perp$. Recall that for a $[[n,k_{\out}]]_{q,b_\out}$ vector space code, the dimension of $V_X$ as a $\F_q$-subspace is $k_{\out}\cdot b_\out$.
   Thus, the dimension of the CSS code $\cE$ is, \[\dim(\ex) -\dim(\ez^\perp) =  \dim(\phi_X(V_X)) = \dim(V_X) = k_{\out}\cdot b_\out = k_{\out}\cdot k_\inn b_\inn \quad .\]
Folding this $\cE$ into blocks of size $b_\inn$, we get a $[[nd, k_\out \cdot k_\inn]]_{q, b_\inn}$ vector space CSS code.    
The LDPC property follows since the generators of $\ex^\perp$ are comprised of generators of $\dz^\perp$ mapped by $\wphix$, and generators corresponding to $\cz^{\perp}$. The former are sparse if $\cD$ is LDPC, and the latter has weight at most $db_\inn$ which is a constant.
\end{proof}


\subsection{AEL Amplification and Folding CSS Codes}

The distance of the concatenated code $\cE$  can be amplified by the AEL procedure using a $d$-regular bipartite expander, $G = (L,R, E)$.
The graph $G$ is chosen such that the size of $L$ and $R$ match the blocklength of the outer code, and the degree matches the blocklength of the inner code. 

 The AEL procedure is a three-step process --- (i) concatenate the outer code $\cD$ with inner code $\cC$ to obtain $\cE$, (ii) shuffle the symbols of concatenated code via edges on a bipartite expander graph $G$, and (iii) collect $d$-symbols on the right vertices and fold them back to produce the final code, $\cF$. 

In this subsection, we will formally state the AEL procedure and set up some useful notation. We start by restating the definitions of the concatenated codes (and their duals) in a manner that will be convenient when working with AEL, and later, with sum-of-squares proofs.



%
\paragraph{Concatenated Codes and AEL} To simplify notation, will use $\Sigma$ to denote $\F_q^{b_\inn}$ as the concatenated code lies inside the space $(\F_q^{b_{\inn}})^{nd}  = \Sigma^E$.  We view the codewords, $z \in \ex$ (or $\ez$), as an assignment of $\Sigma$-values to the edges. Denote by $z_\li \in \Sigma^d$, the restriction of the vector $z$ to the neighborhood, $N(\li)$, of vertex $\li \in L$. We will use similarly use $z_\ri$, to denote restriction to the neighborhood of vertex $\ri \in R$.  
	The concatenated code $\cE = \cD\circ_\varphi \cC$ and its duals can be explicitly described as,
\begin{align*}
 	\ex ~&=~ \Set*{ x \mid x_{\li} =  \phix(u_\li) + z_\li, \text{ for a unique  } z_\li \in \cz^\perp \text{ and } u \in \dx },\\[2.5pt]
 		\ex^\perp ~&=~ \Set*{ x \mid x_{\li} =  \phix(u_\li) + z_\li, \text{ for a unique  } z_\li \in \cx^\perp \text{ and } u \in \dx^\perp },\\[2.5pt]
 		\ez ~&=~ \Set*{ x \mid x_{\li} =  \phix(u_\li) + z_\li, \text{ for a unique  } z_\li \in \cx^\perp \text{ and } v \in \dz },\\[2.5pt]
 		\ez^\perp ~&=~ \Set*{ x \mid x_{\li} =  \phix(u_\li) + z_\li, \text{ for a unique  } z_\li \in \cz^\perp \text{ and } u \in \dx }.
 \end{align*}

 Uniqueness follows as $\mathrm{im}(\phix) = \subs{W}{X}$ is disjoint from $\cz^\perp$ (and similarly for $\phiz$). The concatenated code $\cE$ is folded using the partitions induced by the neighborhoods of the right vertices. Explicitly, the folded code $\cF = \fold(\cE) =  (\fold(\ex), \fold(\ez))$ is given by,
 \begin{align*}
 	\fx:=~ \fold(\ex) ~&=~ \Set[\Big]{\fold(z_{\ri_1}, \cdots , z_{\ri_n}) \big\vert \, z \in \ex,\; \ri_i \in R} \subseteq \parens*{\Sigma^d}^n  \\
 	\fz:=~ \fold(\ez) ~&=~ \Set[\Big]{\fold(z_{\ri_1}, \cdots , z_{\ri_n})  \big\vert\, z \in \ez,\; \ri_i \in R} \subseteq \parens*{\Sigma^d}^n 
 	\end{align*} 

\begin{proposition}[AEL Procedure]\label{cor:ael_rate}
Let $\cD$ be a $[[n, k_\out]]_{q, b_\out}$ CSS code and $\cC$ be a $[[d, k_\inn]]_{q, b_\inn}$ CSS code.\\[2.5pt]
Then, the AEL code, $\cF = (\fx, \fz)$ defines a $[[n, \frac{k_\out \cdot k_\inn}{d}, \delta_R\cdot n]]_{q, db_\inn}$ CSS code where \[\delta_R\cdot n = \min \braces[\big]{\, \abs[\big]{\braces{i \mid z_{r_i}\neq 0}} \;\big\vert\;  z \in (\ex\setminus \ez^\perp) \cup (\ez\setminus \ex^\perp)}. \] 
\end{proposition}

The key property of the amplified code, $\cF$, is that its distance is significantly better than $\cE$. We will prove a lower bound on $\delta_R$ in the next section (\cref{thm:ael_final}). Before that, we define two notions that we use later: \textit{local inversion maps}, and a couple different distance metrics over $\Sigma^{dn}$.

\paragraph{Local Inversion} The uniqueness of decomposition of the local codeword, $x_\li$, let us define an inverse to the maps $\phix$ and $\phiz$.
\begin{definition}[Local Inversion Maps]\label{def:local_inversion}
Let $\cE = \cD\circ_\varphi \cC$ be the concatenated code as above. Then one defines local inversion maps,
\begin{align*}
 \unconx :\cx \rightarrow \F_q^{b_{\out}},\;\; x_\li ~&=~ \phix(u_\li) + z_\li \mapsto u_\li, \\  
 \unconz :\cz \rightarrow \F_q^{b_{\out}},\;\; x_\li ~&=~ \,\phiz(v_\li) + z_\li \mapsto v_\li.
\end{align*}
	
\end{definition}

\paragraph{Distance metrics for $\cE$} Using the graph structure, we can fold the code $\cE$ using the left or right vertices. Moreover, we can define a define a metric for the set of $\cz^\perp$ cosets which is needed for the quantum notion of distance. 
 \begin{align*}
 	\Delta_L(z,h) ~&=~ \Ex{\ell\sim L}{\indi{ z_\li \neq h_{\li}}},  \\
 	\Delta_{L, \cz^\perp}(z,h) ~&=~ \Ex{\ell\sim L}{\indi{ z_\li \not\in h_{\li} + \cz^\perp }}, \\
 	\Delta_{R}(z,h) ~&=~\Ex{r\sim R}{\indi{ z_{r} \neq h_{r}}}.
 \end{align*}
  Now we can reinterpret AEL procedure as changing the metric from the initial $\Delta_L(\cdot, \cdot)$ on the concatenated code $\cE$ to  $\Delta_R(\cdot, \cdot)$. This change is crucial as this is where the pseudorandom properties of the graph (expansion) come in, and imply that the distance between codewords under the $\Delta_R$ metric is much larger than the initial distance under the $\Delta_L$ metric.
  
  \paragraph{Decoders for AEL}
  As defined in \cref{sum:list_dec_vector_space}, a list decoder up to radius $\tau$ for the code $\cF$ should take as input a string $g\in (\Sigma^d)^R$, and output a list of cosets of $\fz^\perp$ containing 
  \begin{align*}
      \calL(g,\tau) ~&=~ \braces[\big]{ \fold(h)+\fz^\perp \suchthat h\in \ex,\, \Delta_R(g,h) <\tau}\\
      ~&\cong~ \braces[\big]{ h+\ez^\perp \suchthat h\in \ex,\, \Delta_R(g,h) <\tau} .
  \end{align*}

  The second set is merely the unfolded version of the first and it is equivalent to work with either. We will work with the latter to simplify notation.

\section{Distance Proofs}\label{sec:distance_quantum}

In this section, we prove that the folded AEL code, $\fold(\cD \circ_\varphi \cC)$ has large (fractional) distance, $\delta_R$, that can be made arbitrarily close to that of the inner code $\delta_\inn$ by picking a good enough expander.

\subsection{Partial Minimizer}
As mentioned before, one of the key objects we will use in the distance proof is a partial minimizer, which does not change the coset, but gets closer to the codeword we are measuring distance from. For a fixed codeword $h \in \ex$, and any vector $z \in \F_q^E$ we define the partial minimizer to be a new vector $\partmin(z,h)$    
\[	\partmin(z,h)_{\li} = \begin{cases}
		z_{\li} \;\; \text { if } \;\;  z_{\li} \not\in\,  h_{\li} + \cz^\perp \\[6pt]
		h_{\li} \;\; \text { if } \;\; z_{\li} \in\,  h_{\li} + \cz^\perp 
	\end{cases}
		\]
We observe that the partial minimizer satisfies two key properties,
\begin{align} 
&\text{ (\small{Coset-preserving}) } \;\;\; &&\partmin(z,h)_{\li} \in\, z_{\li} + \cz^\perp \;\; \forall\, \li \in L  \label{eqn:coset}\\[4pt]
	&\text{ (\small{Monotone}) } \;\;\;	&&\Delta(\partmin(z,h)_\ri, h_\ri) \leq \Delta(z_\ri, h_\ri) \;\; \forall \, \ri \in\, R \label{eqn:monotone}
\end{align}
 
\subsection{Distance proof}

We next prove for any two codewords in $\ex$ that do not share the same coset of $\ez^\perp$, their distance in the $\Delta_R(\cdot, \cdot)$ metric is almost as large as $\delta_{\inn}$.

\begin{lemma}[Distance proof of AEL]\label{lem:zfc_dist}
Let $z$ and $h$ be two non-equivalent codewords in $\ex$, \ie $z \not\in h + \ez^\perp$. Then,
\[\Delta_{R}(z,h) \geq \delta := \delta_\inn - \frac{\lambda}{\Delta_{L, \cz^\perp}(z,h)}  \] 	
 where $\Delta_{L,\cz^\perp} (z,h) = \Ex{\li \in L}{\one \braces[\big]{z_{\li} \not\in  h_{\li} + \cz^\perp}} $. 
\end{lemma}
\begin{proof}
	We will lower bound and upper bound the same quantity $\Delta_E(\partmin(z,h), h)$. Note that by \cref{eqn:coset}, we have $\Delta_{L,\cz^\perp} (z,h) = \Delta_{L,\cz^\perp} (\partmin(z,h),h)$.
	
	\begin{align*}
			\Ex{e\in E}{\partmin(z,h)_e \neq h_e} ~&=~ \Ex{\li\in L}{{\dist{\partmin(z,h)_{\li} , h_{\li} }}} \\
			~&\geq~ \Ex{\li\in L}{\indi{\partmin(z,h)_{\li} \not\in  h_{\li} + \cz^\perp} \cdot \dist{\partmin(z,h)_{\li},h_{\li} }} \\
			~&\geq~ \Ex{\li\in L}{\indi{\partmin(z,h)_{\li} \not\in  h_{\li} + \cz^\perp} \cdot \delta_\inn} \\
			~&=~ \delta_\inn \cdot \Delta_{L,\cz^\perp} (\partmin(z,h),h) \\
			~&=~ \delta_\inn \cdot \Delta_{L,\cz^\perp} (z,h)&& (\text{Using \cref{eqn:coset}}) .
	\end{align*}

	For the upper bound, we will use the expander mixing lemma~(\cref{lem:eml_bipartite}).
	\begin{align*}
		\Ex{e\in E}{{\indi{\partmin(z,h)_e \neq h_e}}} ~&\leq~ \Ex{(\li, \ri)\, \in\, E}{{ \indi{\partmin(z,h)_{\li} \neq h_{\li} } \cdot \indi{\partmin(z,h)_{\ri} \neq h_{\ri} }}} \\[4pt]
		~&\leq~ \Ex{\li, \ri}{{ \indi{\partmin(z,h)_{\li} \neq h_{\li} } \cdot \indi{\partmin(z,h)_{\ri} \neq h_{\ri} }}} + \lambda &&\text{(\cref{lem:eml_bipartite})}\\[4pt]
		~&=~\Delta_{L} (\partmin(z,h),h)\cdot \Delta_R(\partmin(z,h),h) + \lambda  \\[4pt]
		~&=~ \Delta_{L,\cz^\perp} (\partmin(z,h),h)\cdot \Delta_R(\partmin(z,h),h) + \lambda &&\text{(Definition of $\partmin$)} \\[4pt]
		~&=~ \Delta_{L,\cz^\perp} (z,h)\cdot \Delta_R(\partmin(z,h),h) + \lambda &&\text{(\cref{eqn:coset})} \\[4pt]
		~&\leq~ \Delta_{L,\cz^\perp} (z,h)\cdot \Delta_R(z,h) + \lambda &&\text{(\cref{eqn:monotone})}
\end{align*}
	 Comparing the two sides, we get  
	\begin{align*}
		\Delta_{R}(z,h)\cdot \Delta_{L, \cz^\perp}(z,h) + \lambda ~\geq~ \delta_\inn \cdot \Delta_{L, \cz^\perp}(z,h)
	\end{align*}
	Since $z\not\in h+ \ez^\perp$, there exists at least one vertex $\li \in L$ such that $z_\li \not \in h_\li + \cz^\perp$ and thus, $\Delta_{L,\cz^\perp} (z,h) > 0$.  Dividing by it gives the result.	
\end{proof}

We now deduce that using AEL machinery amplifies the distance of the base outer code. 

\begin{theorem}[AEL distance]\label{thm:ael_final}
Let $\cD$ be a $[[n, k_\out, \delta_\out\cdot n]]_{q, b_\out}$ vector space CSS code and $\cC$ be a $[[d, k_\inn, \delta_{\inn}\cdot d]]_{q, b_\inn}$ vector space CSS code.
Let $\cF$ be the AEL code obtained by using an $(n,d,\lambda)$-expander. Then, $\cF$ is a $[[n, \frac{k_\out \cdot k_\inn}{d}, \delta_R\cdot n]]_{q, db_\inn}$ vector space CSS code where $\delta_R \geq \delta_{\inn} - \frac{\lambda}{\delta_{\out}}$.
\end{theorem}
\begin{proof}
	\cref{cor:ael_rate} gives the dimension of the code. To compute the distance, we observe that  $\delta_R = \min_v \disR{v}{0}$ where the min is over non-trivial codewords. The result now follows from the above distance bound~\cref{lem:zfc_dist}. 
\end{proof}


\section{SoS Proof of Distance}\label{sec:sos_distance}



To be able to use SoS notation freely, we start by extending the definition of partial minimizer to be a vector-valued local function so that can be used inside pseudoexpectations. For a fixed codeword $h$, we define
\[
	\partmin(\zee,h)_{\li} = \begin{cases}
		\zee_{\li} \;\; \text { if } \;\;  \zee_{\li} \not\in\,  h_{\li} + \cz^\perp \\[6pt]
		h_{\li} \;\; \text { if } \;\; \zee_{\li} \in\,  h_{\li} + \cz^\perp 
	\end{cases}
\]
More explicitly, $\partmin(\zee,h)_{\li}$ is shorthand for $\indi{\zee_{\li} \not\in h_{\li} + \cz^{\perp}} \cdot \zee_{\li} + \indi{\zee_{\li} \in h_{\li} + \cz^{\perp}} \cdot h_{\li}$.
%
\begin{align} 
	&\text{ (\small{Coset-preserving}) } \;\;\; &&\partmin(\zee,h)_{\li} \in\, \zee_{\li} + \cz^\perp \iff \unconx(\partmin(\zee,h)_{\li}) = \unconx(\zee_{\li}) \;\; \forall\, l \in L  \label{eqn:coset_sos}\\[4pt]
	&\text{ (\small{Monotone}) } \;\;\;	&&\indi{ \partmin(\zee, h)_r \neq h_r} ~\leq~ \indi{\zee_r \neq h_r} \;\; \forall \, r \in\, R \label{eqn:monotone_sos}
\end{align}
Recall that $\unconx$ is the local inversion map defined in \cref{def:local_inversion}. We emphasize that the partial minimizer is importantly defined in a way that the monotonicity property is true \emph{locally}.

Next, we extend the distance proof of last section to pseudocodewords. We use the same definition of pseudocodewords as defined for classical AEL codes, and refer the reader to \cref{sec:sos_and_codes} for details.
We also recall the following non-convex property from \cref{def:eta_good} that allows us to prove a distance property similar to \cref{lem:zfc_dist} for pseudocodewords.

\begin{definition}[Restatement of \cref{def:eta_good}]\label{def:eta_good_duplicate}
	A pseudocodeword is $\eta$-good if,
	\[
		\Ex{\li,\ri}{ \tildeCov{ \zee_{\li}}{\zee_{\ri} } } \leq \eta
	\]
\end{definition}

Now we prove a distance bound analogous to \cref{lem:zfc_dist} for pseudocodewords that satisfy the $\eta$-good property from \cref{def:eta_good_duplicate}. This is the key statement we need to make the framework from \cref{chap:framework} applicable to the quantum AEL setting.

\begin{lemma}[Distance proof]\label{lem:sos_ael_distance_quantum}
Let $\tildeEx{\cdot}$ be an $\eta$-good pseudocodeword and $h \in \Sigma^E$ be such that $\fold(h)$ is a codeword in $ \fx$ and $\distLperp{\tildeEx{\cdot}} > 0$. Then,
\[\distR{\tildeEx{\cdot}} ~\geq~ \delta_\inn - \frac{\lambda + 
\eta}{\distLperp{\tildeEx{\cdot}}}  \mper\] 	
\end{lemma}
\begin{proof}
	We will closely follow the proof of~\cref{lem:zfc_dist} and similarly, will lower bound and upper bound the quantity $\Ex{e\in E}{~\tildeEx{~\indi{\partmin(\zee, h)_e \neq h_e}}}$. 
\begin{align*}
			\Ex{e\in E}{~\tildeEx{~\indi{\partmin(\zee, h)_e \neq h_e}}} ~&=~ \Ex{\li\in L}{\tildeEx{\dist{\partmin(\zee, h)_{\li} , h_{\li} }}} \\
			~&\geq~ \Ex{\li\in L}{\tildeEx{\indi{\partmin(\zee, h)_{\li} \not\in h_{\li} + \cz^\perp} \cdot \dist{\partmin(\zee, h)_{\li} , h_{\li} }}} \\
			~&\geq~ \Ex{\li\in L}{\tildeEx{\indi{\partmin(\zee, h)_{\li} \not\in h_{\li} + \cz^\perp} \cdot \delta_{\inn} }} \\
			~&=~ \delta_{\inn} \cdot \Ex{\li\in L}{\tildeEx{\indi{\partmin(\zee, h)_{\li} \not\in h_{\li} + \cz^\perp} }} \\
			~&=~ \delta_{\inn} \cdot \Ex{\li\in L}{\tildeEx{\indi{\zee_{\li} \not\in h_{\li} + \cz^\perp} }} && (\text{Using \cref{eqn:coset_sos}}) \\
			~&=~ \delta_{\inn} \cdot \distLperp{\tildeEx{{\cdot}}}
	\end{align*}

For the upper bound, we will use the expander mixing lemma~(\cref{lem:eml_for_pexp}).
	\begin{align*}
		& \qquad \Ex{e\in E}{~\tildeEx{~\indi{\partmin(\zee, h)_e \neq h_e}}} \\
		~&\leq~ \Ex{\li \sim \ri}{~\tildeEx{~\indi{\partmin(\zee,h)_{\li} \neq h_{\li} } \cdot \indi{\partmin(\zee,h)_{\ri} \neq h_{\ri} }}} \\[4pt]
		~&\leq~ \Ex{\li, \ri}{\tildeEx{ \indi{\partmin(\zee, h)_{\li} \neq h_{\li} } \cdot \indi{\partmin(\zee, h)_{\ri} \neq h_{\ri} }}} + \lambda &&\text{( \cref{lem:eml_for_pexp} )}\\[4pt]
		~&\leq~ \Ex{\li, \ri}{\tildeEx{ \indi{\partmin(\zee, h)_{\li} \neq h_{\li} }} \cdot \tildeEx{\indi{\partmin(\zee, h)_{\ri} \neq h_{\ri} }}} + \lambda +\eta &&\text{( \cref{def:eta_good_duplicate} )}\\[4pt]
		~&=~ \Ex{\li}{\tildeEx{ \indi{\partmin(\zee, h)_{\li} \neq h_{\li} }}} \cdot \Ex{\ri}{\tildeEx{\indi{\partmin(\zee, h)_{\ri} \neq h_{\ri} }}} + \lambda +\eta \\[4pt]
		~&=~ \Ex{\li}{\tildeEx{ \indi{\zee_{\li} \not\in h_{\li} + \cz^{\perp}}}} \cdot \Ex{\ri}{\tildeEx{\indi{\partmin(\zee, h)_{\ri} \neq h_{\ri} }}} + \lambda +\eta &&\text{ (Definition of $\partmin$)} \\[4pt]
		~&\leq~ \Ex{\li}{\tildeEx{ \indi{\zee_{\li} \not\in h_{\li} + \cz^{\perp}}}} \cdot \Ex{\ri}{\tildeEx{\indi{\zee_{\ri} \neq h_{\ri} }}} + \lambda +\eta &&\text{( \cref{eqn:monotone_sos} )} \\[4pt]
		~&=~\distLperp{\tildeEx{\cdot}} \cdot \distR{\tildeEx{\cdot}}  + \lambda + \eta\;.
\end{align*}

Comparing the two sides, we get  
	\begin{align*}
		\distLperp{\tildeEx{\cdot}} \cdot \distR{\tildeEx{\cdot}} + \lambda + \eta  ~\geq~ \delta_{\inn} \cdot \distLperp{\tildeEx{{\cdot}}} \; .
	\end{align*}
	Dividing by $\distLperp{\tildeEx{{\cdot}}}$ gives us the result.
	\end{proof}

\section{List Decoding Algorithm}
In this section, we combine the algorithmic covering, correlation rounding and the distance proof to give a list decoding algorithm. This part of the chapter mostly follows the details from \cref{chap:framework}, and so we omit some proofs and only mention the relevant lemmas.

For any $f\in \Sigma^E$, recall that $f_r \in  \Sigma^d$ is the restriction of $f$ to the edge-neighborhood $N(\ri)$ of a vertex $r \in R$. Suppose the inner code is defined over alphabet $\Sigma = \F_q^{b_{\inn}}$, so that codewords can be seen as members of $\Sigma^E$. Let us denote the size of $\Sigma$ by $s \defeq q^{b_{\inn}}$.

We first start with an algorithmic covering lemma that states that an SoS solution that is close to every codeword in the list can be found in polynomial time. As in \cref{chap:framework}, such a pseudocodeword can be found by maximizing an entropy proxy. This part of the algorithm is analytic and does not care about the code being classical vs quantum.
\begin{lemma}[Algorithmic Covering Lemma]\label{lem:algo_covering_quantum}
	Let $g\in \Sigma^E$, $\alpha \in (0,1)$ and $\eps \in (0,1-\alpha)$. 
	There exists a pseudocodeword $\tildeEx{\cdot}$ of degree $t\geq d$ such that for every $h \in \fx$ such that $\distR{g} < 1-\alpha - \eps$, it holds that
	\[
		\distR{\tildeEx{\cdot}} < 1-\alpha^2 - 2 \alpha \eps \mper
	\]
	Moreover, such a pseudocodeword can be found in time $n^{\calO(t)}$.
\end{lemma}

Next, we mention a lemma that says that the $\eta$-good property can be obtained by random conditioning. This technique is common in algorithmic applications of Sum-of-Squares, and first appeared in \cite{BRS11}. It was adapted for decoding in \cref{chap:framework} as \cref{lem:low_covariance_solution}.

\begin{lemma}[Restatement of \cref{lem:low_covariance_solution}]\label{lem:conditioning_eta_good_quantum}
	Let $\eta>0$ be arbitrarily small. Given any pseudocodeword $\tildeEx{\cdot}$ of degree $\geq 2d\left(\frac{s^{3d}}{\eta^2}+1\right)$, there exists an integer $u^* \leq s^{3d}/\eta^2 $ such that 
	\[
		\Ex{r_1,r_2,\cdots,r_{u^*}}{\Ex{\li,\ri}{\tildecov\brackets*{\zee_{\li},\zee_{\ri} \vert \zee_{r_1},\zee_{r_2},\cdots,\zee_{r_{u^*}}}}} \leq \eta \mper
	\]
\end{lemma}

Finally, we give the statement of the main decoding theorem, that combines \cref{lem:sos_ael_distance_quantum}, \cref{lem:algo_covering_quantum} and \cref{lem:conditioning_eta_good_quantum} in exactly the same way as in \cref{chap:framework}. Further, we remark that this algorithm can also be derandomized as in \cref{chap:framework}, with the caveat that the output list cannot be pruned as in the classical case.

\begin{theorem}[List Decoding quantum AEL amplification]\label{thm:list_decoding_quantum}
	Let $(\fx, \fz)$ be the code obtained by applying quantum AEL amplification to the outer code $(\dx,\dz)$ and inner code $(\cx,\cz)$ of distance $\delta_{\inn}$ using an $(n,d,\lambda)$--expander graph. The inner code $(\cx,\cz)$ is over an alphabet $\Sigma$ of size $s$, so that $(\fx,\fz)$ is over alphabet $\Sigma^d$ of size $s^d$.

Suppose the code $(\dx, \dz)$ can be unique-decoded from radius $\delta_{\dec}$ in time $\calT(n)$. Assume that $\lambda < \delta_{\dec}$. 
	
	Then for any $\gamma > 0$, there exists an algorithm based on $s^{\calO(d)}/\eps^4$ levels of the SoS hierarchy that given $g\in \Sigma^E$, runs in time $\ln(1/\gamma) \cdot \brackets*{n^{s^{\calO(d)}/\eps^4} + \calT(n)}$ and produces a list $\calL'$ of cosets that contains the list of cosets $\calL \defeq \calL \parens*{g, \calJ\parens*{\delta_{\inn} - \frac{\lambda}{\delta_{\dec}}} - \eps }$ with probability at least $1-\gamma$. The size of $\calL'$ is at most $\widetilde{\Omega}_{s,d} \parens*{ \frac{\ln(1/\gamma)}{\eps^6} }$.
\end{theorem}

\begin{proof}
	The decoding algorithm is presented as \cref{algo:ael-decoding-quantum}. 
	Recall that $\calL = \calL \parens*{g, \calJ\parens*{\delta_{\inn} - \frac{\lambda}{\delta_{\dec}}} - \eps }$ is the list of all cosets in $\ex/\ez^{\perp}$ that intersect the Hamming ball $\calB\parens*{g,\calJ\parens*{\delta_{\inn} - \frac{\lambda}{\delta_{\dec}}} - \eps}$.


\begin{figure}[!ht]
\begin{algorithm}{List Decoding}{$g \in \Sigma^E$, $\gamma\in (0,1)$}{List of cosets $\calL' \subseteq \ex / \ez^{\perp}$ that contains $\calL = \calL \parens*{g, \calJ\parens*{\delta_{\inn} - \frac{\lambda}{\delta_{\dec}}} - \eps }$}\label{algo:ael-decoding-quantum}
\begin{itemize}
\item Pick $\eta = \frac{\eps^2 \delta_{\dec}}{16\delta_{\inn}}$ and $M = \frac{\ln(1 / \gamma)}{p \ln p}$, where $p = \frac{ \parens*{\delta_{\dec}}^2 \cdot \eps^6}{2^{12} \cdot s^{3d} \cdot \delta_{\inn}^4}$.
\item Use \cref{lem:algo_covering_quantum} to obtain a pseudocodeword of degree $t \geq 2d(\frac{s^{3d}}{\eta^2} + 1)$ such that for every $h\in \ex$ which satisfies $\distR{g} < \calJ\parens*{\delta_{\inn} - \frac{\lambda}{\delta_{\dec}}} - \eps$, it holds that
\[
	\distR{\tildeEx{\cdot}} ~<~ \delta_{\inn} - \frac{\lambda}{\delta_{\dec}} - \eps
\]
%
%
%
\item Initialize $\calL' = \emptyset$.
\item Repeat $M$ times:
\begin{enumerate}[(i)]
	\item Choose $u$ uniformly at random from $\{1,2,\cdots , \frac{s^{3d}}{\eta^2}\}$.
	\item Choose a random subset $U \subseteq R$ of size $u$, and let $N(U)$ denote the edge neighborhood of $U$. Sample a random assignment $\sigma$ for $N(U)$ using the local distribution for this set of edges. That is, $\sigma$ is chosen with probability $\tildeEx{\indi{\zee_U = \sigma}}$.
	\item Condition $\tildeEx{\cdot}$ on $\zee_U = \sigma$, and let $\dupPE{\cdot} = \condPE{\cdot}{\zee_U = \sigma}$ be the conditioned pseudocodeword of degree $2d\parens*{\frac{s^{3d}}{\eta^2}+ 1} - 2\cdot d \cdot u \geq 2d$.
	\item Generate $y \in (\F_{q}^{b_{\out}})^L$ by independently sampling $\dupPE{\phi_X^{-1} (\zee_{\li})}$. That is, for $w \in \F_q^{b_{\out}}$, the probability that $y_{\li} = w$ is $\dupPE{\indi{\phi_X^{-1}(\zee_{\li}) = w}}$.
	\item Call the $X$-decoder of $(\dx,\dz)$ on $y$. If a coset $\dxz + \dz^{\perp}$ is found close to $y$,
	then add to $\calL'$ the coset $\wphix(\dxz) + \ez^{\perp}$.
\end{enumerate}
\item Prune $\calL'$ to only include one representative per coset, via Gaussian elimination.
\item Return $\calL'$.
\end{itemize}
\vspace{5pt}
\end{algorithm}
\end{figure}

We now argue that the probability of all the cosets in $\calL$ being included in $\calL'$ is at least $1 - \gamma$. Fix such a coset in $\calL$ and let $h \in \Sigma^E$ be the nearest codeword to $g$ from this coset $h+\ez^\perp$, so that $\distR{g} < \calJ\parens*{\delta_{\inn} - \frac{\lambda}{\delta_{\dec}}} - \eps$. \cref{lem:algo_covering_quantum} implies that $\tildeEx{\cdot}$ satisfies
\begin{align}\label{eqn:agreement_before_conditioning_quantum}
	\distR{\tildeEx{\cdot}} ~<~ \delta_{\inn} - \frac{\lambda}{\delta_{\dec}} - \eps\;.
\end{align}
To be able to use the distance proof of \cref{lem:sos_ael_distance_quantum}, we need the pseudocodeword to be $\eta$-good. \cref{lem:conditioning_eta_good_quantum} shows that this can be obtained by random conditioning. In particular, there exists a $u^*\in \{1,2,\cdots ,\frac{s^{3d}}{\eta^2} \}$ such that,
\begin{equation}\label{eq:conditioning_eta_good_quantum}
		\Ex{r_1,r_2,\cdots,r_{u^*}}{\Ex{\li,\ri}{\tildecov\brackets*{\zee_{\li},\zee_{\ri} \vert \zee_{r_1},\zee_{r_2},\cdots,\zee_{r_{u^*}}}}} ~\leq~ \eta\;.
\end{equation}
Note that we picked $\eta = \frac{\eps^2 \delta_{\dec}}{16\delta_{\inn}}$ in \cref{algo:ael-decoding-quantum}, so that $u^*$ is bounded by a constant independent of $n$.

Suppose $u$ is chosen in step (i) of \cref{algo:ael-decoding-quantum} to be $u^*$, which happens with probability at least $\frac{\eta^2}{s^{3d}}$. Rewriting \cref{eq:conditioning_eta_good_quantum} with $U$ to denote the randomly chosen set $\{r_1,r_2,\cdots ,r_u\}$ and $N(U)\sub E$ to denote the set of all edges incident on $U \sub R$, we get
\begin{align*}
	\Ex{\substack{U\subseteq R \\ |U| = u}}{\Ex{\li,\ri}{\tildecov\brackets*{\zee_{\li},\zee_{\ri} \vert \zee_{N(U)}}}} ~\leq~ \eta, \\
\Ex{\substack{U\subseteq R, |U| = u \\ \sigma \sim \zee_{N(U)}}}{\Ex{\li,\ri}{\tildecov\brackets*{\zee_{\li},\zee_{\ri} \vert \zee_{N(U)} = \sigma}}} ~\leq~ \eta,
\end{align*}
where $\sigma \sim \zee_{N(U)}$ is used to denote that $\sigma \in \Sigma^{N(U)}$ is sampled according to the local distribution induced on $N(U)$ by $\tildeEx{\cdot}$.

That is, on average, we end up with an $\eta$-good pseudocodeword. We actually picked $\eta$ to be much smaller than the bound on average covariance we will be needing, so that the probability (over conditionings) of obtaining a weaker low covariance becomes very close to 1. This is needed to be able to take a union bound with some other low-probability events we will see soon. A simple application of Markov's inequality shows that the probability of obtaining an $\frac{\eps\delta_{\dec}}{4}$-good pseudocodeword is at least $1-\frac{\eps}{4\delta_{\inn}}$.
\begin{align}\label{eqn:eta_good_quantum}
\Pr{U,\sigma}{\Ex{\li,\ri}{\tildecov\brackets*{\zee_{\li},\zee_{\ri} \vert \zee_U = \sigma}} > \frac{\eps \delta_{\dec}}{4}} ~\leq~ \frac{4\eta}{\eps \delta_{\dec}} ~\leq~ \frac{\eps}{4\delta_{\inn}}\;.
\end{align}

Therefore, we started with a pseudocodeword that is close to $h$ (\cref{eqn:agreement_before_conditioning_quantum}), and then condition it to make it $\frac{\eps\delta_{\dec}}{4}$-good. We must also argue that this conditioned pseudocodeword is still close to $h$, at least with some probability. This probability cannot be made too large, and this is why we needed to ensure that the low average correlation property holds with probability close to 1, so that both of these hold simultaneously with some positive probability. To do this, we use the law of total expectation and another application of Markov's inequality,
\begin{align}
	\Ex{U,\sigma}{\distR{\condPE{\cdot}{\zee_U = \sigma}}} ~&=~ \Delta_R(\tildeEx{\cdot},h)\;, \\
	~&<~ \delta_{\inn} - \frac{\lambda}{\delta_{\dec}} - \eps\;, && (\text{Using }\cref{eqn:agreement_before_conditioning_quantum}) \\
	\implies \Pr{U,\sigma}{\distR{\condPE{\cdot}{\zee_U = \sigma}} \leq \delta_{\inn} - \frac{\lambda}{\delta_{\dec}} - \frac{\eps}{2}} ~&\geq~ \frac{\eps/2}{\delta_{\inn} - \frac{\lambda}{\delta_{\dec}} - \frac{\eps}{2}} ~\geq~ \frac{\eps}{2\delta_{\inn}}\;.\label{eqn:agreement_after_conditioning_quantum}
\end{align}

Using a union bound over \cref{eqn:eta_good_quantum} and \cref{eqn:agreement_after_conditioning_quantum}, we get
\begin{align}\label{eqn:union_bound_quantum}
	\Pr{V,\sigma}{\Ex{\li,\ri}{\tildecov\brackets*{\zee_{\li},\zee_{\ri} \vert \zee_U = \sigma}} ~\leq~ \frac{\eps \delta_{\dec}}{4} \text{ and } \distR{\condPE{\cdot}{\zee_U = \sigma}} \leq \delta_{\inn} - \frac{\lambda}{\delta_{\dec}} - \frac{\eps}{2}} ~\geq~ \frac{\eps}{4\delta_{\inn}}\;.
\end{align}

Suppose such a conditioning pair $(U,\sigma)$ is chosen in \cref{algo:ael-decoding-quantum}, and this happens with probability at least $\frac{\eps}{4\delta_{\inn}}$.

We define $\dupPE{\cdot} = \condPE{\cdot}{\zee_U = \sigma}$ as in \cref{algo:ael-decoding-quantum}, and let us call the corresponding covariance operator as $\dupCov\brackets*{\cdot}$. Rewriting \cref{eqn:union_bound_quantum}, we see that $\dupPE{\cdot}$ satisfies the following two properties:
\begin{gather*}
	\Ex{\li,\ri}{\dupCov\brackets*{\zee_{\li},\zee_{\ri}}} ~\leq~ \frac{\eps \delta_{\dec}}{4} \;, \\
	\distR{\dupPE{\cdot}} ~\leq~ \delta_{\inn} - \frac{\lambda}{\delta_{\dec}} - \frac{\eps}{2}\;.
\end{gather*}

Using the SoS distance proof for AEL from \cref{lem:sos_ael_distance_quantum} for $\dupPE{\cdot}$, we can use the above upper bound on $\distR{\dupPE{\cdot}}$ to deduce an upper bound on $\distLperp{\dupPE{\cdot}}$.
\begin{align*}
	&\distR{\dupPE{\cdot}} ~\geq~ \delta_{\inn} - \frac{\lambda + \eps \delta_{\dec} / 4}{\distLperp{\dupPE{\cdot}}} \\
	\implies \quad & \quad \frac{\lambda}{\delta_{\dec}} + \frac{\eps}{2} ~\leq~ \frac{\lambda + \eps \delta_{\dec}/4}{\distLperp{\dupPE{\cdot}}} \\
	\implies \quad & \distLperp{\dupPE{\cdot}} ~\leq~ \delta_{\dec} - \frac{\eps \delta_{\dec}}{4\parens*{\frac{\lambda}{\delta_{\dec}} + \frac{\eps}{2}}} ~\leq~ \delta_{\dec} - \frac{\eps \delta_{\dec}}{4 \delta_{\inn}}\;.
\end{align*}

We next wish to show that $y\in (\F_q^{b_{\out}})^L$ obtained by rounding $\tildeEx{\cdot}$ in step (iv) of \cref{algo:ael-decoding-quantum} can be used to find the coset $h+\ez^{\perp}$. Let $\dxh$ be the codeword in $\dx \sub (\F_q^{b_{\out}})^L$ corresponding to $h \in \ex$. In other words, $\dxh_{\li} = \unconx(h_{\li})$. 

Let $\dist{y, \dxh}$ denote the normalized Hamming distance between $y$ and $\dxh$, viewed as strings of length $n$ over the alphabet $\F_q^{b_{\out}}$. On average, $y$ satisfies
\begin{align*}
	\Ex{y}{ \dist{y, \dxh}} &= \Ex{y}{ \Ex{\li \in L}{\indi{y_{\li} \neq \dxh_{\li}}}} \\
	&= \Ex{\li \in L}{~\Ex{y}{ \indi{y_{\li} \neq \dxh_{\li}}}} \\
	&= \Ex{\li \in L}{~\Ex{y_{\li}}{ \indi{y_{\li} \neq \dxh_{\li}}}} \\
	&= \Ex{\li \in L}{~ \sum_{w}{ \dupPE{\indi{\unconx(\zee_{\li}) = w}} \indi{w \neq \dxh_{\li}}}} \\
	&= \Ex{\li \in L}{~ \dupPE{\indi{\unconx(\zee_{\li}) \neq \dxh_{\li}}}} \\
	&= \Ex{\li \in L}{~ \dupPE{\indi{\unconx(\zee_{\li}) \neq \unconx(h_{\li})}}} \\
	&= \Ex{\li \in L}{~ \dupPE{\indi{\zee_{\li} \not\in h_{\li} + \cz^{\perp}}}}\\
	&= \distLperp{\dupPE{\cdot}} ~\leq~ \delta_{\dec} - \frac{\eps \delta_{\dec}}{4 \delta_{\inn}} \mper
\end{align*}

Using Markov's inequality,
\begin{align*}
	\Pr{y}{ \dist{y, \dxh} \leq \delta_{\dec}} \geq \frac{\eps}{4\delta_{\inn}} \mper
\end{align*}

Suppose a $y$ is found in \cref{algo:ael-decoding-quantum} such that $\dist{y, \dxh} \leq \delta_{\dec}$, which happens with probability at least $\frac{\eps}{4\delta_{\inn}}$. Then the $X$-decoder of the $(\dx,\dz)$ code must return the coset $\dxh+ \dz^{\perp}$. Let $\dxz$ be a coset representative of $\dxh+\dz^\perp$ returned by the $X$-decoder of $(\dx,\dz)$. Using $\ez^{\perp} = \wphix(\dz^{\perp}) + \mathbb{F}_{q}^n \otimes \cz^{\perp}$ from \cref{prop:concat_dual}, 
\begin{align*}
	\dxh - \dxz ~&\in~ \dz^{\perp}\;, \\
	\wphix(\dxh - \dxz) ~&\in~ \ez^\perp\;, \\
	\wphix(\dxh) - \wphix(\dxz) ~&\in~ \ez^\perp\;, \\
	h - \wphix(\dxz) ~&\in~ \ez^\perp + \F_q^n \otimes \cz^\perp = \ez^\perp\;.
\end{align*}

Therefore, \cref{algo:ael-decoding-quantum} adds the coset $\wphix(\dxz)+\ez^\perp = h+\ez^\perp$ to $\calL'$. In conclusion, if the following three events happen, the coset $h + \ez^{\perp}$ is added to the list $\calL'$.
\begin{enumerate}
\item $u=u^*$ is chosen, which happens with probability at least $\frac{\eta^2}{s^{3d}}$.
\item The pair $(U,\sigma)$ to condition on is chosen such that \cref{eqn:union_bound_quantum} holds. Conditioned on previous event, this happens with probability at least $\frac{\eps}{4\delta_{\inn}}$.
\item A $y$ is generated so that $\dist{g,\dxh} \leq \delta_{\dec}$. Conditioned on above two events, this happens with probability at least $\frac{\eps}{4\delta_{\inn}}$.
\end{enumerate}

Therefore, in any iteration, the coset $h + \ez^{\perp}$ is added to the list $\calL'$ with probability at least,
\[
	p ~=~ \frac{\eta^2}{s^{3d}} \cdot \frac{\eps}{4\delta_{\inn}} \cdot \frac{\eps}{4\delta_{\inn}} ~=~ \frac{ \parens*{\delta_{\dec}}^2 \cdot \eps^6}{4096 \cdot s^{3d} \cdot \delta_{\inn}^4} ~=~ \Omega_{s,d,\delta_{\dec}} (\eps^6) \mper
\]

Note that this immediately implies an upper bound of $1/p$ on the list size, although we can get better list sizes (combinatorially) by appealing to the covering lemma and approximate \Caratheodory theorem.

Finally, we show that with enough repetitions, $\calL'$ must contain the entire list $\calL$ with high probability. The probability that a coset in $\calL$ does not get added to $\calL'$ in $M$ iterations is
\[
	(1-p)^M ~\leq~ e^{-p \cdot M}
\]

With a union bound over the entire list, which is of size at most $1/p$, the probability that any coset in $\calL$ is not present in $\calL'$ is at most
\[
	\frac{1}{p} \cdot e^{-p \cdot M} ~=~ e^{-p M \ln p} ~\leq~ \gamma \;\text{ if } M = \frac{\ln(1 / \gamma)}{p \ln p}. \qedhere
\]
\end{proof}

\subsection{Near-MDS Quantum LDPC Codes list decodable up to Johnson bound}\label{sec:near_mds_quantum}

In this section, we show how instantiating the AEL amplification with unique decodable asymptotically good QLDPC codes leads to QLDPC codes near the (quantum) Singleton bound that can be list decoded up to the Johnson bound. 

\begin{theorem}[Near-MDS Codes decodable upto Johnson bound]\label{thm:near_mds_main_quantum}
	For any $0<\rho<1$, and for any $\nfrac{1}{2} >\eps_1, \eps_2>0$, there is an infinite family of quantum LDPC codes $(\fx,\fz)$ with the following properties:
	\begin{enumerate}[(i)]
		\item The rate of the code is at least $\rho$ and distance is at least $\frac{1-\rho-\eps_1}{2}$.
		\item The code is over an alphabet of size $2^{\calO(\eps_1^{-6}\log(1/\eps_1))}$.
		\item The code of blocklength $n$ can be list decoded from radius $\calJ(\frac{1-\rho - \eps_1}{2})-\eps_2$ in time $n^{\calO_{\eps_1}(1/\eps_2^4)}$.
	\end{enumerate}
\end{theorem}

\begin{proof}
We show how to instantiate \cref{thm:list_decoding_quantum} to obtain such codes. 
We will use the AEL distance amplification based on a family of $(n,d,\lambda)$-expanders, with $d$ and $\lambda$ to be chosen later.

The inner code $\cC = (\cx,\cz)$ is chosen to be a quantum Reed-Solomon code of rate $\rho_0 \defeq \frac{\rho}{1-\eps_1}$ and distance $\frac{1-\rho_0}{2}$. Since the inner code has blocklength $d$, the alphabet size of $\cC$ must be at least $d$, and we choose $\cC$ to be defined over a finite field $\F_{2^{b_{\inn}}}$, where $b_{\inn} \leq 1+\log_2 d$. Using \cref{lem:field_downgrade}, we will view $\cC$ as a vector space CSS code with alphabet $\F_2^{b_{\inn}}$.

For the outer code $\cD = (\dx,\dz)$, we fold the binary QLDPC codes of \cite{LZ23:decodable} into blocks of size $b_{\out}$, with $b_{\out} = b_{\inn}\cdot \rho_0 d$ so that concatenation is well defined. The starting binary code from \cite{LZ23:decodable} is chosen to have rate $1-\eps_1$ that can be unique decoded from radius $\delta_{\dec}=\Omega(\eps_1^2)$, and these properties are preserved after folding.

Let $\lambda = \eps_1 \delta_{\dec},\;$ so that $\lambda \leq \Omega(\eps_1^3)$ and $d=\calO(1/\eps_1^6)$. The rate of the AEL-amplified code $(\fx,\fz)$ is $(1-\eps_1) \rho_0 = \rho$, and the distance is at least
\begin{align*}
	\frac{1-\rho_0}{2} - \frac{\lambda}{\delta_1} &~\geq~ \frac{1}{2} \parens*{1-\frac{\rho}{1-\eps_1}} - \frac{\lambda}{\delta_{\dec}} \\
	&~\geq~ \frac{1}{2} \parens*{1-\rho -2\eps_1 \rho} -\eps_1 \\
	&~\geq~ \frac{1-\rho}{2} - 2\eps_1
\end{align*}

The final code is defined over alphabet $\F_2^{b_{\inn}\cdot d}$, which is of size at most $ 2^{(1+\log d))d}  = 2^{ \calO\inparen{\eps_1^{-6}\log(1/\eps_1)}}$. Since AEL amplification preserves the LDPC property of the outer code, our final code $(\fx,\fz)$ is also LDPC. For list decodability, we use \cref{thm:list_decoding_quantum} to claim that the above code can be list decoded from $\calJ(\frac{1-\rho}{2}-2\eps_1)-\eps_2$ in time $n^{\calO_{q,d,\delta_{\dec}}(1/\eps_2^4)} = n^{\calO_{\eps_1}(1/\eps_2^4)}$. The claimed parameters can be obtained by replacing $\eps_1$ by $\eps_1/4$. 
\end{proof}

\chapter{Fast Decoding of Ta-Shma's Code via Regularity Lemmas}\label{chap:regularity}
%
A binary code $\calC \subseteq \F_2^N$ is said to be $\eps$-balanced if any two distinct codewords $x,y \in \calC$ satisfy $\Delta(x,y) \in [\nfrac{(1-\eps)}{2}, \nfrac{(1+\eps)}{2}]$, where $\Delta(x,y)$ denotes the relative distance between the two codewords. 
Finding explicit and optimal constructions of such codes, and indeed of codes where the distances are at least $\nfrac{(1-\eps)}{2}$ is a central problem in coding theory~\cite{G10:ICM,G09Survey}, with many applications to the theory of pseudorandomness~\cite{Vadhan12}. 
Recently, Ta-Shma~\cite{TS17} gave a breakthrough construction of (a family of) explicit $\eps$-balanced codes, with near-optimal rates, for arbitrarily small $\eps > 0$. 
For the case of codes with distance at least $\nfrac{(1-\eps)}{2}$, the existential rate-distance tradeoffs established by Gilbert~\cite{G52} and Varshamov~\cite{V57}, prove the existence of codes with rate $\Omega(\eps^2)$, while McEliece \etal~\cite{MRRW77} prove an upper bound of $O(\eps^2 \log(1/\eps))$ on the rate. 
On the other hand, Ta-Shma's result yields an \emph{explicit} family of codes with rate $\Omega(\eps^{2+o(1)})$.


%
\vspace{-10 pt}
\paragraph{Decoding algorithms.}
The near-optimal $\eps$-balanced codes of Ta-Shma~\cite{TS17} (which we will refer as Ta-Shma codes) were not known to be efficiently decodable at the time of their discovery. 
In later work, polynomial-time unique decoding algorithms for (a slight modification of) these codes were developed in~\cite{JQST20} (building on~\cite{AJQST20}) using the Sum-of-Squares (SoS) hierarchy of semidefinite programming (SDP) relaxations. 
For unique decoding of codes with rates $\Omega(\eps^{2+\alpha})$ (when $\alpha > 0$ is an arbitrarily small constant) these results yield algorithms running in time $N^{O_{\alpha}(1)}$.
These algorithms also extend to the case when $\alpha$ is a vanishing function of $\eps$, and to the problem of list decoding within an error radius of $\nfrac12 - \eps'$ (for $\eps'$ larger than a suitable function of $\eps$) with running time $N^{O_{\eps, \eps',\alpha}(1)}$. 
However, the $O_{\alpha}(1)$ exponent of $N$ obtained in the unique decoding case is quite large even for a fixed constant $\alpha$ (say $\alpha = 0.1$), and the exponent in the list decoding case grows with the parameter $\eps$.

\section{Our Results}
In this work, we use a different approach based on new weak regularity lemmas (for structures identified by the SoS algorithms), resulting in near-linear time algorithms for both the above tasks. The algorithms below work in time $\tilde{O}_{\eps}(N)$ for $\eps$-balanced Ta-Shma codes with rates $\Omega(\eps^{2+\alpha})$, even when $\alpha$ is a (suitable) vanishing function of $\eps$.

\begin{restatable}[Near-linear Time Unique Decoding]{theorem}{TheoMainUniqueDec}\label{theo:main}
  For every $\epsilon > 0$ sufficiently small, there are explicit binary linear Ta-Shma codes
  $\Cc_{N,\epsilon,\alpha} \subseteq \mathbb{F}_2^N$ for infinitely many values $N \in \mathbb{N}$ with
  \begin{enumerate}[(i)]
   \item distance at least $1/2 - \epsilon/2$ (actually $\epsilon$-balanced),
   \item rate $\Omega(\epsilon^{2 + \alpha})$ where $\alpha = O(1/(\log_2(1/\epsilon))^{1/6})$, and
   \item an $r(\epsilon) \cdot \tilde{O}(N)$ time unique decoding algorithm that that decodes within radius $1/4 - \eps/4$ and works with high probability,
  \end{enumerate}
  where $r(\epsilon) = \exp(\exp(\polylog(1/\epsilon)))$.
\end{restatable}

We can also obtain list decoding results as in~\cite{JQST20}, but now in near-linear time.
\begin{restatable}[Near-linear Time Gentle List Decoding]{theorem}{TheoMainGentleListDec}\label{theo:gentle_list_decoding}
  For every $\epsilon > 0$ sufficiently small, there are explicit binary linear Ta-Shma codes $\Cc_{N,\epsilon,\alpha} \subseteq \mathbb{F}_2^N$
  for infinitely many values $N \in \mathbb{N}$ with
  \begin{enumerate}[(i)]
   \item distance at least $1/2 - \epsilon/2$ (actually $\epsilon$-balanced),
   \item rate $\Omega(\epsilon^{2 + \alpha})$ where $\alpha = O(1/(\log_2(1/\epsilon))^{1/6})$, and
   \item an $r(\epsilon) \cdot \tilde{O}(N)$ time list decoding algorithm that decodes
         within radius $1/2 - 2^{-\Theta((\log_2(1/\epsilon))^{1/6})}$ and works with high probability,
  \end{enumerate}
  where $r(\epsilon) = \exp(\exp(\poly(1/\epsilon)))$.
\end{restatable}
While~\cref{theo:gentle_list_decoding} yields a list decoding radius close to $1/2$, we remark that the above tradeoff between the list decoding radius and rate, is far from the state-of-the-art of $1/2-\epsilon$ radius with rate $\Omega(\epsilon^3)$ of Guruswami and Rudra~\cite{GR08}.
Considering a three way trade-off involving distance, rate, and list-decoding radius, ~\cref{theo:gentle_list_decoding} can be seen as close to optimal with respect to the first two parameters, and quite far off with respect to the third one. 
Finding an algorithm for codes with optimal tradeoffs in all three parameters, is a very interesting open problem.
Another interesting problem is understanding the optimal dependence of the ``constant'' factors $r(\eps)$ in the running times. We have not tried to optimize these factors in our work.
%

%
%

\vspace{-10 pt}
\paragraph{Direct-Sum Codes and ``Structured Pseudorandomness''.}
Ta-Shma's code construction can be viewed as a special case of ``distance amplification via direct-sum", an operation with several applications in coding and complexity theory~~\cite{ABNNR92, IW97, GI01, ImpagliazzoKW09, DinurS14, DDGEKS15, Chan16, DinurK17, A02:icm}.
Given a (say) linear code $\Cc_0 \subseteq \F_2^n$ and a collection of tuples $W \subseteq [n]^k$, we define it's ``direct-sum lifting" as $\Cc = \dsum_W(\Cc_0) \subseteq \F_2^{\abs{W}}$ where
\[
\dsum_{W}(\Cc_0) ~\defeq~ \left\{(z_{i_1} + \cdots + z_{i_k})_{(i_1,\ldots,i_k) \in W} ~\mid~ z \in \Cc_0 \right\} \mper
\]
It is easy to see that if $\Cc_0$ is $\eps_0$-balanced for a constant $\eps_0$,
then taking $W=[n]^k$ results in $\dsum_W(\Cc_0)$ being $\eps$-balanced with $\eps = \eps_0^k$ (though with vanishing rate). 
A standard sampling argument shows that a random $W \subseteq [n]^k$ with $\abs{W} = O(n/\eps^2)$ also suffices, while yielding rate $\Omega(\eps^2)$.
Rozenman and Wigderson \cite{RW08} suggested a derandomization of this argument using a ``pseudorandom" $W$ constructed from the collection of all length-$(k-1)$ walks on a suitable expander graph. While this result can be shown to achieve a rate of $\Omega(\eps^{4+o(1)})$, Ta-Shma achieves a rate of $\Omega(\eps^{2 + o(1)})$ using a carefully constructed \emph{sub-collection} of walks on an expander with a special form.

The above results show that pseudorandomness can be used to amplify distance, since the collections $W$ above behave \emph{like} a random $W$. However, finding decoding algorithms for such codes requires understanding properties of these collections which are \emph{unlike} a random $W$, since random collections yield codes with (essentially) random generator matrices, where we do not expect efficient algorithms.

Our results can be viewed as showing that when the collection $W$ satisfies a form of ``structured multi-scale pseudorandomness" property
\footnote{As discussed later, there are several notions of ``structured pseudorandom'' for (ordered and unordered) hypergraphs. We describe splittability here, since this is the one directly relevant for our algorithmic applications. }
called \emph{splittability} (identified in previous work), it can be exploited for algorithm design.
One can think of splittability as capturing properties of the complete set $[n]^k$, which are not present in a (sparse) random $W \subseteq [n]^k$. For the case of $k=4$, when $W = [n]^4$, if we consider a graph between pairs $(i_1,i_2)$ and $(i_3, i_4)$, which are connected when $(i_1,\ldots, i_4) \in W$, then this defines an expanding (complete) graph when $W = [n]^4$. On the other hand, for a random $W$ of size $O(n)$, such a graph is a matching with high probability. 
Splittability requires various such graphs defined in terms of $W$ to be expanders.
\begin{definition}[Splittability, informal]
Given $W \subseteq [n]^k$ and $a,b \in [k]$, let $W[a,b] \subseteq [n]^{b-a+1}$ denote the tuples obtained by considering $(i_a, \ldots, i_b)$ for every $(i_1, \ldots, i_k) \in W$. 
We say $W$ can be $\tau$-split at position $t$, if the bipartite graph with vertex sets $W[1,t]$ and $W[t+1,k]$, edge-set $W$, and (normalized) biadjacency matrix
%
%
$\Ess_t \in \R^{W[1,t] \times W[t+1,k]}$,  is an expander satisfying $\sigma_2(\Ess_t) \leq \tau$. 
We say that $W$ is $\tau$-splittable if for all $1 \leq a \leq t < b \leq k$, $W[a,b]$ can be $\tau$-split at position $t$. 
\end{definition}
Note that when $k=2$, this coincides with the definition of (bipartite) graph expansion. It is also easy to show that collections of length-$(k-1)$ walks on a graph with second singular value $\lambda$, satisfy the above property with $\tau=\lambda$. 
The sub-collections used by Ta-Shma can also be shown to splittable (after a a slight modification) and we recall this proof from~\cite{JQST20} in \cref{appendix:ta-shma}.
%
%

The key algorithmic component in our decoding results, is a general \emph{list decoding} result for codes constructed via direct-sum operations, which reduces the task of list decoding for $\dsum_W(\Cc_0)$ to that of unique decoding for the code $\Cc_0$, when $W$ is $\tau$-splittable for an appropriate $\tau$. 
The splittability property was identified and used in previous work~\cite{AJQST20, JQST20}, for the analysis of SoS based algorithms, which obtained the above reduction in $N^{O_{\eps}(1)}$ time. 
Regularity based methods also allow for near-linear time algorithms in this  general setting of direct-sum codes, with a simpler and more transparent proof (and improved dependence of the list decoding radius on $\tau$ and $k$).
%
%
%
%
\begin{theorem}[List Decoding Direct Sum (informal version of~\cref{thm:direct-sum-decoding})]
Let $\calC_0 \subseteq \F_2^n$ be an $\epsilon_0$-balanced linear code, which is unique-decodable to distance $\nfrac{(1-\eps_0)}{4}$ in time $\calT_0$.
Let $W \subseteq [n]^k$ be a $\tau$-splittable collection of tuples.
Let $\calC = \dsum_W(\calC_0)$ be $\eps$-balanced, and let $\beta$ be such that
\[
\beta ~\gg~ \max\inbraces{\sqrt{\eps}, ~\inparen{\tau \cdot k^3}^{1/2}, ~\inparen{\frac12+2\eps_0}^{k/2}} \mper
\]
Then, there exists a randomized algorithm, which given $\tilde{y} \in \F_2^{W}$, recovers the list
$$
\calL_{\beta}(\tilde{y}) \defeq \inbraces{y \in \calC ~|~ \Delta(\tilde{y}, y) \leq \nfrac12 - \beta},
$$
with probability at least $1 - o(1)$, in time $\tilde{O}(C_{\beta,k,\eps_0} \cdot (\abs{W}+\calT_0))$, where $C_{k,\beta,\eps_0}$ only depends on $k$,
$\beta$ and $\eps_0$.
\end{theorem}
%

\vspace{-10 pt}
\paragraph{Splittable Regularity.}
The technical component of our results is a novel understanding of splittable structures, via weak regularity lemmas.
This provides a different way of exploiting ``structured pseudorandomness" properties in hypergraphs, which may be of interest beyond applications considered here.

For the case of graphs (\ie $k=2)$, several weak regularity lemmas are known which can be applied to (say) dense subgraphs of an expanding graph~\cite{ReingoldTTV08, TrevisanTV09, CCFA09, BV20}. 
As in the Frieze-Kannan~\cite{FK96:focs} weak regularity lemma for dense graphs, these lemmas decompose the adjacency matrix $A_{W'}$ of a subgraph $W' \subseteq W$, as a weighted sum of a small number of cut matrices ($\one_{S_{\ell}}\one_{T_{\ell}}^{\trans}$ for $S_{\ell},T_{\ell} \subseteq [n]$), such that one can use this decomposition to count the number of edges between any subsets $S, T \subseteq [n]$ \ie
\[
\abs{\one_S^{\trans} \inparen{A_{W'} - \sum_{\ell} c_{\ell} \cdot \one_{S_{\ell}}\one_{T_{\ell}}^{\trans}}\one_T} ~\leq~ \eps\cdot \abs{W} \mper
\]
This can be thought of as computing an ``approximation'' of $A_{W'}$ using a small number of cut matrices $\one_{S_j}\one_{T_j}^{\trans}$, which is ``indistinguishable'' by any cut matrix $\one_{S}\one_T^{\trans}$. 

More generally, one can think of the above results as approximating any function $g: W \to [-1,1]$ (with $g = \one_{W'}$ in the example above) with respect to a family of "split" functions $\calF \subseteq \inbraces{f:[n] \to [-1,1]}^{\otimes 2}$, where the approximation itself is a sum of a small number of of functions from $\calF$ \ie for all $f_1, f_2 \in \calF$
\[
\abs{\ip{g - \sum_{\ell} c_{\ell} \cdot f_{\ell, 1} \otimes f_{\ell, 2}~}{f_1 \otimes f_2}} ~\leq~ \eps \cdot \abs{W} \mper
\]
Our regularity lemma for splittable $W \subseteq [n]^k$, extends the above notion of approximation, using $k$-wise split functions of the form
$f_1 \otimes \cdots \otimes f_k$. 
We obtain near-linear time weak regularity decompositions for classes of $k$-wise cut functions of the form
\[
\CUT^{\otimes k} \coloneqq \set{\pm \one_{S_1} \otimes \cdots \otimes \one_{S_k} \mid S_1,\ldots,S_k \subseteq [n]},
\]
and also for signed version of these $k$-wise cut functions
$$
\SCUT^{\otimes k} \coloneqq \set{\pm \sset_{S_1} \otimes \cdots \otimes \sset_{S_k} \mid S_1,\ldots,S_k \subseteq [n]},
$$
where $\sset_{S} = (-1)^{\one_{S}}$. 
For our decoding results, we will use $\SCUT^{\otimes k}$. Our near-linear time weak regularity
decomposition result is given next.

\begin{theorem}[Efficient Weak Regularity (informal version of~\cref{theo:eff_weak_reg})]
  Let $W \subseteq [n]^k$ and let $\calF$ be either $\CUT^{\otimes k}$ or $\SCUT^{\otimes k}$.
  Suppose $g \in \mathbb{R}^{[n]^k}$ is supported on $W$ and has bounded norm.
  For every $\error > 0$,
  if $W$ is $\tau$-splittable with $\tau=O(\error^2/k^3)$, then
  we can find $h = \sum_{\ell=1}^{p} c_\ell \cdot f_\ell$ in $\widetilde{O}_{k,\error}(\abs{W})$ time,
  where $p = O(k^2/\error^2)$, $f_{\ell} \in \calF$ and $c_{\ell} \in \mathbb{R}$,
  such that $h$ is a good approximator to $g$ in the following sense
  $$
  \max_{f \in \calF}\quad \ip{g - h}{f} ~\le~ \error \cdot \abs{W},
  $$
  where the inner product is over the counting measure on $[n]^k$.
\end{theorem}
We note that an existential version of the above theorem follows known abstract versions of the Frieze-Kannan regularity lemma~\cite{TrevisanTV09, BV20}, via a relatively simple use of splittability. 
However, making a black-box application of known regularity lemmas algorithmic, requires computing a form of "tensor cut-norm", which is believed to be hard to even approximate in general\footnote{Strictly speaking, we only need to approximate this for ``splittable" tensors. It is possible that one could use existing regularity lemmas black box, and use splittability to design a fast algorithm for tensor cut-norm. In our proof, we instead choose to use the matrix cut-norm algorithms as black-box, and use splittability to modify the proof of the regularity lemma.} (unlike the matrix case). 
The nontrivial component of the result above, is obtaining a regularity lemma which allows for a \emph{near-linear time computation}, while still achieving parameters close to the existential version.

\paragraph{Related Work.}
As discussed above, the decoding results in this paper, were derived earlier using algorithms based on the SoS hierarchy~\cite{AJQST20,JQST20}, though with significantly larger running times (and somewhat worse dependence on parameters). 
A common thread in the SoS algorithms is to relate the task of decoding, to that of solving instances of constraint satisfaction problems with $k$ variables in each constraint ($k$-CSPs).
The original weak regularity lemma of Frieze and Kannan~\cite{FK96:focs} was indeed motivated by the question of approximately solving $k$-CSPs on dense structures (see also \cite{KV09:spectral}).
Several extensions of the Frieze-Kannan lemma are known, particularly for various families of sparse pseudorandom graphs~\cite{KR02:regularity, ReingoldTTV08, TrevisanTV09, OGT15, BV20}.
Oveis-Gharan and Trevisan~\cite{OGT15} also proved a new weak regularity lemma for ``low threshold-rank" graphs, which was used to obtain approximation algorithms for some 2-CSPs, where the previously known algorithms used the SoS hierarchy~\cite{BRS11, GuruswamiS11}.
Our work can be viewed as an extension of these ideas to the case of $k$-CSPs.

Ideas based on regularity lemmas, were also employed in the context of list decoding of Reed-Muller codes, by Bhowmick and Lovett~\cite{BL18}.
They use analogues of the abstract weak regularity lemma~\cite{TrevisanTV09} and the \Szemeredi regularity lemma over finite fields, but these are only used to prove bounds on the list size, rather than in the algorithm itself. 
On the other hand, our decoding algorithm crucially uses the decomposition obtained via our weak regularity lemma for (real-valued functions on) splittable structures.

In general, expansion phenomena have a rich history of interaction with coding theory (\eg~\cite{GI01,GSurvey04,GI05,RWZ20}) including to the study of linear (or near-linear) time decoding backing to the seminal work of Sipser and Spielman~\cite{SS96}. The codes in~\cite{SS96} were good codes, though not near optimal in terms of distance-rate trade-off.
Several other notions of ``structured pseudorandomness'' for hypergraphs (referred to as high-dimensional expansion) have also been considered in literature, which also have connections to the decoding of good codes. In particular, the notion of ``double sampler'' was used to obtain algorithms for the list decoding for direct-product codes~\cite{DHKLNTS19}. 
The notions of local spectral expansion~\cite{DinurK17}, cosystolic expansion~\cite{EvraK16}, and multilayer agreement samplers~\cite{DDHR20:ltc}, are also used to connect structured pseudorandomness to the design of locally testable codes.
The notion of splittability was also studied for unordered hypergraphs in terms of ``complement walks'' by Dinur and Dikstein~\cite{DD19}, and in terms of ``swap walks'' in~\cite{AJT19}, for high-dimensional expanders defined via local spectral expansion.

In this chapter, unlike the rest of this thesis, we will deal exclusively with binary codes. The results mentioned here were later extended to deal with codes (and $k$-CSPs) over larger alphabets by Jeronimo \cite{Jer23}.



\section{A Technical Overview}

We now give a more detailed overview of some of the technical components of our proof.
\paragraph{Splittability.}
The key structural property used for our algorithmic and structural results, is the ``structured pseudorandomness" of ordered hypergraphs $W \subseteq [n]^k$, which we call \emph{splittability}. 
The canonical example one can think of for this case, is a collection  of all length-$(k-1)$ walks on a (say) $d$-regular expander graph $G$ on $n$ vertices. Note that this satisfies $\abs{W [a,b]} = d^{b-a} \cdot n$, where $W[a,b]$ represents the collection of sub-tuples with coordinates between indices $a$ and $b$ \ie portions of the walks between the $a^{th}$ and $b^{th}$ step.
We will restrict our discussion in this paper only to \emph{$d$-regular collections} $W \subseteq [n]^k$ satisfying $\abs{W[a,b]} = d^{b-a} \cdot n$.

We briefly sketch why the collection of length-3 walks (\ie the case $k=4$) is splittable.
Recall that splittability requires various graphs with sub-tuples to be expanding, and in particular consider the graph between $W[1,2]$ and $W[3,4]$, with edge-set $W[1,4]$.
If $E(G)$ is the set of edges in $G$ included with both orientations, then
note that $W[1,2] = W[3,4] = E(G)$, and $(i_1,i_2), (i_3,i_4)$ are connected iff $(i_2,i_3) \in E(G)$. If $M \in \R^{W[1,2]\times W[3,4]}$ denotes the biadjacency matrix of the bipartite graph $H$ on $W[1,2] \times W[3,4]$, then up to permutations of rows and columns, we can write $M$ as $\matr A_{G} \otimes \matr J_{d}/d$, where $\matr J_d$ denotes the $d \times d$ all-1s matrix and $\matr A_{G}$ is the normalized adjacency matrix of $G$, since each tuple $(i_2, i_3) \in E(G)$ contributes $d^2$ edges in $H$ (for choices of $i_1$ and $i_4$). 
Thus $\sigma_2(M) = \sigma_2(\matr A_G)$, which is small if $G$ is an expander. A similar argument also works for splits in other positions, and for longer walks.

The above argument can also be extended to show that the sub-collections of walks considered by Ta-Shma (after a slight modification) are splittable, though the structure and the corresponding matrices are more involved there (see \cref{appendix:ta-shma}). 
%
%

\paragraph{Regularity for graphs and functions.}
We first consider an analytic form of the Frieze-Kannan regularity lemma (based on~\cite{TrevisanTV09}). 
Let $g: \calX \to [-1,1]$ be any function on a finite space $\calX$ with an associated probability measure $\mu$, and let $\cF \subseteq \inbraces{f: \calX \to [-1,1]}$ be any class of functions closed under negation. Say we want to construct a ``simple approximation/decomposition'' $h$, which is indistinguishable from $g$, for all functions in $f$ \ie 
\[
\ip{g-h}{f}_{\mu} ~=~ \Ex{x \sim \mu}{\inparen{g(x)-h(x)} \cdot f(x)} ~\leq~ \delta \qquad \forall f \in \cF \mper
\]
We can view the regularity lemma as saying that such an $h$ can always be constructed as a sum of $1/\delta^2$ functions from $\cF$. 
Indeed, we can start with $h^{(0)} = 0$, and while there exists $f_{\ell}$ violating the above condition, we update $h^{(\ell+1)} = h^{(\ell)} + \delta \cdot f_{\ell}$. 
The process must stop in $1/\delta^2$ steps, since $\smallnorm{g-h^{(\ell)}}^2$ can be shown to decrease by $\delta^2$ in every step.
\[
\norm{g-h^{(\ell)}}_{\mu}^2 - \norm{g-h^{(\ell+1)}}_{\mu}^2
~=~ 2\delta \cdot \ip{g-h^{(\ell)}}{f_{\ell}}_{\mu} - \delta \cdot \norm{f_{\ell}}_{\mu}^2
~\geq~ \delta^2 \mper
\]
In fact, the above can be seen as gradient descent for minimizing the convex function $F(h) = \sup_{f \in \cF} \ip{g-h}{f}_{\mu}$. 
Taking $\calX = [n]^2$ with $\mu$ as uniform on $[n]^2$, $g = \one_{E(G)}$ for a (dense) graph $G$, and $\cF$ as all functions (cut matrices) of the form $\pm \one_S\one_T^{\trans}$ yields the weak regularity lemma for graphs, since we get $h = \sum_{\ell} c_{\ell} \cdot f_{\ell} = \sum_{\ell} c_{\ell} \cdot \one_{S_{\ell}}\one_{T_{\ell}}^{\trans}$ such that
\[
\ip{g-h}{f}_{\mu} \leq \delta \quad\forall f \in \cF
\quad \Leftrightarrow \quad
\frac{1}{n^2} \cdot \abs{E_G(S,T) - \sum_{\ell} c_{\ell} \abs{S_{\ell} \cap S} \abs{T_{\ell} \cap T}} \leq \delta \quad\forall S,T \subseteq [n] \mper
\]
Note that the inner product in the above analytic argument can be chosen to be according to any measure on $\calX$, and not just the uniform measure. 
In particular, taking $W \subseteq [n]^2$ to be the edge-set of a (sparse) $d$-regular expander with second singular value (say) $\lambda$, and $\mu = \mu_2$ to be uniform over $W$, we obtain the regularity lemma for subgraphs of expanders.
In this case, after obtaining the approximation with respect to $\mu$, one shows using the expander mixing lemma that if $\ip{g-h}{f}_{\mu_2} \leq \delta$, then $\ip{g- \inparen{\nfrac{d}{n}} \cdot h}{f}_{\mu_1 \otimes \mu_1} \leq (\nfrac{d}{n}) \cdot \delta'$, where $\mu_1$ denotes the uniform measure on $[n]$ and $\delta' = \delta + \lambda$. This gives a sparse regularity lemma, since for $G \subseteq W$ and $g = \one_G$,
{\small
\[
\ip{g-\inparen{\frac{d}{n}}h}{f}_{\mu_1^{\otimes 2}} \leq \frac{d}{n} \cdot \delta' \quad\forall f \in \cF
~~~\Leftrightarrow~~~
\abs{E_G(S,T) - \sum_{\ell} c_{\ell} \cdot \frac{d}{n}\abs{S_{\ell} \cap S} \abs{T_{\ell} \cap T}} \leq \delta' \cdot nd \quad\forall S,T \mper
\]
}
The \emph{algorithmic step} in the above proofs, is finding an $f_{\ell}$ such that $\ip{g-h}{f_{\ell}} > \delta$. For the function class $\calF$ corresponding to cut matrices, this corresponds to solving a problem of the form $\max_{S,T} \abs{\one_S^{\trans} M \one_T}$ for an appropriate matrix $M$ at each step. This equals the cut-norm and can be (approximately) computed using the SDP approximation algorithm of Alon and Naor~\cite{AN04}. Moreover, this can be implemented in near-linear time in the sparsity of $M$, using known fast, approximate SDP solvers of Lee and Padmanabhan~\cite{LeeP20}
or of Arora and Kale~\cite{AK07} (see \cref{subsec:matrix_corr_oracle} for details).

\paragraph{Splittable regularity.} 
For our regularity lemma, the class $\cF$ comprises of ``$k$-split functions'' of the form $f_1 \otimes \cdots \otimes f_k$, where for each $f_{t}$ can be thought of as $\one_{S_t}$ (or $(-1)^{\one_{S_t}}$) for some $S_t \subseteq [n]$. An argument similar to the one above, with the measure $\mu_k$ uniform on $W \subseteq [n]^k$, can yield an \emph{existential version} of the splittable regularity lemma, similar to the one for expander graphs (we now transition from $\mu_k$ to $\mu_1^{\otimes k}$ using a simple generalization of the expander mixing lemma to splittable collections). 
However, the algorithmic step in the above procedure, requires computing
\[
\max_{f_1, \ldots, f_k \in \cF} \ip{g-h}{f_1 \otimes \cdots \otimes f_k}
\]
Unfortunately, such an algorithmic problem is hard to even approximate in general, as opposed to the 2-split case for graphs.
Another approach is to first compute an approximation of a given $g: W \to [-1,1]$, in terms of 2-split functions of the form $f_1 \otimes f_2$, where $f_1:W[1,t] \to [-1,1]$ and $f_2: W[t+1,k] \to [-1,1]$, and then inductively approximate $f_1$ and $f_2$ in terms of 2-split functions, and so on. 
Such an induction does yield an algorithmic regularity lemma, though naively approximating the component functions $f_1$ and $f_2$ at each step, leads to a significantly lossy dependence between the final error, the splittability parameter $\tau$, and $k$.

We follow a hybrid of the two approaches above. 
We give an inductive argument, which at step $t$, approximates $g$ via $h_t$ which is a sum of $t$-split functions. 
However, instead of simply applying another 2-split to each term in the decomposition $h_t$ to compute $h_{t+1}$, we build an approximation for \emph{all of $h_t$} using the regularity argument above from scratch.  
We rely on the special structure of $h_t$ to solve the algorithmic problem $\max_{f_1, \ldots, f_{t+1}}\ip{h_{t} - h_{t+1}}{f_1 \otimes \cdots \otimes f_{t+1}}$, reducing it to a matrix cut-norm computation\footnote{Strictly speaking, we also need to be careful about the bit-complexity of our matrix entries, to allow for near-linear time computation. However, all the entries in matrices we consider will have bit-complexity $O_{k,\delta}(\log n)$.}.
This yields near-optimal dependence of the error on $\tau$ and $k$, needed for our coding applications.
\paragraph{Decoding direct-sum codes using regularity.}
We now consider the problem of decoding, from a received, possibly corrupted, $\tilde{y} \in \F_2^W$, to obtain the closest $y \in \dsum_W(\Cc_0)$ (or a list) \ie finding $\argmin_{z_0 \in \Cc_0} \Delta(\tilde{y}, \dsum_W(z_0))$. 
Let $g: [n]^k \to \pmone$ be defined as $g(i_1,\ldots,i_k) = (-1)^{\tilde{y}_{(i_1,\ldots,i_k)}}$ if $(i_1,\ldots,i_k) \in W$ and 0 otherwise. Also, for any $z \in \F_2^n$, define the function $\chi_z$ as $\chi_z(i) = (-1)^{z_i}$. As before, let $\mu_1$ denote the uniform measure on $[n]$. 
Using that $g$ is 0 outside $W$, and that $\abs{W} =d^{k-1} \cdot n$, we get
\begin{align*}
1 - 2\cdot \Delta(\tilde{y}, \dsum_{W}(z)) 
&~=~ \Ex{(i_1,\ldots,i_k) \in W}{g(i_1,\ldots,i_k) \cdot \chi_z(i_1) \cdots \chi_z(i_k)} \\
&~=~ \inparen{\frac{n}{d}}^{k-1} \cdot  \Ex{(i_1,\ldots,i_k) \in [n]^k}{g(i_1,\ldots,i_k) \cdot \chi_z(i_1) \cdots \chi_z(i_k)} \\
&~=~ \inparen{\frac{n}{d}}^{k-1} \cdot \ip{g}{\chi_z^{\otimes k}}_{\mu_1^{\otimes k}} \mper
\end{align*}
At this point, we modify the problem in three ways.
First, instead of restricting the optimization to $z_0 \in \Cc_0$, we widen the search to all $z \in \F_2^n$. We will be able to show that because of the pseudorandom (distance amplification) properties of $W$, a good (random) solution $z$ found by our algorithm, will be within the unique decoding radius of $\Cc_0$ (with high probability).
Secondly, using the fact that for splittable $W$, the function $g$ has an approximation $h = \sum_{\ell=1}^p c_{\ell} \cdot f_{\ell,1} \otimes \cdots \otimes f_{\ell,k}$ given by the regularity lemma, we can restrict our search to $z$ which (approximately) maximize the objective
\[
\ip{h}{\chi_z^{\otimes k}}_{\mu_1^{\otimes k}} 
~=~ 
\sum_{\ell=1}^p c_{\ell} \cdot \prod_{t \in [k]} \ip{f_{\ell,t}~}{\chi_z}_{\mu_1}
\]
Finally, instead of searching for $\chi_z: [n] \to \pmone$, we further widen the search to  $\fbar: [n] \to [-1,1]$. A random ``rounding'' choosing each $\chi_z(i)$ independently so that $\Ex{\chi_z} = \fbar$ should preserve the objective value with high probability.
We now claim that the resulting search for functions $\fbar$ maximizing $\ip{h}{\fbar^{\otimes k}}_{\mu_1^{\otimes k}}$, can be solved via a simple brute-force search. Note that the objective only depends on the inner products with a finite number of functions $\inbraces{f_{\ell, t}}_{\ell \in [p], t \in [k]}$ with range $\pmone$. Partitioning the space $[n]$ in $2^{pk}$ ``atoms'' based on the values of these functions, we can check that it suffices to search over $\fbar$, which are constant on each atom. Moreover, it suffices to search the values in each atom, up to an appropriate discretization $\eta$, which can be done in time $O\inparen{(\nfrac{1}{\eta})^{2^{pk}}}$. 

For the problem of list decoding $\tilde{y}$ up to radius $\nfrac12 - \beta$, we show that each $z_0 \in \Cc_0$, such that $\dsum_W(z_0)$ is in the list, there must be an $\fbar$ achieving a large value of $\ip{h}{\fbar^{\otimes k}}_{\mu_1^{\otimes k}}$ which then yields a $z$ within the unique decoding radius of $z_0$. Since we enumerate over all $\fbar$, this recovers the entire list. Details of the decoding algorithm are given in \cref{sec:abstract_dec}.



\section{Preliminaries}

We now introduce some notation. The asymptotic notation
$\widetilde{O}(r(n))$ hides polylogarithmic factors in $r(n)$.

\subsection{Codes}

We briefly recall some standard code terminology.
Given $z,z' \in \F_2^n$, recall that the relative
Hamming distance between $z$ and $z'$
is \text{$\Delta(z,z') \coloneqq \abs{\set{i \mid z_i\ne
z_i'}}/n$}. A binary code is any subset $\Cc \subseteq \F_2^n$.
The distance of $\Cc$ is defined as $\Delta(\Cc)
\coloneqq \min_{z\ne z'} \Delta(z,z')$ where $z,z' \in \Cc$. We say
that $\Cc$ is a linear code if $\Cc$ is a linear subspace of
$\mathbb{F}_2^n$. The rate of $\Cc$ is $\log_2(\abs{\Cc})/n$, or
equivalently $\dim(\Cc)/n$ if $\Cc$ is linear.

\begin{definition}[Bias]
  The \emph{bias} of a word $z \in \F_2^n$ is defined as $\bias(z) \coloneqq \abs{\E_{i \in [n]} (-1)^{z_i}}$.
  The bias of a code $\Cc$ is the maximum bias of any non-zero codeword in $\Cc$.
\end{definition}

\begin{definition}[$\epsilon$-balanced Code]
  A binary code $\Cc$ is \emph{$\epsilon$-balanced} if $\bias(z+z') \le \epsilon$ for every pair of distinct $z,z' \in \Cc$.
\end{definition}

\begin{remark}
 For linear binary code $\Cc$, the condition $\bias(\Cc) \leq \epsilon$ is equivalent to $\Cc$ being an $\epsilon$-balanced code.
\end{remark}

\subsection{Direct Sum Lifts}

Starting from a code $\Cc \subseteq \F_2^n$, we amplify its distance by considering the \textit{direct sum lifting} operation based on a collection $W(k) \subseteq [n]^k$.
The direct sum lifting maps each codeword of $\Cc$ to a new word in $\F_2^{|W(k)|}$ by taking the $k$-XOR of its entries on each element of $W(k)$.

\begin{definition}[Direct Sum Lifting]
  Let $W(k) \subseteq [n]^k$.
  For $z \in \F_2^n$, we define the \emph{direct sum lifting} as $\dsum_{W(k)}(z) = y$ such that
  $y_{(i_1,\dots,i_k)} = \sum_{j =1}^k z_{i_j}$ for all $(i_1,\dots,i_k) \in W(k)$.
  The direct sum lifting of a code $\Cc \subseteq \F_2^n$ is
  $$
  \dsum_{W(k)}(\Cc) = \{\dsum_{W(k)}(z) \mid z \in \mathcal C\}.
  $$
  We will omit $W(k)$ from this notation when it is clear from context.
\end{definition}

\begin{remark}
  We will be concerned with collections $W(k) \subseteq [n]^k$ arising
  from length-$(k-1)$ walks on expanding structures (mostly arising from
  Ta-Shma's direct sum construction~\cite{TS17}).
\end{remark}

We will be interested in cases where the direct sum lifting reduces
the bias of the base code; in~\cite{TS17}, structures with such a
property are called \emph{parity samplers}, as they emulate the reduction in
bias that occurs by taking the parity of random samples.

\begin{definition}[Parity Sampler]
  A collection $W(k) \subseteq [n]^k$ is called an \emph{$(\epsilon_0, \epsilon)$-parity sampler} if for all $z \in \F_2^{n}$
  with $\bias(z) \leq \epsilon_0$, we have $\bias(\dsum_{W(k)}(z)) \leq \epsilon$.
\end{definition}

\subsection{Splittable Tuples}

We now formally define the \emph{splittability} property for a collection of tuples
$W(k) \subseteq [n]^k$. For $1 \le a \le b \le k$, we define $W[a,b] \subseteq [n]^{(b-a+1)}$ as
$$
W[a,b] \coloneqq \set{(i_a,i_{a+1},\dots,i_b) \mid (i_1,i_2,\dots,i_k) \in W(k)}.
$$
We will work with $d$-regular tuples in the following sense.
\begin{definition}[Regular tuple collection]
  We say that $W(k) \subseteq [n]^k$ is $d$-regular if 
  for every $1 \le a \le b \le k$, we have
  \begin{itemize}
    \item $\abs{W[a,b]} =  d^{b-a} \cdot n$,
    \item $W[a] = [n]$.
  \end{itemize}
\end{definition}
A collection $W(k)$ being $d$-regular is analogous to a graph being
$d$-regular. 
\begin{example}
  The collection $W(k)$ of all length-$(k-1)$ walks on a $d$-regular connected graph $G=([n],E)$
  is a $d$-regular collection of tuples.
\end{example}

The space of functions $\mathbb{R}^{W[a,b]}$ is endowed with an inner
product associated to the uniform measure $\mu_{[a,b]}$ on $W[a,b]$.
We use the shorthand $\mu_{b}$ for $\mu_{[1,b]}$.

\begin{definition}[Splitable tuple collection]\label{def:split}
  Let $\tau > 0$.
  We say that a collection $W(k) \subseteq [n]^k$ is $\tau$-\emph{splittable} if it is $d$-regular
  and either $k = 1$ or for every $1 \le a \le t < b \le k$ we have
  \begin{itemize}
    \item the \emph{split} operator $\swap{W[a,s]}{W[t+1,b]} \in \mathbb{R}^{W[a,t] \times W[t+1,b]}$ defined as
          $$
          \left(\swap{W[a,t]}{W[t+1,b]}\right)_{(i_a,\dots,i_t),(i_{t+1},\dots,i_k)} \coloneqq \frac{\one\left[ (i_a,\dots,i_t, i_{t+1},\dots i_k) \in W[a,b]\right]}{d^{k-t}}
          $$
          satisfy $\sigma_2(\swap{W[a,t]}{W[t+1,b]}) \le \tau$, where $\sigma_2$ denotes the second largest singular value.
  \end{itemize}
\end{definition}

\begin{example}
  The collection $W(k)$ of all length-$(k-1)$ walks on a $d$-regular a
  graph $G=([n],E)$ whose normalized adjacency matrix has second
  largest singular value at most $\tau$ is a collection of $\tau$-splittable tuples
  as shown in~\cite{AJQST20}.
\end{example}

\begin{example}
  The collection $W(k)$ of tuples arising (from a slight modification) of the
  direct sum construction of Ta-Shma~\cite{TS17} is a
  $\tau$-splittable as shown in~\cite{JQST20}. Precise parameters are recalled
  later as~\cref{fact:ta-shma_splittable_tuples} of~\cref{appendix:ta-shma}.
\end{example}


\subsection{Factors}
It will be convenient to use the language of factors, to search the decompositions identified by regularity lemmas, for relevant codewords.
This concept (from ergodic theory) takes a rather simple form in our finite settings: it is just a partition of base set $\calX$, with an associated operation of averaging functions defined on $\calX$, separately over each piece.
\begin{definition}[Factors and measurable functions]
Let $\calX$ be a finite set. A {\deffont factor} $\calB$ is a partition of the set $\calX$, and the subsets of the partition are referred to as \deffont{atoms} of the factor. A function $f: \calX \to \calR$ is said to {\deffont measurable} with respect to $\calB$ ($\calB$-measurable) if $f$ is constant on each atom of $\calB$.
\end{definition}

\begin{definition}[Conditional averages]
If $f: \calX \to \R$ is a function, $\mu$ is a measure on the space $\calX$, and $\calB$ is a factor, then we define the {\deffont conditional average function} $\Ex{f|\calB}$ as
\[
\Ex{f|\calB}(x) ~\defeq~ \Ex{y \sim \mu|\calB(x)}{f(y)} \mcom
\]
where $\calB(x)$ denotes the atom containing $x$. Note that the function $\Ex{f|\cB}$ is measurable with respect to $\cB$.
\end{definition}
We will need the following simple observation regarding conditional averages.
\begin{proposition}\label{prop:measurable-inner-product}
  Let $h: \calX \to \R$ be a $\calB$-measurable function, and let $f: \calX \to \R$ be any function. Then, for any measure $\mu$ over $\calX$, we have
  \[
    \ip{h}{f}_{\mu} ~=~ \ip{h}{\Ex{f|\calB}}_{\mu} \mper
  \]
\end{proposition}
\begin{proof}
By definition of the $\calB$-measurability, $h$ is constant on each atom, and thus we can write $h(x)$ as $h(\calB(x))$. 
  \begin{align*}
    \ip{h}{f}_{\mu} ~=~ \Ex{x \sim \mu}{h(x) \cdot f(x)}
    &~=~ \ExpOp_{x \sim \mu}\Ex{y \sim \mu | \calB(x)}{h(y) \cdot f(y)} \\
    &~=~ \Ex{x \sim \mu}{h(\calB(x)) \cdot \Ex{y \sim \mu | \calB(x)}{f(y)}} \\
    &~=~ \Ex{x \sim \mu}{h(x) \cdot \Ex{f|\calB}(x)} ~=~ \ip{h}{\Ex{f|B}}_{\mu} \mper \qedhere 
  \end{align*}
\end{proof}
The factors we will consider will be defined by a finite collection of functions appearing in a regularity decomposition.
\begin{definition}[Function factors]
Let $\calX$ and $\calR$ be finite sets, and let $\cF_0 = \inbraces{f_1, \ldots, f_r: \calX \to \calR}$ be a finite collection of functions. We consider the factor $\cB_{\cF_0}$ defined by the functions in $\cF_0$, as the factor with atoms $\inbraces{x ~|~ f_1(x) = c_1, \ldots, f_r(x) = c_r}$ for all $(c_1,\ldots, c_r) \in \cR^r$.
\end{definition}

\begin{remark}\label{remark:factor_sigma_algebra_comp}
Note that when the above function are indicators for sets \ie each $f_j = \one_{S_j}$ for some $S_j \subseteq \calX$, then the function factor $\calB_{\calF_0}$ is the same as the $\sigma$-algebra generated by these sets.
Also, given the functions $f_1,\ldots,f_r$ as above, the function factor $\calB_{\calF_0}$ can be computed in time $O(\abs{\calX} \cdot \abs{\calR}^{r})$.
\end{remark}

\subsection{Functions and Measures}
We describe below some classes of functions, and spaces with associated measures, arising in our proof. 
The measures we consider are either uniform on the relevant space, or are products of measures on its component spaces.

\paragraph{Function classes.} 
Let $S \subseteq [n]$. We define $\sset_S \colon [n] \to \set{\pm 1}$
as $\sset_S(i) \coloneqq (-1)^{\one_{i \in S}}$ (we observe that as
defined $\sset_S$ is not a character\footnote{Strictly speaking
  $\sset_S$ is not a character but by identifying the elements of
  $[n]$ with those of a canonical basis of $\F_2^n$ it becomes a
  character for $\F_2^n$.}).  We need the following two collection of
functions for which algorithmic results will be obtained.

\begin{definition}[CUT functions]
We define the set of $0/1$ \emph{CUT} cut functions as
$$
\CUT^{\otimes k} \coloneqq \set{\pm \one_{S_1} \otimes \cdots \otimes \one_{S_k} \mid S_1,\ldots,S_k \subseteq [n]},
$$
and defined the set of $\pm 1$ \emph{CUT} functions as
$$
\SCUT^{\otimes k} \coloneqq \set{\pm \sset_{S_1} \otimes \cdots \otimes \sset_{S_k} \mid S_1,\ldots,S_k \subseteq [n]}.
$$
\end{definition}

We will use a higher-order version of cut norm.
\begin{definition}\label{def:higher_order_cut_norm}
  Let $g \in \mathbb{R}^{[n]^k}$, the $k$-tensor cut norm is
  $$
  \norm{g}_{\square^{\otimes k}} \coloneqq \max_{f \in \textup{CUT}^{\otimes k}} \ip{g}{f},
  $$
  where the inner product is over the counting measure on $[n]^k$. 
\end{definition}

Some of our results hold for more general class of functions.

\begin{definition}[$t$-split functions]\label{def:abstract_tensor_spaces}
  Suppose $W(k)$ is a regular collection of $k$-tuples.
  For $t \in \set{0,\dots,k-1}$, we define a generic class of tensor product functions $\calF_t$ as
  $$
  \calF_t \subseteq \left\{\pm f_1 \otimes \cdots \otimes f_t \otimes f_{t+1} \mid f_j \subseteq \mathbb{R}^{W[1]} \text{ for } i \le t, f_{t+1} \subseteq \mathbb{R}^{W[t+1,k]}, \norm{f_j}_{\infty} \le 1\right\}.
  $$
  To avoid technical issues, we assume that each $\calF_t$ is finite.   
\end{definition}

Fixing some $\calF \subseteq \mathbb{R}^{\calX}$, we define the set of
functions that are linear combinations of function from $\calF$ with
coefficients of bounded support size and bounded $\ell_{1}$-norm as
follows
$$
\calH(R_0,R_1,\calF) \coloneqq \left\{ \sum_{\ell=1}^p c_{\ell} \cdot f_{\ell} \mid p \le R_0, \sum \abs{c_{\ell}} \le R_1, f_{\ell} \in \calF \right\}.
$$

\paragraph{Measures and inner products.}
Recall that $\mu_1 \coloneqq \mu_{[1,1]}$ is the uniform measure on $W[1]$
(equivalently uniform measure on $W[i]$ since $W(k)$ is regular) and
$\mu_{[t+1,k]}$ is the uniform measure on $W[t+1,k]$. 
We define following measure $\nu_t$ as
$$
\nu_t \coloneqq \left(\mu_1\right)^{\otimes t} \otimes \left(\mu_{[t+1,k]}\right).
$$
Note that $\nu_0$ is the equal to $\mu_k$ and $\nu_{k-1}$ is equal to $\mu_1^{\otimes k}$.
We will need to consider inner products of functions according to various measures defined above, which we will denote as $\ip{\cdot}{\cdot}_{\mu}$ for the measure $\mu$. 
When a measure is not indicated, we take the inner product $\ip{f}{g}$ to be according to the counting measure on the domains of the functions $f$ and $g$.
%



\section{Weak Regularity for Splittable Tuples}

We will show how functions supported on a (possibly) sparse splittable
collection of tuples $\tuples{k} \subseteq [n]^k$ admit weak regular
decompositions in the style of Frieze and
Kannan~\cite{FK96:focs}. In~\cref{subsec:abstract_reg_lemma}, we start
by showing an abstract regularity lemma for functions that holds in
some generality and does not require splittability. Next,
in~\cref{subsec:split_mixing_lemma}, we show that splittable
collections of tuples satisfy suitable (simple) generalizations of the
expander mixing lemma for graphs which we call splittable mixing
lemma. By combining this abstract weak regularity decomposition with
splittable mixing lemmas, we obtain \emph{existential} decomposition
results for splittable tuples
in~\cref{subsec:existential_weak_reg_split}. Then, we proceed to make
these existential results not only algorithmic but near-linear time
computable in~\cref{subsec:eff_reg}. These algorithmic results will
rely on fast cut norm like approximation algorithms tailored to our
settings and this is done in~\cref{subsec:matrix_corr_oracle}. As
mentioned previously, this last step borrows heavily from known
results~\cite{AN04,AK07,LeeP20}.

\subsection{Abstract Weak Regularity Lemma}\label{subsec:abstract_reg_lemma}

We now show a weak regularity decomposition lemma for functions that
works in some generality and does not require splittability.
We now fix some notation for this section.
Let $\calX$ be a finite set endowed with a probability measure $\mu$.
Let $\mathbb{R}^{\calX}$ be a Hilbert space endowed with inner product
$\ip{f}{g}_{\mu} = \E_{\mu} \left[f \cdot g\right]$ and associated
norm $\norm{\cdot}_{\mu} = \sqrt{\ip{\cdot}{\cdot}_{\mu}}$.
Let $\calF \subseteq \set{f\colon \calX \to \mathbb{R} \mid
  \norm{f}_\mu \le 1}$ be a finite collection of functions such that
$\calF = - \calF$.

In a nutshell, given any $g \in \mathbb{R}^{\calX}$, the abstract weak
regularity lemma will allow us to find an approximator $h$, with
respect to the semi-norm $g-h \mapsto \max_{f \in \calF} \ip{g-h}{f}$,
which is a linear combinations of a certain \emph{small} number of
functions from $\calF$ (where this number depends only on the
approximation accuracy and the norm $\norm{g}_{\mu}$). This means that $g$ and
$h$ have approximately the same correlations with functions from
$\calF$. We will produce $h$ in an iterative procedure, where at each
step an oracle of the following kind (\cf~\cref{def:correlation_oracle}) is invoked.

\begin{definition}[Correlation Oracle]\label{def:correlation_oracle}
  Let $1 \ge \delta \ge \delta' > 0$ be accuracy parameters and $B > 0$.
  We say that $\calO_{\mu,B}$ is a $(\delta,\delta')$-correlation oracle for $\calF$ if
  given $h \in \mathbb{R}^{\calX}$ with $\norm{h}_{\mu}^2 =O(B)$ if there exists $f \in \calF$ with $\ip{h}{f} \ge \delta$,
  then $\calO_{\mu,B}$ returns some $f' \in \calF$ with $\ip{h}{f'} \ge \delta'$.
\end{definition}

More precisely, our abstract weak regularity decomposition is as
follows.

\begin{lemma}[Abstract Weak Regularity]\label{lemma:abstract_weak_reg}
  Let $\calO_{\mu,B}$ be a $(\delta,\delta')$-correlation oracle for $\calF$
  with $\error \ge \error' > 0$. Let $g \colon \calX \to \mathbb{R}$ 
  satisfy $\norm{g}^2_{\mu} \le B$. Then, we can find $h = \sum_{\ell=1}^p c_\ell \cdot f_\ell \in \calH(B/(\error')^2, B/\error',\calF)$
  with $f_\ell \in \calF$, $c_\ell \in [\error'/(1+\error'/\sqrt{B})^p,\error']$ and $\norm{h}_{\mu}^2 \le B$  such that
  $$
  \max_{f \in \mathcal{F}} \quad \ip{g-h}{f}_\mu ~\le~ \error.
  $$
  Furthermore, if $\calO_{\mu,B}$ runs in time $\calT_{\calO_{\mu,B}}$, then $h$ can be computed in
  $$
  \widetilde{O}\left(\poly(B, 1/\error') \cdot (\calT_{\calO_{\mu,B}}+ \abs{\supp(\mu)})\right)
  $$
  time, where $\supp(\mu)$ is the support of $\mu$.
  The function $h$ is constructed in~\cref{algo:abstract_weak_reg} as the final function
  in a sequence of approximating functions $h^{(\ell)} \in \calH(B/(\error')^2, B/\error',\calF)$.
\end{lemma}

The proof is based on \cref{algo:abstract_weak_reg}. 
\begin{figure}
\begin{algorithm}{Regularity Decomposition Algorithm}{$g \colon \calX \to \mathbb{R}$}{$h = \sum_{\ell=1}^p c_\ell \cdot f_\ell$}\label{algo:abstract_weak_reg}
     \begin{itemize}
      \item Let $\Pi$ be the projector onto the convex ball $\set{g' \in \mathbb{R}^{\calX} \mid \norm{g'}^2_{\mu} \le B}$.
      \item Let $\ell=0$ and $h^{(\ell)} = 0$
      \item While $\max_{f \in \mathcal{F}} \ip{g-h^{(\ell)}}{f}_\mu \ge \error$: \begin{itemize}
            \item $\ell = \ell + 1$
            \item Let $f_\ell \in \calF$ be such that $\ip{g-h^{(\ell-1)}}{f_\ell}_\mu \ge \error'$ \quad   \emph{(Correlation Oracle $\calO_{\mu,B}$ Step)}
            \item Let $c_\ell = \error'$
            \item $h^{(\ell)} = \Pi (h^{(\ell-1)} + c_\ell \cdot f_\ell)$
          \end{itemize}
      \item Let $p = \ell$
      \item return $h = \sum_{\ell=1}^p c_\ell \cdot f_\ell$
     \end{itemize}
\end{algorithm}
\end{figure}
Before getting into the proof, we will need the following general fact about projections onto a convex body.
\begin{fact}[Implicit in Lemma 3.1 of~\cite{B15}]\label{fact:proj_convex_body}
    Let $\mathcal{Y}$ be a compact convex body in a finite dimensional Hilbert space $\mathcal{V}$
    equipped with inner product $\ip{\cdot}{\cdot}_{\nu}$ and associated norm $\norm{\cdot}_{\nu}$.
    Let $\Pi_{\mathcal{Y}}$ be projector onto $\mathcal{Y}$. Then, for $y \in \mathcal{Y}$ and $x \in \mathcal{V}$,
    we have
    $$
    \norm{y-x}_{\nu}^2 ~\ge~ \norm{y - \Pi_{\mathcal{Y}}(x)}_{\nu}^2 ~+~ \norm{\Pi_{\mathcal{Y}}(x) - x}_{\nu}^2.
    $$
\end{fact}

\begin{proof}[Proof of~\cref{lemma:abstract_weak_reg}]
  We will show that the norm of $\norm{g - h^{(\ell)}}_{\mu}$ strictly decreases as the algorithm
  progresses. Computing we obtain
  \begin{align*}
    \norm{g - h^{(\ell)}}_{\mu}^2 ~=~& \norm{g - \Pi (h^{(\ell-1)} + c_\ell \cdot f_\ell)}_{\mu}^2\\
                             ~\le~ & \norm{g - (h^{(\ell-1)} + c_\ell \cdot f_\ell)}_{\mu}^2 ~- \norm{(h^{(\ell-1)} + c_\ell \cdot f_\ell)-\Pi (h^{(\ell-1)} + c_\ell \cdot f_\ell)}_{\mu}^2\\    
                             ~\le~ & \norm{g - (h^{(\ell-1)} + c_\ell \cdot f_\ell)}_{\mu}^2\\
                             ~=~& \norm{g - h^{(\ell-1)}}_{\mu}^2 ~+~ c_\ell^2 \cdot \norm{f_\ell}_{\mu}^2 ~-~ 2 c_\ell \cdot \ip{g-h^{(\ell-1)}}{f_\ell}_{\mu}\\
                             ~\le~& \norm{g - h^{(\ell-1)}}_{\mu}^2 ~-~ (\error')^2
  \end{align*}
  where the first inequality is due to \cref{fact:proj_convex_body}, and the last inequality follows from $c_\ell=\error'$,  the bound $\norm{f_\ell}_{\mu} \le 1$ and
  $$
  \ip{g-h^{(\ell-1)}}{f_\ell}_{\mu} \ge \error'.
  $$
  Since $\norm{g}^2_{\mu} \le B$ and $\norm{g - h^{(\ell)}}^2_{\mu}$ decreases by at least $(\error')^2$ in each iteration,
  we conclude that the algorithm halts in at most $p \le B/(\error')^2$ steps.

  By construction each $c_\ell$ is initialized to $\error'$ and can not increase (it can only decrease due to projections).
  Thus, we obtain $\sum_{\ell=1}^p \abs{c_\ell} \le p \cdot \error' \le B/\error'$. Also by construction
  at termination $\norm{h}^2_{\mu} \le B$. It remains to show that $c_\ell \ge \error'/(1+\error'/\sqrt{B})^p$.
  Note that the projection $\Pi (h^{(\ell-1)} + c_\ell \cdot f_\ell)$ at each iteration  either does
  nothing to the coefficients $c_\ell$'s or scales them by a factor of at most $(1+\error'/\sqrt{B})$
  since $\norm{h^{(\ell-1)}}_{\mu} + \norm{c_\ell \cdot f_\ell}_{\mu} \le \sqrt{B}(1+\error'/\sqrt{B})$. This
  readily implies the claimed lower bound on the coefficients $c_\ell$'s at termination. Moreover,
  we have $h^{(\ell)} \in \calH(B/(\error')^2, B/\error',\calF)$ also by construction.

  \noindent \textbf{Running Time:} The decomposition algorithm calls the correlation oracle at most $p+1$ times.
  Since the coefficients $c_\ell$ always lie in $[\error'/(1+\error'/\sqrt{B})^p,\error'] \subseteq [\error'/\exp(p \error'/\sqrt{B}),\error']$,
  the bit complexity is $C = O(p \error'/\sqrt{B})$ and computing the projection (which amounts to computing $h^{(\ell)}/\norm{h^{(\ell)}}_{\mu}$
  if $\norm{h^{(\ell)}}_{\mu}^2 > B$) takes at most $\widetilde{O}(p^2 \cdot \poly(C) \cdot \abs{\supp(\mu)})$. Then the total
  running time is at most
  $$
  \widetilde{O}(p (\calT_{\calO_{\mu,B}} + p^2 \cdot \poly(C) \cdot \abs{\supp(\mu)})) ~=~ \widetilde{O}\left(\poly(B, 1/\error') \cdot (\calT_{\calO_{\mu,B}}+ \abs{\supp(\mu)})\right),
  $$
  concluding the proof.
\end{proof}

  If we are only interested in an existential version of~\cref{lemma:abstract_weak_reg}, we can always
  use a trivial existential $(\error,\error)$-correlation oracle. However, to obtain weak regularity decompositions
  efficiently in our settings, we will later use efficient $(\error,\error')$-correlation oracle with $\error'=\Omega(\error)$.
  
  As shown in \cref{chap:intro}, such an existential regularity lemma can also be used to prove a weak form of Johnson bound. Improvements in the regularity lemma for Ta-Shma codes therefore could have implications for constructing binary codes closer to list decoding capacity, which is interesting even if algorithmic questions are ignored.

\subsection{Splittable Mixing Lemma}\label{subsec:split_mixing_lemma}

A splittable collection of tuples gives rise to several expanding
split operators (see~\cref{def:split}). This allows us to show that a
splittable collection satisfies some higher-order analogues of the
well known expander mixing lemmas for graphs (\cf
\cite[Section 2.4]{HLW06}) as we make precise next.

\begin{lemma}[Splittable Mixing Lemma]\label{lemma:splittable_mixing}
  Suppose $\tuples{k} \subseteq [n]^k$ is a $\tau$-splittable
  collection of tuples. For every $t \in \set{0,\dots,k-2}$ and
  every $f,f' \in \calF_{t+1}$, we have
  $$
  \abs{\ip{f'}{f}_{\nu_{t+1}} ~-~ \ip{f'}{f}_{\nu_t}} ~\le~ \tau.
  $$
\end{lemma}

\begin{proof}
  Let $f = f_1 \otimes \cdots \otimes f_t \otimes f_{t+1} \otimes f_{t+2}$
  and $f' = f_1' \otimes \cdots \otimes f_t' \otimes f_{t+1}' \otimes f_{t+2}'$.
  We have
  \begin{align*}
    \Big\vert \ip{f'}{f}_{\nu_{t+1}} &- \ip{f'}{f}_{\nu_t} \Big\vert\\
    ~=~& \abs{ \prod_{i=1}^{t} \ExpOp_{\mu_1} f_i f_i' } \cdot \abs{ \ExpOp_{\mu_1 \otimes \mu_{[t+2,k]}} f_{t+1}f_{t+1}' \otimes f_{t+2}f_{t+2}' -
                 \ExpOp_{\mu_{[t+1,k]}} f_{t+1}f_{t+1}' \otimes f_{t+2}f_{t+2}'}\\
    ~\le~& \abs{ \ExpOp_{\mu_1 \otimes \mu_{[t+2,k]}} f_{t+1}f_{t+1}' \otimes f_{t+2}f_{t+2}' -  \ExpOp_{\mu_{[t+1,k]}} f_{t+1}f_{t+1}' \otimes f_{t+2}f_{t+2}' }.
  \end{align*}
  Let $f_{t+1}'' =  f_{t+1}f_{t+1}'$ and $f_{t+2}'' = f_{t+2}f_{t+2}'$. Note that
  $$
  \ExpOp_{\mu_1 \otimes \mu_{[t+2,k]}} f_{t+1}'' \otimes f_{t+2}'' - \ExpOp_{\mu_{[t+1,k]}} f_{t+1}'' \otimes f_{t+2}'' = \ip{f_{t+1}''}{\left(\frac{\matr J_{\textup{rec}}}{\abs{W[t+2,k]}}  - \swap{W[t+1]}{W[t+2,k]}\right) f_{t+2}''}_{\mu_1},
  $$
  where $\matr J_{\textup{rec}}$ is the (rectangular) $\abs{W[t+1]} \times \abs{W[t+2,k]}$ all ones matrix.
  Using the $\tau$-splittability assumption, we have the following bound on the largest singular value
  $$
  \sigma\left(\frac{\matr J_{\textup{rec}}}{\abs{W[t+2,k]}}  - \swap{W[t+1]}{W[t+2,k]}\right) \le \sigma_2\left(\swap{W[t+1]}{W[t+2,k]}\right) \le \tau.
  $$
  Then
  $$
  \abs{ \ExpOp_{\mu_1 \otimes \mu_{[t+2,k]}} f_{t+1}f_{t+1}' \otimes f_{t+2}f_{t+2}' -  \ExpOp_{\mu_{[t+1,k]}} f_{t+1}f_{t+1}' \otimes f_{t+2}f_{t+2}' } \le \tau,
  $$
  concluding the proof.
\end{proof}

We can iterate the preceding lemma to obtain the following.

\begin{lemma}[Splittable Mixing Lemma Iterated]\label{lemma:splittable_mixing_iterated}
  Suppose $\tuples{k} \subseteq [n]^k$ is a $\tau$-splittable
  collection of tuples. For every $f=f_1\otimes \cdots \otimes f_k \in \calF_{k-1}$, we have
  $$
  \abs{ \ExpOp_{\nu_{0}} f ~-~ \ExpOp_{\nu_{k-1}} f } ~\le~ (k-1)\cdot\tau.
  $$
\end{lemma}

\begin{proof}
  Let $1 \in \calF_{k-1}$ be the constant $1$ function.
  Note that for any $t \in \set{0,\dots,k-1}$ the restriction of any
  $f' \in \calF_{k-1}$ to the support of $\nu_{t}$ which we denote by
  $f'\vert_t$ belongs to $\calF_t$. It is immediate that $\ip{f}{1}_{\nu_{t}} = \ip{f\vert_t}{1}_{\nu_{t}}$.
  Computing we obtain
  \begin{align*}
    \abs{ \ExpOp_{\nu_{0}} f ~-~ \ExpOp_{\nu_{k-1}} f } ~=~    \abs{ \ip{f}{1}_{\nu_0} ~-~ \ip{f}{1}_{\nu_{k-1}} }
                                              ~\le~  &  \sum_{i=0}^{k-2} \abs{ \ip{f}{1}_{\nu_i} ~-~ \ip{f}{1}_{\nu_{i+1}} }\\
                                              ~=~  &  \sum_{i=0}^{k-2} \abs{ \ip{f\vert_t}{1\vert_t}_{\nu_i} ~-~ \ip{f\vert_{t+1}}{1\vert_{t+1}}_{\nu_{i+1}} }\\
                                             ~\le~ & \sum_{i=0}^{k-2} \tau, \qquad \qquad \textup{(By~\cref{lemma:splittable_mixing})}
  \end{align*}
  finishing the proof.
\end{proof}

In~\cref{subsec:eff_reg}, we will need two corollaries of the
splittable mixing lemma which we prove now.

\begin{restatable}{claim}{ClaimIndistinguishability}\label{claim:indistinguishability_under_measure_extension}
  Let $\tuples{k} \subseteq [n]^k$ be a $\tau$-splittable collection of tuples.
  Let $t \in \set{0,\dots,k-2}$ and $h_{t+1} \in \calH(R_0,R_1,\calF_{t+1})$.
  For every $f \in \calF_{t+1}$, we have
  $$
  \abs{ \ip{h_{t+1}}{f}_{\nu_{t+1}} ~-~ \ip{h_{t+1}}{f}_{\nu_t} } ~\le~ \tau \cdot R_1.
  $$  
\end{restatable}

\begin{proof}
  Since $h_{t+1} \in \calH(R_0,R_1,\calF_{t+1})$, we can write $h_{t+1} = \sum_{\ell} c_{\ell} \cdot f_{\ell}$,
  where $f_{\ell} \in \calF_{t+1}$ and $\sum_{\ell} \abs{c_{\ell}} \le R_1$. By the splittable
  mixing lemma, \cf~\cref{lemma:splittable_mixing}, we have
  $$
  \abs{ \ip{h_{t+1}}{f}_{\nu_{t+1}} ~-~ \ip{h_{t+1}}{f}_{\nu_t} } ~\le~  \sum_{\ell} \abs{c_{\ell} } \cdot \abs{ \ip{f_{\ell}}{f}_{\nu_{t+1}} ~-~ \ip{f_{\ell}}{f}_{\nu_t} } ~\le~ \tau \cdot R_1. \qedhere
  $$
\end{proof}

\begin{restatable}{claim}{ClaimNormIndistinguishability}\label{claim:norm_indistinguishability_under_measure_extension}
  Let $\tuples{k} \subseteq [n]^k$ be a $\tau$-splittable
  collection of tuples.
  Let $t \in \set{0,\dots,k-2}$ and $h_{t+1} \in \calH(R_0,R_1,\calF_{t+1})$.
  Then
  $$
  \abs{ \norm{h_{t+1}}_{\nu_{t+1}}^2 ~-~ \norm{h_{t+1}}_{\nu_t}^2 } ~\le~ \tau \cdot R_1^2.
  $$
\end{restatable}

\begin{proof}
  Since $h_{t+1} \in \calH(R_0,R_1,\calF_{t+1})$, we can write $h_{t+1} = \sum_{\ell} c_{\ell} \cdot f_{\ell}$,
  where $f_{\ell} \in \calF_{t+1}$ and $\sum_{\ell} \abs{c_{\ell}} \le R_1$. By the splittable
  mixing lemma, \cf~\cref{lemma:splittable_mixing}, we have
  $$
  \abs{ \ip{h_{t+1}}{h_{t+1}}_{\nu_{t+1}} ~-~ \ip{h_{t+1}}{h_{t+1}}_{\nu_t} } ~\le~  \sum_{\ell,\ell'} \abs{c_{\ell}}\cdot \abs{c_{\ell'} } \cdot \abs{\ip{f_{\ell}}{f_{\ell'}}_{\nu_{t+1}} ~-~ \ip{f_{\ell}}{f_{\ell'}}_{\nu_t} } ~\le~ \tau \cdot R_1^2.
  $$ 
\end{proof}

\subsection{Existential Weak Regularity Decomposition}\label{subsec:existential_weak_reg_split}

Using the abstract weak regularity
lemma,~\cref{lemma:abstract_weak_reg}, together splittable mixing
lemmas of~\cref{subsec:split_mixing_lemma}, we can obtain
(non-constructive) existential weak regularity decompositions for
splittable structures.

\begin{lemma}[Existential Weak Regularity for Splittable Tuples]\label{lemma:existential_weak_reg_split}
  Let $\tuples{k} \subseteq [n]^k$ be a $\tau$-splittable structure.
  Let $g \in \mathbb{R}^{W[1]^k}$ be supported on $\tuples{k}$ with $\norm{g}_{\mu_k} \le 1$.  
  Let $\calF=\calF_{k-1}$ (\cf~\cref{def:abstract_tensor_spaces}) be arbitrary.  
  For every $\error > 0$,
  if $\tau \le O(\error^2/(k-1))$,
  then there exists $h \in \mathbb{R}^{W[1]^k}$
  supported on $O(1/\error^2)$ functions in $\calF$ such that
  $$
  \max_{f \in \calF}~ \ip{g - h}{f} ~\le~ \error \cdot \abs{\tuples{k}},
  $$
  where the inner product is over the counting measure on $W[1]^k$.
\end{lemma}

\begin{proof}
  Apply the weak regularity~\cref{lemma:abstract_weak_reg}, with parameters $\error$ and $\error'$ equal to $\error$,
  collection $\calF$, input function $g$, measure $\mu=\mu_k$ (\ie uniform measure on $\tuples{k}$)
  and a non-explicit correlation oracle based on the existential guarantee. This yields
  $h = \sum_{\ell=1}^p c_\ell \cdot f_\ell \in \calH(1/\error^2,1/\error,\calF)$ where
  $$
  \max_{f \in \calF} \quad \ip{g-h}{f}_{\mu_k} ~\le~ \error.
  $$
  Let $f \in \calF$.
  We claim that $h' = h \cdot \abs{\tuples{k}}/\abs{W[1]}^k$ satisfies the conclusion of the current lemma.
  For this, we bound
  \begin{align*}
    \abs{\abs{\tuples{k}} \ip{g - h}{f}_{\mu_k} ~-~ \ip{g - h'}{f}}  ~\le~ & \abs{\abs{\tuples{k}} \ip{g}{f}_{\mu_k} - \ip{g}{f}} ~+~\\
                                                              & \sum_{\ell=1}^p \abs{c_\ell} \cdot \abs{ \abs{\tuples{k}}\ip{f_\ell}{f}_{\mu_k} ~-~  \frac{\abs{\tuples{k}}}{\abs{W[1]}^k} \ip{f_\ell}{f}}.
  \end{align*}
  The first term in the RHS above is zero since
  $$
  \abs{\tuples{k}} \ip{g}{f}_{\mu_k} ~=~ \sum_{\ess \in \tuples{k}} g(\ess) \cdot f(\ess) ~=~ \ip{g}{f}, 
  $$
  where in the second equality we used that $g$ is supported on $\tuples{k}$.
  Suppose that $f= f_1 \otimes \cdots \otimes f_k$ and $f_\ell = f_{\ell,1} \otimes \cdots \otimes f_{\ell,k}$.
  Set $f_\ell' = (f_1\cdot f_{\ell,1}) \otimes \cdots \otimes (f_k\cdot f_{k,1})$ where $(f_j\cdot f_{j,1})$ is
  the pointwise product of $f_j$ and $f_{j,1}$. Note that
  \begin{align*}
         \ip{f_\ell}{f}_{\mu_k} ~=~ \ExpOp_{\nu_0} \left[f_\ell'\right] \quad \text{ and } \quad
         \frac{\ip{f_\ell}{f}}{\abs{W[1]}^k}  ~=~ \ExpOp_{\nu_{k-1}}\left[ f_\ell'\right],
  \end{align*}
  where we recall that $\mu_k$ is equal to $\nu_0$ and $\mu_1^{\otimes k}$ is equal to $\nu_{k-1}$.
  Moreover, $f_\ell'$ is the tensor product of $k$ functions in $\mathbb{R}^{X[1]}$ of $\ell_{\infty}$-norm at most $1$.
  By the splittable mixing lemma (\cf~\cref{lemma:splittable_mixing_iterated}), we have
  $$
  \abs{\ExpOp_{\nu_0}\left[f_\ell'\right] ~-~ \ExpOp_{\nu_{k-1}}\left[f_\ell'\right]} ~\le~ (k-1) \cdot \tau.
  $$
  Hence, we obtain
  \begin{align*}
    \abs{\abs{\tuples{k}} \ip{g - h}{f}_{\mu_k} ~-~ \ip{g - h'}{f}} ~\le~ &  \sum_{\ell=1}^p \abs{c_\ell} \cdot \abs{\tuples{k}} \cdot \abs{\ExpOp_{\nu_0 }\left[ f_\ell' \right] ~-~ \ExpOp_{\nu_{k-1}}\left[ f_\ell' \right]}\\
    ~\le~ &\sum_{\ell=1}^p \abs{c_\ell} \cdot (k-1) \cdot \tau \cdot \abs{\tuples{k}} ~\le~ \error \cdot \abs{\tuples{k}},
  \end{align*}
  from which the lemma readily follows.
\end{proof}

\subsection{Efficient Weak Regularity Decomposition}\label{subsec:eff_reg}

The goal of this section is to prove an efficient version of weak
regularity that can be computed in near-linear time. We obtain
parameters somewhat comparable to those parameters of the existential
weak regularity in~\cref{lemma:existential_weak_reg_split} above with a mild
polynomial factor loss of $\Theta(1/k^2)$ on the splittability
requirement.

\begin{restatable}{theorem}{EffWeakReg}[Efficient Weak Regularity]\label{theo:eff_weak_reg}
  Let $\tuples{k} \subseteq [n]^k$ be a $\tau$-splittable collection of tuples.
  Let $g \in \mathbb{R}^{W[1]^k}$ be supported on $\tuples{k}$ with $\norm{g}_{\mu_k} \le 1$. 
  Suppose $\calF$ is either $\CUT^{\otimes k}$ or $\SCUT^{\otimes k}$.
  For every $\error > 0$,
  if $\tau \le \error^2/(k^3 \cdot 2^{20})$, then
  we can find $h = \sum_{\ell=1}^{p} c_\ell \cdot f_\ell$ with $p = O(k^2/\error^2)$,
  $c_1,\dots,c_p \in \mathbb{R}$ and functions $f_1,\dots,f_p \in \calF$, such
  that $\norm{h}_{\mu_1^{\otimes k}} \le 2$ and $h$ is a good approximator to $g$ in
  the following sense
  $$
  \max_{f \in \calF}\quad \ip{g - \inparen{\frac{d}{n}}^{k-1}  h}{f} ~\le~ \error \cdot \abs{\tuples{k}},
  $$
  where the inner product is over the counting measure on $W[1]^k$.
  Furthermore, $h$ can be found in $\widetilde{O}(2^{2^{\widetilde{O}(k^2/\error^2)}} \cdot \abs{\tuples{k}})$ time.
\end{restatable}

\noindent \textbf{Warm-up:} We first sketch a simpler naive
algorithmic weak regularity decompoistion for $\CUT^{\otimes k}$ whose
parameters are much worse than the existential parameters
of~\cref{lemma:existential_weak_reg_split}, but it can be computed in
near-linear time. The fast accumulation of errors will explain our
motivation in designing the efficient algorithm
underlying~\cref{theo:eff_weak_reg}. The reader only interested in the
latter is welcome to skip ahead.

\begin{lemma}[Naive Efficient Weak Regularity]\label{lemma:naive_eff_weak_reg}
  Let $W' \subseteq \tuples{k}$ where $\tuples{k}$ is $\tau$-splittable.
  Let $\calF$ be either $\CUT^{\otimes k}$ or $\SCUT^{\otimes k}$.
  For every $\error > 0$, if $\tau \le (O(\error))^{2^k}$,
  then we can find $h$ supported on $(O(1/\error))^{2^k}$
  functions of $\calF$ such that
  $$
  \max_{f \in \calF}\quad \ip{\one_{W'} - h}{f} ~\le~ (k-1) \cdot \error \cdot \abs{\tuples{k}},
  $$
  where the inner product is over the counting measure on $W[1]^k$.
  Furthermore, this can be done in time $\tilde{O}_{\error}(\abs{\tuples{k}})$.
\end{lemma}

\begin{proofsketch}
  In this sketch, our goal is to show the fast accumulation of errors when applying the
  weak regularity decomposition for matrices. For simplicity, we assume that this
  can be done in near-linear time on the number of non-zero entries of the matrix.
  Precise details and much better parameters are given in the proof
  of~\cref{theo:eff_weak_reg}.
  
  Applying the matrix regularity decomposition to $\one_{W'}$, viewed a matrix in $\mathbb{R}^{W[1,k-1] \times W[k]}$ supported on $W[1,k]$, with accuracy parameter $\error_1 > 0$,
  we get in $\tilde{O}_{\error_1}(\abs{W[1,k]})$ time
  $$
  \norm{\one_{W'} ~-~ \frac{d}{n} \sum_{\ell_1=1}^{p_1} c_{\ell_1} \cdot \one_{S_{\ell_1}} \otimes \one_{T_{\ell_1}}}_{\square} ~\le~ \error_1 \cdot \abs{W[1,k]},
  $$
  where $p_1 = O(1/\error_1^2)$ and $\sum_{\ell_1} \abs{c_{\ell_1}} \le O(1/\error_1)$.

  In turn, for each $\one_{S_{\ell_1}}$ viewed a matrix in $\mathbb{R}^{W[1,k-2] \times W[k-1]}$ supported on $W[1,k-1]$, we apply the matrix regularity decomposition with
  accuracy parameter $\error_2 > 0$ getting  in $\tilde{O}_{\error_2}(\abs{W[1,k-1]})$ time
  $$
  \norm{\one_{S_{\ell_1}} ~-~ \frac{d}{n} \sum_{\ell_2=1}^{p_2} c_{\ell_2,\ell_1} \cdot \one_{S_{\ell_2,\ell_1}} \otimes \one_{T_{\ell_2,\ell_1}}}_{\square} ~\le~ \error_2 \cdot \abs{W[1,k-1]},
  $$
  where $p_2 = O(1/\error_2^2)$ and $\sum_{\ell_2} \abs{c_{\ell_2,\ell_1}} \le O(1/\error_2)$.
  Continuing this process inductively with accuracy parameters $\error_3,\dots, \error_{k-1}$, we obtain
  $$
  h \coloneqq \left(\frac{d}{n}\right)^{k-1} \sum_{\ell_1}^{p_1} \cdots \sum_{\ell_{k-1}=1}^{p_{k-1}} c_{\ell_1} \ldots c_{\ell_1,\ldots,\ell_{k-1}} \cdot \one_{T_{\ell_1,\ldots,\ell_{k-1}}} \otimes \cdots \otimes \one_{T_{\ell_1}},
  $$
  in time $\widetilde{O}_{\error_1,\ldots,\error_{k-1}}(\abs{\tuples{k}})$.
  We show that $h$ is close in $k$-tensor cut norm (\cf~\cref{def:higher_order_cut_norm}) to $\one_{W'}$.
  Computing we have
\begin{align*}
  &\norm{\one_{W'} - h}_{\square^{\otimes k}} \le \\
  & \quad \sum_{j=0}^{k-2} \sum_{\ell_1=1}^{p_1} \cdots \sum_{\ell_{j}=1}^{p_j} \abs{c_{\ell_1} \ldots c_{\ell_1,\ldots,\ell_{j}}} \cdot \\
  & \qquad\qquad\qquad \norm{\one_{S_{\ell_1,\ldots,\ell_{j}}} ~-~ \left(\frac{d}{n}\right)^{k-j-1} \sum_{\ell_{j+1}=1}^{p_{j+1}} c_{\ell_1,\ldots,\ell_{j+1}} \cdot \one_{S_{\ell_1,\ldots,\ell_{j+1}}} \otimes \one_{T_{\ell_1,\ldots,\ell_{j+1}}} }_{\square^{\otimes k-j}} \cdot \\
 & \qquad\qquad\qquad \left(\frac{d}{n}\right)^{j} \cdot \norm{\one_{T_{\ell_1,\ldots,\ell_{j}}} \otimes \cdots \otimes \one_{T_{\ell_1}}}_{\square^{\otimes j}}  \\
 & \le \quad \sum_{j=0}^{k-2} \sum_{\ell_1=1}^{p_1} \cdots \sum_{\ell_{j}=1}^{p_j} d^j \cdot \abs{c_{\ell_1} \ldots c_{\ell_1,\ldots,\ell_{j}}} \cdot \\
 & \qquad\qquad \qquad \norm{\one_{S_{\ell_1,\ldots,\ell_{j}}} ~-~ \left(\frac{d}{n}\right)^{k-j-1} \sum_{\ell_{j+1}=1}^p c_{\ell_1,\ldots,\ell_{j+1}} \cdot \one_{S_{\ell_1,\ldots,\ell_{j+1}}} \otimes \one_{T_{\ell_1,\ldots,\ell_{j+1}}} }_{\square}  \\
 & \le \quad \sum_{j=0}^{k-2} \sum_{\ell_1=1}^{p_1} \cdots \sum_{\ell_{j}=1}^{p_j} d^j \cdot \abs{c_{\ell_1} \ldots c_{\ell_1,\ldots,\ell_{j}}} \cdot \error_{j+1} \cdot \abs{W[1,k-j]} \\
 & \le \quad \abs{\tuples{k}} \sum_{j=0}^{k-2} \error_{j+1}  \prod_{\ell=1}^{j} O(1/\error_\ell).
\end{align*}
By setting $\error_j = \Theta(\error^{2^{j}})$, the LHS becomes at
most $(k-1) \cdot \error \cdot \abs{\tuples{k}}$.
\end{proofsketch}

We now proceed to prove our main result in this section,
namely~\cref{theo:eff_weak_reg}. Fist, we establish some extra
notation now. Let $\tuples{k}$ be a $d$-regular collection of
tuples. Most of our derivations which are existential hold for a
generic $\calF_t$ (\cf~\cref{def:abstract_tensor_spaces}).  However,
we only derive near-linear time algorithmic results when $\calF_t$ is
either the CUT functions
$$
\calF_t^{0/1} \coloneqq \left\{\pm \one_{S_1} \otimes \cdots \otimes \one_{S_t} \otimes \one_{T} \mid S_j \subseteq W[1], T \subseteq W[t+1,k] \right\},
$$
or ``signed'' CUT functions
$$
\calF_t^{\pm 1} \coloneqq \left\{\pm \sset_{S_1} \otimes \cdots \otimes \sset_{S_t} \otimes \sset_{T} \mid S_j \subseteq W[1], T \subseteq W[t+1,k] \right\},
$$
where above  we recall that for $S \subseteq [n]$, we have $\sset_{S}(i) = (-1)^{\one_{i \in S}}$ for $i \in [n]$.
Observe that the condition $S_j \subseteq W[1]$ is equivalent to $S_j \subseteq W[i]$ since
$\tuples{k}$ is $d$-regular. 

For quick reference, we collect the notation needed in our algorithmic weak regularity decomposition in the following table.
\begin{center}
  \fbox{\begin{minipage}{35em}
$\calF_t \coloneqq \left\{\pm f_1 \otimes \cdots \otimes f_t \otimes f_{t+1} \mid f_j \subseteq \mathbb{R}^{W[1]} \text{ for } i \le t, f_{t+1} \subseteq \mathbb{R}^{W[t+1,k]}, \norm{f_j}_{\infty} \le 1\right\}$
\vskip 0.3cm      
$\calF_t^{0/1} \coloneqq \left\{\pm \one_{S_1} \otimes \cdots \otimes \one_{S_t} \otimes \one_{T} \mid S_j \subseteq W[1], T \subseteq W[t+1,k] \right\} \subseteq \calF_t$
\vskip 0.3cm
$\calF_t^{\pm 1} \coloneqq \left\{\pm \sset_{S_1} \otimes \cdots \otimes \sset_{S_t} \otimes \sset_{T} \mid S_j \subseteq W[1], T \subseteq W[t+1,k] \right\} \subseteq \calF_t$
\vskip 0.3cm    
$\calH(R_0,R_1,\calF) \coloneqq \left\{ \sum_{\ell=1}^p c_{\ell} \cdot f_{\ell} \mid p \le R_0, \sum \abs{c_{\ell}} \le R_1, f_{\ell} \in \calF \right\}$
\vskip 0.3cm
$\mu_1$ is the uniform distribution on $W[1]$
\vskip 0.3cm
$\mu_{[t+1,k]}$ is the uniform distribution on $W[t+1,k]$
\vskip 0.3cm
$\nu_t \coloneqq \left(\mu_1\right)^{\otimes t} \otimes \left(\mu_{[t+1,k]}\right)$
\end{minipage}}
\end{center}

Our main result of this section, namely, the near-linear time weak
regularity decomposition~\cref{theo:eff_weak_reg}, can be readily
deduced from~\cref{lemma:eff_weak_reg_induction} below.

\begin{lemma}[Efficient Weak Regularity Induction]\label{lemma:eff_weak_reg_induction}
  Let $\tuples{k} \subseteq [n]^k$ be a $\tau$-splittable $d$-regular collection of tuples.
  Let $g \in \calF_0$ and $t \in \set{0,\ldots,k-1}$ with $\norm{g}_{\mu_k} \le 1$.
  For every $\error > 0$,
  if $\tau \le \error^2/(k\cdot 2^{18})$, then
  there exists $h_t \in \calH(O(1/\error^2),2^{8} (1+1/k)^t/\error,\calF_t)$
  with $\norm{h_t}_{\nu_t}^2 \le (1+1/k)^t$ 
  such that
  $$
  \max_{f \in \calF_t} \quad \ip{g - \left(\frac{d}{n}\right)^t h_t}{f}_{\nu_t} \le~ 2 \cdot \left(\frac{d}{n}\right)^t \cdot t \cdot \error.
  $$
  Furthermore, the function $h_t$ can be found in $\widetilde{O}((2t)^{2^{O(1/\error^2)}} \cdot \abs{\tuples{k}})$ time.
\end{lemma}

We restate~\cref{theo:eff_weak_reg} below and then prove it
assuming~\cref{lemma:eff_weak_reg_induction}.

\EffWeakReg*

\begin{proof}
  Set $\calF_t = \calF_t^{0/1}$ if $\calF = \CUT^{\otimes k}$ or set $\calF_t = \calF_t^{\pm 1}$ if $\calF = \SCUT^{\otimes k}$.
  We apply~\cref{lemma:eff_weak_reg_induction} with $t=k-1$, accuracy $\error$ as $\error/(2k)$ and input function $g$. This gives
  $h_{t} = \sum_{\ell=1}^{p} c_\ell' \cdot f_\ell \in \calH(O(k^2/\error^2),O(k/\error),\calF_t)$ 
  such that
  \begin{equation}\label{eq:reg_induction_guarantee}
    \max_{f \in \calF_t} \quad \ip{g - \left(\frac{d}{n}\right)^t h_t}{f}_{\nu_{t}} ~\le~ 2 \cdot \left(\frac{d}{n}\right)^t \cdot t \cdot \error.
  \end{equation}
  Note that $\nu_t =\nu_{k-1} = \mu_1^{\otimes k}$ is the uniform measure on $W[1]^k$. Since $\tuples{k}$ is $d$-regular, $\abs{\tuples{k}} = \abs{W[1]}^k \cdot  (d/n)^{k-1}$.
  Set $h =  \cdot h_t$. Then the guarantee in~\cref{eq:reg_induction_guarantee} becomes
  $$
  \max_{f \in \calF}\quad \ip{g - \inparen{\frac{d}{n}}^{k-1} h}{f} ~\le~ \error \cdot \abs{\tuples{k}} ,
  $$
  where the inner product is under the counting measure. By~\cref{lemma:eff_weak_reg_induction},
  we have $\norm{h_t}_{\nu_t}^2 \le (1+1/k)^t \le e$, so $\norm{h_t}_{\nu_t} \le 2$. Then
  $\norm{h}_{\mu_1^{\otimes k}} \le 2$.
  The running time follows from \cref{lemma:eff_weak_reg_induction} completing the proof.
\end{proof}


We now prove~\cref{lemma:eff_weak_reg_induction} above assuming the
following algorithmic result which we prove later.

\begin{restatable}{lemma}{AlgoWeakRegStep}[Algorithmic Weak Regularity Step]\label{lemma:algo_weak_reg_step}
  Let $\error > 0$ and $t \in \set{0,\ldots,k-2}$.
  Let $h_t \in \calH(O(B/\error^2),O(B/\error),\calF_t)$ with $\norm{h_t}^2_{\nu_t} \le B$.
  Then there exists $h_{t+1} \in \calH(O(B/\error^2), 2^8 B/\error,\calF_{t+1})$
  with $\norm{h_{t+1}}_{\nu_t}^2 \le~ B$
  such that
  $$
  \max_{f \in \calF_{t+1}} \quad \ip{h_t - h_{t+1}}{f}_{\nu_t} ~\le~ \error.
  $$
  Furthermore, each $h_{t+1}$ can be found in time $\widetilde{O}((2t)^{2^{O(1/\error^2)}} \cdot \abs{\tuples{k}})$.
\end{restatable}

\begin{proof}[Proof of~\cref{lemma:eff_weak_reg_induction}]
  We will prove the lemma with the following simple equivalent conclusion
  \begin{align*}
   \ip{g - \left(\frac{d}{n}\right)^t h_t}{f}_{\nu_t} ~\le~ 2\cdot \left(\frac{d}{n}\right)^t \cdot t \cdot \error \quad\qquad \iff \quad\qquad \ip{ \left(\frac{n}{d}\right)^t g -  h_t}{f}_{\nu_t} ~\le~  2\cdot t \cdot \error,
  \end{align*}
  which we will prove holds for every $f \in \calF_t$.
  The base case $t=0$ follows immediately by setting $h_0 = g$.
  Let $t \in \set{0,\ldots,k-2}$. Since $h_t \in \calH(O(1/\error^2),2^8(1+1/k)^t/\error,\calF_t)$,
  invoking~\cref{lemma:algo_weak_reg_step} with accuracy parameter $\error$ and input function
  $h_t$, we obtain $h_{t+1} \in \calH(O(1/\error^2),2^8(1+1/k)^{t+1}/\error,\calF_{t+1})$ satisfying
  \begin{equation}\label{eq:algo_step_guarantee}
    \max_{f \in \calF_{t+1}} ~\ip{h_t - h_{t+1}}{f}_{\nu_t} ~\le~ \error.
  \end{equation}
  Let $f \in \calF_{t+1}$.  We will show that $h_{t+1}$ satisfies the conclusion of the lemma.
  Expanding we have
  \begin{align*}
    \ip{ \left(\frac{n}{d}\right)^{t+1} g -  h_{t+1}}{f}_{\nu_{t+1}} ~=~ & \underbrace{\ip{ \left(\frac{n}{d}\right)^t g -  h_t}{f}_{\nu_t}}_{(i)}
    ~+~ \left( \frac{n}{d}\right)^t \cdot \underbrace{\left(\frac{n}{d} \ip{g}{f}_{\nu_{t+1}} ~-~ \ip{g}{f}_{\nu_t} \right)}_{(ii)}\\
    &+~ \underbrace{\ip{h_t - h_{t+1}}{f}_{\nu_t}}_{(iii)}
    ~+~ \underbrace{\ip{h_{t+1}}{f}_{\nu_t} ~-~ \ip{h_{t+1}}{f}_{\nu_{t+1}}}_{(iv)}.
  \end{align*}
  We will bound each of the terms in RHS above.
  
  \noindent \textbf{Term (i):} Suppose $f = f_{1} \otimes \cdots \otimes f_{t+1} \otimes f_{t+2} \in \calF_{t+1}$.
  Let $f' = f_{1} \otimes \cdots \otimes f_t \otimes f'_{t+1}$, where $f_{t+1}' = (f_{t+1} \otimes f_{t+2})\vert_{W[t+2,k]}$,
  so that $f' \in \calF_t$. Using the induction hypothesis, we have
  $$
  \ip{ \left(\frac{n}{d}\right)^t g -  h_t}{f}_{\nu_t} ~=~ \ip{ \left(\frac{n}{d}\right)^t g -  h_t}{f'}_{\nu_t} ~\le~ 2 \cdot t \cdot \error.
  $$

  \noindent \textbf{Term (ii):} Since $g \in \calF_0$, it is supported on $\tuples{k}$ and so we have
  \begin{align*}
    \ip{g}{f}_{\nu_t} ~=~& \frac{1}{\abs{W[1]}^t \abs{W[t+1,k]}} \sum_{\ess \in \tuples{k}} g(\ess) \cdot f(\ess)\\
                    ~=~& \frac{n}{d} \cdot \frac{1}{\abs{W[1]}^{t+1} \abs{W[t+2,k]}} \sum_{\ess \in \tuples{k}} g(\ess) \cdot f(\ess) ~=~ \frac{n}{d} \cdot \ip{g}{f}_{\nu_{t+1}}.
  \end{align*}
  where the second equality follows from $\abs{W[t+1,k]} = d\cdot \abs{W[t+2,k]}$
  by the $d$-regular assumption.

  \noindent \textbf{Term (iii):} By~\cref{eq:algo_step_guarantee}, we have $\ip{h_t - h_{t+1}}{f}_{\nu_t} ~\le~ \error$.

  \noindent \textbf{Term (iv):} For notional convenience, set $R_1= 2^8(1+1/k)^{t+1}/\error$.
  Since $h_{t+1} \in \calH(\infty,R_1,\calF_{t+1})$ and the splittability parameter $\tau$ satisfies $\tau \le \error^2/(k\cdot 2^{18})$,
  from \cref{claim:indistinguishability_under_measure_extension} we obtain
  $$
  \ip{h_{t+1}}{f}_{\nu_t} - \ip{h_{t+1}}{f}_{\nu_{t+1}} ~\le~ \tau \cdot R_1 ~\le~ \error.
  $$
  Putting everything together yields
  \begin{align*}
    \ip{ \left(\frac{n}{d}\right)^{t+1} g -  h_t}{f}_{\nu_{t+1}} ~\le~ & \underbrace{2 \cdot t \cdot \error}_{(i)}
    + \left( \frac{n}{d}\right)^t \cdot \underbrace{0}_{(ii)}
    + \underbrace{\error}_{(iii)}
    + \underbrace{\error}_{(iv)}
    ~\le~ 2 \cdot (t+1) \cdot \error,
  \end{align*}
  concluding the claimed inequality.

  Now we use the bound $\norm{h_{t+1}}_{\nu_t}^2 \le
  \norm{h_{t}}_{\nu_t}^2$ from~\cref{lemma:algo_weak_reg_step} together with
  the splittability assumption $\tau \le \error^2/(k\cdot 2^{18})$
  to bound the norm $\norm{h_{t+1}}_{\nu_{t+1}}^2$ under the new
  measure $\nu_{t+1}$. Under these assumptions and
  using~\cref{claim:norm_indistinguishability_under_measure_extension}
  we get
  \begin{align*}
    \abs{ \norm{h_{t+1}}_{\nu_{t+1}}^2 - \norm{h_{t+1}}_{\nu_t}^2 } ~\le~ \tau \cdot R_1^2 ~\le~ & \frac{\error^2}{k \cdot 2^{18}} \cdot \frac{2^{16}(1+1/k)^{2(t+1)}}{\error^2}\\
                                                                                            ~\le~ & \frac{(1+1/k)^{t}}{k}.
  \end{align*}
  where we used the bounds on $\tau$, $R_1$ and $(1+1/k)^{(t+2)} \le 4$ for $0 \le t \le k-2$.
  From the previous inequality and the induction hypothesis $\norm{h_{t}}_{\nu_t}^2 \le (1+1/k)^{t}$,
  we finally get $\norm{h_{t+1}}_{\nu_{t+1}}^2 \le (1+1/k)^{t+1}$ as desired.
\end{proof}

We now show a near-linear time weak regularity decomposition for
special functions of the form $h_t \in
\calH(O(1/\error^2),O(1/\error),\calF_t)$ that admit a tensor product
structure. The goal is to design a correlation oracle that exploits
the special tensor product structure of the function $(h_t - h_{t+1}^{(\ell)})$,
where $h_{t+1}^{(\ell)}$ is the $\ell$th approximator of $h_t$ in the
abstract weak regularity algorithm (\cf~\cref{algo:abstract_weak_reg}).

\AlgoWeakRegStep*

Our correlation oracle for higher-order tensors will make calls to a
correlation oracle for matrices~\cref{theo:alon_naor_oracle} (\ie
$2$-tensors) stated below. This matrix oracle is presented
in~\cref{subsec:matrix_corr_oracle} and it follows from a simple
combination of a matrix cut norm approximation algorithm by Alon and
Naor~\cite{AN04} with known fast SDP solvers for sparse matrices such
as those by Lee and Padmanabhan~\cite{LeeP20} and Arora and
Kale~\cite{AK07}.

\begin{restatable*}{theorem}{TheoAlonNaorOracle}[Alon--Naor Correlation Oracle]\label{theo:alon_naor_oracle}
  Let $\calF$ be either $\CUT^{\otimes 2}$ or $\SCUT^{\otimes 2}$ and
  $\mu$ be the uniform measure supported on at most $m$ elements of $[n'] \times [n']$.
  There exists an algorithmic $(\delta, \AN \cdot \delta)$-correlation oracle $\calO_{\mu,B}$
  running in time $\calT_{\calO_{\mu,B}} = \tilde{O}\left(\poly(B/\delta) \cdot (m+n')\right)$,
  where $\AN \ge 1/2^4$ is an approximation ratio constant. 
\end{restatable*}

\begin{proof}
  We will apply the abstract weak regularity lemma, \cf \cref{lemma:abstract_weak_reg},
  with $\calF = \calF_{t+1}$, $\error$, $\error' = \error/2^8$ and $\mu = \nu_t$.
  This will result in a function from $\calH(O(B/\error^2), 2^8 B/\error,\calF_{t+1})$.
  
  \noindent \textbf{Correlation oracle task:}
  To make this application take near-linear time, we need to specify a correlation oracle $\calO_{\nu_t}=\calO_{\nu_t,O(1)}$
  and now we take advantage of the special tensor structure in our setting. We want an oracle that given
  \begin{align*}
    h_t =& \sum_{\ell=1}^p c_{\ell} \cdot g_{\ell},\quad g_{\ell} \in \calF_t,\qquad g_{\ell} = g_{\ell,1} \otimes \cdots \otimes g_{\ell,t} \otimes \underbrace{g_{\ell,t+1}}_{\in \R^{W[t+1,k]}} \text{ and}\\
    h_{t+1} =& \sum_{\ell=1}^p c_{\ell}' \cdot g_{\ell}',\quad  g_{\ell}' \in \calF_{t+1},\quad g_{\ell}' = g_{\ell,1}' \otimes \cdots \otimes g_{\ell,t}' \otimes \underbrace{g_{\ell,t+1}'}_{\in \R^{W[1]}} \otimes \underbrace{g_{\ell,t+2}'}_{\in \R^{W[t+2,k]}},
  \end{align*}  
  if there exists
  $$
  f = f_1 \otimes \cdots \otimes f_t \otimes \underbrace{f_{t+1}}_{\in \R^{W[1]}} \otimes \underbrace{f_{t+2}}_{\in \R^{W[t+2,k]}} \in \calF_{t+1}
  $$
  satisfying 
  $$
  \ip{h_t - h_{t+1}}{f}_{\nu_t} ~\ge~ \error,
  $$
  for some $f \in \calF_{t+1}$, finds $f' \in \calF_{t+1}$ in near-linear time such that
  $$
  \ip{h_t - h_{t+1}}{f'}_{\nu_t} ~\ge~ \error' ~=~ \frac{\error}{2^8}.
  $$
  Here, $h_{t+1}$ is the current approximator of $h_t$ 
  in the abstract weak regularity algorithm and, by~\cref{lemma:abstract_weak_reg},
  $h_{t+1}  \in \calH(O(1/\error^2),2^8(1+1/k)^{t+1}/\error,\calF_{t+1})$.
  Expanding $\ip{h_t - h_{t+1}}{f}_{\nu_t}$ we get
  \begin{align*}
    \ip{h_t - h_{t+1}}{f}_{\nu_t} ~=~ &\sum_{\ell=1}^p c_{\ell} \underbrace{\prod_{j=1}^{t} \ip{g_{\ell,j}}{f_j}_{\mu_1}}_{\gamma_{\ell}} \cdot \ip{g_{\ell,t+1}}{f_{t+1} \otimes f_{t+2}}_{\mu_{[t+1,k]}} ~-~\\
&    \sum_{\ell=1}^p c_{\ell}' \underbrace{ \prod_{j=1}^{t} \ip{g_{\ell,j}'}{f_j}_{\mu_1}}_{\gamma'_{\ell}} \cdot \ip{g_{\ell,t+1}' \otimes g_{\ell,t+2}'}{f_{t+1} \otimes f_{t+2}}_{\mu_{[t+1,k]}},
  \end{align*}
  where we define $\gamma_{\ell} \coloneqq \prod_{j=1}^t \ip{g_{\ell,j}}{f_j}_{\mu_1}$
  and $\gamma'_{\ell} \coloneqq \prod_{j=1}^t \ip{g_{\ell,j}'}{f_j}_{\mu_1}$ for $\ell \in [p]$, $j \in [t]$.
  Suppose $g_{\ell,j} = f_{S_{\ell,j}}$ and $g_{\ell,j}' = f_{S_{\ell,j}'}$ for $\ell \in [p]$,
  $j \in [t]$, where $f_{S_{\ell,j}},f_{S_{\ell,j}'}$ are either $\one_{S_{\ell,j}},\one_{S_{\ell,j}'}$
  or $\sset_{S_{\ell,j}},\sset_{S_{\ell,j}'}$ depending on $\calF_t$ being $\calF_t^{0/1}$ or $\calF_t^{\pm 1}$,
  respectively.
  
  \noindent \textbf{Sigma-algebra brute force:} Now for each $j \in [t]$, we form the $\sigma$-algebra $\Sigma_j$ generated
  by $\set{S_{\ell,j},S_{\ell,j}'}_{\ell \in [p]}$ which can be done in 
  $2^{p} \cdot \widetilde{O}(\abs{W[1]})$ time by~\cref{remark:factor_sigma_algebra_comp}
  and yields
  at most $2^{p}$ atoms. Hence, the generation of all these $\sigma$-algebras takes
  at most $t \cdot 2^{p} \cdot \widetilde{O}(\abs{W[1]})$ time. Suppose $f_j = f_{S_j}$ for
  some $S_j \subseteq W[1]$. Let $\eta > 0$ be an approximation parameter to be specified
  shortly.  For each atom $\sigma_{j'} \in \Sigma_j$, we enumerate over all possible
  values for the ratio $\abs{\sigma_{j'} \cap S_j}/\abs{\sigma_{j'}}$ up to
  accuracy $\eta$. More precisely, if $\abs{\sigma_{j'}} \ge 1/\eta$,
  we consider the values
  $$
  0,1\cdot \eta,2\cdot \eta,\ldots,\lfloor 1/\eta \rfloor \cdot \eta,
  $$
  and we consider $0,1/\abs{\sigma_{j'}},2/\abs{\sigma_{j'}},\ldots,\abs{\sigma_{j'}}/\abs{\sigma_{j'}}$
  otherwise. Let $\abs{\Sigma_j}$ denote the number of atoms in $\Sigma_j$.
  This enumeration results in $\prod_{j=1}^t (1/\eta)^{\abs{\Sigma_j}}$ configurations which
  allows us to approximate any realizable values for $\ip{g_{\ell,j}}{f_j}_{\mu_1}$ 
  within additive error at most $4\cdot \eta$ since either
  $$
  \ip{g_{\ell,j}}{f_j}_{\mu_1} ~=~ \E_{\mu_1}\left[ \one_{S_{\ell,j}} \cdot  \one_{S_j} \right] = \frac{\abs{S_{\ell,j} \cap S_j}}{\abs{W[1]}} = \frac{1}{\abs{W[1]}} \sum_{\sigma_{j'} \subseteq S_{\ell,j}} \abs{\sigma_{j'} \cap S_j} \quad \text{or}
  $$  
  \begin{align*}
    \ip{g_{\ell,j}}{f_j}_{\mu_1} ~=~ \E_{\mu_1}\left[ \chi_{S_{\ell,j}} \cdot  \chi_{S_j} \right] &~=~ \frac{\abs{W[1]}  - 2(\abs{S_{\ell,j}} + \abs{S_j} - 2\abs{S_{\ell,j} \cap S_j})}{\abs{W[1]}}\\
                                                                                        &~=~ \frac{\abs{W[1]}  - 2(\abs{S_{\ell,j}} + \sum_{\sigma_{j'}} \abs{\sigma_{j'} \cap S_j} - 2 \sum_{\sigma_{j'} \subseteq S_{\ell,j}} \abs{\sigma_{j'} \cap S_j})}{\abs{W[1]}},
  \end{align*}
  according to $\calF_{t+1}$.
  We can approximate $\ip{g_{\ell,j}'}{f_j}_{\mu_1}$ similarly.
  In turn, we can approximate each of the realizable values in $\set{\gamma_{\ell},\gamma'_{\ell}}_{\ell \in [p]}$ within
  additive error $4\cdot t \cdot \eta$ by some configuration of fractional value assignment to the atoms of each $\sigma$-algebra.  
  
  \noindent \textbf{Invoking the matrix correlation oracle:}
  Let $\matr A \coloneqq \sum_{\ell} \left(c_{\ell} \cdot \gamma_{\ell} \cdot g_{\ell,t+1} + c_{\ell}' \cdot \gamma'_{\ell} \cdot g_{\ell,t+1}' \otimes g_{\ell,t+2}'\right)$.
  We conveniently view $\matr A$ as a \emph{sparse} matrix of dimension $\abs{W[t+1]}\times \abs{W[t+2,k]}$ with at most $\abs{W[t+1,k]}$ non-zeros entries.
  Define  $\varphi_{\matr A}(f_{t+1},f_{t+2}) \coloneqq \ip{\matr A}{f_{t+1} \otimes f_{t+2}}_{\mu_{[t+1,k]}}$. Define
  \begin{equation}\label{eq:cut_norm_like}
    \OPT(\matr A) \coloneqq \max_{f_{t+1},f_{t+2}}\quad \varphi_{\matr A}(f_{t+1},f_{t+2}),
  \end{equation}
  where $f_{t+1},f_{t+2}$ range over valid $f_{S_{t+1}},f_{S_{t+2}}$ (again according to kind of $\calF_{t+1}$ we have).
  In the computation of $\OPT(\matr A)$,
  we have incurred so far an additive error of at most
  $$
  4 \cdot t \cdot \eta \cdot \sum_{\ell} (\abs{c_{\ell}} + \abs{c_{\ell}'}).
  $$
  Let $\widetilde{\matr A}$ be obtained from $\matr A$ by zeroing out all entries of absolute value smaller than $\error/8$.
  Note that $\OPT(\widetilde{\matr A}) \ge \OPT(\matr A) - \error/8$ and the absolute value of the entries of $\widetilde{\matr A}$ lie
  $[\error/8, O(1/\error)]$. For each entry of $\matr A$, we compute a rational approximation $\pm P/Q$ where $Q = \Theta(1/\error)$
  and $P \in [1,O(1/\error)]$ obtaining $\widetilde{\matr A}'$ such that
  $$
  \OPT(\widetilde{\matr A}') ~\ge~ \OPT(\widetilde{\matr A}) - \error/8 ~\ge~ \OPT(\widetilde{\matr A}) ~\ge~ \OPT(\matr A) - \error/4.
  $$
  Using~\cref{theo:alon_naor_oracle} with accuracy parameter $\error/4$ and input matrix $\widetilde{\matr A}'$, we obtain
  in $\calT_{\matr A} \coloneqq \widetilde{O}(\poly(1/\error) \cdot \abs{W[t+1,k]})$ time, with an extra additive error of $\error/4$
  and a multiplicative guarantee of $\AN$, a $2$-tensor $\tilde{f}_{t+1} \otimes \tilde{f}_{t+2}$ satisfying 
  $$
  \varphi_{\widetilde{\matr A}}(\tilde{f}_{t+1},\tilde{f}_{t+2}) ~\ge~ \AN \cdot \left(\OPT(\matr A) ~-~ 2 \cdot \frac{\error}{4}  ~-~ 4 \cdot t \cdot \eta \cdot \sum_{\ell} (\abs{c_{\ell}} + \abs{c_{\ell}'})\right).
  $$
  Since $h_t \in \calH(O(1/\error^2),2^8\cdot (1+1/k)^t/\error,\calF_t)$ and $h_{t+1} \in \calH(O(1/\error^2),2^8\cdot (1+1/k)^{t+1}/\error,\calF_{t+1})$,
  we have $\sum_{\ell} (\abs{c_{\ell}} + \abs{c_{\ell}'}) \le 2^{10}/\error$ and
  $p = O(1/\error^2)$.
  By choosing $\eta \le O(\error^2/t)$ appropriately, we can bound 
  $$
  4 \cdot t \cdot \eta \cdot \sum_{\ell} (\abs{c_{\ell}} + \abs{c_{\ell}'}) ~\le~ 4 \cdot t \cdot \frac{2^{10}}{\error} \cdot \eta ~\le~ \frac{\error}{4}.
  $$
  Hence, $\varphi_{\widetilde{\matr A}}(\tilde{f}_{t+1},\tilde{f}_{t+2}) \ge \AN \cdot \error/4$ since we are under the assumption that
  $\OPT(\matr A) \ge \error$.

  \noindent \textbf{Running Time:} First, observe that with our choices of parameters the total number of configurations $m_{\textup{config}}$
  is at most
  $$
  m_{\textup{config}} ~\le~ \prod_{j=1}^t (1/\eta)^{\abs{\Sigma_j}} ~\le~ \left( \frac{t}{\error^2} \right)^{2^{p}} ~\le~ (2t)^{2^{O(1/\error^2)}},
  $$
  so that the correlation oracle $\calO_{\nu_t}$ takes time at most
  $$
  m_{\textup{config}} \cdot \calT_{\matr A} ~\le~ (2t)^{2^{O(1/\error^2)}} \cdot \widetilde{O}(\poly(1/\error) \cdot \abs{W[t+1,k]}) ~=~ \widetilde{O}((2t)^{2^{O(1/\error^2)}} \cdot \abs{W[t+1,k]}).
  $$
  Using the running time of the oracle $\calO_{\nu_t}$, the total running time of the weak regularity
  decomposition follows from~\cref{lemma:abstract_weak_reg} which concludes the proof.
\end{proof}

\subsection{Near-linear Time Matrix Correlation Oracles}\label{subsec:matrix_corr_oracle}

The main result of this section,~\cref{theo:alon_naor_oracle} below,
is a near-linear time correlation oracle for $\CUT^{\otimes 2}$ and
$\SCUT^{\otimes 2}$. We combine the constant factor approximation
algorithms of Alon--Naor~\cite{AN04} for $\norm{\matr A}_{\infty \to
  1}$ and $\norm{\matr A}_{\square}$ based on semi-definite
programming (SDP) with the faster SDP solvers for sparse matrices such
as those by Lee and Padmanabhan \cite{LeeP20} and by Arora and
Kale~\cite{AK07}. We point out that these SDP solvers provide additive
approximation guarantees which are sufficient for approximating
several CSPs, \eg MaxCut, but they do not seem to provide non-trivial
multiplicative approximation guarantees for $\norm{\matr A}_{\infty
  \to 1}$ or $\norm{\matr A}_{\square}$ in general. Since in our
applications of computing regularity decomposition we are only
interested in additive approximations, those solvers provide
non-trivial sufficient approximation guarantees for $\norm{\matr
  A}_{\infty \to 1}$ or $\norm{\matr A}_{\square}$ in our settings.

\TheoAlonNaorOracle

\cref{theo:alon_naor_oracle} is a simple consequence of the following
theorem.

\begin{theorem}\label{theorem:cut_like_alg}
  Let $\matr A \in \mathbb{R}^{n \times n}$ be a matrix of integers with at most $m$ non-zero entries.
  Let $\delta \in (0,2^{-5}]$ be an accuracy parameter. Suppose that
  $$
  \OPT \coloneqq \max_{x_i,y_i \in \set{\pm 1}} \sum_{i,j=1}^n \matr A_{i,j} x_i y_j ~\ge~ \delta \cdot m.
  $$
  Then, with high probability,\ie $o_n(1)$, we we can find in $\tilde{O}\left(\poly(\norm{\matr A}_{\infty}/\delta) \cdot (m+n)\right)$ time
  vectors $\tilde{x},\tilde{y} \in \set{\pm 1}^n$ such that
  $$
  \sum_{i,j=1}^n \matr A_{i,j} \tilde{x}_i \tilde{y}_j ~\ge~ \frac{1}{4} \cdot \OPT,
  $$
  and find sets $\tilde{S},\tilde{T} \subseteq [n]$ such that
  $$
  \abs{\sum_{i \in \tilde{S}, j\in \tilde{T}} \matr A_{i,j}} ~\ge~ \frac{1}{2^4} \cdot \norm{\matr A}_{\square},
  $$
  where $\norm{\matr A}_{\square}$ is the cut norm of $\matr A$.
\end{theorem}

The proof of the preceding theorem will rely on the following result
which encapsulates the known sparse SDP solvers~\cite{AK07,LeeP20}.
For concreteness, we will rely on~\cite{LeeP20} although the guarantee
from~\cite{AK07} also suffice for us.

\begin{restatable}{lemma}{SparseSDPSolverWrapper}[Sparse SDP Solver Wrapper based on~\cite{LeeP20} and partially on \cite{AK07}]\label{lemma:sparse_sdp_solver_gram_matrix} 
  Let $\matr C \in \mathbb{R}^{n\times n}$ be a matrix with at most $m$ non-zero entries.
  For every accuracy $\gamma > 0$, with high probability we can find in time $\widetilde{O}((m+n)/\poly(\gamma))$
  vectors $u_1,\ldots,u_n \in \mathbb{R}^n$ in the unit ball (\ie $\norm{u_i} \le 1$) such that that the
  matrix $\widetilde{\matr X}_{i,j} \coloneqq \ip{u_i}{u_j}$ satisfies
  $$
  \Tr\left(\matr C \cdot \widetilde{ \matr X}\right) ~\ge~ \max_{\matr X \succeq 0, X_{i,i} \le 1} \Tr\left(\matr C \cdot \matr X\right) - \gamma \sum_{i,j} \abs{\matr C_{i,j}}.
  $$
\end{restatable}

\begin{proof}[Proof of~\cref{theorem:cut_like_alg}]
  We now implement the strategy mentioned above of combing
  the approximation algorithms of Alon--Naor~\cite{AN04} with the near-linear time sparse SDP solvers.
  We still need to argue that this 
  indeed leads to the claimed approximation guarantees while being computable in
  near-linear time overall. We point out that Alon--Naor actually give a constant
  factor SDP based approximation algorithm for $\norm{\matr A}_{\infty \to 1}$ 
  from which a constant factor approximation algorithm  for
  $\norm{\matr A}_{\square}$ can be readily deduced from
  in near-linear time incurring an extra $1/4$ factor approximation loss\footnote{In Section 5.4 of
    Alon--Naor~\cite{AN04}, there is a transformation avoiding any loss in the approximation ratio.   
    Since constant factors are not asymptotically important for us, we rely on the simpler transformation
    which loses a factor of $1/4$. It simply consists in choosing  $\widetilde{S} \in \set{\set{i\mid \widetilde{x}_i=1} ,\set{i \mid \widetilde{x}_i=-1 }}$
    and $\widetilde{T} \in \set{\set{j\mid \widetilde{y}_j=1} ,\set{j \mid \widetilde{y}_j=-1 }}$ maximizing $\one_{\widetilde{S}}^t \matr A \one_{\widetilde{T}}$, which can
    be done in near-linear time given as input $\widetilde{x},\widetilde{y}$.}.
  Using the matrix $\matr A$, we set
  $$
  \matr C \coloneqq \frac{1}{2} \begin{pmatrix}
                                  0 & \matr A\\
                                  \matr A^{\dag} & 0     
                                \end{pmatrix}.
  $$
  The SDP relaxation of Alon--Naor for $\norm{\matr A}_{\infty \to 1}$ becomes 
  \begin{align*}
                            \max &\quad \Tr(\matr C \cdot \matr X) && \eqqcolon \SDP^* \\
                    \text{s.t.}  &\quad \matr X_{i,i} ~\le~ 1 && \forall i \in [2n]\\
                                 &\quad \matr X ~\succeq~ 0,
  \end{align*}
  except for the constraints $X_{i,i} \le 1$ which they instead take
  to be $X_{i,i} = 1$. This technical difference will play a (small)
  role in the rounding of this SDP since Alon--Naor analysis relies on
  Gram vectors of $\matr X$ being on the unit sphere. Moreover, 
  we will be solving this SDP within only a weak additive approximation
  guarantee\footnote{This may not be sufficient to obtain $X_{i,i} \approx 1$
    by an extremality argument}. Although these technical differences need to
  be handled, this will be simple to do.
  
  Applying the solver of~\cref{lemma:sparse_sdp_solver_gram_matrix} with
  accuracy parameter $\gamma = \delta^2/\norm{\matr A}_{\infty}$ to the above SDP, we obtain in $\widetilde{O}(\poly(\norm{\matr A}_{\infty}/\delta)\cdot (m+n))$
  time vectors $u_1,\ldots,u_{2n} \in \mathbb{R}^{2n}$ in the unit ball so that the
  matrix $\widetilde{\matr X}_{i,j} \coloneqq \ip{u_i}{u_j}$ satisfy
  $$
  \Tr\left(\matr C \cdot \widetilde{ \matr X}\right) ~\ge~ \max_{\matr X \succeq 0, X_{i,i} \le 1} \Tr\left(\matr C \cdot \matr X\right) ~-~ \delta^2 \cdot m.
  $$
  By assumption, we have $\SDP^* \coloneqq \max_{\matr X \succeq 0, X_{i,i} \le 1} \Tr\left(\matr C \cdot \matr X\right) \ge \OPT \ge \delta \cdot m$, in which case
  the above guarantee becomes
  $$
  \Tr\left(\matr C \cdot \widetilde{ \matr X}\right) ~\ge~ (1-\delta) \cdot \SDP^*.
  $$  

  To obtain diagonal entries equal to $1$ in our SDP solution we simply consider the
  new SDP solution $\widetilde{ \matr X}' = \widetilde{ \matr X} + \Lambda$,
  where $\Lambda$ is the diagonal matrix defined as $\Lambda_{i,i} \coloneqq 1-\widetilde{ \matr X}_{i,i}$.
  Gram vectors $u_1',\ldots,u_{2n}'$ of $\widetilde{ \matr X}'$ can be obtained in near-linear
  time from $u_1,\dots,u_{2n}$ and $\Lambda$ by setting
  $$
  u_i' \coloneqq u_i \oplus \sqrt{\Lambda_{i,i}} \cdot e_{i} ~\in~ \mathbb{R}^{2m} \oplus \mathbb{R}^{2m},
  $$
  where $e_i \in \mathbb{R}^{2m}$ has a one at the $i$th position and zero everywhere else.
  Observe that for our particular $\matr C$, we have
  $$
  \Tr\left(\matr C \cdot \widetilde{ \matr X}' \right) ~=~ \Tr\left(\matr C \cdot \widetilde{ \matr X}\right).
  $$

  We now proceed to round $\widetilde{ \matr X}'$ according to the rounding scheme of
  Alon--Naor \cite{AN04} (\cf Section 5.1) which was chosen because it is simple enough to easily afford a
  near-linear time computation while providing a $\approx 0.27 \ge 1/4$ approximation guarantee~\footnote{Alon--Naor \cite{AN04}
    have a more sophisticated rounding scheme that achieves $0.56 \ge 1/2$ approximation. In our applications, it is important
    to have a constant factor approximation, but the distinction between $1/2$ and the weaker $1/4$ factor approximation guarantee
    is not asymptotically relevant.}
  This rounding consists in sampling a Gaussian vector $g \sim N(0,\matr I_d)$ and setting
  $\widetilde{x}_i \coloneqq \sgn{\ip{u_i'}{g}}$ and $\widetilde{y}_{i+n} \coloneqq \sgn{\ip{u_{i+n}'}{g}}$
  for $i \in [n]$. To analyze the approximation guarantee, the following identity is used.
  \begin{fact}[Alon--Naor~\cite{AN04}, \cf Eq.\ 5]\label{fact:alon_naor_identity}
  Let $u,w \in \mathbb{R}^d$ be unit vectors in $\ell_2$-norm.
  Then
  {\footnotesize
  \begin{align*}
    \frac{\pi}{2} \cdot \E \left[\sgn{\ip{u}{g}}\sgn{\ip{w}{g}}\right] = \ip{u}{w} + \E\left[\left(\ip{u}{g} - \sqrt{\frac{\pi}{2}} \sgn{\ip{u}{g}} \right) \left(\ip{w}{g} - \sqrt{\frac{\pi}{2}} \sgn{\ip{w}{g}} \right)\right],
  \end{align*}
  }
  where the expectations are taken with respect to a random Gaussian vector $g \sim N(0,\matr I_d)$.
  \end{fact}
  Using~\cref{fact:alon_naor_identity}, the expected value of the rounding, \ie
  $$
  \E \left[ \sum_{i,j} \matr A_{i,j} \sgn{\ip{u_{i}'}{g}}\sgn{\ip{u_{j+n}'}{g}}\right],
  $$
  becomes
  {\footnotesize
  \begin{align*}
    \frac{2}{\pi} \cdot \sum_{i,j} \matr A_{i,j} \ip{u_i'}{u_{j+n}'} + \frac{2}{\pi} \cdot \sum_{i,j} \matr A_{i,j} \E\left[\left(\ip{u_i'}{g} - \sqrt{\frac{\pi}{2}} \sgn{\ip{u_i'}{g}} \right) \left(\ip{u_{j+n}'}{g} - \sqrt{\frac{\pi}{2}} \sgn{\ip{u_{j+n}'}{g}} \right)\right],
  \end{align*}}
  As in Alon--Naor~\cite{AN04}, we will use the fact that $\ip{u_i'}{g} - \sqrt{\frac{\pi}{2}} \sgn{\ip{u_i'}{g}}$ and $\ip{u_{j+n}'}{g} - \sqrt{\frac{\pi}{2}} \sgn{\ip{u_{j+n}'}{g}}$
  are themselves vectors on a Hilbert space with norm squared $\pi/2-1$. Then, in our setting we obtain 
  \begin{align*}
    \E \left[ \sum_{i,j} \matr A_{i,j} \sgn{\ip{u_{i}'}{g}}\sgn{\ip{u_{j+n}'}{g}}\right] ~\ge~ & \frac{2}{\pi} (1-\delta) \cdot \SDP^* - \left(1-\frac{2}{\pi}\right) \cdot \SDP^*\\
                                                                                     ~\ge~ & \frac{2}{\pi} \left(2 - \frac{\pi}{2} -\delta\right) \cdot \SDP^*\\
                                                                                     ~\ge~ & \left(\frac{1}{4} + \Omega(1)\right) \cdot \SDP^* && \text{(Since $\delta \le 2^{-5}$)}\\
                                                                                     ~\ge~ & \left(\frac{1}{4} + \Omega(1)\right) \cdot \OPT,                                                                                        
  \end{align*}
  as claimed. By standard techniques, this guarantee on the expected value of the rounded solution
  can be used to give with high probability a guarantee of $1/4 \cdot \OPT$ (namely, by repeating
  this rounding scheme $O(\poly(1/\gamma)\cdot \log(n))$ times).
\end{proof}

We now proceed to establish the sparse SDP solver wrapper claimed
in~\cref{lemma:sparse_sdp_solver_gram_matrix}. For concreteness, we
will use the following sparse SDP solver result of Lee--Padmanabhan
\cite{LeeP20}. The analogous result of Arora--Kale~\cite{AK07} with
slightly worse parameters also suffices for our purposes, but the main
result of~\cite{LeeP20} is stated in more convenient form.

\begin{theorem}[Adapted from Theorem 1.1 of~\cite{LeeP20}]\label{theo:fast_sdp}
  Given a matrix $\matr C \in \mathbb{R}^{n\times n}$ with $m$ non-zero entries, parameter
  $\gamma \in (0,1/2]$, with high probability, in time $\widetilde{O}((m+n)/\gamma^{3.5})$,
  it is possible to find a symmetric matrix $\matr Y \in \mathbb{R}^{n\times n}$ with $O(m)$ non-zero
  entries and diagonal matrix $\matr S \in \mathbb{R}^{n\times n}$ so that $\widetilde{X} = \matr S \cdot \exp \matr Y \cdot \matr S$
  satisfies
  \begin{itemize}
    \item $\widetilde{X} \succeq 0$,
    \item $\widetilde{X}_{i,i} \le 1$ for every $1\le i \le n$, and
    \item $\Tr(\matr C \cdot \widetilde{X}) ~\ge~ \max_{\matr X \succeq 0, X_{i,i} \le 1} \Tr\left(\matr C \cdot \matr X\right) - \gamma \sum_{i,j} \abs{\matr C_{i,j}}$.
  \end{itemize}
  Furthermore, we have $\norm{\matr Y}_{\textup{op}} \le O(\log(n)/\gamma)$ (\cf Lemma C.2.3 of~\cite{LeeP20}).
\end{theorem}

\begin{remark}
  We observe that~\cref{theo:fast_sdp} differs from Theorem 1.1 of~\cite{LeeP20} only by
  an additional bound on $\norm{\matr Y}_{\textup{op}}$.
  This bound is important in analyzing the error when approximating (matrix)
  exponential of $\matr Y$.
\end{remark}

We now show how we can approximate the Gram vectors of the SDP
solution of~\cref{theo:fast_sdp}. We rely on part of the analysis in
Arora--Kale~\cite{AK07}.

\begin{claim}\label{claim:gram_vector_from_sdp_sol}
  Let $\matr C \in \mathbb{R}^{n\times n}$ be a matrix with at most $m$ non-zero entries and $\gamma \in (0,1/2]$.
  Suppose $\widetilde{X} = \matr S \cdot \exp \matr Y \cdot \matr S$ satisfy the conclusions of~\cref{theo:fast_sdp}
  given $\matr C \in \mathbb{R}^{n\times n}$ and accuracy $\gamma$. 
  Then with high probability we can find in $\widetilde{O}(\poly(1/\gamma)\cdot (m+n))$ time approximate Gram vectors
  $u_1,\ldots,u_n \in \mathbb{R}^n$ such that $\widetilde{X}'_{i,j} \coloneqq \ip{u_i}{u_j}$ satisfy
  \begin{itemize}
    \item $\widetilde{X}_{i,i}' \le 1$ for every $1 \le i \le n$, and
    \item $\Tr(\matr C \cdot \widetilde{X}') ~\ge~ \Tr\left( \matr C \cdot \widetilde{\matr X} \right) - \gamma \sum_{i,j} \abs{\matr C_{i,j}}$.
  \end{itemize}  
\end{claim}

\begin{proof}
  Since $\widetilde{X} = (\matr S \cdot \exp( \matr Y/2)) (\matr S \cdot \exp( \matr Y/2))^t$, the rows of
  $\matr S \cdot \exp( \matr Y/2)$ can be taken as Gram vectors $u_1,\dots,u_n \in \mathbb{R}^n$ of $\widetilde{X}$.
  If we knew the rows of $\exp( \matr Y/2)$, we could readily recover these Gram vectors since $\matr S$ is diagonal.
  As observed in Arora--Kale~\cite{AK07}, computing $\exp( \matr Y/2)$ may be computationally expensive,
  so instead one can approximate the matrix-vector product $\exp( \matr Y/2) u$ using $d=O(\log(n)/\gamma^2)$
  random Gaussian vectors $u \sim N(0,I_n)$. By the Johnson--Lindenstrauss Lemma and scaling by $\sqrt{n/d}$,
  with high probability we obtain vectors $\widetilde{u}_1,\ldots,\widetilde{u}_n$ satisfying for every $i,j \in [n]$ say
  $$
  \abs{\ip{u_i}{u_j} - \ip{\widetilde{u}_i}{\widetilde{u}_j}} ~\le~ \frac{\gamma}{6}.
  $$
  In particular, whp $\norm{\widetilde{u}_i}_2^2 \le 1 + \gamma/6$. Thus, by normalizing the vectors $\widetilde{u}_i$ with
  $\norm{\widetilde{u}_i}_2 > 1$ to have $\ell_2$-norm one the preceding approximation deteriorates to
  $$
  \abs{\ip{u_i}{u_j} - \ip{\widetilde{u}_i}{\widetilde{u}_j}} ~\le~ \gamma/2.
  $$  
  To compute each the matrix-vector product $\exp( \matr Y/2) u$ in $\widetilde{O}(\poly(1/\gamma) \cdot (m+n))$,
  we rely on the following lemma.    
  \begin{lemma}[Arora--Kale~\cite{AK07}, \cf Lemma 6]\label{lemma:matrix_exp_vec_mul}
    Let $\calT_{\matr Y}$ be the time needed to compute the matrix-vector product $\matr Y u$. Then the vector
    $v \coloneqq \sum_{i=0}^k \matr Y^i u/(i!)$ can be computed in $O(k\cdot \calT_{\matr Y})$ time and if
    $k \ge \max\set{e^2\cdot\norm{\matr Y}_{\textup{op}},\ln(1/\delta)}$, then $\norm{\exp(\matr Y)u -v}_2 \le \delta$.
  \end{lemma}
  By noting that $\norm{\matr Y}_{\textup{op}} \le O(\log(n)/\gamma)$ and the time $\calT_{\matr Y}$ (\cf~\cref{lemma:matrix_exp_vec_mul})
  $\matr Y u$ is $\widetilde{O}((m+n)/\gamma)$, applying~\cref{lemma:matrix_exp_vec_mul} with say $\delta \le \poly(\gamma/n)$ 
  we can approximate each $\exp( \matr Y/2) u$ in time $\widetilde{O}((m+n)/\gamma)$. Therefore, the total running is $\widetilde{O}(\poly(1/\gamma) \cdot (m+n))$
  as claimed. Then the actual Gram vectors still satisfy
  $$
  \abs{\ip{u_i}{u_j} - \ip{\widetilde{u}_i}{\widetilde{u}_j}} ~\le~ \gamma.
  $$
  Hence, we get  
  $$
  \Tr(\matr C \cdot \widetilde{X}') ~\ge~ \Tr\left( \matr C \cdot \widetilde{\matr X} \right) - \gamma \sum_{i,j} \abs{\matr C_{i,j}},
  $$
  concluding the proof.
\end{proof}

We are ready to prove~\cref{lemma:sparse_sdp_solver_gram_matrix} which
is restated below for convenience.

\SparseSDPSolverWrapper*

\begin{proof}[Proof of~\cref{lemma:sparse_sdp_solver_gram_matrix}]
  Follows by combining the SDP solution $\widetilde{X}$ of~\cref{theo:fast_sdp} with
  the fast approximate Gram vector computation of~\cref{claim:gram_vector_from_sdp_sol}, the latter
  yielding another approximated SDP solution $\widetilde{X}'$. In both of these computations, we use accuracy parameter $\gamma/2$
  so that
  \begin{align*}
    \Tr(\matr C \cdot \widetilde{X}') ~\ge~& \Tr\left( \matr C \cdot \widetilde{\matr X} \right) - \frac{\gamma}{2} \sum_{i,j} \abs{\matr C_{i,j}}\\
                                      ~\ge~ &  \max_{\matr X \succeq 0, X_{i,i} \le 1} \Tr\left(\matr C \cdot \matr X\right) - \frac{\gamma}{2} \sum_{i,j} \abs{\matr C_{i,j}} - \frac{\gamma}{2} \sum_{i,j} \abs{\matr C_{i,j}}.
  \end{align*}
  Moreover, each step takes $\widetilde{O}(\poly(1/\gamma)\cdot (m+n))$ which concludes the proof.
\end{proof}

\section{Regularity Based Decoding}\label{sec:abstract_dec}

\subsection{List Decoding of Direct-Sum Codes}
We now develop list-decoding algorithms for direct-sum codes, using the regularity lemmas obtained in the previous section. We will prove the following theorem.
\begin{theorem}\label{thm:direct-sum-decoding}
Let $\calC_0 \subset \F_2^n$ be a code with $\bias(\calC_0) \leq \eps_0$, which is unique-decodable to distance $\nfrac{(1-\eps_0)}{4}$ in time $\calT_0$. Let $W \subseteq [n]^k$ be a $d$-regular, $\tau$-splittable collection of tuples, and let $\calC = \dsum_W(\calC_0)$ be the corresponding direct-sum lifting of $\calC_0$ with $\bias(\calC) \leq \eps$. Let $\beta$ be such that
\[
\beta ~\geq~ \max\inbraces{\sqrt{\eps}, ~\inparen{2^{20} \cdot \tau \cdot k^3}^{1/2}, ~2 \cdot \inparen{\frac12+2\eps_0}^{k/2}} \mper
\]
Then, there exists a randomized algorithm, which given $\tilde{y} \in \F_2^{W}$, recovers the list $\calL_{\beta}(\tilde{y}) \defeq \inbraces{y \in \calC ~|~ \Delta(\tilde{y}, y) \leq \nfrac12 - \beta}$ with probability $1-o(1)$, in time $\tilde{O}(C_{\beta,k,\eps_0} \cdot (\abs{W}+\calT_0))$, where $C_{k,\beta,\eps_0} = (\nfrac{6}{\eps_0})^{2^{O(\nfrac{k^3}{\beta^2})}}$. 
\end{theorem}
To obtain the decoding algorithm, we first define a function $g: [n]^k \to \pmone$ supported on $W$ as
\[
g(i_1, \ldots, i_k) ~\defeq~ 
\begin{cases}
(-1)^{\tilde{y}_{(i_1,\ldots,i_k)}} & ~\text{if}~ (i_1,\ldots,i_k) \in W \\
0 & ~\text{otherwise}
\end{cases}
\]
For each $z \in \F_2^n$, we also consider the similar function $\chi_z: [n] \to \pmone$ defined as $\chi_z(i) = (-1)^{z_i}$. We first re-state the decoding problem in terms of the functions $g$ and $\chi_z$.
\begin{claim}\label{clm:decoding-correlation}
Let $z \in \F_2^n$, and let the functions $g$ and $\chi_z$ be as above. Then,
\[
\Delta(\tilde{y}, \dsum_W(z)) \leq \frac12 - \beta 
\quad \Leftrightarrow \quad
\ip{g}{\chi_z^{\otimes k}}_{\mu_k} 
= 
\inparen{\frac{n}{d}}^{k-1} \cdot \ip{g}{\chi_z^{\otimes k}}_{\mu_1^{\otimes k}}
\geq~
2\beta \mper
\]
\end{claim}
\begin{proof}
We have
\begin{align*}
\Delta(\tilde{y}, \dsum_W(z))
&~=~ \Ex{(i_1, \ldots, i_k) \sim W}{\indicator{\tilde{y}_{(i_1,\ldots,i_k)} ~\neq~ z_{i_1} + \cdots + z_{i_k} \mod 2}} \\
&~=~ \Ex{(i_1, \ldots, i_k) \sim \mu_k}{\frac{1-g(i_1,\ldots, i_k) \cdot \prod_{t \in [k]} \chi_z(i_t)}{2}}
~=~ \frac12 - \frac12 \cdot \ip{g}{\chi_z^{\otimes k}}_{\mu_k}    \mper
\end{align*}
Finally, using the fact that $g$ is only supported on $W$, and $\abs{W} = d^{k-1} \cdot n$ by $d$-regularity, we have $\ip{g}{f}_{\mu_k} = (\nfrac{n}{d})^{k-1} \cdot \ip{g}{f}_{\mu_1^{\otimes k}}$ for any function $f:[n]^k \to \R$.
\end{proof}
Note that each element of the list $\calL_{\beta}(\tilde{y})$ must be equal to $\dsum_W(z)$ for some $z \in \calC_0$. 
Thus, to search for all such $z$, we will consider the decomposition $h$ of the function $g$, given by \cref{theo:eff_weak_reg} with respect to the class of functions $\calF = \SCUT^{\otimes k}$. Since the functions $\chi_z^{\otimes k}$ belong to $\calF$, it will suffice to only consider the inner product $\ip{h}{\chi_z^{\otimes k}}_{\mu_1^{\otimes k}}$. 

Also, since the approximating function $h$ is determined by a small number of functions, say $\inbraces{f_1, \ldots, f_r: [n] \to \pmone}$, it will suffice to (essentially) consider only the functions measurable in the factor $\calB$ determined by $f_1, \ldots, f_r$. Recall that the factor $\calB$ is simply a partition of $[n]$ in $2^r$ pieces according to the values of $f_1, \ldots, f_r$. Also, since any $\calB$-measurable function is constant on each piece, it is completely specified by $\abs{\calB}$ real values. We will only consider functions taking values in $[-1,1]$, and discretize this space to an appropriate accuracy $\eta$, to identify all relevant $\calB$-measurable functions with the set $\inbraces{0, \pm \eta, \pm 2\eta, \ldots, \pm 1}^{\abs{\calB}}$.
The decoding procedure is described in the following algorithm.

\begin{algorithm}{List Decoding}{$\tilde{y} \in \F_2^W$}{List $\calL \subseteq \calC$}\label{algo:direct-sum-decoding}
\begin{itemize}
\item Obtain the approximator $h$ given by~\cref{theo:eff_weak_reg} for $\calF = \SCUT^{\otimes k}$, $\delta = \beta$, and the function $g:[n]^k \to \pmone$ defined as
\[
g(i_1, \ldots, i_k) ~\defeq~ 
\begin{cases}
(-1)^{\tilde{y}_{(i_1,\ldots,i_k)}} & ~\text{if}~ (i_1,\ldots,i_k) \in W \\
0 & ~\text{otherwise}
\end{cases}
\]
%
%
\item Let $h$ be of the form $h = \sum_{j=1}^p c_j \cdot f_{j_1} \otimes \cdots \otimes f_{j_k}$, with each $f_{j_t}: [n] \to \pmone$. Let $\calB$ be the factor determined by the functions $\inbraces{f_{j_t}}_{j \in [p], t \in [k]}$.
\item Let $\calL = \emptyset$ and let $\eta = 1/\lceil(2/\eps_0)\rceil$.
 For each $\calB$-measurable function $\fbar$ given by a value in $D_{\eta} \defeq \inbraces{0, \pm \eta, \pm 2\eta, \ldots, \pm 1}$ for every atom of $\calB$:
%
  \begin{itemize}
  \item Sample a random function $\chi:[n] \to \pmone$ by independently sampling $\chi(i) \in \pmone$ for each $i$, such that $\Ex{\chi(i)} = \fbar(i)$. Take $\tilde{z} \in \F_2^n$ to be such that $\chi = \chi_{\tilde{z}}$.
  \item If there exists $z \in \calC_0$ such that
    \[
      \Delta(\tilde{z},z) ~\leq~ \frac{(1-\eps_0)}{4}
      \quad \text{and} \quad
      \Delta(\tilde{y}, \dsum_W(z)) ~\leq~ \frac12 - \beta \mcom
    \]
    then $\calL \leftarrow \calL \cup \{\dsum_{W}(z)\}$.
  \end{itemize}
\item Return $\calL$.
\end{itemize}
\end{algorithm}
Note that by our choice of the $\beta$ in \cref{thm:direct-sum-decoding}, we have that $\tau \leq \beta^2/(2^{20} k^3)$. Thus, we can indeed apply \cref{theo:eff_weak_reg} to obtain the function $h$ as required by the algorithm. To show that the algorithm can recover the list, we will need to show that for each $z$ such that $\dsum_W(z) \in \calL_{\beta}$, the sampling procedure finds a $\tilde{z}$ close to $z$ with significant probability. To analyze this probability, we first prove the following claim.
\begin{claim}\label{clm:sampling}
Let $z \in \F_2^n$ and let $\fbar:[n] \to D_{\eta}$ be a minimizer of $\smallnorm{\Ex{\chi_z|\calB} - \fbar}_{\infty}$ among all $\calB$-measurable functions in $D_{\eta}^{\abs{\calB}}$. Then, over the random choice of $\chi$ such that $\Ex{\chi} = \fbar$, we have
  \[
  \Ex{\chi}{\ip{\chi}{\chi_z}_{\mu_1}} ~=~ \ip{\fbar}{\chi_z}_{\mu_1} ~\geq~ \norm{\Ex{\chi_z|\calB}}_{\mu_1}^2 - \eta \mper   
  \]
\end{claim}
\begin{proof}
By linearity of the inner product, we have
\[
    \Ex{\chi}{\ip{\chi}{\chi_z}_{\mu_1}}
    ~=~ \ip{\Ex\chi}{\chi_z}_{\mu_1}
    ~=~ \ip{\fbar}{\chi_z}_{\mu_1}
    ~=~ \ip{\fbar}{\Ex{\chi_z|\calB}}_{\mu_1} \mcom
\]
where the last equality used \cref{prop:measurable-inner-product} and the fact that $\fbar$ is $\calB$-measurable.  Since $\Ex{\chi_z|\calB}$ takes values in $[-1,1]$ and $\fbar$ is the minimizer over all functions in $D_{\eta}^{\abs{\calB}}$, we must have $\smallnorm{\Ex{\chi_z|\calB} - \fbar}_{\infty} \leq \eta$. Using this pointwise bound, we get
\begin{align*}
  \ip{\fbar}{\Ex{\chi_z|\calB}}_{\mu_1}
  &~=~ \Ex{i \sim \mu_1}{\fbar(i) \cdot \Ex{\chi_z|\calB}(i)} \\
  &~\geq~ \Ex{i \sim \mu_1}{\inparen{\Ex{\chi_z|\calB}(i)}^2 - \eta \cdot \abs{\Ex{\chi_z|\calB}(i)}}
  ~\geq~  \norm{\Ex{\chi_z|\calB}}_{\mu_1}^2 - \eta \mper \qedhere
\end{align*}
\end{proof}
We next show that when $z \in \F_2^n$ is such that $\ip{g}{\chi_z^{\otimes k}}$ is large, then the norm of the conditional expectation $\Ex{\chi_z|\calB}$ is also large, and hence the sampling procedure finds a $\tilde{z}$ close to $z$. When we have a $z \in \calC_0$ with such a property, we can use $\tilde{z}$ to recover  $z$ using the unique decoding algorithm for $\calC_0$.
\begin{lemma}\label{lem:large-norm}
  Let $z \in \F_2^n$ be such that
  \[
    \ip{g}{\chi_z^{\otimes k}}_{\mu_k}  =
    \inparen{\frac{n}{d}}^{k-1} \cdot \ip{g}{\chi_z^{\otimes k}}_{\mu_1^{\otimes k}}
    \geq~ 2\beta \mper
  \]
  Then, we have $\smallnorm{\Ex{\chi_z|\calB}}_{\mu_1}^2 \geq (\beta/2)^{2/k}$.
\end{lemma}
\begin{proof}
Let $h$ be the approximating function obtained by applying~\cref{theo:eff_weak_reg} to $g$ with approximation error $\delta = \beta$.
 Note that we have $\norm{h}_{\mu_1^{\otimes k}} \leq 2$, and for any $f \in \SCUT^{\otimes k}$,
    \[
      \inparen{\frac{n}{d}}^{k-1} \cdot \ip{g - \inparen{\frac{d}{n}}^{k-1} \cdot h~}{f}_{\mu_1^{\otimes k}} ~\leq~ \delta \mper
    \]
    Using $f = \chi_z^{\otimes k}$ and $\delta = \beta$, we get
    \[
      \ip{h}{\chi_z^{\otimes k}}_{\mu_1^{\otimes k}} ~\geq~ 2\beta - \delta ~\geq~ \beta \mper
    \]
    Using \cref{prop:measurable-inner-product}, and the fact that $\calB$ is defined so that all functions in the decomposition of $h$ are (by definition) $\calB$-measurable, we have
    \begin{align*}
      \ip{h}{\chi_z^{\otimes k}}_{\mu_1^{\otimes k}}
      ~=~ \sum_{j=1}^p c_j \prod_{t = 1}^k \ip{f_{j_t}}{\chi_z}_{\mu_1}
      ~=~ \sum_{j=1}^p c_j \prod_{t = 1}^k \ip{f_{j_t}}{\Ex{\chi_z|\calB}}_{\mu_1}
      ~=~ \ip{h}{\inparen{\Ex{\chi_z|\calB}}^{\otimes k}}_{\mu_1^{\otimes k}} \mper
    \end{align*}
    Combining the above with Cauchy-Schwarz, we get
    \[
      \beta ~\leq~ \ip{h}{\chi_z^{\otimes k}}_{\mu_1^{\otimes k}}
      ~\leq~ \norm{h}_{\mu_1^{\otimes k}} \cdot \norm{\inparen{\Ex{\chi_z|\calB}}^{\otimes k}}_{\mu_1^{\otimes k}}
      ~=~ \norm{h}_{\mu_1^{\otimes k}} \cdot \norm{\Ex{\chi_z|\calB}}_{\mu_1}^{k} \mper
    \]
  Using $\norm{h}_{\mu_1^{\otimes k}} \leq 2$ then gives $\smallnorm{\Ex{\chi_z|\calB}}_{\mu_1}^2 \geq (\beta/2)^{2/k}$.
  \end{proof}
Using the above results, we can now complete the analysis of the algorithm.
\begin{proof}[Proof of \cref{thm:direct-sum-decoding}]
We first argue that for any codeword $z \in \calC_0$ such that $\dsum_W(z) \in \calL_{\beta}$, sampling a random function $\chi$ (with $\Ex{\chi} = \fbar$ for an appropriate $\fbar$) finds a $\tilde{z}$ close to $z$ with significant probability.
Let $\fbar \in D_{\eta}^{\calB}$ be the minimizer of $\smallnorm{\chi_z - \fbar}_{\infty}$, for such a $z \in \calC_0$. We have by \cref{clm:sampling} that $\Ex{\chi}{\ip{\chi}{\chi_z}_{\mu_1}} \geq \smallnorm{\Ex{\chi_z|\calB}}_{\mu_1}^2 - \eta$.
Since $\Delta(\tilde{y},\dsum_W(z)) \leq \nfrac12 - \beta$, we have by \cref{clm:decoding-correlation} that $\ip{g}{\chi_z^{\otimes k}}_{\mu_k} \geq 2\beta$. Thus, by \cref{lem:large-norm}, we have that $\smallnorm{\Ex{\chi_z|\calB}}_{\mu_1}^2 \geq (\nfrac{\beta}{2})^{2/k}$. Combining these, and using the lower bound on $\beta$, we get that
\[
\Ex{\chi}{\ip{\chi}{\chi_z}_{\mu_1}}
~~\geq~~ \inparen{\frac{\beta}{2}}^{2/k} - \eta
~~\geq~~ \frac12 + 2\eps_0 - \eta
~~\geq~~ \frac12 + \frac{3\eps_0}{2} \mper
\]
Since $\ip{\chi}{\chi_z}_{\mu_1}$ is the average of $n$ independent (not necessarily identical) random variables $\inbraces{\chi(i)\cdot \chi_z(i)}_{i \in [n]}$ in the range $[-1,1]$, we get by Hoeffding's inequality that
\[
\Pr{\chi}{\ip{\chi}{\chi_z}_{\mu_1} \leq \frac12 + \eps_0}
~\leq~ \Pr{\chi}{\abs{\ip{\chi}{\chi_z}_{\mu_1} - \Ex{\chi}{\ip{\chi}{\chi_z}_{\mu_1}}} \geq \frac{\eps_0}{2}}
~\leq~ \exp\inparen{-\eps_0^2 \cdot n/8} \mper
\]

Thus, given a good sample $\chi$ satisfying $\ip{\chi}{\chi_z}_{\mu_1} \geq \nfrac12 + \eps_0$, we can recover the above $z \in \calC_0$ such that $\dsum_{W}(z) \in \calL_{\beta}$, via the unique decoding algorithm for $\calC_0$. Also, given the right $\fbar$, we sample a good $\chi$ with probability at least $1 - \exp(-\eps_0^2 \cdot n/8)$. 
A union bound then gives
\[
\Pr{\calL = \calL_{\beta}} ~~\geq~~ 1 ~-~ \abs{\calL_{\beta}} \cdot \exp(-\eps_0^2 \cdot n/8) \mper
\]
Using $\beta \geq \sqrt{\eps}$, we get that $\abs{\calL_{\beta}} \leq (\nfrac{1}{\eps})$ by the Johnson bound, which yields the desired probability bound.
%
\paragraph{Running time.} Using \cref{theo:eff_weak_reg}, the decomposition $h$ can be computed in time $\tilde{O}(C_{\beta,k} \cdot \abs{W})$. Given the functions $f_1, \ldots, f_r$ forming the decomposition $h$, the factor $\calB$ can be computed in time $O(2^r \cdot n)$. For a chosen $\fbar$ in the sampling step, a sample $\chi$ can be computed in time $O(n)$, and the decoding problem for the corresponding $\tilde{z}$ can be solved in time $\calT_0$. Also, the distance $\Delta(\tilde{y},\dsum_W(z))$ can be computed in time $O(\abs{W})$. Since the total number of sampling steps is at most $(\nfrac{3}{\eta})^{\abs{B}}$ and the number of functions in the decomposition $h$ is $O(\nfrac{k^3}{\beta^2})$ from \cref{theo:eff_weak_reg}, we get that the total number of sampling steps is $(\nfrac{6}{\eps_0})^{2^{O(\nfrac{k^3}{\beta^2})}}$. Thus, the total running time is bounded by $\tilde{O}(C_{k,\beta,\eps_0} \cdot (\abs{W} + \calT_0))$, where $C_{k,\beta,\eps_0} = (\nfrac{6}{\eps_0})^{2^{O(\nfrac{k^3}{\beta^2})}}$. 
\end{proof}

\subsection{List Decoding of Direct-Product Codes}
\mnote{Maybe include direct-product in intro/prelims?}
We now show that a slight modification of the above algorithm for direct-sum codes can also be used for list decoding direct-product codes. For $W \subseteq [n]^k$ and $z \in \F_2^n$, the lifting $\dprod_W(z) \in (\F_2^k)^{W}$ is defined as
\[
\dprod_W(z) ~\defeq~ y \quad\text{s.t.}\quad y_{i_1,\ldots,i_k} = (z_{i_1},\ldots, z_{i_k}) \quad\forall (i_1,\ldots,i_k) \in W
\]
As before, $\dprod_{W}(\Cc_0) = \inbraces{\dprod_{W}(z) ~|~ z \in \Cc_0}$.
Since $\dprod_{W}(\Cc_0)$ is a code over alphabet $\F_2^k$, the distance is now close to 1. We prove the following theorem.
\begin{theorem}\label{thm:direct-product-decoding}
Let $\calC_0 \subset \F_2^n$ be a code with $\bias(\calC_0) \leq \eps_0$, which is unique-decodable to distance $\nfrac{(1-\eps_0)}{4}$ in time $\calT_0$. Let $W \subseteq [n]^k$ be a $d$-regular, $\tau$-splittable collection of tuples, and let $\calC = \dprod_W(\calC_0)$ be the corresponding direct-product lifting of $\calC_0$ with $\Delta(\calC) \geq 1-\eps$. Let $\beta$ be such that
\[
\beta ~\geq~ \max\inbraces{ \sqrt{\eps}, ~\inparen{2^{24} \cdot \tau \cdot k^3}^{1/2}, ~8 \cdot \inparen{\frac12+2\eps_0}^{k/6}, 2\cdot e^{-k/54} } \mper
\]
Then, there exists a randomized algorithm, which given $\tilde{y} \in (\F_2^k)^{W}$, recovers the list $\calL_{\beta}(\tilde{y}) \defeq \inbraces{y \in \calC ~|~ \Delta(\tilde{y}, y) \leq 1 - \beta}$ with probability $1-o(1)$, in time $\tilde{O}(C_{\beta,k,\eps_0} \cdot (\abs{W}+\calT_0))$, where $C_{k,\beta,\eps_0} = (\nfrac{6}{\eps_0})^{2^{O(\nfrac{k^3\log k}{\beta^2})}}$. 
\end{theorem}
\mnote{Will need to check the conditions on $\beta$ and the running time.}
As in the case of direct-sum decoding, we will apply regularity to function supported on $W$, taking values in $\pmone$. Let $K \subseteq [k]$ and $z \in \F_2^n$. We define the functions $g^{(K)}, \chi_z^{(K)}: [n]^k \to \pmone$ as 
\begin{align*}
g^{(K)}(i_1, \ldots, i_k) &~\defeq~ 
\begin{cases}
\prod_{t \in K}(-1)^{\tilde{y}_{(i_1,\ldots,i_k), t}} & ~\text{if}~ (i_1,\ldots,i_k) \in W \\
0 & ~\text{otherwise}
\end{cases} \mper \\
\chi_z^{(K)}(i_1,\ldots,i_k) &~\defeq~ \prod_{t \in K} (-1)^{z_{i_t}} \mper
\end{align*}
We can now state the decoding problem in terms of the correlation of these functions. 
\begin{claim}\label{clm:dprod-correlation}
Let $z \in \F_2^n$, and let the functions $g^{(K)}$ and $\chi_z^{(K)}$ be as above. Then,
\[
1 - \Delta(\tilde{y}, \dprod_W(z)) ~=~
\Ex{K \subseteq [k]}{\ip{g^{(K)}}{\chi_z^{(K)}}_{\mu_k}}
~=~ 
\inparen{\frac{n}{d}}^{k-1} \cdot \Ex{K \subseteq [k]}{\ip{g^{(K)}}{\chi_z^{(K)}}_{\mu_1^{\otimes k}}}
\mper
\]
\end{claim}
\begin{proof}
The second equality follows from the fact that $g^{(K)}$ is supported on $W$ and from $d$-regularity, as in the case of direct sum. We focus on the proving the first equality.
\begingroup
\allowdisplaybreaks
\begin{align*}
\mathbb{E}_{K\subseteq [k]} \left[ \ip{g^{(K)}}{\chi_z^{(K)}}_{\mu_k} \right] &= 
\mathbb{E}_{K\subseteq [k]} \left[ \mathbb{E}_{(i_1,i_2,\cdots ,i_k) \sim W} \left[ g^{(K)}(i_1,i_2,\cdots ,i_k) \cdot \chi_z^{(K)}(i_1,i_2,\cdots ,i_k)\right]  \right] \\
&= \mathbb{E}_{(i_1,i_2,\cdots ,i_k) \sim W} \left[ \mathbb{E}_{K\subseteq [k]} \left[ g^{(K)}(i_1,i_2,\cdots ,i_k) \cdot \chi_z^{(K)}(i_1,i_2,\cdots ,i_k)\right]  \right]\\
&= \mathbb{E}_{(i_1,i_2,\cdots ,i_k) \sim W} \left[ \mathbb{E}_{K\subseteq [k]} \left[ \prod_{t\in K} (-1)^{\tilde{y}_{(i_1,i_2,\cdots ,i_k),t}} \cdot (-1)^{z_{i_t}} \right]  \right]\\
&= \mathbb{E}_{(i_1,i_2,\cdots ,i_k) \sim W} \left[ \mathbb{E}_{j_1,j_2,\cdots ,j_k \sim \{0,1\}} \left[ \prod_{t\in [k]} (-1)^{\left( \tilde{y}_{(i_1,i_2,\cdots ,i_k),t} + z_{i_t}\right) \cdot j_t} \right]  \right]\\
&= \mathbb{E}_{(i_1,i_2,\cdots ,i_k) \sim W} \left[ \prod_{t\in [k]} \left( \mathbb{E}_{j_t \sim \{0,1\}} \left[ (-1)^{\left( \tilde{y}_{(i_1,i_2,\cdots ,i_k),t} + z_{i_t}\right) \cdot j_t} \right]  \right) \right]\\
&= \mathbb{E}_{(i_1,i_2,\cdots ,i_k) \sim W} \left[ \prod_{t\in [k]} \left( \frac{1}{2}\cdot 1 + \frac{1}{2}\cdot (-1)^{\mathbb{1}_{\tilde{y}_{(i_1,i_2,\cdots ,i_k),t} \neq z_{i_t}}} \right) \right]\\
&= \mathbb{E}_{(i_1,i_2,\cdots ,i_k) \sim W} \left[ \prod_{t\in [k]} \left( \mathbb{1}_{\tilde{y}_{(i_1,i_2,\cdots ,i_k),t} = z_{i_t}} \right) \right]\\
&= \mathbb{E}_{(i_1,i_2,\cdots ,i_k) \sim W} \left[ \mathbb{1}_{\tilde{y}_{(i_1,i_2,\cdots ,i_k)} = (z_{i_1},z_{i_2},\cdots ,z_{i_k})} \right]\\
&= 1- \mathbb{P}_{(i_1,i_2,\cdots ,i_k) \sim W} \left[ \tilde{y}_{(i_1,i_2,\cdots ,i_k)} \neq (z_{i_1},z_{i_2},\cdots,z_{i_k}) \right] = 1-\Delta(\tilde{y},\dprod_W(z))
\end{align*}
\endgroup
\end{proof}
Using Chernoff bounds, we can prove the following corollary.
\begin{claim}
Let $\Delta(\tilde{y}, \dprod_W(z)) \leq 1-\beta$ and let $\beta \geq 2\cdot e^{-k/54}$. Then there exists $K \subseteq [k]$ with $\abs{K} \geq k/3$ such that
\[
\ip{g^{(K)}}{\chi_z^{(K)}}_{\mu_k} 
~=~ 
\inparen{\frac{n}{d}}^{k-1} \cdot \ip{g^{(K)}}{\chi_z^{(K)}}_{\mu_1^{\otimes k}}
~\geq~
\frac{\beta}{2}
\]
\end{claim}
\begin{proof}
From Chernoff bound, we get that 
\[
\mathbb{P}_{K\subseteq [k]} [|K|\leq k/3] = \mathbb{P}_{X_i \in \{0,1\} \forall i\in [k]} [\sum_{i\in [k]} X_i \leq (1-\frac{1}{3})\frac{k}{2}] \leq e^{-\frac{1}{3} \left( \frac{1}{3} \right)^2 \frac{k}{2}} = e^{-k/54}
\].

Suppose the claim is not true. Then $\ip{g^{(K)}}{\chi_z^{(K)}}_{\mu_k}< \frac{\beta}{2}$ for all $K\subseteq [k]$ with $|K|\geq k/3$. Then,
\begin{align*}
\mathbb{E}_{K\subseteq [k]} [\ip{g^{(K)}}{\chi_z^{(K)}}_{\mu_k}] & =\ \mathbb{E}_{K\subseteq [k]} [\ip{g^{(K)}}{\chi_z^{(K)}}_{\mu_k} \: \Big \vert \: |K| <k/3] \cdot \mathbb{P}_{K\subseteq [k]} [|K|<k/3] \\
&\quad + \mathbb{E}_{K\subseteq [k]} [\ip{g^{(K)}}{\chi_z^{(K)}}_{\mu_k} \:\Big\vert\: |K|\geq k/3] \cdot \mathbb{P}_{K\subseteq [k]} [|K|\geq k/3]\\
&\leq \mathbb{P}_{K\subseteq [k]} [|K|<k/3] + \mathbb{E}_{K\subseteq [k]} [\ip{g^{(K)}}{\chi_z^{(K)}}_{\mu_k} \:\Big\vert\: |K|\geq k/3] \\
&< \mathbb{P}_{K\subseteq [k]} [|K|\leq k/3] + \beta/2 \\
&\leq e^{-k/54} + \beta/2 \leq \beta/2+\beta/2 =\beta
\end{align*}

which is a contradiction, as
\[
\mathbb{E}_{K\subseteq [k]} [\ip{g^{(K)}}{\chi_z^{(K)}}_{\mu_k}] = 1 - \Delta(\tilde{y}, \dprod_W(z)) \geq \beta .
\]
\end{proof}
\mnote{We will apply the direct-sum decoding algorithm to $g^{(K)}$. The lower bound on $\abs{K}$ will be needed for parity sampling, though the choices of $k/3$ and $\beta/2$ are arbitrary.} 
As before, we will consider decompositions $h^{(K)}$ of the functions $g^{(K)}$, given by \cref{theo:eff_weak_reg} with respect to the class of functions $\calF = \SCUT^{\otimes k}$, since the functions $\chi_z^{(K)}$ also belong to $\calF$. 
The only change to the algorithm is the fact that now we consider all sufficiently large $K \subseteq [k]$.

\begin{algorithm}{List Decoding of Direct-Product }{$\tilde{y} \in \F_2^W$}{List $\calL \subseteq \calC$}\label{algo:direct-product-decoding}
\begin{itemize}
\item Let $\calL = \emptyset$. For each $K \subseteq [k]$, with $\abs{K} \geq k/3$:
\begin{itemize}
\item Obtain the approximator $h$ given by~\cref{theo:eff_weak_reg} for $\calF = \SCUT^{\otimes k}$, $\delta = \beta/4$, and the function $g^{(K)}:[n]^k \to \pmone$ defined as
\[
g^{(K)}(i_1, \ldots, i_k) ~\defeq~ 
\begin{cases}
\prod_{t \in K}(-1)^{\tilde{y}_{(i_1,\ldots,i_k), t}} & ~\text{if}~ (i_1,\ldots,i_k) \in W \\
0 & ~\text{otherwise}
\end{cases} 
\]
%
%
\item Let $h$ be of the form $h = \sum_{j=1}^p c_j \cdot f_{j_1} \otimes \cdots \otimes f_{j_k}$, with each $f_{j_t}: [n] \to \pmone$. Let $\calB$ be the factor determined by the functions $\inbraces{f_{j_t}}_{j \in [p], t \in [k]}$.
\item Let $\eta = 1/\lceil(2/\eps_0)\rceil$.
 For each $\calB$-measurable function $\fbar$ given by a value in $D_{\eta} \defeq \inbraces{0, \pm \eta, \pm 2\eta, \ldots, \pm 1}$ for every atom of $\calB$:
%
  \begin{itemize}
  \item Sample a random function $\chi:[n] \to \pmone$ by independently sampling $\chi(i) \in \pmone$ for each $i$, such that $\Ex{\chi(i)} = \fbar(i)$. Take $\tilde{z} \in \F_2^n$ to be such that $\chi = \chi_{\tilde{z}}$.
  \item If there exists $z \in \calC_0$ such that
    \[
      \Delta(\tilde{z},z) ~\leq~ \frac{(1-\eps_0)}{4}
      \quad \text{and} \quad
      \Delta(\tilde{y}, \dprod_W(z)) ~\leq~ 1 - \beta \mcom
    \]
    then $\calL \leftarrow \calL \cup \{\dsum_{W}(z)\}$.
  \end{itemize}
\end{itemize}
\item Return $\calL$.
\end{itemize}
\end{algorithm}
%



\section{Near-linear Time Decoding of Ta-Shma's Codes}\label{sec:concrete_dec}

We now proceed to prove our main result, namely~\cref{theo:main},
which establishes a near-linear time \emph{unique} decoding algorithm
for Ta-Shma's codes~\cite{TS17}. It will follow from the
regularity based list decoding algorithm for direct sum
codes,~\cref{thm:direct-sum-decoding}, applied to the decoding of a
slight modification of Ta-Shma's construction from~\cite{JQST20} that
yields a splittable collection of tuples for the direct sum.

\TheoMainUniqueDec*

We now state the properties and guarantees needed in our work of this
slightly modified version of Ta-Shma's direct sum construction of near
optimal $\epsilon$-balanced codes. To make the decoding task more
transparent, we will additionally require the base code in Ta-Shma's
construction have the following technical property.

\begin{definition}\label{def:sym_mult}
  We say that a code has symbol multiplicity $m \in \mathbb{N}$ if it can be obtained from
  another code by repeating each symbol of its codeword $m$ times.
\end{definition}


\begin{restatable*}{theorem}{TaShmaConsFact}[Ta-Shma's Codes (implicit in~\cite{TS17})]\label{fact:ta-shma_splittable_tuples}
  Let $c > 0$ be an universal constant. For every $\epsilon > 0$ sufficiently small, there exists    $k=k(\epsilon)$ satisfying $\Omega(\log(1/\epsilon)^{1/3}) \le k \le O(\log(1/\epsilon))$,
  $\epsilon_0 = \epsilon_0(\epsilon) > 0$, and positive integer
  $m=m(\epsilon) \le (1/\epsilon)^{o(1)}$ such that Ta-Shma's construction yields
  a collection of $\tau$-splittable tuples $W=\tuples{k} \subseteq [n]^k$ satisfying:
  \begin{itemize}
    \item[(i)] For every linear $\epsilon_0$-balanced code $\Cc_0 \subseteq \F_2^{n}$ with symbol multiplicity $m$, the direct sum code $\dsum_{W}(\Cc_0)$ is: \begin{itemize}
                  \item[(i.1)] $\epsilon$-balanced (parity sampling).
                  \item[(i.2)] if $\Cc_0$ has rate $\Omega(\epsilon_0^c/m)$,  then $\dsum_{W}(\Cc_0)$
                                has rate $\Omega(\epsilon^{2+o(1)})$    (near optimal rate)
                 \end{itemize}
    \item[(ii)] $\tau ~\le~ \exp(-\Theta(\log(1/\epsilon)^{1/6}))$  (splittability).
    \item[(iii)] $W$ is constructible in $\poly(\abs{W})$ time (explicit construction).                     
  \end{itemize}
\end{restatable*}

Ta-Shma's construction is based on a generalization of the zig-zag
product of Reingold, Vadhan and Wigderson~\cite{RVW00}. To make the
exposition more self-contained, we recall the slight modification
from~\cite{JQST20} in~\cref{appendix:ta-shma}, but it is not
exhaustive exposition. The interested reader is referred to
Ta-Shma~\cite{TS17} for the original construction for aspects not
covered here.

Ta-Shma's code construction requires an $\epsilon_0$-balanced base
code $\Cc_0 \subseteq \F_2^n$ whose distance will be amplified by
taking the direct sum with a carefully chosen collection of tuples $W$
yielding an $\epsilon$-balanced code $\Cc=\dsum_W(\Cc_0)$. Since we
our goal is to achieve near-linear time encoding and decoding of
$\Cc$, we take an ``off-the-shelf'' base code $\Cc_0$ that is linear
time encodable and decodable (near-linear time also suffices). A
convenient choice is the linear binary code family of
Guruswami--Indyk~\cite{GI05} that can be encoded and decoded in linear
time. The rate versus distance trade-off is at the so-called Zyablov
bound. In particular, it yields codes of distance $1/2-\epsilon_0$
with rate $\Omega(\epsilon_0^3)$, but for our applications rate
$\poly(\epsilon_0)$ suffices (or with some extra steps even
any rate depending only on $\epsilon_0$ suffices,
see~\cref{remark:choice_of_base_code}). We will
use~\cref{cor:base_code} implicit in~\cite{GI05}.

\begin{restatable*}{corollary}{BaseCodeGI}[Implicit in Guruswami--Indyk~\cite{GI05}]\label{cor:base_code}
   For every $\epsilon_0 > 0$, there exists a family of $\epsilon_0$-balanced binary
   linear codes $\Cc_0 \subseteq \F_2^{n}$ of rate $\Omega(\epsilon_0^3)$ which can be encoded in 
   $O_{\epsilon_0}(n)$ time and can be decoded in $O(\exp(\poly(1/\epsilon_0)) \cdot n)$
   time from up to a fraction $1/4-\epsilon_0$ of errors.
   Furthermore, every code in the family is explicitly specified given a binary
   linear code of blocklength $\poly(1/\epsilon_0)$ which can be constructed in
   probabilistic $O(\poly(1/\epsilon_0))$ or deterministic $2^{O(\poly(1/\epsilon_0))}$ time.
\end{restatable*}

We first prove the (gentle) \emph{list} decoding result of Ta-Shma's codes. 
\TheoMainGentleListDec*

\begin{proof}
 We start by dealing with a simple technical issue of making the base code
 in Ta-Shma's construction have the required symbol multiplicity.
 Let $\Cc_0' \subseteq \F_2^{n'}$ be an $\epsilon_0$-balanced code
 from~\cref{cor:base_code} which we will use to obtain a base code in
 Ta-Shma's construction where $\epsilon_0 > 0$ is a suitable value
 prescribed by this construction.

 Ta-Shma's construction then takes $\Cc_0' \subseteq \F_2^{n'}$ and
 forms a new code $\Cc_0 \subseteq \F_2^n$ by repeating each codeword symbol $m \le (1/\epsilon)^{o(1)}$ times.
 By~\cref{claim:dec_replication_lifting}, $\Cc_0$ is an $\epsilon_0$-balanced code that can be unique decoded within
 the same (fractional) radius of $\Cc_0'$ in time $\calT_0(n) = r \cdot \calT_0'(n') + \widetilde{O}(r^2 \cdot n')$,
 where $\calT_0(n)'$ is the running time of an unique decoder for $\Cc_0'$. Since by~\cref{cor:base_code}
 $\calT_0(n') = O(\exp(\poly(1/\epsilon_0)) \cdot n')$ and $\epsilon_0 \gg \epsilon$, the decoding time of $\Cc_0$ becomes $\calT_0(n)= O(\exp(\poly(1/\epsilon)) \cdot n)$.

 Let $W=\tuples{k}$ be a collection of tuples from Ta-Shma's construction~\cref{fact:ta-shma_splittable_tuples} so that
 $\Cc = \dsum_W(\Cc_0)$ is $\epsilon$-balanced, $\tau \le \exp(-\Theta(\log(1/\epsilon)^{1/6}))$ and $k=\Omega(\log(1/\epsilon)^{1/3})$.
 We will invoke our list decoding algorithm~\cref{thm:direct-sum-decoding} whose list decoding
 radius $1/2-\beta$ has to satisfy
 \[
  \beta ~\geq~ \max\inbraces{\sqrt{\eps}, ~\inparen{2^{20} \cdot \tau \cdot k^3}^{1/2}, ~2 \cdot \inparen{\frac12+2\eps_0}^{k/2}} \mper
 \]
 Using our values of $\tau$ and $k$ together with the fact that $\eps_0 < 1$ is  bounded away form $1$ by a constant amount gives
 \[
  \beta ~\geq~ \max\inbraces{\sqrt{\eps}, ~, ~\exp(-\Theta((\log(1/\epsilon))^{1/6})), \exp(-\Theta((\log(1/\epsilon))^{1/3}))} \mper
 \]
 Hence, we can take $\beta = \exp(-\Theta(\log(1/\epsilon)^{1/6}))$. Now, we compute the list decoding
 running proving a (crude) upper bound on its dependence on $\epsilon$. By~\cref{thm:direct-sum-decoding},
 the list decoding time
 $$
 \tilde{O}(C_{\beta,k,\eps_0} \cdot (\abs{W}+\calT_0(n))), 
 $$
 where $C_{k,\beta,\eps_0} = (\nfrac{6}{\eps_0})^{2^{O(\nfrac{k^3}{\beta^2})}}$. For our choices of parameters,
 this decoding time can be (crudely) bounded by $\tilde{O}(\exp(\exp(\poly(1/\epsilon))) \cdot N)$.
\end{proof}

The gentle \emph{list} decoding theorem above readily implies our main
result for \emph{unique} decoding if we are only interested in
$\widetilde{O}_{\epsilon}(N)$ decoding time without a more precise
dependence on $\epsilon$. We prove our main result,~\cref{theo:main},
for \emph{unique} decoding making more precise the dependence of the
running time on $\epsilon$.

\begin{proof}{Proof of~\cref{theo:main}}
  We proceed as in the proof of~\cref{theo:gentle_list_decoding}
  expect that we take $\beta=1/4$ in the list decoding radius $1/2-\beta$
  so that by performing list decoding we can recover all codewords in
  the unique decoding radius of the corrupted codeword regardless of
  the bias of the code $\Cc_{N,\epsilon,\alpha}$.

  We now recompute the running time.
  By~\cref{thm:direct-sum-decoding},
  the list decoding time
  $$
  \tilde{O}(C_{\beta,k,\eps_0} \cdot (\abs{W}+\calT_0(n))), 
  $$
  where $C_{k,\beta,\eps_0} = (\nfrac{6}{\eps_0})^{2^{O(\nfrac{k^3}{\beta^2})}}$. For our choices of parameters,
  this decoding time can be (crudely) bounded by $\tilde{O}(\exp(\exp(\polylog(1/\epsilon))) \cdot N)$.
\end{proof}

\subsection{Choosing the Base Code}

We now describe the (essentially) ``off-the-shelf'' base codes from
Guruswami and Indyk~\cite{GI05} which we use in Ta-Shma's
construction. We will need to prove that balanced codes can be easily
obtained from~\cite{GI05}. The argument is quite simple and borrows
from standard considerations related to the Zyablov and
Gilbert--Varshamov bounds.

\BaseCodeGI

\begin{theorem}[Guruswami--Indyk~\cite{GI05}, \cf Theorem 5]\label{theo:guruswami_indyk_cor}
   For every $\gamma > 0$ and for every $0 < R < 1$, there exists a family of binary
   linear \emph{concatenated} codes of rate $R$, which can be encoded in linear time and can be decoded
   in linear time from up to a fraction $e$ of errors, where
   \begin{equation}\label{eq:half_zyablov_bound}
      e ~\ge~ \max_{R < r < 1} \frac{(1-r-\gamma) \cdot H^{-1}_2(1-R/r)}{2}.
   \end{equation}
   $H^{-1}_2(x)$ is defined as the unique $\rho$ in the range $0 \le \rho \le 1/2$ satisfying $H_2(\rho) =x$.
   Every code in the family is explicitly specified given a constant sized binary linear code which
   can be constructed in probabilistic $O(\log(1/\gamma)R^{-1}/\gamma^4)$ or deterministic $2^{O(\log(1/\gamma)R^{-1}/\gamma^4)}$ time~\footnote{Note that
   dependence $\log(1/\gamma)R^{-1}/\gamma^4$ is slightly worse than that claimed in~\cite{GI05}, but not qualitatively relevant here nor in~\cite{GI05}.}.
\end{theorem}

As stated the codes in~\cref{theo:guruswami_indyk_cor} are not
necessarily balanced. We will see shortly that this can be easily
achieved by choosing balanced inner codes in the concatenated code
construction of Guruswami--Indyk~\cite{GI05}. To compute bounds on the
parameters, we will use the following property about binary entropy.

\begin{fact}[\cite{GRS23},\cf Lemma 3.3.9 abridged]\label{fact:inverse_of_bin_entropy}
  Let $H_2^{-1}$ be the inverse of the restriction of $H_2$ to
  $[0,1/2]$ (where $H_2$ is bijective).  For every small enough
  $\epsilon > 0$,
  $$
  H_2^{-1}(x-\epsilon^2/C_2) ~\ge~ H_2^{-1}(x) - \epsilon,
  $$
  where $C_2$ is a constant.
\end{fact}

\begin{proof}[Proof of~\cref{cor:base_code}]
  To achieve a final binary code of rate $R$, Guruswami and
  Indyk~\cite{GI05} concatenate an outer code of rate $r > R$ and
  distance $1-r-\gamma$ (over a non-binary alphabet of size $O_{\gamma}(1)$)
  with an inner binary linear code of rate $R/r$
  at the GV bound whose distance $\rho \in [0,1/2]$ satisfy $R/r =
  1-H_2(\rho)$ (since it is at the GV bound), or equivalently $\rho =
  H_2^{-1}(1-R/r)$.
  By choosing $\gamma= \Theta(\epsilon_0)$ and $R= \Theta(\epsilon_0^3)$ in~\cref{theo:guruswami_indyk_cor},
  the decoding error $e$ can be lower bounded by letting $r = \Theta(\epsilon_0)$ so
  that~\cref{fact:inverse_of_bin_entropy} implies that~\cref{eq:half_zyablov_bound} becomes
  $$
  e ~\ge~ \max_{R < r < 1} \frac{(1-r-\gamma) \cdot H^{-1}_2(1-R/r)}{2} ~\ge~ \frac{1}{4} - \epsilon_0.
  $$  
  To obtain codes that are $\epsilon_0$-balanced, we require that the
  inner codes used in this code concatenation not only
  lie on the Gilbert--Varshamov bound but are also balanced. It is
  well known that with high probability a random
  binary linear code at the GV bound designed to have minimum distance
  $1/2-\gamma/2$ also has maximum distance at most $1/2+\gamma/2$, \ie the
  code is $\gamma$-balanced. Therefore, we assume that our inner codes
  are balanced.

  For our concrete choices of parameters, $\rho =
  1/2-\Theta(\epsilon_0)$ and we also require the inner code to be
  $\Theta(\epsilon_0)$-balanced. Note that any non-zero codeword
  of the concatenated is obtained as follows: each
  of the $\ge (1-r-\gamma)$ non-zero symbols of the outer codeword
  is replaced by an inner codeword of bias  bias $\Theta(\epsilon_0)$
  and the remaining $\le r + \gamma$ zero symbols are mapped to zero
  (since the inner code is linear). Hence, the bias of the concatenated
  codeword is at most
  $$
  (1-r-\gamma) \cdot \Theta(\epsilon_0) ~+~ 1 \cdot (r + \gamma),
  $$
  which can be taken to be $\epsilon_0$ by suitable choices of hidden constants.
\end{proof}

\begin{remark}\label{remark:choice_of_base_code}
   Guruswami--Indyk~\cite{GI05} codes have several nice properties
   making them a convenient choice for base codes in Ta-Shma's construction,
   but they are not crucial here.
   We observe that for our purposes we could have started with any family of
   good binary linear codes admitting near-linear time encoding and decoding.
   From this family, we could boost its distance
   using a simpler version of Ta-Shma's construction (rounds I and II of~\cite{JQST20}[Section 8])
   and our near-linear time decoder~\cref{thm:direct-sum-decoding} for direct sum.
   This would result in an alternative family of linear binary
   $\epsilon_0$-balanced codes of rate $\Omega(\epsilon_0^{2+\alpha})$, for some arbitrarily
   small constant $\alpha >0$, that can be encoded and decoded in near-linear time.
   We also point out that for these base codes  any rate $\poly(\epsilon_0)$ suffices
   our purposes.
\end{remark}

To handle the technical requirement of a base code in Ta-Shma's
construction having a symbol multiplicity property
(\cf~\cref{def:sym_mult}), we use the following observation.

\begin{claim}\label{claim:dec_replication_lifting}
  Let $\Cc_0 \subseteq \F^n_2$ be an $\epsilon_0$-balanced linear code of dimension $D_0$.
  Suppose that $\Cc_0$ is uniquely decodable within (fractional) radius $\delta_0 \in (0,1]$ in time $\calT_0(n)$.
  Let $m \in \mathbb{N}$ and $\Cc \subseteq \F^{m\cdot n}_2$ be the code formed by replicating $m$ times each
  codeword from $\Cc_0$, \ie
  $$
  \Cc \coloneqq \set{z_1\cdots z_m \in \F^{m\cdot n}_2 \mid z_1=\cdots=z_m \in \Cc_0}.
  $$
  Then, $\Cc$ is an $\epsilon_0$-balanced linear code of dimension $D_0$ that
  can be uniquely decoded within (fractional) radius $\delta_0$ in time $m\cdot \calT_0(n) + \widetilde{O}(m^2 \cdot n)$.
\end{claim}

\begin{proof}
  The only non-immediate property is the unique decoding guarantees of $\Cc$.
  Given $\tilde{y} \in \F_2^{m\cdot n}$ within $\delta_0$ (relative) distance of $\Cc$.
  Let $\beta_i$ be the fraction of errors in the $i$th $\F_2^n$ component $\tilde{y}$.
  By assumption $\E_{i \in [m]} \beta_i \le \delta_0$, so there is at least one of such component
  that can be correctly uniquely decoded. 
  We issue unique decoding calls for $\Cc_o$ on each component $i \in [m]$. For each successful
  decoding say $z \in \Cc_0$, we let $y = z\ldots z \in \F_2^{m\cdot n}$ and check whether
  $\Delta(\tilde{y},y) \le \delta_0$ returning $y$ if this succeeds. Finally, observe that
  this procedure indeed takes at most the claimed running time.
\end{proof}




\chapter{Making AEL Amplification Achieve List Decoding Capacity} \label{chap:capacity}
What is the optimal error correction radius for a given rate $R \in (0,1)$? It is not difficult to use the Singleton bound to show that the error correction radius, even with list decoding, cannot be more than $1-R$, and random codes of rate $R$ are list decodable upto radius $1-R-\eps$. However, an explicit code family with such a strong error correction guarantee remained elusive until the works of Parvaresh and Vardy \cite{PV05} and Guruswami and Rudra \cite{GR08}, who showed that the folded Reed-Solomon codes provide such guarantees for a large enough folding parameter.

Since then, a number of codes that achieve list decoding capacity have been discovered, including some with better alphabet size, list size and/or decoding time \cite{GW11, Kop15, KMRZS17, GX22, KRSW23, GHKS24}. However, almost all of these continue to rely upon the interpolation based techniques, and are therefore based on algebra. 

In this chapter, we obtain new codes achieving list decoding capacity, based on non-algebraic properties such as spectral expansion. Before going into our results, we mention a few reasons why such codes are of interest.

\begin{enumerate}
\item We would like new techniques for studying list decodability other than interpolation in the hope that the new techniques will enjoy additional flexibility. There are numerous examples where expanders and other combinatorial operations have been used to replace algebra \cite{BSS04, Dinur07, Meir13}. As examples for such flexibility, one would like such codes to have features such as the LDPC property, linear-time unique decodability, etc. Moreover, a graph-based code achieving capacity would open up an avenue towards achieving list decoding capacity with truly linear-time decoding.
\item While the covering lemma of \cref{chap:framework} works for any code, are there properties of a specific code that allow for bounded list size beyond the Johnson bound? We know very few techniques for ensuring list decodability beyond the Johnson bound, and most of these are quite different from the Johnson bound argument \cite{GGR09, BL18}. Also, such codes often tend to involve significant random components, or are far from optimal rate-distance tradeoffs. One notable exception is the argument by Parvaresh and Vardy \cite{PV05}, which was crucial to the result of Guruswami and Rudra \cite{GR08}, that replaces bivariate interpolation by multivariate interpolation to get a smooth improvement in decoding radius. A combinatorial argument might offer more insights into decodability beyond Johnson bound, somewhat similar to how our covering lemma from \cref{chap:framework} provides a combinatorial explanation to Guruswami-Sudan list decoder for RS codes \cite{GS99}.
\item Understanding how to improve upon the Johnson bound via combinatorial arguments instead of multivariate interpolation might have implications for constructing codes achieving list decoding capacity over \emph{binary} alphabet. For binary alphabet, Ta-Shma codes \cite{TS17} achieve near-optimal tradeoff between rate and distance, however the best we know about their list decodability is the radius guaranteed by Johnson bound. Improving their list decoding radius all the way up to their distance would improve upon all existing explicit code constructions in terms of rate vs list decoding radius tradeoff.
\end{enumerate}

\section{Our Results}
Our main result is that the AEL amplification, when its inner code is chosen to be a capacity achieving code with constant list sizes and expansion is strong enough, is list decodable up to capacity with constant alphabet size and constant list size. This means that these codes have rate $R$ and the list size up to decoding radius $1-R-\eps$ is bounded by a constant dependent only $\eps$ and independent of the blocklength. However, unlike algebraic codes where the proof of list size being bounded often comes with a natural polynomial time algorithm, we do not yet know an efficient algorithm that would decode up to $1-R-\eps$. Since they are just AEL codes based on very strong spectral expanders, they are still linear-time unique decodable \cite{GI05} and polynomial time list decodable up to $1-\sqrt{R}-\eps$ (using results from \cref{chap:framework}).

Above guarantees are most useful when the decoding radius is close to $1-R$. We show that the AEL amplification can also be adapted to get explicit codes with list size 2 up to decoding radius $\frac{2}{3}(1-R)$. Such "higher order MDS" property for lists of size 2 was not known for explicit codes with constant sized alphabets. We note that $\frac{2}{3}(1-R) \geq 1-\sqrt{R}$ for $R\in [\frac{1}{4},1]$, and so for high rate codes, this already beats the Johnson bound!

Unfortunately, as we try to decode upto $\frac{k}{k+1}(1-R)$ for $k=3,4,\cdots$, the list size (and alphabet size) blows up rather quickly, and these codes are nowhere close to higher order MDS (that is, list size $k$) for $k>2$. The final dependence of list size, alphabet size and the degree of the graph is a tower function of height $\poly(1/\eps)$. However, we expect these parameters to improve to more reasonable functions of $\eps$ with a better proof technique.

\subsection{Overview of Techniques}
We briefly recall the Guruswami-Sudan approach to list decoding RS codes \cite{GS99}, and the subsequent improvement by Parvaresh and Vardy \cite{PV05} to rate vs list decoding radius tradeoff. Given a received word $g$, the Guruswami-Sudan algorithm learns a bivariate polynomial $Q_1(Y,X)$. We then argue that for any codeword $f(X)\in \calL(g,1-2\sqrt{R})$, it must hold that $Q_1(f(X),X)=0$ as a polynomial in $X$. Therefore, $f(X)$ can be found as a factor of the form $Y-f(X)$ using bivariate polynomial factorization algorithms for $Q_1$.

The change in \cite{PV05} (and also considered earlier by Coppersmith and Sudan \cite{CS03}) is to consider interleaved RS codes, and to interpolate to a multivariate polynomial from interleaved received words. We restrict our attention to 2-interleavings for simplicity, and in this case one interpolates to a polynomial $Q_2(Y_1, Y_2, X)$ in 3 variables. An analogous argument now shows that for any $f_1 \odot f_2 \in \calL(g,1-3R^{2/3})$, it must be the case that $Q_2(f_1(X), f_2(X), X)=0$ as a polynomial in $X$.

However, at this point, we run into the key difference between bivariate and multivariate cases. While the number of codewords $f(X)$ that satisfy $Q_1(f(X),X)=0$ is immediately bounded by the degree of $Y$ in $Q_1$, the number of codewords $f_1 \odot f_2$ that satisfy $Q_2(f_1(X),f_2(X),X)=0$ need not be constant, or even polynomial. In fact, if the RS code being interleaved has exponential list sizes at a certain radius, then so much the interleaved code. Therefore, polynomial list sizes at the decoding radius $1-3R^{2/3}$ would have strong implications for the open problem of whether RS codes have small list sized beyond the Johnson bound.

Nevertheless, \cite{PV05} impose a fixed algebraic condition between $f_1, f_2$, the codewords being interleaved, and this allows them to cut down the list size to constant. Let us ignore this preconditioning step for now, and focus on what we can learn from $Q_2(f_1(X),f_2(X),X)=0$. While this is not sufficient to extract all $(f_1,f_2)$ pairs, it does decrease the number of choices from $|\calC|^2$ to $\calO(|C|)$, using the Schwartz-Zippel lemma. 

We show that such a mild decrease in list size holds for interleaving of arbitrary codes. In fact, our proof follows the exact same structure as the proof of Schwartz-Zippel, despite the fact that there are no multivariate polynomials when dealing with interleaving of general codes! The argument can be extended to higher order interleaving as well. 

We can adapt these arguments so that starting from any near-MDS code $\calC$, one gets a code $\calC'$ which has a list size of $|\calC'|^{\eps}$ up to a decoding radius of $1-R-\eps$. Then, an argument of Rudra and Wootters \cite{RW15} about random subcodes shows that a random relationship between the codewords being interleaved can bring down the list size to near-optimal $\calO(1/\eps)$, while causing a negligible loss in rate. This answers an open question from \cite{RW15} about list decodability of randomly interleaved codes.

For our applications, we however wish to find this relationship among interleaved codewords explicitly. For RS codes, the Parvaresh-Vardy condition $f_2(X) = f_1(X)^d \mod E(X)$, for some irreducible $E(X)$ and $d$ large enough, is such a relationship among interleaved codewords. However, it is not clear what this relationship should be when the code being interleaved is obtained via AEL.

Faced with this obstruction, we instead use the fact that AEL can be seen as a sparsification of interleaving itself. Therefore, in some sense, the AEL amplification procedure has interleaving built-in! We use this connection to redo the combinatorial analog of Schwartz-Zippel on AEL instead of interleaved codes, and this leads to our capacity achieving codes.
\subsection{Future Work}
Our work leaves open several questions. Two natural questions that arise for these codes are an efficient algorithm to decode upto $1-R-\eps$, and better list sizes to avoid the tower-type dependence on $1/\eps$. Of special interest would be a generalization of known linear-time unique decoders to the list decoding setting, just as the algorithm of \cite{Gur11} can be seen as a generalization of the Berlekamp-Welch algorithm for unique decoding RS codes.

Secondly, one wonders if there is a combinatorial explanation to the excellent coding theoretic performance of codes based on polynomials over finite fields. For example, are there formal connections between our combinatorial analog of Schwartz-Zippel and the argument of Guruswami and Xing \cite{GX13} for list decoding RS codes evaluated on a subfield? Can we find an explanation for the differing behavior of RS codes and folded RS codes when it comes to list decoding radius? Can this help with an explicit evaluation set for RS codes so that they are list decodable up to their distance?

A broader question is whether other applications of algebra in pseudorandomness can also be replaced by expander graphs. Some of the applications indeed go via capacity achieving codes, but maybe we should look closer at the applications of Schwarz-Zippel lemma as well as other algebraic primitives for combinatorial statements hiding underneath.

There are also some key differences between our combinatorial bound on list size vs the multivariate polynomial of \cite{PV05}. In particular, the polynomial $Q_2(Y_1,Y_2,X)$ is a succinct (polynomial sized) object that contains all the information about the exponential sized list. Can we find a similar object that works for general codes? The analogous object for bivariate polynomials turned out to be the degree-1 marginals of a distribution over codewords, as seen in \cref{chap:framework}. Of course, the received word itself is such an object, but we would like it to be more structured, such as the low-degree trivariate polynomial $Q_2$, and this object could help us determine explicit pre-conditioning on interleaved codewords to get small list sizes.
\section{Inspiration from Schwartz-Zippel Lemma}

In this section, we present our combinatorial argument that yields the same list size bound as interpolation combined with Schwartz-Zippel lemma. These results can also be seen as a hierarchy of Johnson bounds. We start with the simplest case of order-2 interleaving of a code.
\begin{lemma}
    Let $\calC$ be a code over large alphabet $[q]$ with rate $R$ and distance $\Delta$. Let $\calC^{\odot 2}$ be the 2-interleaved code. Then, for any $g = g_1 \odot g_2 \in ([q]^2)^n$,
    \[
        \abs{\calL(g,1-(1-\Delta)^{3/4})} \leq 2(q-1)n |\calC l.
    \]
\end{lemma}

\begin{proof}
    Let
    \[
        \calL = \inbraces{f_1 \odot f_2\in \calC^{\odot 2} : \agr(f_1 \odot f_2, g_1\odot g_2) > (1-\Delta)^{3/4}\cdot n}
    \]
    We wish to prove that $|\calL|$ is at most $2(q-1)n| \calC |$.

    First, consider the codewords $f_1 \odot f_2 \in \calL$ such that $\agr(f_2, g_2) > \sqrt{1-\Delta} \cdot n$. There can be at most $|\calC| \times (q-1)n =(q-1)n| \calC |$ many such codewords by Johnson bound.

    For the remaining codewords in the list, it must be that $(1-\Delta)^{3/4}\cdot n < agr(f_2,g_2) \leq (1-\Delta)^{1/2} \cdot n$. Fix such an $f_2 \in \calC$, and let $S \sub [n]$ be the set of indices where $g_2$ and $f_2$ agree. Then, if $(f_1,f_2)\in \calL$, then $f_1$ and $g_1$ must agree on $>((1-\Delta)^{3/4} \cdot n$ positions even when restricted to the set of indices $S$. Let $|S| = (1-\Delta)^{\alpha} \cdot n$, with $\alpha \in [\frac{1}{2},\frac{3}{4})$.

    We will show that the number of $f_1\in \calC$ such that $f_1 \odot f_2 \in \calL$ is at most $(q-1)|S|$ via another appeal to Johnson bound. When the code is restricted to the set $S$, two distinct codewords $f_1,f'_1$ can agree on at most $(1-\Delta)n = (1-\Delta)^{1-\alpha} \cdot |S|$ coordinates, and any $f_1$ such that $f_1 \odot f_2\in \calL$ agrees with $g_1$ in $>(1-\Delta)^{3/4}\cdot n = (1-\Delta)^{3/4-\alpha} \cdot |S|$ coordinates. There can only be $(q-1)|S|$ many such $f_1$ codewords (for a fixed $f_2$) if 
    \[
        \frac{3}{4}-\alpha \leq \frac{1}{2}(1-\alpha) \implies \alpha\geq 1/2
    \]
    which is true since $\alpha \in [\frac{1}{2},\frac{3}{4})$. 
    
    Thus for every $f_2$ such that $(1-\Delta)^{3/4}\cdot n < \agr(f_2,g_2) \leq (1-\Delta)^{1/2} \cdot n$, there can be only $(q-1)|S| \leq (q-1)n$ many $f_1$ such that $(f_1,f_2)\in \calL$.
\end{proof}

Next we generalize this claim to higher order interleavings via induction.
\begin{lemma}
    Let $\calC$ be a code over large alphabet $[q]$ with rate $R$ and distance $\Delta$. Let $\calC^{\odot t}$ be the order-$t$ interleaved code. Then, for any $g = g_1\odot \cdots \odot g_t\in ([q]^t)^n$,
    \[
        \abs{\calL(g,1-(1-\Delta)^{1-\frac{1}{2^t}})} \leq t\cdot (q-1)n| \calC |^{t-1}
    \]
\end{lemma}

\begin{proof}
    The statement is true for $t=1$ by the Johnson bound. Now we assume it is true for $t-1$, and prove it for $t$.

    First, there are at most $| \calC |^{t-1}\cdot (q-1)n$ codewords $f_1 \odot f_2 \odot \cdots\odot f_t$ such that $\agr(f_t,g_t) > \sqrt{1-\Delta}\cdot n$.

    Now, fix $f_t$ to be such that $\frac{\agr(f_t, g_t)}{n} \in ((1-\Delta)^{1-1/2^t},(1-\Delta)^{1/2}]$. Let $S\sub [n]$ be the set of agreement indices between $f_t$ and $g_t$, with $|S| = (1-\Delta)^{\alpha}\cdot n$. Any two codewords of $C^{\odot t-1}$, say $f_1\odot \cdots \odot f_{t-1}$ and $f'_1\odot \cdots \odot f'_{t-1}$, agree in at most $(1-\Delta)n = (1-\Delta)^{1-\alpha}\cdot |S|$ coordinates. Also, if $f_1 \odot \cdots \odot f_t \in \calL$, then $g_1 \odot \cdots \odot g_{t-1}$ and $f_1 \odot \cdots \odot f_{t-1}$ agree on at least $(1-\Delta)^{1-1/2^t}\cdot n = (1-\Delta)^{1-1/2^t-\alpha} \cdot |S|$ coordinates in $S$. By the inductive hypothesis applied to $C^{\odot t-1}$, $f_t$ appears in at most $(t-1)\cdot (q-1)n | \calC |^{t-2}$ codewords in $\calL$ if the following is satisfied:
    \begin{align*}
        1-\frac{1}{2^t}-\alpha &\leq (1-\frac{1}{2^{t-1}})(1-\alpha) \\
        \iff 2^t-1-2^t \alpha &\leq (2^t-2)(1-\alpha) = 2^t-2-2^t\alpha+2\alpha \\
        \iff 1\leq 2\alpha
    \end{align*}
    The above holds because $\alpha \in [\frac{1}{2},1-\frac{1}{2^t})$. 
    
    Thus, for every $f_t$ we fixed above, there are at most $(t-1)\cdot (q-1)n | \calC |^{t-2}$ many $f_1 \odot \cdots  \odot f_{t-1}$ such that $f_1 \odot \cdots \odot f_t \in \calL$. The total number of codewords in $\calL$ therefore is,
    \begin{align*}
    	| \calC |^{t-1}\cdot (q-1)n &+ \abs{ \{f_t\in \calC \suchthat \agr(f_t,g_t) \in ((1-\Delta)^{1-1/2^t},(1-\Delta)^{1/2}] \} } \times (t-1)\cdot (q-1)n | \calC |^{t-2} \\
    	&\leq | \calC |^{t-1}\cdot (q-1)n + | \calC | \times (t-1)\cdot (q-1)n | \calC |^{t-2} \\
    	&= t | \calC |^{t-1}\cdot (q-1)n
    \end{align*}
\end{proof}

The version above roughly corresponds to interpolating with individual degrees of $Y_i$ bounded. We next prove a version that corresponds to interpolating to a multivariate polynomial such that the \emph{total degree} of $Y$ variables is at most 1. This is the version used by Guruswami in \cite{Gur11}, building on the exposition by \cite{Vadhan12}. Roughly speaking, this corresponds to using unique decoding instead of list decoding in the proof above.

\begin{lemma}
	Let $\calC$ be a code over alphabet $[q]$ with rate $R$ and distance $\Delta$. Let $\calC^{\odot t}$ be the order-$t$ interleaved code. Then, for any integer $s \in[1, t]$ and $g = g_1\odot \cdots \odot g_t\in ([q]^t)^n$,
    \[
        \abs{\calL(g,\frac{s}{s+1}\Delta)} \leq |\calC|^{s-1}
    \]
\end{lemma}

\begin{proof}
	Once again, we rely on induction on $t$. The case of $t=1$ is just unique decoding. We now assume the statement for $t-1$ and prove it for $t$.
	
	The case $s=1$ is again unique decoding, and henceforth we assume $s\geq 2$. Consider any $i\in [t]$ such that the closest codeword in $\calC$ to $g_i$ is at a distance at least $\frac{\Delta}{s+1}$ from $g_i$. If no such $i$ exists, then every $g_i$ has a codeword $\frac{\Delta}{s+1}$ close to it, which precludes the existence of another codeword in $\calL(g_i,\frac{s}{s+1}\Delta)$, and then $|\calL(g,\frac{s}{s+1}\Delta)| \leq 1$. So, we may assume that such an $i$ exists.
	
	We fix any $f_i \in \calL(g_i,\frac{s}{s+1}\Delta)$, and there are at most $|\calC|$ choices here. Let $S \sub [n]$ be the agreement set between $f_i$ and $g_i$, with $n(1-\frac{s}{s+1}\Delta) \leq |S|\leq n(1-\frac{\Delta}{s+1})$. Consider the code $\calC'$ obtained from $\calC$ by restricting the coordinate set to $S$. The distance $\Delta'$ of $\calC'$ is at least $\frac{\Delta n - (n-|S|)}{|S|}$. Let $\calC'^{\odot t-1}$ be the order-$(t-1)$ interleaving of $\calC'$, and let $g' = (g_1 \odot \cdots \odot g_{i-1} \odot g_{i+1} \odot \cdots \odot g_t)\Big\vert_S$ be the modification of $g$ where $g_i$ is removed and the coordinate set is restricted to $S$. Therefore, $g' \in ([q]^{t-1})^S$ and can be a received word for $\calC'^{\odot t-1}$.
	
	Recall that $f_i$ is fixed, and $S$ is defined based on $f_1$. 
	\begin{claim}\label{claim:interleave}
	For any $f_1,\cdots,f_{i-1},f_{i+1},\cdots,f_t$ such that $f = f_1\odot \cdots \odot f_t \in \calL(g,\frac{s}{s+1}\Delta)$, it holds that $f' = f_1\Big\vert_S \odot \cdots \odot f_{i-1}\Big\vert_S\odot f_{i+1}\Big\vert_S \odot \cdots \odot f_t\Big\vert_S$ belongs to $\calL(g',\frac{s-1}{s}\Delta')$. 
	\end{claim}
	Note that $f \in \calC^{\odot t}, f'\in \calC'^{\odot t-1}$ and $\calL(g',\frac{s-1}{s}\Delta')$ is defined using the code $\calC'^{\odot t-1}$. By the inductive hypothesis, the number of such $f'$ is at most $|\calC'|^{s-2} = |\calC|^{s-2}$, and so the total number of codewords in $\calL(g,\frac{s}{s+1}\Delta)$ is at most $|\calC| \times |\calC|^{s-1} = |\calC|^{s-1}$.
	\begin{proof}[Proof of \cref{claim:interleave}]
		The distance between $f$ and $g$ is at most $\frac{s}{s+1}\Delta \cdot n$, so that the distance between $f\Big\vert_S$ and $g\Big\vert_S$ is at most $\frac{s}{s+1}\Delta \cdot n - (n-|S|)$, since every coordinate in $\overline{S}$ has a disagreement between $f$ and $g$ due to $f_i$ and $g_i$. Further, removing $f_i \Big\vert_S$ and $g_i \Big\vert_S$ from $f \Big\vert_S$ and $g\Big\vert_S$ respectively, we get that the distance between $f'$ and $g'$ is also at most $\frac{s}{s+1}\Delta \cdot n - (n-|S|)$. We will be done if
		\begin{align*}
			\frac{s}{s+1}\Delta \cdot n - (n-|S|) \leq \frac{s-1}{s} \Delta' \cdot |S|
		\end{align*}
		Using the lower bound $\Delta' \geq \frac{\Delta n - (n-|S|)}{|S|}$, it suffices to prove
		\begin{align*}
			& \frac{s}{s+1}\Delta \cdot n - (n-|S|) \leq \frac{s-1}{s} \inparen{\Delta n - (n-|S|)} \\
			\iff & \frac{1}{s(s+1)} \Delta - \frac{1}{s}(1-\frac{|S|}{n})\leq 0 \\
			\iff & \frac{|S|}{n} \leq 1-\frac{\Delta}{s+1}
		\end{align*}
		which is true since $f_i$ and $g_i$ have distance at least $\frac{\Delta}{s+1}$.
	\end{proof}
\end{proof}
We note that similar ideas were also used in the work of \cite{GGR09} on list decodability of interleaved codes. There, the error locations after fixing a codeword $f_i$ in position $i$ were treated as erasures for the rest of the codewords. That is, every location where $f_i$ and $g_i$ disagree can be treated as an erasure for the entire interleaved code, if we insist on codewords with a fixed $f_i$. We will use this viewpoint as it simplifies some of the exposition when dealing with AEL.
\section{From interleaved codes to AEL}

Suppose the outer code has distance $\delta_{out}$, and the inner code has rate $R$, distance $1-R$ and has list size 2 up to the decoding radius $\frac{2}{3}(1-R)$. Random Reed-Solomon codes can be used as an inner code with these properties, as shown by \cite{BGM23}. Suppose $\lambda \leq \eps \cdot \delta_{out}$ so that distance of AEL code is at least $1-R - \eps$. We actually choose $\lambda$ to be even smaller, $\lambda \leq \eps \cdot \delta_{out} / C$ for a large constant $C$.

Suppose we are given $g\in (\F_q^d)^R$ such that 
\[
    \calL = \inbraces{ h\in \calC \suchthat \Delta_R(g,h) < \frac{2}{3}(1-R) - \eps}
\]
We wish to prove an upper bound on $|\calL|$. In fact, we will be proving that $|\calL| \leq 2$.

Given $g$, we obtain corresponding received words $\inbraces{g_{\ell}}_{\ell \in L}$, where each $g_{\ell} \in \F_q^d$, for left vertices. Let $\calL_{\ell}$ be the list of (inner) codewords around $g_{\ell}$ at distance $\frac{2}{3}(1-R)$. By assumption, $|\calL_{\ell}| \leq 2$ for all $\ell \in L$. We divide the set of left vertices $L$ into $L_0, L_1, L_2$ based on the list sizes obtained via this local decoding procedure.

For any $h\in \calL$, let $S_h \sub L$ be the set of vertices where $h_{\ell}$ and $g_{\ell}$ differ in more than $\frac{2}{3}(1-R)d$ coordinates.

\begin{claim}
    For every $h\in \calL$, $\frac{|S_h|}{n} \leq \delta_{out}/C$.
\end{claim}

\begin{proof}
    Let $T_h \sub R$ be the set of right vertices that are touched by at least one "error edge" from $S_h$. Note that $T_h$ vertices are errors between $g$ and $h$, and so $|T_h| \leq \frac{2}{3}(1-R) - \eps_2$. By the AEL argument,
    \begin{align*}
        |S_h| \frac{2}{3}(1-R) d \leq E(S_h, T_h) &\leq \frac{d}{n} |S_h|\cdot |T_h| + \lambda d n \\
        \frac{2}{3}(1-R) - \frac{\lambda}{|S_h|/n} &\leq |T_h| \leq \frac{2}{3}(1-R) - \eps \\
        \frac{\lambda}{|S_h|/n} &\geq \eps \\
        \frac{|S_h|}{n} &\leq \frac{\lambda}{\eps} \leq \frac{\delta_{out}}{C}
    \end{align*}
\end{proof}

\begin{observation}
    $L_0 \sub S_h$ for any $h \in \calL$.
\end{observation}
\begin{lemma}
    If $|L_2| \leq \frac{\delta_{out}}{4}$, then $|\calL| \leq 1$.
\end{lemma}
\begin{proof}
    Pick an element from every list in $L_1$ and $L_2$ arbitrarily, and any local codeword in $L_0$. The distance of string constructed thus from any codeword in the list is at most $|L_2| + |S_h| \leq \frac{\delta_{out}}{2}\cdot n$, and therefore there is only one codeword in the list.
\end{proof}
Henceforth, we assume that $|L_2| > \frac{\delta_{out}}{4}$.
For any $h \in \calL$, define $h' \in \calC_{in}^L$ as follows:
\[
    h'_{\ell} = \begin{cases}
        h_{\ell} \qquad \ell \not\in S_h \\
        f^* \qquad \ell \in L_0 \\
        f_{\ell} \qquad \ell \in S_h \backslash L_0, \text{ where } f_{\ell} \in \calL_{\ell} \text{ (arbitrarily)}
    \end{cases}
\]
This ensures that 
\[
    h' \in \prod_{\ell\in L_0} \inbraces{f^*} \times \prod_{\ell\in L_1} \calL_{\ell} \times \prod_{\ell \in L_2} \calL_{\ell}
\]
Note also that since $h$ and $h'$ differ only on $S_h$, $\Delta_L(h,h') \leq |S_h| \leq \frac{\delta_{out}}{C} \cdot n$. Thus there is a one-to-one mapping between $h$ and $h'$.
\begin{lemma}\label{lem:low_agreement}
For any $h_1, h_2 \in \calL$, if it holds that $h'_1$ and $h'_2$ agree on at least $3n\delta_{out}/C$ coordinates in $L_2$, then $h_1 = h_2$.
\end{lemma}
\begin{proof}
    If $h'_1$ and $h'_2$ agree on at least $3n\delta_{out}/C$ coordinates in $L_2$, it follows that they agree on at least $n\delta_{out}/C$ coordinates in $L_2 \backslash (S_{h_1} \cup S_{h_2})$. Let this agreement set in $L_2 \backslash (S_{h_1} \cup S_{h_2})$ be called $U$, with $|U| \geq n\delta_{out}/C$.

    By definition of $h'_1$ and $h'_2$, it also follows that $h_1$ and $h_2$ agree on $U$.

    Now we observe that for any $\ell \in U$,  $\Delta(h_{\ell}, g_{\ell}) \geq \frac{1}{3}(1-R)d$, otherwise $\ell$ would not be in $L_2$ at all. This means that fixing $h_{\ell}$ for $\ell \in U$ allows us to fix at least $\frac{1}{3}(1-R)d$ vertices on the right as error locations. Call this set of error location vertices as $V(\ell)$, where $V(\ell) \sub R$.

    Extending this argument to the entire set $U$, we may fix the union of $V(\ell)$ over all $\ell \in U$, which we call $V(U)$, to be error locations. By the standard AEL argument,
    \begin{align*}
        \frac{1}{3}(1-R)d |U| \leq |E(U,V(U))| \leq \frac{d}{n} |U| \cdot |V(U)| + \lambda d n \\
        \implies \frac{|V(U)|}{n} \geq \frac{1}{3}(1-R) - \frac{\lambda}{|U|/n} \geq \frac{1}{3}(1-R) - \eps
    \end{align*}

	Therefore, $h_1$ and $h_2$ share $\frac{1}{3}(1-R) - \eps$ fraction of error locations. We claim that the only way this can happen is if $h_1=h_2$. If not, then the distance between $h_1$ and $h_2$ is at most
	\[
		\frac{2}{3}(1-R)-\eps + \frac{2}{3}(1-R)-\eps - \frac{1}{3}(1-R) + \eps = 1-R-\eps
	\]
	contradicting the fact that the distance of $\AELC$ is $>1-R-\eps$.
%
\end{proof}

\cref{lem:low_agreement} shows that for any two distinct $h_1$ and $h_2$, their corresponding $h'_1$ and $h'_2$ differ in at least $|L_2| - \frac{3n\delta_{out}}{C}$ coordinates of $L_2$. That is, when restricted to $L_2$, the pairwise fractional distance between these strings is at least
\[
    1 - \frac{3n\delta_{out}}{C \cdot |L_2|} \geq 1 - \frac{12}{C}
\]
    which leaves room for only 2 such strings due to \cref{thm:strong_plotkin} when $C > 48$.

\section{Plotkin bound}

We will be needing this version of the Plotkin bound for the recursive structure of the proof. Since we couldn't find a bound in this form in the literature, we include a full proof.

\begin{theorem}\label{thm:plotkin}
Let $\calC \sub [K]^n$ be a collection of strings such that no $L$ strings in $\calC$ agree on a set of size $\beta \cdot n$. 
If $\beta \leq \frac{1}{K^{2(L-1)}}$, then $|\calC| \leq (L-1)^2 \cdot (1+K^{L-1}) $.
\end{theorem}

\begin{proof}
	Let $M = |\calC|$. Any $L$ strings in $\calC$ must differ in at least $(1-\beta)n$ positions. We will lower bound and upper bound the total number of positions where all possible $L$-tuples of strings differ.
	\begin{align*}
		M(M-1)(M-2)\cdots (M-L+1) (1-\beta)n &\leq \sum_{c_1, c_2, \cdots ,c_L \in \calC} \sum_{i=1}^n \inparen{1- \one_{ (c_1)_i = (c_2)_i = \cdots = (c_L)_i}} \\
		&= \sum_{i=1}^n \sum_{c_1, c_2, \cdots ,c_L \in \calC} \inparen{1- \one_{ (c_1)_i = (c_2)_i = \cdots = (c_L)_i}} \\
		&= \sum_{i=1}^n M^L - m_{i,1}^L - m_{i,2}^L - \cdots - m_{i,K}^L \\
		&\leq \sum_{i=1}^n M^L - K\cdot \inparen{\frac{m_{i,1} + m_{i,2} + \cdots + m_{i,K}}{K}}^L \\
		&= n \cdot M^L \inparen{1-\frac{1}{K^{L-1}}}
	\end{align*}
	
	Crudely bounding,
	\begin{align*}
		\inparen{1-\frac{L-1}{M}}^{L-1} (1-\beta) \leq 1-\frac{1}{K^{L-1}} \\
		\inparen{1-\frac{(L-1)^2}{M}} (1-\beta) \leq 1-\frac{1}{K^{L-1}}
	\end{align*}
	If $\beta \leq \frac{1}{K^{2(L-1)}}$, this can be simplified to
	\begin{align*}
		1-\frac{(L-1)^2}{M}\leq \frac{1}{1+\frac{1}{K^{L-1}}} \\
		\frac{1}{1+K^{L-1}}\leq \frac{(L-1)^2}{M} \\
		M \leq (L-1)^2 \cdot (1+K^{L-1})
	\end{align*}
\end{proof}

\section{List decoding up to Capacity}

\begin{definition}[$(\beta,\gamma)$-expander]
	Let $\beta,\gamma\in (0,1)$. A $d$-regular bipartite graph $G(L,R,E)$ with $|L|=|R|=n$ is a $(\beta,\gamma)$-expander if the following is true for every $\alpha \in (\gamma,1)$ and for every $S\sub L$ with $|S|\geq \beta\cdot n$: if for every vertex in $S$, $\alpha d$ edges among its neighborhood are colored red, then at least $(\alpha-\gamma)n$ vertices in $R$ have one or more red edges incident on it.
\end{definition}

\begin{proposition}
	The complete bipartite graph $K_{n,n}$ is a $(\frac{1}{n},0)$-expander.
\end{proposition}

\begin{proposition}
	The bipartite spectral expander with second largest normalized singular value $\lambda$ is a $(\beta,\gamma)$-expander if $\lambda \leq \beta \gamma$.
\end{proposition}

\begin{theorem}\label{thm:main_technical}
Suppose $k\geq 1$ is an integer. Additionally, suppose for every $\eta>0$, there exists an (inner) code $\calC_{in}$ that is list-decodable (with erasures) up to $\Delta - \frac{\eta}{2}$ with list size $M(\eta)$.

For every $\eps >0$, there exists a $\beta = \frac{1}{16\cdot (\tow_{(1+M(\eps))^3}(k-1))^2} \cdot \delta_{out}$ such that if $\AELC$ is a code obtained by passing $\calC_{out}$ through a $(\beta,\eps/6)$-expander, then for any received word $g\in(\Sigma_{in}^d \cup \{?\})^R$ with $s$ fraction of erasures, the number of codewords $h\in \AELC$ that satisfy $\frac{k+1}{k}\Delta(g,h)+s<\Delta-\eps$ is bounded by $\tow_{(1+M(\eps))^3}(k-1)$.
\end{theorem}

\begin{proof}
	The proof is by induction on $k$. The base $k=1$ is just unique decoding as long as the distance is at least $\Delta - \eps$, which it is if the graph $G$ is a $(\delta_{out},\eps)$-expander. The conditions in the theorem for $k=1$ provide for a $(\frac{\delta}{16},\frac{\eps}{6})$-expander which is stronger.
	
	Now assume $k>1$ and the statement is true up to $k-1$. Moreover, assume the graph $G$ is a $(\beta,\frac{\eps}{6})$-expander for some $\beta<\delta_{out}$ to be chosen later. Let $g\in (\Sigma_{in}^d \cup \{?\})^R$ be a received word with $s$ fraction of erasures, and let $\calL$ be the list of codewords around $g$ whose size we wish to bound. That is,
	\begin{align*}
		\calL &= \inbraces{ h\in \AELC \suchthat \inparen{\frac{k+1}{k}} \Delta_R(g,h) + s < \Delta - \eps} \\
		&= \inbraces{ h\in \AELC \suchthat  \Delta_R(g,h) < \frac{k}{k+1} \inparen{\Delta - s - \eps}}
	\end{align*}

	For $\li \in L$, let $g_{\li}$ be the local snapshot of $g$ as seen in the neighborhood of $\li$. Note that $g_{\li}$ will also have erasures inherited from $g$, and let the fraction of erasures in $g_{\li}$ be $s_{\li}$. Let $T\sub L$ be the set of left vertices where $s_{\li} > s+\eps/6 $ erasures. Then, by the $(\beta, \frac{\eps}{6})$-expander property of $G$, we get that $|T| \leq \beta n$.

	Similarly, we may bound the set of vertices where $g_{\li}$ has too many errors - but this set will depend on the codeword from $\calL$. Define $e \defeq \frac{k}{k+1} \inparen{\Delta - s - \eps}$. Let $h\in \calL$ be a codeword so that $\Delta_R(g,h) <e$, and let $S_h \sub L$ be the set of vertices so that $\Delta(g_{\li},h_{\li}) > e+\eps/6$. Again, by the $(\beta,\frac{\eps}{6})$-expander property of $G$, we get that $|S_h| \leq \beta n$.
	
	For every $g_{\li}$, consider the list $\calL_{\li}$ of $\calC_{in}$ codewords defined as
	\[
		\calL_{\li} = \inbraces{h_{\li} \in \calC_{in} \suchthat \inparen{\frac{k+1}{k}} \Delta(g_{\li},h_{\li}) + s_{\li} < \Delta - \frac{\eps}{2} }
	\]
	Using the inner code's list decodability up to capacity, $|\calL_{\li}| \leq M(\frac{\eps}{2})$ for every $\li \in L$. We divide $L$ into three sets $L_0,L_1,L_{>1}$ based on the size of $\calL_{\li}$ being $0,1$ or $>1$ respectively. Note that the sets $L_0,L_1,L_{>1}$ do not depend on $h\in \calL$, unlike $S_h$.
	
	\begin{lemma}\label{lem:adjust}
		To every $h\in \calL$, we can associate an $\tilde{h}$ with the following properties:
		\begin{enumerate}[(i)]
			\item $\tilde{h} \in \Pi_{\li\in L_0} \{f^*\} \times \Pi_{\li \in L_1 \cup L_{>1}} \calL_{\li}$.
			\item $\tilde{h}$ and $h$ only differ on $S_h \cup T$. In particular, $\Delta_L(\tilde{h},h) \leq 2\beta n$.
		\end{enumerate}
	\end{lemma}
	
	That is, every $h\in \calL$ is $2\beta$-close to the space $\Pi_{\li\in L_0} \{f^*\} \times \Pi_{\li \in L_1 \cup L_{>1}} \calL_{\li}$. From property (i), it follows that for any distinct $h_1,h_2\in \calL$, $\tilde{h}_1$ and $\tilde{h}_2$ can only differ on $L_{>1}$. We will next show that they must in fact differ a lot on $L_{>1}$, but this difference need not be pairwise. 
	
	Let $\tilde{h}\vert_{L_{>1}}$ denote the restriction of $\tilde{h}$ to coordinates in $L_{>1}$, and
	\[
		\widetilde{\calL} \defeq \inbraces{ \tilde{h}\vert_{L_{>1}} \suchthat h\in \calL} \sub \Pi_{\li \in L_{>1}} \calL_{\li}
	\]
	$\widetilde{\calL}$ can be seen as a collection of strings of length $|L_{>1}|$ over an alphabet of size $M(\frac{\eps}{2})$. We will use \cref{thm:plotkin} in conjunction with the next lemma to prove an upper bound on the size of $\widetilde{\calL}$, and therefore on the size of $\calL$.
	
	\begin{lemma}\label{lem:plotkin_disagreement}
		No set of $\tow_{(1+M(\eps))^3}(k-2)+1$ strings in $\widetilde{\calL}$ can agree on $(\tow_{(1+M(\eps))^3}(k-2)+3)\cdot \beta n$ positions.
	\end{lemma}

	With the choice of $\beta$ in the theorem statement, we get
	\begin{align*}
		\frac{\inparen{\tow_{(1+M(\eps))^3}(k-2)+3}\cdot \beta n}{|L_{>1}|} &\leq \frac{\inparen{\tow_{(1+M(\eps))^3}(k-2)+3}\cdot \beta n}{\frac{\delta_{out}n}{4}}\\		
		&= \frac{4\inparen{\tow_{(1+M(\eps))^3}(k-2)+3}\cdot}{\delta_{out}} \frac{1}{16\cdot (\tow_{(1+M(\eps))^3}(k-1))^2} \cdot \delta_{out} \\
		&= \frac{\inparen{\tow_{(1+M(\eps))^3}(k-2)+3}}{4(\tow_{(1+M(\eps))^3}(k-1))^2}\\
		&\leq \frac{\tow_{(1+M(\eps))^3}(k-2)}{\tow_{(1+M(\eps))^3}(k-1)^2}\\
		&\leq \frac{1}{\tow_{(1+M(\eps))^3}(k-1)}\\
		&\leq \frac{1}{(1+M(\eps))^{3\tow_{(1+M(\eps))^3}(k-2)}} \\
		&\leq \frac{1}{M(\eps)^{2\tow_{(1+M(\eps))^3}(k-2)}}.
	\end{align*}
	Therefore, \cref{thm:plotkin} can be applied to get that the size of $\widetilde{\calL}$, and therefore the size of $\calL$, is at most 
	\begin{align*}
		&\tow_{(1+M(\eps))^3}(k-2)^2\cdot (1+M(\eps)^{\tow_{(1+M(\eps))^3}(k-2)}) \\
		\leq& (1+M(\eps))^{2\tow_{(1+M(\eps))^3}(k-2)}\cdot (1+M(\eps))^{\tow_{(1+M(\eps))^3}(k-2)} \\
		=& (1+M(\eps))^{3\tow_{(1+M(\eps))^3}(k-2)} \\
		=& \tow_{(1+M(\eps))^3}(k-1)
	\end{align*}
	completing the induction step.
\end{proof}

\begin{remark}
	We always use $M(\eps)$ as the local list size bound on $|\calL_{\li}|$ in the proof above, and there is room for much better bounds for small $k$. However, this will not qualitatively change the tower-type bounds. To avoid unnecessary distraction from list size changing with $k$, we choose to work with a single list-size bound for the inner code - the bound one gets when $\frac{\eps}{2}$-close to capacity.
\end{remark}

\begin{corollary}
	For any $R\in (0,1)$ and $\eps>0$, there is an infinite family of codes based on AEL amplification starting with an arbitrary high rate, constant distance code with the following properties:
	\begin{enumerate}[(i)]
		\item The code has rate $R$ and distance at least $1-R-\eps$.
		\item The code is list decodable up to $1-R-\eps$ with list size bounded by a tower of base $(1/\eps)^{\calO(1/\eps)}$ and height $\calO(1/\eps)$.
		\item The alphabet size of the code is bounded by a tower of base $(1/\eps)^{\calO(1/\eps)}$ and height $\calO(1/\eps)$.
	\end{enumerate}
\end{corollary}

\begin{proof}
	We use \cref{thm:main_technical} with $s=0$ and $k=1/\eps$. Using the folded Reed-Solomon codes, we can take $M(\eps) = (1/\eps)^{\calO(1/\eps)}$.
\end{proof}

\begin{proof}[Proof of \cref{lem:adjust}]
\[
    \tilde{h}_{\li} = \begin{cases}
        h_{\li} \qquad \li \not\in S_h \cup T\\
        f^* \qquad \li \in L_0 \\
        f_{\li} \qquad \li \in (S_h\cup T) \backslash L_0, \text{ where } f_{\ell} \in \calL_{\ell} \text{ (arbitrarily)}
    \end{cases}
\]
\end{proof}

\begin{proof}[Proof of \cref{lem:plotkin_disagreement}]
	Let $a = \tow_{(1+M(\eps))^3}(k-2)+1$. Towards a contradiction, we assume that there are $a$ strings $f_1, f_2,\cdots ,f_a$ in $\widetilde{\calL}$ that agree on $(a+2)\beta n$ positions in $L_{>1}$. Suppose the strings $f^{(1)}, f^{(2)},\cdots ,f^{(a)}$ in $\widetilde{\calL}$ are restrictions of $\tilde{h}^{(1)},\tilde{h}^{(2)},\cdots ,\tilde{h}^{(a)}$ to coordinates in $L_{>1}$. Then for the codewords $h^{(1)},h^{(2)},\cdots ,h^{(a)}$ in $\calL$, it holds that $\tilde{h}^{(1)},\tilde{h}^{(2)},\cdots ,\tilde{h}^{(a)}$ agree on $(a+2)\beta n$ positions in $L_{>1}$. Let this agreement set be $A \sub L_{>1}$.
	
	Since $h$ and $\tilde{h}$ only differ on $S_h \cup T$, the codewords $h^{(1)},h^{(2)},\cdots ,h^{(a)} \in \calL$ agree on the set $U \defeq A\setminus \inparen{S_{h^{(1)}} \cup S_{h^{(2)}} \cup \cdots \cup S_{h^{(a)}} \cup T}$, which is of size at least $(a+2)\beta n - a\cdot \beta n - \beta n  = \beta n$. 
	
	Let $h_{\li}$ be the common value of $h^{(1)},h^{(2)},\cdots ,h^{(a)}$ on an $\li \in U$. Since $\li\in A\sub L_{>1}$, there exists an $h'_{\li}\in \calL_{\li}$ distinct from $h_{\li}$. By the triangle inequality,
	\begin{align*}
		\Delta(g_{\li},h_{\li}) + \Delta(g_{\li},h'_{\li}) &\geq \Delta(h_{\li},h'_{\li}) \geq \Delta - s_{\li} \\
		\Delta(g_{\li},h_{\li}) &\geq \Delta - s_{\li} - \Delta(g_{\li},h'_{\li}) \\
		\Delta(g_{\li},h_{\li}) &\geq \Delta - s_{\li} - \frac{k}{k+1}\inparen{\Delta - s_{\li} - \frac{\eps}{2}} \\
		\Delta(g_{\li},h_{\li}) &\geq \frac{\Delta - s_{\li}}{k+1} + \frac{k}{k+1}\cdot \frac{\eps}{2}
	\end{align*}
	Therefore, every $\li \in U$ identifies at least $\frac{\Delta- s_{\li}+(k\eps/2)}{k+1}\cdot d \geq \frac{\Delta- s - (\eps/6)+(k\eps/2)}{k+1}\cdot d$ edges in its neighborhood such that the right vertices touched by these edges are common error locations (with respect to $g$) for all $h^{(1)},h^{(2)},\cdots h^{(a)}$. Let us call this set of error locations in $R$ identified by $\li$ as $V(\li)$. Define $V(U) \defeq \bigcup_{\li \in U} V(\li)$. Since $|U| \geq \beta n$, using the $(\beta,\eps/6)$-expander property, we conclude that $|V(U)| \geq \frac{\Delta- s - (\eps/6)+(k\eps/2)}{k+1} - \frac{\eps}{6}$. 
	
	In conclusion, we have that $V(U)$ is a common error location set for $h^{(1)},h^{(2)},\cdots h^{(a)}$, and $\frac{|V(U)|}{n} \geq \frac{\Delta- s - (\eps/6)+(k\eps/2)}{k+1} - \frac{\eps}{6}$.
	
	Next we construct a $g' \in \inparen{\Sigma_{in}^d \cup \{?\}}^R$ with $s'\cdot n$ erasures such that all $h^{(1)}, h^{(2)},\cdots ,h^{(a)}$ satisfy $\inparen{\frac{k}{k-1}}\Delta_R(g',h^{(i)}) + s' < \Delta-\eps$. This would contradict the inductive hypothesis since $a > \tow_{(1+M(\eps))^3}(k-2)$. The new $g'$ is simply $g$ with the modification that the symbols in the set $V(U)$ are erased. Clearly, $s' = s+|V(U)|$, and since $V(U)$ was the location of errors between any $h^{(i)}$ and $g$, it also follows that $\Delta(g',h^{(i)}) = \Delta(g,h^{(i)}) - |V(U)|$.
	
	\begin{align*}
		\inparen{\frac{k}{k-1}}\Delta_R(g',h^{(i)}) + s' &= \inparen{\frac{k}{k-1}}\inparen{\Delta_R(g,h^{(i)}) - |V(U)|} + (s+|V(U)|) \\
		&= \inparen{\frac{k}{k-1}}\Delta_R(g,h^{(i)}) - \frac{|V(U)|}{k-1} + s \\
		&\leq \inparen{\frac{k}{k-1}}\Delta_R(g,h^{(i)}) - \frac{\Delta - s -\eps/6 + (k\eps/2)}{k^2-1} +\frac{\eps/6}{k-1} + s \\
		&\leq \inparen{\frac{k}{k-1}}\inparen{\frac{k}{k+1}}\inparen{\Delta -s-\eps} - \frac{\Delta - s -\eps/6 + (k\eps/2)}{k^2-1} +\frac{\eps/6}{k-1} + s \\
		&= \Delta - s - \frac{\eps}{k^2-1}\inparen{ k^2-\frac{1}{6} + \frac{k}{2} - \frac{k+1}{6}} + s \\
		&\leq \Delta - \frac{\eps}{k^2-1}\inparen{ k^2 + \frac{2k-2}{6}} \\
		&< \Delta - \frac{\eps}{k^2-1}\inparen{ k^2} < \Delta - \eps
	\end{align*}
 \end{proof}

\chapter{Improved list size bounds for Folded Reed-Solomon Codes} \label{chap:list_size}
In the last chapter, we came up with new codes that achieve list decoding capacity, albeit not efficiently. As mentioned then, the first codes to achieve the list decoding capacity were the folded RS codes \cite{GR08}. Originally, their list size was proven to be $n^{\calO(1/\eps)}$, where $n$ is the blocklength and $\eps$ is the gap to capacity. In fact, \cite{Gur11} proved a stronger statement that the list is contained in an affine subspace of dimension $1/\eps$. 

This list size was brought down over time with subspace evasive sets \cite{DL12, BAS14}, or combinatorially bounding intersections of Hamming balls and affine subspaces \cite{KRSW23, Tamo24}. \cref{tab:list_size_capacity} lists some of these improvements. We note that there are other explicit codes achieving list decoding capacity based on multiplicity codes and algebraic-geometric (AG) codes \cite{GX13, GX22, GRZ21}. but to the best of our knowledge, the state of the art list size for any explicit capacity achieving code remains $(1/\eps)^{\calO(1/\eps)}$.

\section{Our Results}
We extend the above line of work to improve the list size of folded RS codes to $\calO(1/\eps^2)$, and thereby improving the state of the art. First, we give an elementary proof that generalizes the results of \cite{KRSW23, Tamo24}. This is again based on upper bounds on the intersection of Hamming balls and affine subspaces, and gives the same asymptotic bound of $(1/\eps)^{\calO(1/\eps)}$ that was known before. Then, for the specific case of folded RS codes, we improve this analysis to get a list size of $\calO(1/\eps^2)$.

We use a bottom up inductive proof, that gives us precise bounds on the list size for fixed decoding radii of the form $\approx \frac{k}{k+1}(1-R)$.
\begin{theorem}
For $t$-folded Reed Solomon codes, and any integer $k \in [t]$,
\[
	\abs{\calL \inparen{g,\frac{k}{k+1} \cdot \inparen{1-\frac{t}{t-k+1}R}}} \leq (k-1)^2 + 1
\]
\end{theorem}
By choosing $t > k/\eps$, we get that the list size for decoding up to $\frac{k}{k+1}(1-R-\eps)$ is at most $(k-1)^2+1$. For example, if we had a fixed budget of being able to deal with an output list of size 50, this theorem shows that we can approach a decoding radius of $\frac{8}{9}(1-R)$ by increasing $t$. We also note that the decoding radius of $\frac{k}{k+1}(1-R)$ is larger than the Johnson bound $1-\sqrt{R}$ whenever $R\geq \frac{1}{k^2}$.

\begin{table}
\centering
\begin{tabular}{|m{5cm}|c|c|c|}
    \hline
    Code & List size & Explicit? & Reference \\ [0.5ex]
    \hline
    \hline
     Random code & $1/\epsilon$ & Non-explicit & \cite{ZP81} \\ [0.5ex]
     \hline
     Random linear code & $1/\epsilon$ & Non-explicit & \cite{AGL24}\\ [0.5ex]
     \hline
     Randomly evaluated Reed-Solomon & $1/\epsilon$ & Non-explicit & \cite{BGM23}\\ [0.5ex]
     \hline
     \hline
     Folded Reed-Solomon & $n^{1/\epsilon}$ & Explicit & \cite{GR08}\\
     \hline
     Subspace-evasive subcode of Folded RS& $(1/\epsilon)^{1/\epsilon}$ & Explicit & \cite{DL12} \\ [0.5ex]
     \hline
     Folded Reed-Solomon
      & $(1/\epsilon)^{\frac{1}{\epsilon}\cdot \log(\frac{1}{R})}$ & Explicit & \cite{KRSW23}\\ [0.5ex]
     \hline
     \qquad \texttt{"} & $(1/\epsilon)^{1/\epsilon}$ & Explicit & \cite{Tamo24}\\  [0.5ex]
     \hline
     \qquad \texttt{"} & $1/\epsilon^2$ & Explicit & This work\\  [0.5ex]
     \hline
\end{tabular}
\caption{List sizes for some codes of rate $R$ at decoding radius $1-R-\eps$.}
\label{tab:list_size_capacity}
\end{table}

\subsection{Future Work}
The original motivation for this work was to apply the ideas of \cite{KRSW23, Tamo24} to the AEL-based capacity achieving codes of \cref{chap:capacity}. It remains to be seen whether the techniques of this chapter can help us avoid the large list sizes for those codes.

One advantage of the arguments of \cite{KRSW23, Tamo24} is that they immediately suggest randomized algorithms to find the list in linear time, given a basis for the affine subspace. 
One wonders whether our proof technique can be used to give a deterministic near-linear time algorithm to obtain the list given a basis for the affine subspace in which it is contained. If true, this would give a near-linear time \emph{deterministic} algorithm for decoding folded RS codes using the work of \cite{GHKS24}.

Indeed, when decoding up to $\frac{2}{3}(1-R)$, which means we are dealing with a 1-dimensional affine subspace, a simple near-linear time deterministic algorithm can be obtained. If the affine subspace is $\{f_0 + \alpha f_1 \suchthat \alpha \in \F_q\}$, and the received word is $g$, we look at the most frequent values appearing among $\{\frac{g(i)-f_0(i)}{f_1(i)} \}$ over $i$ such that $f_1(i)\neq 0$. This avoids having to try all possible values in $\F_q$ for $\alpha$. Can this idea be generalized to higher dimensional affine subspaces?

Coming back to combinatorial bounds, can a better analysis allow us to improve the list size to optimal $(k-1)+1=k$ instead of $(k-1)^2+1$? We do not know of any explicit constructions for such higher-order MDS codes when $k\geq 3$.

Finally, the notion of Wronskian determinants is tailored to the algebraic structure of folded RS and multiplicity codes. Can we generalize it to general linear codes, and what further applications does it have?
\section{Intersection of affine subspace and Hamming balls}

In this section, we show that the intersection of a low-dimensional affine subspace and a Hamming ball cannot be too large for any code, giving alphabet-independent bounds on the list size. Let us start with the easiest case where we show that a 1-dimensional affine subspace (essentially, a line) intersects Hamming balls of radius $\frac{2\Delta}{3}$ in at most 2 places.

\begin{lemma}
	Let $\calC$ be a linear code of distance $\Delta$ and blocklength $n$ over alphabet $\F_q$, and let $\calH \sub \calC$ be an affine subspace of dimension 1. Then, for any $g\in \F_q^n$,
	\[
		\abs{\calH \cap \calL \inparen{g,\frac{2\Delta}{3} } } \leq 2.
	\]
\end{lemma}

\begin{proof}
	Let $\calH = \{f_0 + \alpha \cdot f_1 \suchthat \alpha \in \F_q\}$ for some $f_0$ and $f_1$ in $\calC$, and let $S\sub [n]$ be the set of coordinates where $f_1$ is non-zero. Clearly, $|S| \geq \Delta \cdot n$.
	
	Let $S_h \sub S$ denote the set of coordinates in $S$ where $g$ and $h\in \calH$ agree. Note that for any two distinct $h_1,h_2 \in \calH$, they differ on every coordinate in $S$. This means that for any distinct $h_1, h_2 \in \calH$, the sets $S_{h_1}$ and $S_{h_2}$ are disjoint.
	
	Now for the sake of contradiction, assume there are three codewords $h_1,h_2,h_3 \in \calH \cap \calL\inparen{g,\frac{2\Delta}{3}}$. Then for at least one of these $h_i$, its $S$-agreement with $g$ must be small so that $|S_{h_i}| \leq \frac{|S|}{3}$. For this $h_i$, it therefore also holds that its disagreement with $g$ is at least $\frac{2|S|}{3} \geq \frac{2\Delta}{3}$, which contradicts $h_i \in \calL\inparen{g,\frac{2\Delta}{3}}$.
\end{proof}

With essentially the same proof, this lemma can be generalized to larger radii as long as we are still working with a line. We will need this version for higher dimensional $\calH$.

\begin{lemma}\label{lem:dim1}
	Let $\calC$ be a linear code of distance $\Delta$ and blocklength $n$ over alphabet $\F_q$, and let $\calH \sub \calC$ be an affine subspace of dimension 1. Then, for any $g\in \F_q^n$,
	\[
		\abs{\calH \cap \calL \inparen{g,\frac{k}{k+1} \Delta} } \leq k.
	\]
\end{lemma}
Next, we prove a list size bound of 12 for 2-dimensional affine planes when decoding up to $\frac{3\Delta}{4}$. 

\begin{lemma}\label{lem:dim2}
	Let $\calC$ be a linear code of distance $\Delta$ and blocklength $n$ over alphabet $\F_q$, and let $\calH \sub \calC$ be an affine subspace of dimension 2. Then, for any $g\in \F_q^n$,
	\[
		\abs{\calH \cap \calL \inparen{g,\frac{k}{k+1} \Delta} } \leq k(k+1).
	\]
\end{lemma}

\begin{proof}

Denote $\calH_g = \calH \cap \calL \inparen{g,\frac{k}{k+1} \Delta}$.

	As before, let $\calH = \{f_0 + \alpha_1 \cdot f_1 + \alpha_2 \cdot f_2 \suchthat \alpha_1,\alpha_2 \in \F_q\}$ for some $f_0, f_1$ and $f_2$ in $\calC$, and let $S\sub [n]$ be the set of coordinates where at least one of $f_1$ and $f_2$ is non-zero. As before, we define $S_h \sub S$ to be the set of coordinates in $S$ where $g$ and $h \in \calH$ agree.

	Next, we would like an analog of the disjointness property for agreement sets $\{S_h\}_{h\in \calH}$. We claim that any coordinate $i \in S$ will appear in at most $k$ sets in $\{S_h\}_{h\in \calH_g}$. This is because every $h\in \calH_g$ whose $S_h$ contains $i$ must have $h_i = g_i$, and so the collection of these $h$ are restricted to a 1-dimensional affine subspace inside $\calH$. Appealing to \cref{lem:dim1}, the number of such $h$ is at most $k$. Therefore,
	\[
		\sum_{h \in \calH_g} |S_h| \leq k \cdot |S|.
	\]
	It is easy to observe that every $h \in \calH_g$ must have $|S_h| > \frac{|S|}{k+1}$. If not, $g$ and $h$ disagree on at least $ \frac{k}{k+1} |S|$ positions, which is at least $\frac{k}{k+1} \Delta n$, contradicting $h\in \calL \inparen{g,\frac{k}{k+1} \Delta}$. Combining the two,
	\begin{gather*}
		k \cdot |S| \geq \sum_{h \in \calH_g} |S_h| > \sum_{h \in \calH_g} \frac{|S|}{k+1} = |\calH_g| \frac{|S|}{k+1} \\
		|\calH_g| < k(k+1)
	\end{gather*}
\end{proof}
Finally, we prove the general case via induction whose base cases were the lemmas above.

\begin{lemma}\label{lem:dimd}
	Let $\calC$ be a linear code of distance $\Delta$ and blocklength $n$ over alphabet $\F_q$, and let $\calH \sub \calC$ be an affine subspace of dimension $d$. Then, for any $g\in \F_q^n$,
	\[
		\abs{\calH \cap \calL \inparen{g,\frac{k}{k+1} \Delta} } \leq k(k+1)^{d-1}.
	\]
\end{lemma}

\begin{proof}
	The proof is very similar the proof of \cref{lem:dim2}. Each coordinate in $S$ will appear in at most $k(k+1)^{d-2}$ sets out of $\{S_h\}_{h\in \calH_g}$. Moreover, each $|S_h| > \frac{|S|}{k+1}$ due to same reason as before. Combining, we get
	\[
		|\calH_g| < k(k+1)^{d-2} \cdot (k+1) = k (k+1)^{d-1}
	\]
\end{proof}

\section{Getting more out of the Folded RS code}

The key idea we used in the previous section was that fixing any coordinate to be in the agreement set reduces the search space dimension by 1. However, here we only used agreement of $g$ with a Reed-Solomon codeword, whereas we have the opportunity to decrease the dimension much more by using the agreement of $g$ with a codeword on the \emph{folded} symbol. In an ideal case, such a fixing will uniquely determine the codeword, giving us disjointness of agreement sets as in the case of \cref{lem:dim1} and an optimal list size.

Unfortunately, the set of $t$ constraints imposed by a $t$-folded symbol need not be linearly independent. In fact, there might not even be $d$ linearly independent constraints (recall that $t$ is typically chosen so that $t \gg d$), which is what would suffice to pin down a codeword. However, these linear dependencies can be bounded in number globally using the Wronskian of (a basis of) the affine subspace we are working with.

Let us set some notation for folded RS codes. Let $q > n$ and $m = Rn$ be the field size and degree parameters respectively. The $t$-folded Reed-Solomon code $\calC$ is then of rate $R$, distance at least $1-R$, alphabet $\F_q^t$ and blocklength $N = n/t$. We assume that this folding is according to a primitive element $\gamma$ of $\F_q$. For this code, we will denote the list of codewords in a ball of radius $\delta$ around $g\in (\F_q^t)^N$ by $\calL(g,\delta)$.

Let $\calH$ be an affine subspace of $\F_q[X]^{<m}$ with dimension $d$, so that there exist vectors $h_0, h_1,h_2,\cdots ,h_d$ such that
\[
	\calH = \inbraces{h_0 + \sum_{j=1}^d \alpha_j h_j \suchthat \forall j\in [d], \alpha_j \in \F_q }
\]
Moreover, the set of polynomials $\inbraces{h_1,h_2,\cdots ,h_d}$ is linearly independent over $\F_q$.

The condition that a polynomial $h = h_0 + \sum_{j=1}^d \alpha_j h_j$ agrees with $g$ on position $i\in [N]$ after folding can be written as the collection of $t$ equations: 
\[
	 \forall j\in [t], \quad h(\gamma^{(i-1)t+j-1}) = g(\gamma^{(i-1)t+j-1})
\]
Writing as a linear system,

\begin{align*}
\begin{bmatrix}
h_1(\gamma^{(i-1)t}) & h_2(\gamma^{(i-1)t}) & \cdots & h_d(\gamma^{(i-1)t})\\
h_1(\gamma^{(i-1)t+1}) & h_2(\gamma^{(i-1)t+1}) & \cdots & h_d(\gamma^{(i-1)t+1})\\
\vdots & \vdots & \cdots & \vdots \\
\vdots & \vdots & \cdots & \vdots \\
h_1(\gamma^{(i-1)t+t-1}) & h_2(\gamma^{(i-1)t+t-1}) & \cdots & h_d(\gamma^{(i-1)t+t-1})
\end{bmatrix}
\begin{bmatrix}
\alpha_1 \\
\alpha_2 \\
\vdots \\
\alpha_d
\end{bmatrix}
=
\begin{bmatrix}
(g-h_0)(\gamma^{(i-1)t}) \\
(g-h_0)(\gamma^{(i-1)t+1}) \\
\vdots \\
\vdots \\
(g-h_0)(\gamma^{(i-1)t+t-1})
\end{bmatrix}
\end{align*}

Let us call the $t\times d$ matrix appearing above as $A_i$ for $i\in [N]$, and denote $r_i = \rank(A_i)$. If $r_i$ is always $d$, that is $A_i$ is always full rank, then each agreement $h_i$ and $g_i$ would fix all $\alpha_j$ for $j\in [d]$, and we would get the best case scenario where all agreement sets must be disjoint. However, this need not be true. Guruswami and Kopparty used folded Wronskian determinants to show that a weakening of this statement is true in an average sense globally. They wrote this in the language of strong subspace designs, and for completeness we present their proof in our simplified setting.

We first start with the following folded Wronskian criterion for linear independence, whose proof can be found in \cite{GK16}.

\begin{lemma}\label{lem:wronskian}
	Let $\gamma \in \F_q^*$ be a generator. The polynomials $p_1,p_2,\cdots ,p_d \in \F_q[X]^{<m}$ are linearly independent over $\F_q$ if and only if the determinant
	\[
		\begin{pmatrix}
			p_1(X) & p_2(X) & \cdots & p_d(X) \\
			p_1(\gamma X) & p_2(\gamma X) & \cdots & p_d(\gamma X) \\
			\vdots & \vdots & \vdots & \vdots \\
			p_1(\gamma^{d-1} X) & p_2(\gamma^{d-1}X) & \cdots & p_d(\gamma^{d-1}X)
		\end{pmatrix}
	\]
	is non-zero as a polynomial in $\F_q[X]$.
\end{lemma}
Next, we use the lemma above to bound the sum of "rank deficit" over all coordinates.
\begin{theorem}[Guruswami-Kopparty \cite{GK16}]
$\sum_{i=1}^N (d-r_i) \leq \frac{d\cdot (m-1)}{t-d+1}$.
\end{theorem}

\begin{proof}
	We start with instantiating \cref{lem:wronskian} with $p_j = h_j$ for $j\in [d]$, which are linearly independent polynomials used in the definition of $\calH$. By \cref{lem:wronskian}, the determinant of the following matrix
	\[
		H(X) \defeq \begin{bmatrix}
			h_1(X) & h_2(X) & \cdots & h_d(X) \\
			h_1(\gamma X) & h_2(\gamma X) & \cdots & h_d(\gamma X) \\
			\vdots & \vdots & \vdots & \vdots \\
			h_1(\gamma^{d-1} X) & h_2(\gamma^{d-1}X) & \cdots & h_d(\gamma^{d-1}X)
		\end{bmatrix}
	\]
	is non-zero. Denote this determinant by $D(X) = \det(H(X))$. Since each $h_i$ is of degree at most $m-1$, we note that $D(X)$ is a polynomial of degree at most $(m-1)d$, so that the number of zeros of $D(X)$ (with multiplicity) is bounded by $(m-1)d$. Therefore, it suffices to show that the number of zeros of $D(X)$ is at least $(t-d+1)\cdot \sum_{i=1}^N (d-r_i)$.
	
	In fact, we will describe the exact set of zeros with their mutliplicities that illustrates this. The next claim immediately completes the proof. Note that we say that a non-root is a root with multiplicity 0. 
	\begin{claim}\label{claim:root_with_mult}
		For every $i\in [N]$, for every $j\in [t-d+1]$, $\gamma^{(i-1)t+j-1}$ is a root of $D(X)$ with multiplicity at least $d-r_i$.
	\end{claim}
	\begin{proof}[Proof of \cref{claim:root_with_mult}]
		Recall that $r_i$ is the rank of matrix $A_i$. For $j\in [t-d+1]$, let $A_{ij}$ denote the $d\times d$ submatrix of $A_i$ formed by selecting all $d$ columns and rows from $j$ to $j+d-1$. That is,
		\begin{align*}
A_{ij} = \begin{bmatrix}
h_1(\gamma^{(i-1)t+j-1}) & h_2(\gamma^{(i-1)t+j-1}) & \cdots & h_d(\gamma^{(i-1)t+j-1})\\
h_1(\gamma^{(i-1)t+j}) & h_2(\gamma^{(i-1)t+j}) & \cdots & h_d(\gamma^{(i-1)t+j})\\
\vdots & \vdots & \cdots & \vdots \\
\vdots & \vdots & \cdots & \vdots \\
h_1(\gamma^{(i-1)t+j+d-2}) & h_2(\gamma^{(i-1)t+j+d-2}) & \cdots & h_d(\gamma^{(i-1)t+j+d-2})
\end{bmatrix}
\end{align*}
	Since $A_{ij}$ is a submatrix of $A_i$, $\rank(A_{ij}) \leq \rank(A_i) = r_i$. If $r_i<d$, then $A_{ij}$ is not full rank and $\det(A_{ij}) = 0$. However, note that $A_{ij} = H(\gamma^{(i-1)t+j-1})$. In conclusion, if $d-r_i>0$, then $\gamma^{(i-1)t+j-1}$ is a root of $D(X)$.
	
	Extending this argument to multiplicities, let $D^{(\ell)}(X)$ be the $\ell^{th}$ derivative of $D(X)$ for $\ell \in \{0,1,\cdots,d\}$. Then this derivative can be written as a sum of $d^{\ell}$ determinants such that every determinant has at least $d-\ell$ columns common with $H(X)$. This follows by writing out the determinant as a signed sum of monomials, applying the product rule of differentiation, and packing them back into determinants. \snote{Move to prelims.}
	
	Therefore, $D^{(\ell)}(\gamma^{(i-1)t+j-1})$ can be written as a sum of determinants where each determinant has at least $d-\ell$ columns in common with $A_{ij}$. For $\ell = 0,1,\cdots ,d-r_i-1$, this leaves at least $r_i+1$ columns in each determinant from $A_{ij}$. Recall that $\rank(A_{ij})\leq r_i$, which implies that any set of $r_i+1$ columns in $A_{ij}$ are linearly dependent, causing each of the $d^{\ell}$ determinants in the sum for $H^{(\ell)}(\gamma^{(i-1)t+j-1})$ to vanish. We conclude that $H^{(\ell)}(\gamma^{(i-1)t+j-1}) = 0$ for $\ell = 0,1,\cdots , d-r_i-1$, and so $\gamma^{(i-1)t+j-1}$ is a root of $D(X)$ with multiplicity at least $d-r_i$.
	\end{proof}
\end{proof}

Now we use the above global upper bound on rank deficit to prove a list size bound with induction.

\begin{theorem}
Suppose $k>d$ and $t \geq k$. Let $\calH$ be an affine subspace of dimension $d$. Then, for every $g \in (\F_q^t)^N$,
\[
	\abs{ \calH \cap \calL\inparen{g, \frac{k}{k+1} \cdot \inparen{1-\frac{t}{t-k+1} \cdot R}} } \leq (k-1)\cdot d + 1.
\]
\end{theorem}
\begin{proof}
We prove this by induction on $d$. The case $d=0$ is trivial, and the case $d=1$ follows by \cref{lem:dim1} and using $\abs{\calL\inparen{g, \frac{k}{k+1} \cdot \inparen{1-\frac{t}{t-k+1} \cdot R}}} \leq \abs{\calL\inparen{g, \frac{k}{k+1} \cdot  \inparen{1- R}}}$.

Henceforth, let $d\geq 2$, and denote $\calH_g = \calH \cap \calL\inparen{g, \frac{k}{k+1} \cdot \inparen{1-\frac{t}{t-k+1} \cdot R}}$, and $S_h$ be the agreement set between $g$ and $h$ (over all of $[n]$). Using the lower bound on the size of agreement sets,
\[
	\inparen{\frac{1}{k+1}+\frac{kR}{k+1} \cdot \frac{t}{t-k+1}} N |\calH_g| \leq \sum_{h\in \calH_g} |S_h|
\]
An upper bound on $\sum_{h\in \calH_g} |S_h|$ can again be proved using double counting. Again, we will consider two cases depending on $r_i=0$ or $r_i>0$. In the latter, we can reduce dimension of the affine space $\calH$ by at least 1 when we decide to assume $h_i= g_i$. Let $B\sub [N]$ be the bad set with $r_i = 0$, and $b = |B|/N$. It is easy to see that $b < R$.
\begin{align*}
	\sum_{h\in \calH_g} |S_h| &= \sum_{i=1}^N \abs{ \{ h\in \calH_g \suchthat \forall j\in [t], \  h(\gamma^{(i-1)t+j-1}) = g(\gamma^{(i-1)t+j-1}) \}} \\
	&\leq \sum_{i\in \bar{B}} \insquare{(k-1)(d-r_i) + 1} + \sum_{i\in B}|\calH_g|\\
	&= N-|B|+(k-1)\sum_{i\not\in B} \insquare{d-r_i} + |B|\cdot |\calH_g|\\
	&\leq |B|\cdot |\calH_g| +N-|B| + (k-1)\inparen{\frac{d\cdot (m-1)}{t-d+1} - d|B|} \\
	&\leq |B|\cdot |\calH_g| + N\inparen{1-b+(k-1)d\inparen{\frac{t}{t-d+1}R - b}}
\end{align*}

Comparing the lower bound and upper bound,
\begin{align*}
	|\calH_g| &\leq \frac{1-b+(k-1)d \inparen{\frac{t}{t-d+1}R-b}}{\inparen{\frac{1}{k+1}+\frac{kR}{k+1} \cdot \frac{t}{t-k+1}-b}} \\
	&< \frac{1-b+(k-1)d \inparen{\frac{t}{t-k+1}R-b}}{\inparen{\frac{1}{k+1}+\frac{kR}{k+1} \cdot \frac{t}{t-k+1}-b}} \\
\end{align*}
We show that $|\calH_g| < 1+(k-1)d$ by showing that 
\[
\inparen{\frac{1}{k+1}+\frac{kR}{k+1} \cdot \frac{t}{t-k+1}-b} \inparen{|\calH_g| - 1 - (k-1)d}<0
\]
This suffices to conclude our induction.
\begin{align*}
&\inparen{\frac{1}{k+1}+\frac{kR}{k+1} \cdot \frac{t}{t-k+1}-b} \inparen{|\calH_g| - 1 - (k-1)d}\\
&< 1+ \frac{t}{t-k+1} (k-1)dR - \frac{1}{k+1} - \frac{kR}{k+1} \cdot \frac{t}{t-k+1} - \frac{(k-1)d}{k+1} - \frac{kR}{k+1} \cdot \frac{t}{t-k+1}\cdot (k-1)d \\
&= \frac{k}{k+1}\left(1-\frac{t}{t-k+1}R\right) + \frac{t}{t-k+1} (k-1)dR - \frac{(k-1)d}{k+1} - \frac{kR}{k+1} \cdot \frac{t}{t-k+1}\cdot (k-1)d \\
&= \frac{k}{k+1}\left(1-\frac{t}{t-k+1}R\right) - \frac{(k-1)d}{k+1} + \frac{R}{k+1} \cdot \frac{t}{t-k+1}\cdot (k-1)d \\
&=  \frac{k}{k+1}\left(1-\frac{t}{t-k+1}R\right) - \frac{(k-1)d}{k+1}\inparen{1 - \frac{t}{t-k+1}R}\\
&=  \left( \frac{k-(k-1)d}{k+1} \right) \cdot \inparen{1-\frac{t}{t-k+1}R}
\end{align*}
The last term is $\leq 0$ as long as $k\leq (k-1)d$, which is always true for $d\geq 2$.
\end{proof}

We can now use \cref{thm:lin_alg_rs} to claim that for $t$-folded RS codes, the list $\calL \inparen{g,\frac{k}{k+1} \cdot \inparen{1-\frac{t}{t-k+1}R}}$ is contained in an affine subspace of dimension $k-1$, and this leads to the following corollary.
\begin{corollary}
For $t$-folded Reed Solomon codes,
\[
	\abs{\calL \inparen{g,\frac{k}{k+1} \cdot \inparen{1-\frac{t}{t-k+1}R}}} \leq (k-1)^2 + 1
\]
\end{corollary}

\renewcommand{\bibname}{References}
\bibliographystyle{alphaurl}
\addcontentsline{toc}{chapter}{References}
\bibliography{macros, madhur}

\appendix
\begin{appendices}

\chapter{Properties of Ta-Shma's Construction}\label{appendix:ta-shma}

The goal of this section is to provide a reasonably self-contained
compilation of the properties of the slightly modified version of
Ta-Shma code construction~\cite{TS17} from~\cite{JQST20}. The
properties we need are collected
in~\cref{fact:ta-shma_splittable_tuples}.

\TaShmaConsFact

We first recall the $s$-wide replacement product in \cref{sec:s_wide_replacement_prod}, then describe Ta-Shma's original construction based on it in \cref{sec:ta_shma_construction}, describe our modification to obtain splittability in \cref{sec:tweaks}, derive the splittability property in \cref{sec:ta-shma:splittability}, and finally choose parameters in terms of desired bias $\epsilon$ of the code we construct in \cref{sec:ta-shma_param}. We refer the reader to~\cite{TS17} for formal details
beyond those we actually need here.


\section{The s-wide Replacement Product}\label{sec:s_wide_replacement_prod}

Ta-Shma's code construction is based on the so-called $s$-wide
replacement product~\cite{TS17}.  This is a derandomization of
random walks on a graph $G$ that will be defined via a product
operation of $G$ with another graph $H$
(see \autoref{def:s_wide_replacement} for a formal definition). We
will refer to $G$ as the
\emph{outer} graph and $H$ as the \emph{inner} graph in this construction.

Let $G$ be a $d_1$-regular graph on vertex set $[n]$ and $H$ be a
$d_2$-regular graph on vertex set $[d_1]^s$, where $s$ is any positive
integer.  Suppose the neighbors of each vertex of $G$ are labeled 1,
2, \dots, $d_1$.  For $v \in V(G)$, let $v_G[j]$ be the $j$-th
neighbor of $v$.  The $s$-wide replacement product is defined by
replacing each vertex of $G$ with a copy of $H$, called a ``cloud''.
While the edges within each cloud are determined by $H$, the edges
between clouds are based on the edges of $G$, which we will define via
operators $\matr G_0, \matr G_1, \dots, \matr G_{s-1}$.  The $i$-th
operator $\matr G_i$ specifies one inter-cloud edge for each vertex
$(v, (a_0, \dots, a_{s-1})) \in V(G) \times V(H)$, which goes to the
cloud whose $G$ component is $v_G[a_i]$, the neighbor of $v$ in $G$
indexed by the $i$-th coordinate of the $H$ component.  (We will
resolve the question of what happens to the $H$ component after taking
such a step momentarily.)

Walks on the $s$-wide replacement product consist of steps with two
different parts: an intra-cloud part followed by an inter-cloud part.
All of the intra-cloud substeps simply move to a random neighbor in
the current cloud, which corresponds to applying the operator $\matr
I \otimes \matr A_H$, where $\matr A_H$ is the normalized adjacency
matrix of $H$.  The inter-cloud substeps are all deterministic, with
the first moving according to $\matr G_0$, the second according to
$\matr G_1$, and so on, returning to $\matr G_0$ for step number
$s+1$.  The operator for such a walk taking $k-1$ steps on the $s$-wide
replacement product is
$$
\prod_{i=0}^{k-2} \matr G_{i \bmod s} (\matr I \otimes \matr A_H).
$$

Observe that a walk on the $s$-wide replacement product yields a walk
on the outer graph $G$ by recording the $G$ component after each step
of the walk.  The number of $(k-1)$-step walks on the $s$-wide replacement
product is
$$
|V(G)| \cdot |V(H)| \cdot d_2^{k-1} = n \cdot d_1^s \cdot d_2^{k-1},
$$
since a walk is completely determined by its intra-cloud steps.  If
$d_2$ is much smaller than $d_1$ and $k$ is large compared to $s$,
this is less than $n d_1^{k-1}$, the number of $(k-1)$-step walks on $G$
itself.  Thus the $s$-wide replacement product will be used to
simulate random walks on $G$ while requiring a reduced amount of
randomness (of course this simulation is only possible under special
conditions, namely, when we are uniformly distributed on each cloud).

To formally define the $s$-wide replacement product, we must consider
the labeling of neighbors in $G$ more carefully.

\begin{definition}[Rotation Map]
	Suppose $G$ is a $d_1$-regular graph on $[n]$.
	For each $v \in [n]$ and $j \in [d_1]$, let $v_G[j]$ be the $j$-th neighbor of $v$ in $G$.
	Based on the indexing of the neighbors of each vertex, we define the rotation map~\footnote{This kind of map is denoted rotation map in the zig-zag terminology~\cite{RVW00}.}
        $\textup{rot}_G \colon [n] \times [d_1] \to [n] \times [d_1]$ such that for every $(v,j) \in [n] \times [d_1]$,
        $$
        \textup{rot}_G((v,j)) = (v',j') \iff v_G[j] = v' \text{ and } v'_G[j']=v.
        $$
	Furthermore, if there exists a bijection $\phi \colon [d_1] \to [d_1]$ such that for every $(v,j) \in [n] \times [d_1]$,
        $$
        \textup{rot}_G((v,j)) = (v_G[j],\phi(j)),
        $$
	then we call $\textup{rot}_G$ \emph{locally invertible}.
\end{definition}

If $G$ has a locally invertible rotation map, the cloud label after
applying the rotation map only depends on the current cloud label, not
the vertex of $G$.  In the $s$-wide replacement product, this
corresponds to the $H$ component of the rotation map only depending on
a vertex's $H$ component, not its $G$ component.  We define the
$s$-wide replacement product as described before, with the inter-cloud
operator $\matr G_i$ using the $i$-th coordinate of the $H$ component,
which is a value in $[d_1]$, to determine the inter-cloud step.

\begin{definition}[$s$-wide replacement product]\label{def:s_wide_replacement}
  Suppose we are given the following:
  \begin{itemize}
    \item A $d_1$-regular graph $G=([n'],E)$ together with a locally invertible rotation map
          $\textup{rot}_G \colon [n'] \times [d_1] \to [n'] \times [d_1]$.
    \item A $d_2$-regular graph $H = ([d_1]^s,E')$.
  \end{itemize}
  And we define:
  \begin{itemize}
    \item For $i \in \set{0,1,\dots,s-1}$, we define $\textup{Rot}_i \colon [n'] \times [d_1]^s \to [n'] \times [d_1]^s$ as,
          for every $v \in [n']$ and $(a_0,\dots,a_{s-1}) \in [d_1]^s$,
          $$
          \textup{Rot}_i((v, (a_0,\dots,a_{s-1}))) \coloneqq (v', (a_0,\dots,a_{i-1},a_i',a_{i+1},\dots,a_{s-1})),
          $$
          where $(v',a_i') = \textup{rot}_G(v,a_i)$.
    \item Denote by $\matr G_i$ the operator realizing $\textup{Rot}_i$ and
          let $\matr A_H$ be the normalized random walk operator of $H$. Note that $\matr G_i$ is
          a permutation operator corresponding to a product of transpositions.
  \end{itemize}

  Then $k-1$ steps of the $s$-wide replacement product are given by the operator
  $$
  \prod_{i=0}^{k-2} \matr G_{i \bmod s} (\matr I \otimes \matr A_H).
  $$
\end{definition}

Ta-Shma instantiates the $s$-wide replacement product with an outer
graph $G$ that is a Cayley graph, for which locally invertible
rotation maps exist generically.

\begin{remark}
   Let $R$ be a group and $A \subseteq R$ where the set $A$ is closed under inversion. For every Cayley graph
   $\textup{Cay}(R,A)$, the map $\phi \colon A \to A$ defined as
   $\phi(g) = g^{-1}$ gives rise to the locally invertible rotation
   map
   $$
   \textup{rot}_{\textup{Cay}(R,A)}((r,a)) = (r\cdot a ,a^{-1}),
   $$
   for every $r \in R$, $a \in A$.
\end{remark}

\begin{figure}[ht!]
	\centering
	\begin{tikzpicture}[scale=0.8,every node/.style={scale=0.8}]
		\foreach \i in {0,...,4}
		{
			\node (v1\i) [circle,draw] at (102+72*\i:4.2) {1};
			\node (v2\i) [circle,draw] at (78+72*\i:4.2) {2};
			\node (v3\i) [circle,draw] at (78+72*\i:3) {3};
			\node (v4\i) [circle,draw] at (102+72*\i:3) {4};
		}
		\foreach \i in {0,...,4}
		{
			\draw (v1\i) -- (v2\i) -- (v3\i) -- (v4\i) -- (v1\i);
			
			\pgfmathparse{int(mod(\i+1,5))};
			\def\j{\pgfmathresult};
			\draw (v1\i) -- (v2\j);
			
			\pgfmathparse{int(mod(\i+2,5))};
			\def\j{\pgfmathresult};
			\draw (v4\i) -- (v3\j);
			
			\node[cloud,draw,minimum width=4cm,minimum height=3cm,rotate=72*\i,cloud puffs=12] at (90+72*\i:3.6) {};
		}
	\end{tikzpicture}
	\caption{An example of the 1-wide replacement product with outer graph $G = K_5$ and inner graph $H = C_4$.
		Vertices are labeled by their $H$ components.
		Note that the rotation map is locally invertible, with $\phi(1) = 2$, $\phi(2) = 1$, $\phi(3) = 4$, and $\phi(4) = 3$.}
\end{figure}
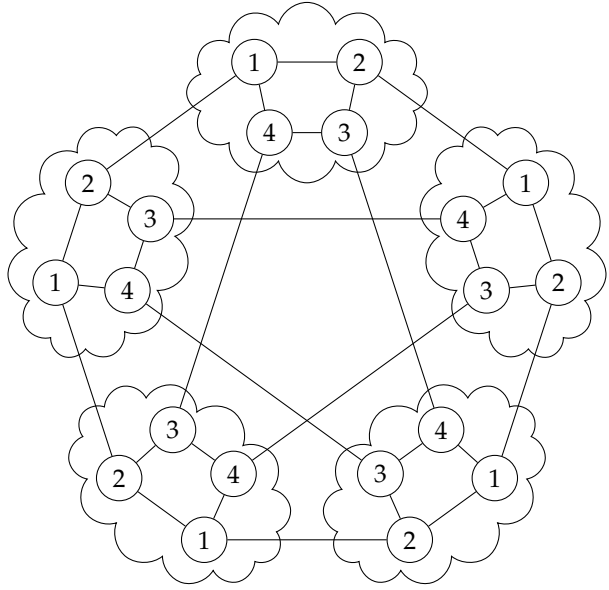

\section{The Construction}\label{sec:ta_shma_construction}

Let $n'=|V(G)|, m=d_1^s= |V(H)|$ and $n = n'\cdot m = |V(G)\times V(H)|$. Ta-Shma's code construction works by starting with a constant bias
code $\Cc_0'$ in $\mathbb{F}_2^{n'}$, repeating each codeword $m=d_1^s$ times to get a new $\epsilon_0$-biased code $\Cc_0$ in $\mathbb{F}_2^{n}$, and boosting $\Cc_0$ to arbitrarily
small bias using direct sum liftings. Recall that the direct sum
lifting is based on a collection $\tuples{k} \subseteq [n]^{k}$, which
Ta-Shma obtains using $k-1$ steps of random walk on the $s$-wide
replacement product of two regular expander graphs $G$ and $H$. The
graph $G$ is on $n'$ vertices 
and other parameters like degrees $d_1$ and $d_2$ of $G$ and $H$
respectively are chosen based on target code parameters.

To elaborate, every $k-1$ length walk on the replacement product gives
a sequence of $k$ vertices in the replacement product graph, which can be seen as
an element of $[n]^{k}$. This gives the collection $\tuples{k}$ with $|\tuples{k}|
= n' \cdot d_1^s \cdot d_2^{k-1}$ which means the rate of lifted code
is smaller than the rate of $\Cc_0'$ by a factor of $d_1^s
d_2^{k-1}$. However, the collection $\tuples{k}$ is a parity sampler and
this means that the bias decreases (or the distance increases) from that of $\Cc_0$. The
relationship between this decrease in bias and decrease in rate with
some careful parameter choices allows Ta-Shma to obtain nearly optimal
$\epsilon$-balanced codes.

\section{Tweaking the Construction}\label{sec:tweaks}


Recall the first $s$ steps in Ta-Shma's construction are given by the
operator
$$
\matr G_{s-1}(\matr I \otimes \Aye_H) \matr G_{s-2} \cdots G_{1}(\matr I \otimes \Aye_H) \matr G_{0} (\matr I \otimes \Aye_H).
$$
Naively decomposing the above operator into the product of operators
$\prod_{i=0}^{s-1} \matr G_{i} (\matr I \otimes \Aye_H)$ is not good
enough to obtain the \emph{splittability} property which would hold
provided $\sigma_2(\matr G_{i}(\matr I \otimes \Aye_H))$ was small
for every $i$ in $\set{0,\ldots,s-1}$. However, each $\matr G_i (\matr
I \otimes \Aye_H)$ has $\abs{V(G)}$ singular values equal to $1$
since $G_{i}$ is an orthogonal operator and $(\matr I \otimes \Aye_H)$
has $\abs{V(G)}$ singular values equal to $1$. To avoid this issue we
will tweak the construction to be the following product
$$
\prod_{i=0}^{s-1} (\matr I \otimes \Aye_H) \matr G_{i} (\matr I \otimes \Aye_H).
$$

The operator $(\matr I \otimes \Aye_H) \matr G_{i} (\matr
I \otimes \Aye_H)$ is exactly the walk operator of the zig-zag product
$G \zigzag H$ of $G$ and $H$ with a rotation map given by the
(rotation map) operator $\matr G_i$. This tweaked construction is slightly
simpler in the sense that $G \zigzag H$ is an undirected graph. We
know by the zig-zag analysis that $(\matr I \otimes \Aye_H) \matr
G_{i} (\matr I \otimes \Aye_H)$ is expanding as long $G$ and $H$ are
themselves expanders.  More precisely, we have a bound that follows
from~\cite{RVW00}.

\begin{fact}\label{fact:zig_zag_bound}
  Let $G$ be an outer graph and $H$ be an inner graph used in the $s$-wide replacement product.
  For any integer $0 \leq i \leq s-1$,
  $$
  \sigma_2((I \otimes \Aye_H) G_i (I \otimes \Aye_H)) \leq \sigma_2(G) + 2 \cdot \sigma_2(H) + \sigma_2(H)^2.
  $$
\end{fact}

This bound will imply \emph{splittability} as shown
in~\cref{sec:ta-shma:splittability}. We will need to argue that this
modification still preserves the correctness of the parity sampling
and that it can be achieved with similar parameter trade-offs.

The formal definition of a length-$t$ walk on this slightly modified
construction is given below.
\begin{definition}
  Let $k \in \mathbb{N}$, $G$ be a $d_1$-regular graph and $H$
  be a $d_2$-regular graph on $d_1^s$ vertices. Given a starting vertex $(v,u) \in V(G) \times
  V(H)$, a $(k-1)$-step walk on the tweaked $s$-wide replacement product of $G$ and $H$ is
  a tuple $((v_1,u_1),\dots, (v_{k},u_{k})) \in (V(G) \times V(H))^{k}$ such that
  \begin{itemize}
    \item $(v_1,u_1) = (v,u)$, and
    \item for every $1 \le i < k$, we have $(v_i,u_i)$ adjacent to $(v_{i+1},u_{i+1})$
          in $(\matr I \otimes \Aye_H) \matr G_{(i-1) \bmod s} (\matr I \otimes \Aye_H)$.
  \end{itemize}
  Note that each $(\matr I \otimes \Aye_H) \matr G_{(i-1) \bmod s} (\matr I \otimes \Aye_H)$
  is a walk operator of a $d_2^2$-regular graph. Therefore, the starting vertex $(v,u)$ together
  with a degree sequence $(m_1,\dots,m_k) \in [d_2^2]^{k-1}$ uniquely defines a $(k-1)$-step walk.
\end{definition}

\subsubsection{Parity Sampling}\label{sec:tweaked_parity_sampling}

We argue informally why parity sampling still holds with similar
parameter trade-offs. In particular,
we formalize a key result underlying parity sampling and,
in~\cref{sec:ta-shma_param}, we compute the new trade-off
between bias and rate in some regimes. In~\cref{sec:s_wide_replacement_prod},
the definition of the original $s$-wide replacement product as a
purely graph theoretic operation was given. Now, we explain how
Ta-Shma used this construction for parity sampling obtaining codes
near the GV bound.

For a word $z \in \mathbb{F}_2^{V(G)}$ in the base code, let $\matr
P_z$ be the diagonal matrix, whose rows and columns are indexed by
$V(G) \times V(H)$, with $(\matr P_z)_{(v,u), (v,u)} =
(-1)^{z_v}$. Proving parity sampling requires analyzing the operator
norm of the following product
\begin{equation}\label{eq:tweaked_tashma}
  \matr P_z \prod_{i=0}^{s-1} \matr (\matr I \otimes \Aye_H) \matr G_{i} \matr P_z (\matr I \otimes \Aye_H),
\end{equation}
when $\bias(z) \le \epsilon_0$. Let $\one \in \mathbb{R}^{V(G)\times
V(H)}$ be the all-ones vector, scaled to be of unit length under the $\ell_2$ norm, and
$W$ be the collection of all $(t-1)$-step walks on the tweaked $s$-wide
replacement product. Ta-Shma showed (and it is not difficult to
verify) that
$$
\bias\left(\dsum_{W}(z) \right) = \left\lvert \ip{\one}{\matr P_z \prod_{i=0}^{k-2} \matr (\matr I \otimes \Aye_H) \matr G_{i \bmod s} \matr P_z (\matr I \otimes \Aye_H) \one} \right\rvert.
$$
The measure used in this inner product is the usual counting measure over $\mathbb{R}^{V(G)\times V(H)}$. From the previous equation, one readily deduces that
$$
\bias\left(\dsum_{W}(z) \right) \le \sigma_1\left(\matr P_z \prod_{i=0}^{s-1} \matr (\matr I \otimes \Aye_H) \matr G_{i} \matr P_z (\matr I \otimes \Aye_H)\right)^{\lfloor (k-1)/s \rfloor}.
$$

The
key technical result obtained by Ta-Shma is the following, which is
used to analyze the bias reduction as a function of the total number
walk steps $k-1$. Here $\theta$ is a parameter used in obtaining explicit Ramanujan graphs.
 
\begin{fact}[Theorem 24 abridged~\cite{TS17}]\label{fact:ta-shma_main}
 If $H$ is a Cayley graph on $\mathbb{F}_2^{s\log d_1}$ and $\epsilon_0 + 2 \cdot \theta + 2 \cdot \sigma_2(G) \le \sigma_2(H)^2$, then
 $$
 \norm{\prod_{i=0}^{s-1} \matr P_z \matr G_i (\matr I \otimes \matr A_H)}_{\textup{op}} \le \sigma_2(H)^s + s \cdot \sigma_2(H)^{s-1} + s^2 \cdot \sigma_2(H)^{s-3},
 $$
 where $\matr P_z \in \mathbb{R}^{(V(G)\times V(H)) \times (V(G)\times V(H))}$ is the \emph{sign operator} of a $\epsilon_0$ biased word $z \in \mathbb{F}_2^{V(G)}$
 defined as a diagonal matrix with $(P_z)_{(v,u),(v,u)} = (-1)^{z_v}$ for every $(v,u) \in V(G) \times V(H)$.
\end{fact}

We reduce the analysis of Ta-Shma's tweaked construction
to an analog of ~\cref{fact:ta-shma_main}. In doing so, we only lose one extra step
as shown below.

\begin{corollary}\label{cor:tweaked_ta-shma_spectral_analysis}
 If $H^2$ is a Cayley graph on $\mathbb{F}_2^{s\log d_1}$ and $\epsilon_0 + 2 \cdot \theta + 2 \cdot \sigma_2(G) \le \sigma_2(H)^4$, then
 $$
 \norm{\prod_{i=0}^{s-1} (\matr I \otimes \matr A_H) \matr P_z \matr G_i (\matr I \otimes \matr A_H)}_{\textup{op}} \le \sigma_2(H^2)^{s-1} + (s-1) \cdot \sigma_2(H^2)^{s-2} + (s-1)^2 \cdot \sigma_2(H^2)^{s-4},
 $$
 where $\matr P_z$ is the \emph{sign operator} of an $\epsilon_0$-biased word $z \in \mathbb{F}_2^{V(G)}$
 as in~\cref{fact:ta-shma_main}.
\end{corollary}

\begin{proof}
  We have
  \begin{align*}
    \norm{\prod_{i=0}^{s-1} (\matr I \otimes \matr A_H) \matr P_z \matr G_i (\matr I \otimes \matr A_H)}_{\text{op}} &\le \norm{(\matr I \otimes \matr A_H)}_{\text{op}} \norm{\prod_{i=1}^{s-1}  \matr P_z \matr G_i (\matr I \otimes \matr A_H^2)}_{\text{op}} \norm{\matr P_z \matr G_{0} (\matr I \otimes \matr A_H)}_{\text{op}}\\
                                                  &\le \norm{\prod_{i=1}^{s-1}  \matr P_z \matr G_i (\matr I \otimes \matr A_H^2)}_{\text{op}}\\
                                                  &\le \sigma_2(H^2)^{s-1} + (s-1) \cdot \sigma_2(H^2)^{s-2} + (s-1)^2 \cdot \sigma_2(H^2)^{s-4},
  \end{align*}
  where the last inequality follows from~\cref{fact:ta-shma_main}.
\end{proof}

\begin{remark}
  We know that in the modified construction $H^2$ is a Cayley graph since $H$ is a Cayley graph.
\end{remark}


\section{Splittability}\label{sec:ta-shma:splittability}

In this subsection, we focus on the splittability parameters arising out of the construction described above. The collection $\tuples{k}\subseteq [n]^k$ is obtained from taking $k-1$ step walks on $s$-wide replacement as described above, which is $d_2^2$-regular. Recall from \cref{def:split} that we need to show $\sigma_2(\swap{W[a,t]}{W[t+1,b]}) \le \tau$ for all $1\leq a<t<b\leq k$, where,

\begin{align*}
\left( \swap{W[a,t]}{W[t+1,b]} \right)_{(i_a,\cdots,i_t),(i_{t+1},\cdots ,i_b)}  := \frac{\one[(i_a,\cdots, i_t,i_{t+1},\cdots ,i_b)\in W[a,b]]}{d_2^{2(b-s)}}
\end{align*}

\begin{lemma}\label{lemma:swap_matrix_rep}
  Let $1 \leq a < t < b  \leq k$.
  Suppose $G$ is a $d_1$-regular outer graph on vertex set $[n]$ with walk operator $G_{t}$ used at step $s$ of a walk on the $s$-wide replacement product and $H$ is a $d_2$-regular inner graph on vertex set $[m]$ with normalized random walk operator $\Aye_H$.
  Then there are orderings of the rows and columns of the representations of $\swap{W[a,t]}{W[t+1,b]}$ and $\Aye_H$ as matrices such that
  $$
  \swap{W[a,t]}{W[t+1,b]} = \left( (I \otimes \Aye_H) G_{t} (I \otimes \Aye_H) \right) \otimes \Jay/d_2^{2(b-t-1)},
  $$
  where $\Jay \in \mathbb{R}^{[d_2]^{2(t-a)} \times [d_2]^{2(b-t-1)}}$ is the all ones matrix.
\end{lemma}

\begin{proof}
  Partition the set of walks $W[a,t]$ into the sets
  $W_{1,1}, \dots, W_{n',m}$, where $w \in W_{i,j}$ if the last vertex
  of the walk $i_{t} = (v_{t},u_{t})$ satisfies $v_{t} = i$
  and $u_{t} = j$. Similarly, partition $W[t+1,b]$ into the sets
  $W_{1,1}', \dots, W_{n',m}'$, where $(i_{t+1},\cdots,i_b) \in W_{i,j}'$ if the first
  vertex of the walk $i_{t+1} = (v_{t+1},u_{t+1})$ satisfies $v_{t+1} = i$ and $u_{t+1} =
  j$.  Note that $\abs{W_{i,j}} = d_2^{2(t-a)}$ and $\abs{W_{i,j}'} =
  d_2^{2(b-t-1)}$ for all $(i,j) \in [n']\times[m]$, since there are $d_2^2$ choices for each step of the walk.

  Now order the rows of the matrix $\swap{W[a,t]}{W[t+1,b]}$ so that all
  of the rows corresponding to walks in $W_{1,1}$ appear first,
  followed by those for walks in $W_{1,2}$, and so on in
  lexicographic order of the indices $(i,j)$ of $W_{i,j}$, with an
  arbitrary order within each set.  Do a similar re-ordering of the
  columns for the sets $W_{1,1}', \dots, W_{n',m}'$.  Observe that
  \begin{align*}
  \left(\swap{W[a,t]}{W[t+1,b]}\right)&_{(i_a,\cdots,i_t),(i_{t+1},\cdots ,i_b)} = \frac{\one_{{(i_a,\cdots,i_t,i_{t+1},\cdots ,i_b)} \in W[a,b]}}{d_2^{2(b-t)}} \\
                                          & = \frac{d_2^2 \cdot (\text{weight of  transition from } i_t \text{ to } i_{t+1}  \text{ in } (I \otimes \Aye_H)G_{t} (I \otimes \Aye_H))}{d_2^{2(b-t)}},
  \end{align*}
  which only depends on the adjacency of the last vertex of $(i_a,\cdots,i_t)$ and
  the first vertex of $(i_{t+1},\cdots,i_b)$.  If the vertices $i_{t} = (v_{t},u_{t})$ and $i_{t+1}=(v_{t+1},u_{t+1})$
  are adjacent, then
  $$
  \left(\swap{W[a,t]}{W[t+1,b]}\right)_{(i_a,\cdots,i_t),(i_{t+1},\cdots ,i_b)} = \left((I \otimes \Aye_H)G_{t}(I \otimes \Aye_H)\right)_{(v_t,u_t),(v_{t+1}, u_{t+1})}/d_2^{2(b-t-1)},
  $$
  for every $(i_a,\cdots,i_t) \in W[a,t]$ and $(i_{t+1},\cdots ,i_b) \in W[t+1,b]$; and otherwise \\
  $\left(\swap{W[a,t]}{W[t+1,b]}\right)_{(i_a,\cdots,i_t),(i_{t+1},\cdots ,i_b)} = 0$.  Since the walks in
  the rows and columns are sorted according to their last and first
  vertices, respectively, the matrix $\swap{W[a,t]}{W[t+1,b]}$ exactly
  matches the tensor product
  $((I \otimes \Aye_H)G_{t} (I \otimes \Aye_H)) \otimes \Jay/d_2^{2(b-t-1)}$.
\end{proof}

\begin{corollary}\label{cor:split_walk_spectral_gap}
  Let $1 \leq a < t < b \leq k$.
  Suppose $G$ is a $d_1$-regular outer graph with walk operator $G_{t}$ used at step $t$ of a walk on the $s$-wide replacement product and $H$ is a $d_2$-regular inner graph with normalized random walk operator $\Aye_H$.
  Then
  $$
  \sigma_2(\swap{W[a,t]}{W[t+1,b]}) = \sigma_2((I \otimes \Aye_H) G_{t} (I \otimes \Aye_H)).
  $$
\end{corollary}

\begin{proof}
  Using~\cref{lemma:swap_matrix_rep} and the fact that
  $$
  \sigma_2(((I \otimes \Aye_H)G_{t} (I \otimes \Aye_H)) \otimes \Jay/d_2^{2(b-t-1)})  = \sigma_2((I \otimes \Aye_H)G_{t} (I \otimes \Aye_H)),
  $$
  the result follows.
\end{proof}

\begin{remark}
\Cref{cor:split_walk_spectral_gap} is what causes the splittability argument to break down for Ta-Shma's original construction, as $\sigma_2(\matr G_{t} (\matr I \otimes \matr A_H)) = 1$.
\end{remark}

\section{Parameter Choices}\label{sec:ta-shma_param}

In this section, we choose parameters to finally obtain \cref{fact:ta-shma_splittable_tuples}, for which we must argue about bias, rate and splittability.

A graph is said to be an
$(n,d,\lambda)$-graph provided it has $n$ vertices, is $d$-regular, and
has second largest singular value of its normalized adjacency matrix
at most $\lambda$.
\begin{notation}
  We use the following notation for the graphs $G$ and $H$ used in
  the $s$-wide replacement product.
  \begin{itemize}
    \item The outer graph $G$ will be an $(n'',d_1,\lambda_1)$-graph.
    \item The inner graph $H$ will be a $(d_1^s,d_2,\lambda_2)$-graph.
  \end{itemize}
  The parameters $n'', d_1,d_2, \lambda_1,\lambda_2$ and $s$ are yet to be chosen.
\end{notation}

We are given the dimension $D$ of the desired code and its bias $\epsilon \in (0,1/2)$.
We set a parameter $\alpha \le 1/128$ such that (for convenience)
$1/\alpha$ is a power of $2$ and
\begin{equation}\label{eq:alpha_epsilon}
\frac{\alpha^5}{4 \log_2(1/\alpha)} \ge \frac{1}{\log_2(1/\epsilon)}.
\end{equation}

By replacing $\log_2(1/\alpha)$ with its upper bound $1 / \alpha$, we observe that $\alpha = \Theta(1/\log_2(1/\epsilon)^{1/6})$ satisfies this bound, and so we choose $s = \Theta(\log_2(1/\epsilon)^{1/6})$.

\noindent \textbf{The inner graph $H$.} \enspace The choice of $H$ is
same as Ta-Shma's choice. More precisely, we set $s=1/\alpha$ and
$d_2 = s^{4s}$.
We obtain a Cayley graph $H
= \textup{Cay}(\mathbb{F}_2^{4s\log_2(d_2)}, A)$ such that $H$ is an
$(n_2=d_2^{4s}, d_2, \lambda_2)$ graph where
$\lambda_2 = b_2/\sqrt{d_2}$ and $b_2 = 4 s \log_2(d_2)$.
(The set of generators, $A$, comes from a small bias code derived from
a construction of Alon et al.~\cite{AGHP92}.)



\noindent \textbf{The base code $\Cc_0$.} This is dealt with in detail in \cref{sec:abstract_dec}. We choose $\epsilon_0=1/d_2^2$ and use \cref{cor:base_code} to obtain a code $\Cc_0'$ in $\mathbb{F}_2^{n'}$ that is $\epsilon_0$-biased and has a blocklength $\Omega(D/ \epsilon_0^c)$ for some constant $c$. Call this blocklength of $\Cc_0'$ to be $n'$.
Next we replicate the codewords $m=d_1^s$ times to get code $\Cc_0$ in $\mathbb{F}_2^n$ with the same bias but a rate that is worse by a factor of $m$. In the proofs below, we only use properties of $\Cc_0$ that is of multiplicity $m$, has rate $\Omega(\epsilon_0^c)/m$ and has bias $\epsilon_0$, as specified in \cref{fact:ta-shma_splittable_tuples}.

\noindent \textbf{The outer graph $G$.} \enspace
Set $d_1 = d_2^4$ so that $n_2 = d_1^s$ as required by the
$s$-wide replacement product. We apply Ta-Shma's explicit Ramanujan
graph lemma (Lemma 2.10 in \cite{TS17}) with parameters $n'$, $d_1$ and
$\theta$ to obtain an $(n'',d_1,\lambda_1)$ Ramanujan graph $G$ with
$\lambda_1 \le 2\sqrt{2}/\sqrt{d_1}$ and $n'' \in [(1-\theta)n',n']$ or $n''\in [(1-\theta)2n',2n']$. Here, $\theta$ is an error parameter that we set as $\theta = \lambda_2^4/6$ (this choice of $\theta$ differs from
Ta-Shma). Because we can construct words with block length $2n'$
(if needed) by duplicating each codeword, we may assume w.l.o.g. that $n''$ is close to $n'$ and
$(n'-n'') \le \theta n' \le 2 \theta n''$. See \cite{TS17} for a more formal description of this graph.

Note that
$\lambda_1 \le \lambda_2^4/6$ since $\lambda_1 \le 3/\sqrt{d_1} =
3/d_2^2 = 3 \cdot \lambda_2^4/b_2^4 \le \lambda_2^4/6$. Hence,
$\epsilon_0 + 2\theta + 2\lambda_1 \le \lambda_2^4$, as needed to apply \cref{cor:tweaked_ta-shma_spectral_analysis}.

\noindent \textbf{The walk length.} \enspace
Set the walk length $k-1$ to be the smallest integer such that
$$
(\lambda_2^2)^{(1-5\alpha)(1-\alpha)(k-1)} \le \epsilon.
$$
This will imply using Ta-Shma's analysis that the bias of the final code is at most $\epsilon$ as shown later.

\begin{center}
\fbox{\begin{minipage}{33em}

$s=1/\alpha,\quad s = \Theta(\log(1/\epsilon)^{1/6}), \text{ so that } \frac{\alpha^3}{4 \log_2(1/\alpha)} \ge \frac{1}{\log_2(1/\epsilon)}$
\vskip 0.3cm
$H: (n_2,d_2,\lambda_2),\quad n_2=d_1^s,\quad d_2=s^{4s},\quad \lambda_2=\frac{b_2}{\sqrt{d_2}},\quad b_2 = 4s\log d_2$
\vskip 0.3cm
$\Cc_0': $ bias $\epsilon_0=1/d_2^2$,\quad blocklength $n' = O(D/\epsilon_0^c)$
\vskip 0.3cm
$\Cc_0: $ bias $\epsilon_0 =1/d_2^2$,\quad multiplicity $m=d_1^s$,\quad blocklength $n = O(mD/\epsilon_0^c)$
\vskip 0.3cm
$G: (n'',d_1,\lambda_1),\quad n'' \approx n' = O(D/\epsilon_0^c),\quad d_1=d_2^4,\quad \lambda_1\leq \frac{2\sqrt{2}}{d_1}$
\vskip 0.3cm
$k: \text{ smallest integer such that } (\lambda_2^2)^{(1-5\alpha)(1-\alpha)(k-1)} \leq \epsilon$

\end{minipage}}
\end{center}

\begin{proof}[Proof of \cref{fact:ta-shma_splittable_tuples}]

We will prove it in the following claims. We denote by $\tuples{k} \subseteq [n]^k$ the collection of walks on the $s$-wide replacement product obtained above, and we denote by $\Cc$ the final code obtained by doing the direct sum operation on $\Cc_0$ using the collection of tuples $\tuples{k}$. The explicitness of $\tuples{k}$ follows from Ta-Shma's construction since all the objects used in the construction have explicit constructions. 

Next, the multiplicity $m=d_1^s = d_2^{4s} = s^{16s^2} = 2^{16s^2\log s} \leq (2^{s^6})^{o(1)} = (1/\epsilon)^{o(1)}$.

\begin{claim}\label{claim:k_lower_bound}
  We have $k-1 \ge s/\alpha = s^2$, and that $k-1\le 2s^5$, so that 
  \begin{gather*}
  \Theta(\log(1/\epsilon)^{1/3}) \leq k \leq \Theta(\log(1/\epsilon))
  \end{gather*}
\end{claim}

\begin{proof}
  Using $d_2 = s^{4s}$ and~\cref{eq:alpha_epsilon}, we have
  \begin{align*}
  \left(\frac{1}{\lambda_2^2}\right)^{(1-5\alpha)(1-\alpha)s/\alpha} & \le \left(\frac{1}{\lambda_2^2}\right)^{s/\alpha} = \left(\frac{d_2}{b_2^2}\right)^{s/\alpha} \le \left(d_2\right)^{s/\alpha} = s^{4s^2/\alpha}\\
                                                              & = 2^{4s^2 \log_2(s)/\alpha} = 2^{4 \log_2(1/\alpha)/\alpha^3} \le  2^{\log_2(1/\epsilon)} = \frac{1}{\epsilon}.
  \end{align*}
  Hence, $\epsilon \ge (\lambda_2^2)^{(1-5\alpha)(1-\alpha)s/\alpha}$ and thus
  $k-1$ must be at least $s/\alpha$.
  
    In the other direction, we show that $(\lambda_2^2)^{(1-5\alpha)(1-\alpha)2s^5} \le \epsilon$, which will imply $k\leq \Theta(s^5)\implies k\leq \Theta(s^6)=\Theta(\log(1/\epsilon))$.
    
    \begin{align*}
    (\lambda_2^2)^{(1-5\alpha)(1-\alpha)2s^5} \leq \left( \frac{b_2^2}{d_2}\right)^{s^5} \leq \left( \frac{1}{s^{3s}}\right)^{s^5} = 2^{-\Theta(s^6\log s)} \leq 2^{-\Theta(s^6)} = 2^{-\log(1/\epsilon)} \leq \epsilon
    \end{align*}
\end{proof}

\begin{remark}\label{remark:choice_of_k}
 By the minimality of $k$, we have $(\lambda_2^2)^{(1-5\alpha)(1-\alpha)(k-2)} \ge \epsilon$.
 Since $1/(k-1) \le \alpha$, we get $(\lambda_2^2)^{(1-5\alpha)(1-\alpha)^2 (k-1)} \ge \epsilon$. This will be useful in rate computation.
\end{remark}

\begin{claim}
  The code $\Cc$ is $\epsilon$-balanced.
\end{claim}

\begin{proof}
  Using~\cref{cor:tweaked_ta-shma_spectral_analysis}, we have that the final bias  
  \begin{align*}
    b \coloneqq \left(\sigma_2(H^2)^{s-1} + (s-1) \cdot \sigma_2(H^2)^{s-2} + (s-1)^2 \cdot \sigma_2(H^2)^{s-4}\right)^{\lfloor (k-1)/s \rfloor}
  \end{align*}
  is bounded by
  \begin{align*}
    b &\le (3(s-1)^2\sigma_2(H^2)^{s-4})^{((k-1)/s)-1} && (\text{Using } \sigma_2(H^2)\leq 1/3s^2) \\
      &\le ((\sigma_2(H^2)^{s-5})^{(k-1-s)/s} \\
      &= \sigma_2(H^2)^{(1-5/s)(1-s/(k-1))(k-1)}\\
      &\le \sigma_2(H^2)^{(1-5\alpha)(1-\alpha)(k-1)}\\
      &= \left(\lambda_2^2\right)^{(1-5\alpha)(1-\alpha)(k-1)} \le \epsilon,                                                                                                                          
  \end{align*}
  where the last inequality follows from $s = 1/\alpha$ and $k-1 \ge s/\alpha$, the latter from~\cref{claim:k_lower_bound}.
\end{proof}

\begin{claim}\label{claim:rate_round_i}
  $\Cc$ has rate $\Omega(\epsilon^{2+ 28 \cdot \alpha})$.
\end{claim}

\begin{proof}
  The support size is the number of walks of length $k$ on the
  $s$-wide replacement product of $G$ and $H$ (each step of the walk has $d_2^2$ options),
  which is
  \begin{align*}
    |V(G)| |V(H)| d_2^{2(k-1)} = n'' \cdot d_1^s \cdot d_2^{2(k-1)} &= n'' \cdot  d_2^{2(k-1)+4s} \le n' \cdot  d_2^{2(k-1)+4s}\\
                                                      &= \Theta\left(\frac{D}{\epsilon_0^c} \cdot d_2^{2(k-1)+4s}\right)\\
                                                      &= \Theta\left(D \cdot (d_2^2)^{k-1+2s+c}\right)\\
                                                      &= O\left(D \cdot (d_2^2)^{(1+3\alpha)(k-1)}\right),                                                      
  \end{align*}
  where the penultimate equality follows from the assumption that $\epsilon_0$ is a constant.

  Note that $d_2^{\alpha} = d_2^{1/s} = s^{4} \ge b_2$ since
  $b_2 = 4 s \log_2(d_2) = 16 s^2 \log_2(s) \le s^4$
  . Thus,
  $$
  d_2^{1-2\alpha} = \frac{d_2}{d_2^{2\alpha}} \le \frac{d_2}{b_2^2} = \frac{1}{\sigma_2(H^2)}.
  $$
  We obtain
  \begin{align*}
  (d_2^2)^{(k-1)} &\le \left(\frac{1}{\sigma_2(H^2)}\right)^{\frac{2(k-1)}{1-2\alpha}} \\
  &\le \left(\frac{1}{\epsilon}\right)^{\frac{2}{(1-2\alpha)(1-5\alpha)(1-\alpha)^2}} && \text{(Using~\cref{remark:choice_of_k})}\\
  &\le \left(\frac{1}{\epsilon}\right)^{2(1+10\alpha)},
  \end{align*}
  which implies a block length of
  $$
  O\left(D \cdot (d_2^2)^{(1+3\alpha)(k-1)}\right) = O\left(D \left(\frac{1}{\epsilon}\right)^{2(1+10\alpha)(1+3\alpha)}\right) = O\left(D \left(\frac{1}{\epsilon}\right)^{2(1+14\alpha)}\right).
  $$
\end{proof}


\begin{claim}
$\tuples{k}$ is $\tau$-splittable for $\tau \leq 2^{-\Theta(\log(1/\epsilon)^{1/6})}$.
\end{claim}

\begin{proof}
As we saw in Corollary \cref{cor:split_walk_spectral_gap}, the splittability $\tau$ can be upper bounded by $\sigma_2((I\otimes \Aye_H)G_t(I\otimes \Aye_H))$, which is at most $\sigma_2(G)+2\cdot \sigma_2(H)+\sigma_2(H)^2$ by Fact \ref{fact:zig_zag_bound}. So, the collection $\tuples{k}$ is $\tau$-splittable for

\begin{align*}
\tau \leq \sigma_2(G)+2\cdot \sigma_2(H)+\sigma_2(H)^2 \leq 4\lambda_2 &= 4b_2/d_2^{1/2} \\
&= 64s^2\log s/s^{2s} \\
&= 2^{-\Theta(s\log s)} \\
&\leq 2^{-\Theta(s)}\\
&= 2^{-\Theta(\log(1/\epsilon)^{1/6})}
\end{align*}
\end{proof}
\end{proof}

\end{appendices}

\end{document}